\PassOptionsToPackage{unicode}{hyperref}
\PassOptionsToPackage{hyphens}{url}
\PassOptionsToPackage{dvipsnames,svgnames,x11names}{xcolor}
\documentclass[12pt]{article}

\usepackage{amssymb, amsthm, amsfonts}
\usepackage{booktabs}       
\usepackage{nicefrac}       
\usepackage{microtype}      
\usepackage{mathtools}
\usepackage{enumerate, enumitem}
\usepackage{color, soul, threeparttable, multirow, multicol, makecell}
\usepackage{bbm, caption, tikz, colortbl, dcolumn, float, prodint, setspace}

\newcommand{\mc}{\mathcal}
\newcommand{\mb}{\mathbf}
\newcommand{\bd}{\boldsymbol}

\newcommand{\Ex}{\mathbb{E}}
\newcommand{\Vx}{\mathbb{V}}
\newcommand{\Px}{\mathbb{P}}
\newcommand{\Qx}{\mathbb{Q}}
\newcommand{\Gx}{\mathbb{G}}
\newcommand{\Ix}{\mathbb{I}}
\newcommand{\glbEst}{\widehat\theta^{\text{CCOD}}_n(t,a)}
\newcommand{\tgtEst}{\widehat\theta^0_n(t,a)}
\newcommand{\fedEst}{\widehat\theta^{\text{fed}}_n(t,a)}
\newcommand{\fedEstfixed}{\widehat\theta^{\text{fed}}_n(t,a;\bd\eta_{t,a})}
\newcommand{\skEst}{\widehat\theta^{k,0}_n(t,a)}
\newcommand{\estimand}{\theta^0(t,a)}
\newcommand{\bigCI}{\perp\mkern-10mu\perp}

\theoremstyle{plain}
\newtheorem{condition}{Condition}[section]
\newtheorem{theorem}{Theorem}[section]

\newtheorem{lemma}[theorem]{Lemma}

\theoremstyle{definition}

\newtheorem{assumption}[theorem]{Assumption}
\newtheorem{remark}[theorem]{Remark}

\usepackage{iftex}
\ifPDFTeX
  \usepackage[T1]{fontenc}
  \usepackage[utf8]{inputenc}
  \usepackage{textcomp} 
\else 
  \usepackage{unicode-math}
  \defaultfontfeatures{Scale=MatchLowercase}
  \defaultfontfeatures[\rmfamily]{Ligatures=TeX,Scale=1}
\fi
\usepackage{lmodern}
\ifPDFTeX\else  
\fi
\IfFileExists{upquote.sty}{\usepackage{upquote}}{}
\IfFileExists{microtype.sty}{
  \usepackage[]{microtype}
  \UseMicrotypeSet[protrusion]{basicmath} 
}{}
\makeatletter
\@ifundefined{KOMAClassName}{
  \IfFileExists{parskip.sty}{%
    \usepackage{parskip}
  }{
    \setlength{\parindent}{0pt}
    \setlength{\parskip}{6pt plus 2pt minus 1pt}}
}{
  \KOMAoptions{parskip=half}}
\makeatother
\usepackage{xcolor}
\makeatletter
\ifx\paragraph\undefined\else
  \let\oldparagraph\paragraph
  \renewcommand{\paragraph}{
    \@ifstar
      \xxxParagraphStar
      \xxxParagraphNoStar
  }
  \newcommand{\xxxParagraphStar}[1]{\oldparagraph*{#1}\mbox{}}
  \newcommand{\xxxParagraphNoStar}[1]{\oldparagraph{#1}\mbox{}}
\fi
\ifx\subparagraph\undefined\else
  \let\oldsubparagraph\subparagraph
  \renewcommand{\subparagraph}{
    \@ifstar
      \xxxSubParagraphStar
      \xxxSubParagraphNoStar
  }
  \newcommand{\xxxSubParagraphStar}[1]{\oldsubparagraph*{#1}\mbox{}}
  \newcommand{\xxxSubParagraphNoStar}[1]{\oldsubparagraph{#1}\mbox{}}
\fi
\makeatother

\usepackage{longtable,booktabs,array}
\usepackage{calc} 
\usepackage{etoolbox}
\makeatletter
\patchcmd\longtable{\par}{\if@noskipsec\mbox{}\fi\par}{}{}
\makeatother
\IfFileExists{footnotehyper.sty}{\usepackage{footnotehyper}}{\usepackage{footnote}}
\makesavenoteenv{longtable}
\usepackage{graphicx}
\makeatletter
\def\maxwidth{\ifdim\Gin@nat@width>\linewidth\linewidth\else\Gin@nat@width\fi}
\def\maxheight{\ifdim\Gin@nat@height>\textheight\textheight\else\Gin@nat@height\fi}
\makeatother
\setkeys{Gin}{width=\maxwidth,height=\maxheight,keepaspectratio}
\makeatletter
\def\fps@figure{htbp}
\makeatother

\makeatletter
\@ifpackageloaded{caption}{}{\usepackage{caption}}
\AtBeginDocument{%
\ifdefined\contentsname
  \renewcommand*\contentsname{Table of contents}
\else
  \newcommand\contentsname{Table of contents}
\fi
\ifdefined\listfigurename
  \renewcommand*\listfigurename{List of Figures}
\else
  \newcommand\listfigurename{List of Figures}
\fi
\ifdefined\listtablename
  \renewcommand*\listtablename{List of Tables}
\else
  \newcommand\listtablename{List of Tables}
\fi
\ifdefined\figurename
  \renewcommand*\figurename{Figure}
\else
  \newcommand\figurename{Figure}
\fi
\ifdefined\tablename
  \renewcommand*\tablename{Table}
\else
  \newcommand\tablename{Table}
\fi
}
\@ifpackageloaded{float}{}{\usepackage{float}}
\floatstyle{ruled}
\@ifundefined{c@chapter}{\newfloat{codelisting}{h}{lop}}{\newfloat{codelisting}{h}{lop}[chapter]}
\floatname{codelisting}{Listing}

\makeatother
\makeatletter
\@ifpackageloaded{caption}{}{\usepackage{caption}}
\@ifpackageloaded{subcaption}{}{\usepackage{subcaption}}
\makeatother

\ifLuaTeX
  \usepackage{selnolig}  
\fi
\usepackage[]{natbib}
\usepackage{bookmark}

\IfFileExists{xurl.sty}{\usepackage{xurl}}{} 
\hypersetup{
  pdftitle={DataFusionAMP},
  pdfauthor={},
  pdfkeywords={AMP HIV-1 prevention trials; Time-to-event outcome; Distribution shift; Semiparametric efficiency theory; Federated learning},
  colorlinks=true,
  linkcolor={blue},
  filecolor={Maroon},
  citecolor={Blue},
  urlcolor={Blue},
  pdfcreator={LaTeX via pandoc}}

\newcommand{\anon}{1}

\begin{document}

\def\spacingset#1{\renewcommand{\baselinestretch}%
{#1}\small\normalsize} \spacingset{1}


\if1\anon
  {\title{\bf\Large Targeted Data Fusion for Region-Specific Survival Effects in the AMP HIV Prevention Trials}
\author{\small
Yi Liu$^{1}$, Alexander W. Levis$^2$, Ke Zhu$^{1,3}$, Shu Yang$^1$, Peter B. Gilbert$^4$, and Larry Han$^{*5}$
\bigskip \\
\small
$^1$Department of Statistics, North Carolina State University\\
\small
$^2$Department of Biostatistics, Epidemiology and Informatics, University of Pennsylvania \\
\small
$^3$Department of Biostatistics and Bioinformatics, Duke University\\
\small
$^4$Vaccine and Infectious Disease and Public Health Sciences Divisions, Fred Hutch Cancer Center \\
\small
$^5$Department of Biostatistics and Data Science Institute, Brown University \\
\small
$^*$Corresponding author: Larry Han, \texttt{larry\_han@brown.edu}
}
  \maketitle
} \fi

\if0\anon
{
  \bigskip
  \bigskip
  \bigskip
    \title{\Large\bf Targeted Data Fusion for Region-Specific Survival Effects in the AMP HIV Prevention Trials}
    \maketitle
  \medskip
} \fi

\bigskip
\begin{abstract}
    The Antibody Mediated Prevention (AMP) trials opened a new scientific frontier by showing that passively administered monoclonal broadly neutralizing antibodies (bnAbs) could prevent HIV-1 acquisition. The trials enrolled 4,611 participants across four epidemiologically distinct regions, with 174 HIV-1 diagnoses observed during follow-up. Regional heterogeneity in treatment efficacy, together with privacy and regulatory limits on central data pooling, calls for methods that borrow strength across regions without sharing individual-level data. To estimate region- and treatment-specific survival curves under distributional heterogeneity, we develop a federated learning approach that combines site-specific estimators via an $\ell_1$-regularized criterion that downweights data sources not aligned with the target. We further extend the framework to a general class of causal contrasts, including the survival difference, survival ratio, and restricted mean survival time difference. In the AMP analysis, for example, the standard error of the day-148 survival difference decreased from 0.007 using the target-only estimator to 0.004 using the federated estimator. For women in South Africa, the region-specific analyses provided no strong evidence of an overall VRC01 benefit, consistent with the primary AMP findings. Overall, the proposed approach provides region-adaptive inference with improved precision without sharing individual-level data, while preserving estimation of region-specific treatment effects under cross-region heterogeneity. 
\end{abstract}

\noindent%

\newpage
\spacingset{1.8} 

\section{Introduction}\label{sec:intro}

Four decades into the global HIV pandemic, more than 1.5 million people acquire HIV-1 each year, including 150,000 infant infections worldwide. Although risk-reduction counseling, treatment-as-prevention, and daily oral pre-exposure prophylaxis (PrEP) have reduced HIV-1 incidence in some settings, broadly neutralizing antibodies (bnAbs) have emerged as a promising longer-acting prevention strategy that may improve adherence and expand coverage beyond what is achievable with daily oral prophylaxis alone.

VRC01, an IgG1 bnAb targeting the CD4-binding site on the HIV-1 envelope, exhibited broad in vitro activity and favorable pharmacokinetics in early clinical studies. Motivated by these results, the Antibody Mediated Prevention (AMP) program initiated the first large-scale test of whether a monoclonal bnAb can prevent sexual acquisition of HIV-1, and thus opened an important new direction for HIV prevention research \citep{corey2021two}. More broadly, AMP illustrates recurring challenges in modern clinical research: \textit{how to draw scientifically meaningful inference from multi-region trials in which populations differ, event rates are low, and direct pooling of individual-level data may be restricted.} 

These challenges are especially acute for time-to-event outcomes. Classical methods such as the Kaplan--Meier estimator and Cox proportional hazards model \citep{kaplan1958nonparametric, cox1972regression} were developed primarily for single-site settings and may be inadequate for heterogeneous multi-source data without strong assumptions. Simple pooled analyses may target the wrong population when inference is sought for a specific region, while region-specific analyses can be valid but inefficient. Existing data-fusion and federated-learning methods have focused largely on binary or continuous outcomes \citep{han2025federated, han2024privacy, yang2020combining, liu2024multi, han2023multiply, li2023targeting, makhija2024federated} and do not directly address right censoring, cross-region heterogeneity, semiparametric efficiency, and target-region causal survival inference. In the AMP trials and other time-to-event studies, restricting attention to binary event indicators discards valuable timing information about event dynamics and risk evolution. Recent semiparametrically efficient and doubly robust survival methods within a single source provide important building blocks but do not solve the multi-region problem \citep{westling2024inference, bai2013doubly, wolock2024framework}. 

To address the challenges in the AMP trials, we develop data-fusion methods that target a scientifically chosen region while adaptively borrowing information from others. Our framework accommodates heterogeneity in covariate, outcome, and censoring distributions; avoids individual-level data exchange; and supports clinically interpretable causal contrasts, including the survival difference, survival ratio, and restricted mean survival time. We first study an efficiency benchmark under a common conditional outcome distribution (CCOD) assumption with full data sharing, and then develop a federated estimator requiring neither this assumption nor individual-level data sharing. Prior privacy-preserving survival analysis \citep{liuprivacy} did not address the efficiency benchmark or these clinically meaningful causal contrasts. A tabular comparison with  related work \citep{liuprivacy, han2025federated, westling2024inference} is provided in Appendix \ref{subapp:notation}. 

This paper is organized as follows. Section \ref{sec:setup} introduces the setting and scientific questions. Section \ref{sec:framework} develops the data fusion framework. Section \ref{sec:estimands} considers other causal contrasts. Section \ref{sec:data} presents the AMP analysis, Section \ref{sec:simu} reports simulations, and Section \ref{sec:conclude} concludes. 

\section{AMP Trials Setting and Scientific Questions}\label{sec:setup}

\subsection{AMP trials data and basic set-up}

We begin by describing the AMP trials, the scientific/clinical questions that motivate our analysis, and the statistical setup tailored to this application. The AMP trials, which included HVTN 704/HPTN 085 in the Americas and Europe, and HVTN 703/HPTN 081 in sub-Saharan Africa, enrolled 4,611 participants across four epidemiologically distinct regions and measured time to HIV-1 diagnosis with centralized endpoint adjudication and viral neutralization phenotyping. The prespecified primary analyses found that VRC01 did not significantly reduce overall HIV-1 acquisition, with estimated reductions of 26.6\% {(95\% confidence interval [CI]: –11.7\% to 51.8\%)} in the Americas/Europe and 8.8\% in sub-Saharan Africa {(95\% CI: –45.1\% to 42.6\%)}. The four regions: (i) South Africa (women), (ii) other sub-Saharan African countries (women), (iii) Brazil/Peru (men and transgender persons), and (iv) the United States/Switzerland (men and transgender persons), differed in demographic composition, baseline HIV risk, HIV-1 subtype distribution, viral susceptibility, and patterns of right-censoring. 

{Based on the AMP trials data structure, we now introduce the tailored notation, causal estimands, and identifying assumptions used throughout the paper. A complete description of all notation is provided in Appendix \ref{subapp:notation}.} Consider $K$ data sources (or ``sites'' or ``regions''), where covariate and outcome heterogeneity may arise either by chance or from structural differences. This setting reflects the reality of large international trials, in which participant characteristics, outcome and censoring patterns may vary across sites. In the AMP setting, heterogeneity in viral resistance is driven by the evolution of HIV that balances retaining replicative fitness with evading effective immune responses, which has resulted in an increase of about 10\% in neutralization resistance to VRC01 and other antibodies per decade of evolution \citep{mkhize2023neutralization}. Heterogeneity in baseline risk is caused by several inter-related factors including local prevalence of HIV-1, the amount of exposure to HIV-1 that depends on demographics, the transmissibility of HIV-1 exposures that depends on antiretroviral therapy use and viral load, and differential use of HIV-1 prevention tools such as pre-exposure prophylaxis. 

For each participant, let $\mb X$ denote baseline covariates (age, weight, body mass index [BMI] and a machine-learning [ML] risk score in the AMP trials) and $A\in\{0,1\}$ denote the treatment assignment ($A=1$ for the combined high- and low-dose bnAb group, $A=0$ for placebo). Therefore, our analysis targets the causal contrast between any bnAb treatment and placebo, which follows the primary analysis by \cite{corey2021two}. Under the potential outcomes framework \citep{Rubin1974}, we define potential event times $T^{(1)}$ and $T^{(0)}$, representing the time to HIV-1 diagnosis by week 80 under bnAb and placebo, respectively, and potential censoring times $C^{(1)}$ and $C^{(0)}$. The observed event time is 
$T = T^{(A)} = A T^{(1)} + (1 - A) T^{(0)},$ and the observed censoring time is 
$C = C^{(A)} = A C^{(1)} + (1 - A) C^{(0)},$ where we adopt the stable unit treatment value assumption (SUTVA; see \cite{rubin1980randomization}). Under right censoring, we observe the censored outcome $Y = \min(T, C)$ and event indicator $\Delta = \Ix(T \le C)$, where $\Ix(\cdot)$ is the indicator function. We denote an arbitrary observation as $\mc O=(\mb X, A, Y, \Delta, R)$, where $R\in\{0,1,\ldots,K-1\}$ {denotes region membership among individuals already enrolled in the AMP trials.} The combined dataset from all regions is $\{\mc O_i, i=1,\ldots,n\}$, with $R=0$ indicating the \textit{target region} and $R=1,\ldots,K-1$ the \textit{source regions}. Our sampling framework assumes a single superpopulation for $\mc O$, and the target population is the subpopulation corresponding to $R=0$. 

\begin{table}[ht]
\footnotesize
\singlespacing
\caption{Summary statistics of AMP trial data by treatment group and region.}\label{tab:AMP-sumstat}
\centering
    \begin{tabular}{rccccc}
    \toprule
    & \multicolumn{5}{c}{\textbf{Treated (bnAb) group}} \\
    & \makecell[c]{\textbf{Total} \\ ($n=3,076$)} & \makecell[c]{\textbf{SA (women)} \\ ($n = 679$)} & \makecell[c]{\textbf{OA (women)} \\ ($n = 608$)} & \makecell[c]{\textbf{BP (men, TG)} \\ ($n = 846$)} & \makecell[c]{\textbf{US (men, TG)} \\ ($n = 943$)} \\
    \midrule
        Age (years) & 27.9 (6.90) & 25.4 (5.09) & 27.1 (5.37) & 27.5 (7.24) & 30.4 (7.73) \\
        Baseline BMI & 25.9 (4.60) & 27.0 (5.19) & 25.4 (4.59) & 25.1 (3.71) & 26.2 (4.69) \\
        Weight (kg) & 72.8 (15.64) & 68.8 (14.24) & 65.2 (12.63) & 70.9 (12.42) & 82.3 (16.43) \\
        ML risk score & 0.0 (1.00) & -0.02 (1.00) & 0.02 (1.00) & 0.76 (0.67) & -0.68 (0.71) \\
        HIV-1 diagnosis & 107 (3.48\%) & 27 (3.98\%) & 20 (3.29\%) & 46 (5.44\%) & 14 (1.48\%) \\
    \midrule
    & \multicolumn{5}{c}{\textbf{Control (placebo) group}} \\
    & \makecell[c]{\textbf{Total} \\ ($n=1,535$)} & \makecell[c]{\textbf{SA (women)} \\ ($n = 340$)} & \makecell[c]{\textbf{OA (women)} \\ ($n = 297$)} & \makecell[c]{\textbf{BP (men, TG)} \\ ($n = 428$)} & \makecell[c]{\textbf{US (men, TG)} \\ ($n = 470$)} \\
    \midrule
        Age (years) & 27.9 (6.71) & 25.6 (5.22) & 27.4 (5.21) & 27.6 (6.91) & 30.2 (7.59) \\
        Baseline BMI & 25.9 (4.72) & 26.6 (5.28) & 25.4 (4.79) & 25.3 (3.94) & 26.2 (4.82) \\
         Weight (kg) & 72.5 (16.35) & 67.6 (14.77) & 65.1 (13.64) & 71.1 (12.84) & 81.8 (17.55) \\
        ML risk score & 0.0 (1.00) & 0.02 (0.92) & -0.02 (1.08) & 0.75 (0.67) & -0.68 (0.73) \\
        HIV-1 diagnosis & 67 (4.36\%) & 16 (4.71\%) & 13 (4.38\%) & 29 (6.78\%) & 9 (1.91\%) \\
        \bottomrule
    \end{tabular}
    \begin{tablenotes}\scriptsize
    \item The ML risk score is a baseline standardized machine-learned score predictive of HIV acquisition \citep{corey2021two}. BMI: body mass index. Age, BMI, ML risk score, and weight are measured at baseline and summarized by mean (standard deviation), while HIV-1 diagnosis is measured by week 80 and is summarized by count (percentage). TG: transgender.
    \end{tablenotes}
\end{table}

To illustrate the regional heterogeneity in the AMP trials, Table~\ref{tab:AMP-sumstat} summarizes baseline covariates and HIV-1 outcomes by region and treatment group. HIV acquisition was uniformly rare but varied across populations; baseline risk scores and body weight differed markedly; and subtype distributions aligned with known regional epidemiology.

\subsection{Scientific question, estimand and identification assumptions}

Because HIV acquisition was rare in all AMP regions, the original AMP statistical analysis plan was not powered to estimate region-specific prevention efficacy, nor to assess whether causal survival effects varied across regions. However, each region represents a scientifically meaningful target population because HIV epidemics differ in scale, subtype, and risk structure, and next-generation bnAb deployment is likely to require region-tailored strategies. 

Further practical constraints complicate statistical inference. In many multi-region studies, privacy, regulatory, or data-use considerations may restrict individual-level data sharing. Because the AMP data are jointly available here, we can benchmark federated region-specific inference against pooled estimators. Purely stratified analyses can be imprecise and cannot borrow strength adaptively. 

These considerations motivate two related questions: \textit{(Q1) How can we obtain valid region-specific survival estimates and clinically interpretable treatment contrasts while adaptively borrowing information across heterogeneous regions with low event rates? (Q2) How do the estimated treatment effects vary when each region is considered in turn as the target population?} We address Q1 through the methodological development below, while Q2 is examined in Appendix \ref{app:add-data} by rotating the target region across the AMP regions. 

We define the estimand of interest as the target-site, treatment-specific survival function: $\estimand = \Px(T^{(a)}>t\mid R=0)$ for $a\in\{0,1\}, t\in[0,\tau]$ and a fixed follow-up horizon $\tau<\infty$. In the AMP trials, follow-up was conducted for up to $\tau = 601$ days (approximately 80 weeks) post enrollment. Here, $\Px$ denotes the underlying data-generating distribution and the superscript ``0'' indicates that the attached quantity is specific to the target site. The estimand $\theta^0(t,a)$ represents the probability that an individual in the target region remains HIV-infection-free up to time $t$ if assigned treatment $a$. 

Several \textit{nuisance functions} play a key role in our estimators. For each site $k=0,\dots,K-1,$ define the conditional survival function $S^k(t\mid a, \mb X)=\Px(T > t\mid A=a, \mb X, R=k).$ To handle both continuous- and discrete-time settings, we adopt the Riemann–Stieltjes product integral notation $\prodi$ \citep{gill1990survey}:
\begin{align*}
    S^k(t\mid a, \mb X) = \Prodi_{(0,t]}\{1-\Lambda^k(du\mid a,\mb X)\}, \quad\text{ where }
    \Lambda^k(t\mid a,\mb X) = \int_0^t\frac{N^k_{1}(du\mid a,\mb X)}{D^k(u\mid a,\mb X)}
\end{align*}
is the conditional cumulative hazard function under $R=k$, with $N^k_{\delta}(t\mid a,\mb X) = \Px(Y\leq t,\Delta=\delta\mid A=a,\mb X, R=k)$ the conditional cumulative incidence function of event $(\delta=1)$ or censoring $(\delta=0)$, and $D^k(t\mid a,\mb X)=\Px(Y\geq t\mid A=a, \mb X, R=k)$. Under discrete time $\prodi$ reduces to an ordinary product $\prod$, and under continuous time to $\exp\{-\Lambda^k(t\mid a,\mb X)\}$. 

To interpret the causal estimand $\theta^0(t,a)$ as the treatment-specific survival probability in the target region, several  assumptions for identification (in addition to SUTVA) are made.

\begin{assumption}[Unconfoundedness]\label{asp:unconf}
$A \bigCI T^{(a)} \mid \mb X, R$ and $A \bigCI C^{(a)} \mid \mb X, R$, $a \in \{0,1\}$.
\end{assumption}

This assumption states that, within each region and conditional on the
measured baseline covariates, treatment assignment is independent of the potential event and censoring times under each treatment level. In our analysis, we merge the low- and high-dose groups into one active treatment group. In fact, under complete, uniform randomization in the AMP trials in every region, $A \bigCI (T^{(a)}, C^{(a)}, \mb X, R)$ holds because every participant had an equal chance to receive any of the three treatments (and so the binary treatment $A$ we defined), which substantially mitigates confounding concerns. In this randomized trial setting, covariate adjustment is therefore used primarily to improve efficiency and to accommodate covariate-dependent censoring, rather than to remove treatment confounding. We adopt the weaker form in Assumption~\ref{asp:unconf} to accommodate more general settings where randomization probabilities may differ across regions. In such cases, covariate adjustment can effectively remove residual biases arising from random or structural imbalances. 

\begin{assumption}[Non-informative censoring]\label{asp:inde-cen} 
$C^{(a)}\bigCI T^{(a)}\mid A=a, \mb X, R$, for $a \in \{0,1\}$.
\end{assumption}
{This is a conditional independent-censoring assumption: censoring may depend on treatment, baseline covariates, and region, but is assumed independent of the potential event time after conditioning on these variables.} In the AMP trials, retention was high and censoring was largely administrative, reflecting the implementation science preparedness research that supported high acceptability of infusions, and rigorous community engagement and trial conduct, making this assumption more plausible. 

\begin{assumption}[Positivity]\label{asp:positivity}
    There exists an $\eta\in(0,\infty)$ such that $\Px(R=k)\geq 1/\eta$ and, for $\Px$-almost all $\mb X$, $\min_{k=0,\ldots,K-1}\{\pi^k(a\mid\mb X), G^k(t-\mid a,\mb X), 1/\omega^{k,0}(\mb X)\}\geq 1/\eta$, $\min_k S^k(t\mid a,\mb X)>0$. Here, $\pi^k(a\mid\mb X)=\Px(A=a\mid\mb X,R=k)$ is the propensity score for treatment $A=a$, $G^k(t-\mid a,\mb X)=\Px(C\geq t\mid A=a,\mb X,R=k)$ is the left limit of the conditional censoring survival function in site $k$, and $\omega^{k,0}(\mb X)=\Px(\mb X\mid R=0)/\Px(\mb X\mid R=k)$ is the covariate density ratio, for $k=0,1,\dots,K-1$, $a\in\{0,1\}$, and $t\in[0,\tau]$.
\end{assumption} 


This assumption ensures that all relevant covariate patterns appear with non-negligible probability in every region and that both treatment groups and censoring distributions are well represented. In the AMP data, both bnAb and placebo recipients were enrolled from each region {and randomized at baseline}, censoring before week 80 was rare, {and the retention rate was high in each region,} supporting this assumption. The additional condition $\Px(R=k)\ge1/\eta$ guarantees that no site contributes an extremely small sample share to the pooled analysis \citep{lee2022doubly, cao2024transporting}, and the four regions in the AMP trials each contribute sufficiently large sample sizes to satisfy this requirement. {This condition allows unequal and skewed region sizes, provided that no included region has an asymptotically vanishing sample proportion.}

\subsection{Estimation by region-specific data}\label{subsec:tgtonlyest}


The region-level results can be formally characterized using the \textit{semiparametric efficient estimator} developed by \cite{westling2024inference}. For later use, define $\mc H_{t,a}$ that appears in the augmented term of the efficient influence functions (EIFs): 
\begin{align}\label{def:H}
    \mc H_{t,a}(\mc O;S,G) = \frac{\Ix(Y\le t,\Delta=1)} {S(Y\mid a,\mb X)G(Y-\mid a,\mb X)} - \int_0^{t\wedge Y} \frac{\Lambda(du\mid a,\mb X)}{S(u\mid a,\mb X)G(u-\mid a,\mb X)}.
\end{align}
This term is a survival-process residual: it compares the observed event by time $t$ with its cumulative expectation under treatment $a$, after adjusting for survival and censoring. Since $S$ and $\Lambda$ have a one-to-one mapping, we omit $\Lambda$ as an input to $\mc H_{t,a}(\mc O;S,G)$. Under a completely nonparametric model, the EIF $\varphi^{*0}_{t,a}(\mc O;\Px)$ for each $(t,a)\in[0,\tau]\times\{0,1\}$, which uses only target-site data \citep{westling2024inference} is $\varphi^{*0}_{t,a}(\mc O;\Px) =$
\begin{align*}
    \frac{\Ix(R=0)}{\Px(R=0)}\{S^0(t\mid a, \mb X)-\theta^0(t,a)\}-\frac{\Ix(R=0)}{\Px(R=0)}\frac{\Ix(A=a)}{\pi^0(a\mid\mb X)}\mc H_{t,a}(\mc O;S^0,G^0)S^0(t\mid a, \mb X). 
\end{align*}
To simplify notation, we use $\varphi^{*0}_{t,a}(\mc O;\Px)$ as shorthand for $\varphi^{*0}_{t,a}(\mc O; S^0, G^0, \pi^0)$, where $\Px$ means that the EIF is evaluated by the true nuisance functions $(S^0, G^0, \pi^0)$. $\widehat\Px$ denotes the same functional with the nuisance functions replaced by their estimators $(\widehat S^0, \widehat G^0, \widehat\pi^0)$. This should not be confused with $\Px_n$, where $\Px_n[f(\mc O)] = n^{-1}\sum_{i=1}^n f(\mc O_i)$ denotes the empirical average. 

The above EIF motivates a semiparametrically efficient estimator $\tgtEst$ by solving an estimating equation based on target-only data: $\Px_n[\widehat\varphi^{*0}_{t,a}(\mc O;\widehat\Px)] = 0.$ We implement this estimator for each region with results shown in Figure \ref{fig:AMP_siteSurvs}. Nuisance functions are fitted by ensemble learning \citep{westling2024inference}. The estimated survival curves suggest lower HIV-1 acquisition risk under VRC01 than placebo across regions, but with heterogeneity in magnitude consistent with differences in circulating viral susceptibility and baseline risk. 
\begin{figure}[ht]
    \centering
    \includegraphics[width=0.88\textwidth]{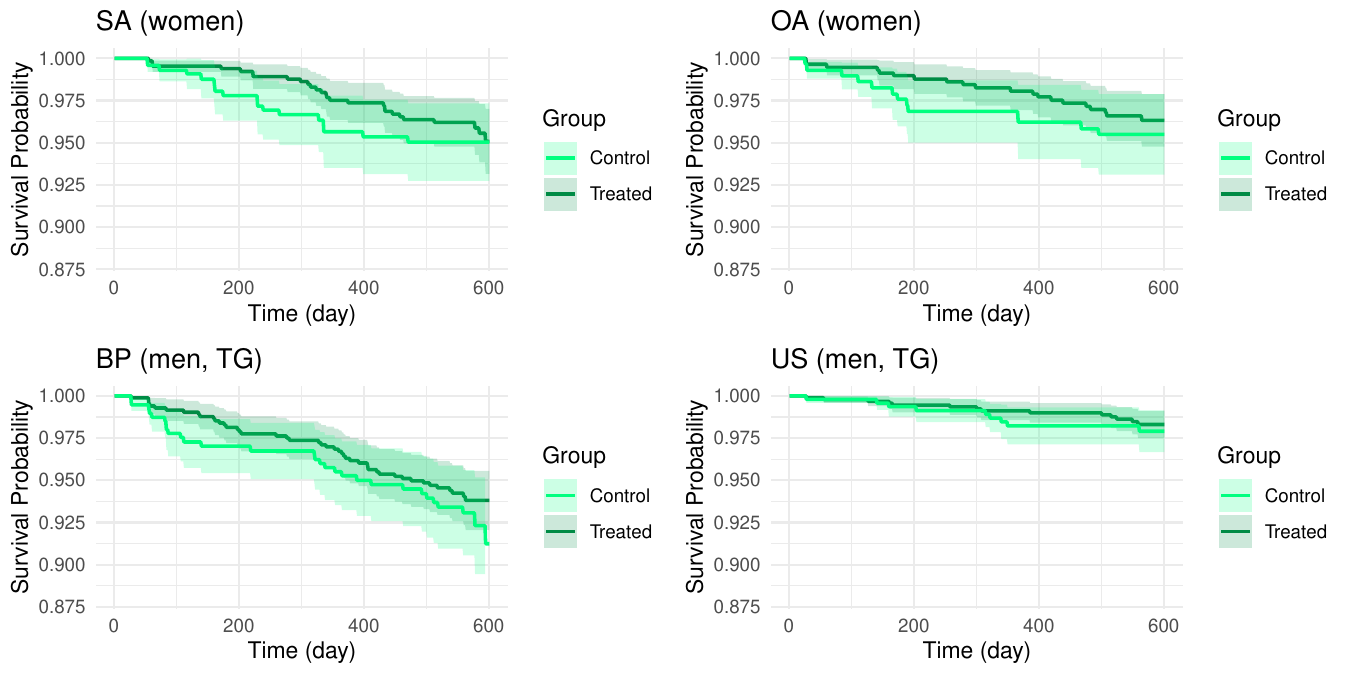}
    \caption{Estimated region-specific survival curves of the HVTN 704/HPTN 085 and HVTN 703/HPTN 081 trials. TG: transgender. }
    \label{fig:AMP_siteSurvs} 
\end{figure}

To further improve the precision of region-specific estimates, we develop targeted data-fusion approaches in the next section. The AMP trials motivate methodology that (i) treats each region as its own target population; (ii) adaptively borrows information across regions using locally computed site-level estimates and shared summary-level quantities, without requiring individual-level data exchange; (iii) supports semiparametric and algorithmic estimation for right-censored time-to-event outcomes; and (iv) accommodates modern causal survival contrasts. This need extends beyond AMP: future bnAb programs will increasingly rely on multi-region, multi-population clinical trials where scientific inference depends on contextually appropriate, region-specific causal estimation.




\section{Targeted Data Fusion for the AMP Trials}\label{sec:framework}

\subsection{The efficiency benchmark}\label{subsec:CCOD}

Consider defining one of the sites, e.g., South Africa (SA), as the target $R=0$. Leveraging other data sources such as other sub-Saharan African (OA), Brazil/Peru (BP), and United States/Switzerland (US) regions ($R=1,2,3$) as sources in the AMP trials can potentially improve the estimation efficiency of the SA's survival curves. A common assumption in data fusion literature is that the conditional distribution of the potential event time is identical across regions, i.e., CCOD, which we formalize as follows.  

\begin{assumption}[Common conditional outcome distribution]\label{asp:ccod}
$T^{(a)} \bigCI R \mid \mb X$ for $a \in \{0,1\}$.
\end{assumption}

Assumption \ref{asp:ccod} states that the conditional distributions of the potential event times $T^{(a)}$ are homogeneous across {AMP regions SA, OA, US and BP}, given baseline covariates. Equivalently, this implies
$S^k(\,\cdot\mid a,\mb X) = \Px(T>\cdot\mid A=a,\mb X, R=k) = \Px(T>\cdot\mid A=a,\mb X) \eqqcolon \bar S(\,\cdot\mid a,\mb X)$
for all regions $k=0,1,\dots,3$, while allowing arbitrary shifts in the covariate distributions across regions. Under this restriction, a semiparametrically efficient  estimator for the target estimand $\estimand$ can be constructed by pooling all AMP regions. 

Consider the following ``global'' nuisance functions {for the entire AMP trial population} (each defined without conditioning on region indicator $R$): the global propensity score $\bar\pi(a\mid\mb X) = \Px(A=a\mid\mb X)$, the global conditional survival function for the event time $\bar S(t\mid a,\mb X)=\Px(T>t\mid A=a,\mb X)$, and the global conditional survival function for the censoring time $\bar G(t\mid a,\mb X)=\Px(C>t\mid A=a,\mb X)$. Moreover, let the target site propensity be given by $q^0(\mb X) = \Px(R=0\mid\mb X)$. {To identify the target-region survival curve under CCOD, we require positivity conditions for the global treatment, conditional censoring and survival function over the time window of interest, stated in Assumption \ref{asp:positivity-ccod} below.} 

\begin{assumption}\label{asp:positivity-ccod}
    There exists $\eta\in(0,\infty)$ such that for $\Px$-almost all $\mb X$, $\min\{\bar\pi(a\mid\mb X), \bar G(t-\mid a,\mb X)\}\geq 1/\eta$ and $\bar S(t\mid a,\mb X)>0$, for $a\in\{0,1\}$ and $t\in[0,\tau]$.  
\end{assumption}

\begin{theorem}\label{thm:CCOD-EIF}
     Under Assumptions \ref{asp:unconf}, \ref{asp:inde-cen}, \ref{asp:ccod} and \ref{asp:positivity-ccod}, the semiparametric EIF of $\estimand$ for $t\in[0,\tau]$ and $a\in\{0,1\}$ is given by $ \varphi^{*\text{CCOD}}_{t,a}(\mc O;\Px) = $
\begin{align*}
   \frac{\Ix(R=0)}{\Px(R=0)}\{\bar S(t\mid a, \mb X)-\estimand\} - \frac{q^0(\mb X)}{\Px(R=0)}\frac{\Ix(A=a)}{\bar\pi(a\mid\mb X)}\mc H_{t,a}(\mc O;\bar S,\bar G)\bar S(t\mid a, \mb X). 
\end{align*}
\end{theorem}
The proof of Theorem \ref{thm:CCOD-EIF} is provided in Appendix \ref{app:theory-CCOD}. Heuristically, this EIF has the familiar structure of EIFs in causal inference: an anchor term based on the global (pooled) survival function $\bar S(t\mid a,\mb X)$, together with an augmentation term involving the nuisance functions. The augmentation term corrects for treatment assignment and censoring, while the target-site propensity score $q^0(\mb X)$ plays a transport role, ensuring that the EIF is centered on the target-region estimand.

A semiparametrically efficient and doubly robust ``CCOD estimator'' $\glbEst$ is given by solving the estimating equation $\Px_n[\varphi^{*\text{CCOD}}_{t,a}(\mc O;\widehat\Px)] = 0.$ {Note that the term $\Px(R=0)$ is a constant probability and therefore does not affect the solution to the estimating equation. However, it is needed in the EIF as a scaling factor to ensure the correct normalization and semiparametric efficiency interpretation in this multi-source setting. } 

{Furthermore, we estimate these global nuisance parameters using pooled data from all AMP regions.} In practice, we use $M$-fold cross-fitting for $\glbEst$: the entire sample is split into $M$ folds, where a popular choice is $M=5$ \citep{chernozhukov2018double}, nuisance functions are fitted on training folds, and estimating equations are evaluated on the corresponding validation fold. The results are then combined across folds to obtain $\glbEst$. 

\begin{remark}\label{rmk:global-funcs}
Theorem \ref{thm:CCOD-EIF} reveals a somewhat counterintuitive result. Although we do not impose conditional homogeneity assumptions on the treatment $A \mid \mb X$ or the censoring $C \mid (\mb X, A)$ across study sites $R$, the resulting EIF involves only global models for these nuisance functions rather than site-specific ones. This occurs because the EIF is derived from the identification formula for $\theta^0(t,a)$ (see Appendix \ref{app:theory-CCOD}), which, under the CCOD assumption, does not depend on the site indicator $R$. As a result, site-specific models for treatment and censoring are not required for achieving semiparametric efficiency. In addition, because the global nuisance functions are estimated using pooled data, they can benefit from a larger sample size, which improves finite-sample prediction accuracy.
\end{remark}

To establish the regular and asymptotically linear (RAL) property of the CCOD estimator, we require three types of conditions from the causal inference literature: consistency of the estimated nuisance functions, positivity of the estimated treatment, censoring, and survival components, and sufficiently small second-order remainder terms. Together, these conditions justify the EIF representation derived in Theorem \ref{thm:CCOD-EIF} and imply the asymptotic normality used for Wald-type inference in the AMP analysis later in the paper. For readability, we defer their precise technical statements to Appendix \ref{app:theory-CCOD} (Conditions C.1--C.3).

\begin{theorem}\label{thm:RAL-ccod}
    If Conditions C.1--C.3 in Appendix \ref{app:theory-CCOD} hold,  $\glbEst=\estimand+\Px_n(\varphi^{*\text{CCOD}}_{t,a}) + o_p(n^{-1/2})$. In particular, $n^{1/2}(\glbEst-\estimand)$  converges in distribution to a Gaussian limit with mean zero and variance $\sigma^2 = \Px[(\varphi^{*\text{CCOD}}_{t,a})^2]$.
\end{theorem}

\begin{remark}[Double robustness of the CCOD estimator]\label{rmk:DR-ccod}
If we are only concerned with consistency of the CCOD estimator, we do not require all nuisance functions to be correctly specified at all time points. For $\Px$-almost all $\mb X$, we only require that there exist measurable sets $\bar{\mc S}_x, \bar{\mc G}_x \subseteq [0,t]$ such that $\bar{\mc S}_x \cup \bar{\mc G}_x = [0,t]$ and
$\bar\Lambda(u\mid a,\mb X)$ is the probability limit of its estimator for all $u \in \bar{\mc S}_x$, and $\bar G(u-\mid a,\mb X)$ is the probability limit of its estimator for all $u \in \bar{\mc G}_x$.
In addition, if $\bar{\mc S}_x$ is a strict subset of $[0,t]$, then $\bar\pi(a\mid\mb X)$ and $q^0(\mb X)$ should be the probability limits of their estimators. Under these conditions, $\glbEst$ is consistent.
\end{remark}

Remark \ref{rmk:DR-ccod} can be interpreted as follows: at a given time $t$, if either (i) the conditional survival model $\bar S$ or (ii) all other nuisance functions $\bar G$, $\bar\pi$, and $q^0$ are correctly specified, then $\glbEst$ remains consistent. A formal technical version of Remark \ref{rmk:DR-ccod} and its proof can be found in Appendix \ref{app:theory-CCOD}. 

\subsection{Federated estimation under distribution shifts}\label{subsec:FED}

In the AMP trials, regional heterogeneity (e.g., in background HIV-1 incidence as impacted by use of various HIV-1 prevention tools such as antiviral pre-exposure prophylaxis) and privacy rules limit central data pooling. Moreover, the CCOD assumption is often too strong. Nonetheless, several AMP regions can still be informative for the target region. Motivated by this, we develop a data-adaptive, privacy-preserving \textit{federated estimation} strategy {whose source weights are learned from the observed data rather than prespecified} and borrows strength across AMP regions using only summary statistics, not individual-level data. Here, by ``privacy-preserving,'' we mean that the federated procedure does not require sharing individual-level data across sites. Instead, each site computes the required nuisance estimators and summary-level quantities locally, and only these aggregated summaries are communicated for downstream estimation and inference. The shared summaries are not assumed to be bias-free; their usefulness depends on how well the corresponding local estimators target the target-site estimand. {Misaligned source estimators are penalized toward zero weight, limiting propagation of source-specific bias.} 

\subsubsection{Estimation by each local site}\label{subsec:local}

We first construct robust site-specific survival curve estimates using data from site $k\in\{0,1,\dots,K-1\}$, making a working ``site-$k$ CCOD'' assumption that potential event times have the same conditional distribution in site $k$ as in the target site when $k\not=0$. Note that the site-$k$ CCOD assumption is used only to derive the influence-function form underlying the local estimator. To aggregate information from source sites, we derive federated weights to account for possible violations of site-$k$ CCOD in Section \ref{subsec:aggregate}. For each site $k$, define
$\theta^{k,0}(t,a) = \Ex\{S^k(t\mid a,\mb X)\mid R=0\},$ which coincides with the target estimand $\theta^0(t,a)$ under the site-$k$ CCOD assumption. The influence function used to construct the local estimator is given in the following theorem.
\begin{theorem}\label{thm:site-k-EIF} 
    For $k\in\{0,1,\dots,K-1\}$, $t\in[0,\tau]$, and $a\in\{0,1\}$, under the working assumption that $S^k(t\mid a,\mb X)=S^0(t\mid a,\mb X)$ almost surely, an influence function $\varphi^{*k,0}_{t,a}(\mc O;\Px)$ is given by $\varphi^{*k,0}_{t,a}(\mc O;\Px)=$
\begin{align*}
    \frac{\Ix(R=0)}{\Px(R=0)}\{S^0(t\mid a, \mb X)-\theta^0(t,a)\} -\frac{\Ix(R=k)}{\Px(R=k)}\omega^{k,0}(\mb X)S^k(t\mid a, \mb X) \frac{\Ix(A=a)}{\pi^k(a\mid\mb X)}\mc H_{t,a}(\mc O;S^k,G^k). 
\end{align*}
\end{theorem}
The first term in $\varphi^{*k,0}_{t,a}(\mc O;\Px)$ is the same anchor term that appears in the EIF for the target-only estimator, reflecting the contribution from the target site. The second term is a density-ratio-weighted augmentation term that transports information from source site $k$ to the target site under site-$k$ CCOD, which correctly targets the target-site estimand. 

A local estimator can be motivated from the influence function above, namely $\skEst$ by solving $\Px_n[\widehat\varphi^{*k,0}_{t,a}(\mc O;\widehat\Px)]=0$. Of note regarding $\skEst$: (i) the ``anchor'' term in the EIF uses target-site data ($R=0$), while the ``augmentation'' term leverages site-$k$ data ($R=k$); (ii) all information is obtained independently from each site except for the density ratio $\omega^{k,0}(\mb X)$, which is estimated by sharing only coarse summary statistics {(see Remark \ref{rmk:dens-ratio} below)}; and (iii) for $S^k$ in the augmentation term, we train the target-site model $S^0$ on data from $R=0$ and apply its predictions to site $k$, rather than fitting a separate local model at site $k$. This is necessary for detecting violations of site-$k$ CCOD in the federated weighting step. If site-$k$ CCOD holds, then $S^0=S^k$, so the target-site model remains appropriate and the resulting discrepancy between $\skEst$ and $\tgtEst$ is close to zero. If site-$k$ CCOD fails, then $S^0\neq S^k$, and this mismatch is reflected in a nonzero discrepancy, which causes site $k$ to be downweighted in our adaptive weighting procedure specified later. By contrast, using a locally trained $S^k$ would largely mask this heterogeneity, since the augmented term is mean zero when all nuisance functions are evaluated on its own site.

\begin{remark}[Density ratio model]\label{rmk:dens-ratio}
    To estimate the density ratio under data-sharing constraints, a common approach is the exponential tilt model of \cite{han2025federated}: $\omega^{k,0}(\mb X) = \exp(\bd\gamma_k'\psi(\mb X))$, where $\bd\gamma_k$ is a parameter vector and $\psi(\mb X)$ is a chosen set of basis functions, e.g., $\psi(\mb X)=\mb X$. Higher-order terms can be included to capture nonlinearities. Estimation of $\bd\gamma_k$ by maximum likelihood requires only the target-site sample mean of $\psi(\mb X)$ to be shared with source sites. More flexible nonparametric or machine learning approaches are possible, but they may require sharing higher-dimensional summaries, such as covariance. 
\end{remark}

\subsubsection{Adaptive aggregation across sites}\label{subsec:aggregate}

To aggregate information from the target and source sites, for a fixed time $t$ and treatment $a$, we use an $\ell_1$-penalization procedure to calculate federated weights for each site. We define the site-specific discrepancy measure as  $\widehat\chi_{n,t,a}^{k,0}=\skEst-\tgtEst,$ and write the treatment- and time-specific federated weights vector as $\bd\eta_{t,a}=(\eta^0_{t,a},\eta^1_{t,a},\dots,\eta^{K-1}_{t,a}).$ To obtain the vector of weights $\bd\eta_{t,a}$, we minimize the $\ell_1$-penalized objective function: 
\begin{align}\label{eq:fed-obj}
    Q(\bd\eta_{t,a})
    =
    \Px_n\left[
    \left\{
    \widehat\varphi^{*0}_{t,a}(\mc O;\widehat\Px)
    - \sum_{k=1}^{K-1}
    \eta^k_{t,a}
    \left[
    \widehat\varphi^{*0}_{t,a}(\mc O;\widehat\Px)
    - \widehat\varphi^{*k,0}_{t,a}(\mc O;\widehat\Px)
    \right]
    \right\}^2
    \right]
    +
    \frac{\lambda}{n}
    \sum_{k=1}^{K-1}
    \vert\eta^k_{t,a}\vert
    (\widehat\chi_{n,t,a}^{k,0})^2.
\end{align}
subject to $0\leq\eta^k_{t,a}\leq 1$ for all $k=0,1,\dots,K-1$ and $\sum_{k=0}^{K-1}\eta^k_{t,a}=1$, and $\lambda$ is a tuning parameter that controls the bias-variance trade-off and is chosen by cross-validation.

\begin{figure}[ht]
    \centering
    \includegraphics[width=\textwidth]{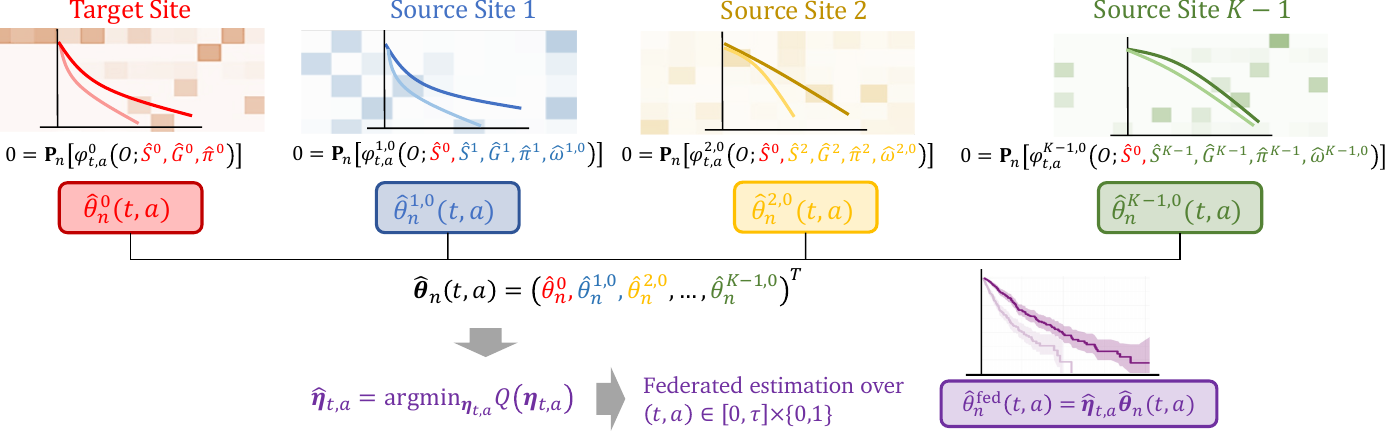}
    \caption{Federated method flow. Each region learns a site-specific survival estimate using local data and receives summaries; treatment- and time-specific weights are learned; the final estimate is a weighted average targeting the AMP region of interest.}
    \label{fig:FED}
\end{figure}

The objective function prioritizes sites that are well aligned with the target survival distribution while asymptotically excluding uninformative ones. We employ a weighted $\ell_1$ penalty to encourage a sparse weighting structure, favoring source sites with smaller estimated discrepancies while allowing substantially misaligned sites to receive zero weight. As a result, only relevant sites contribute to the federated estimator. Our approach anchors on the target-site estimate $\tgtEst$ and assigns weight to site $k$ only when it is sufficiently similar to the target estimand {(see Theorem \ref{thm:RAL-fed} for a more formal characterization)}. The federated estimator $\fedEst$ is then obtained as a weighted average of the site-specific survival estimates: 
$$
\fedEst=\sum_{k=0}^{K-1}\widehat\eta^k_{t,a}\widehat\theta^{k,0}_n(t,a).
$$
Figure \ref{fig:FED} depicts the federated procedure for the AMP trials. Appendix \ref{app:theory-FED} presents the variance estimator for $\fedEst$, derived from its influence function. 

\subsubsection{Theoretical guarantees}\label{subsec:FED-theory}

We establish the asymptotic properties and efficiency gain of the federated estimator in this section. For each local estimator, we establish similar regularity conditions to those of the CCOD estimator, with details deferred to Appendix \ref{subsubapp:cond-local}. These regularity conditions include (i) convergence of the local nuisance estimators to general limits; (ii) positivity conditions that bound the nuisance functions, including the density ratio, away from extreme values so that their estimators are stable and their probability limits are well defined; and (iii) three product-type estimation errors. Together with the oracle selection conditions, these conditions establish the consistency and asymptotic normality of FED in Theorem \ref{thm:RAL-fed}. Regular asymptotic efficiency gains require valid cross-site borrowing, ensured in our framework when at least one source region satisfies CCOD with the target region. 

\begin{theorem}\label{thm:RAL-fed}
At a given time point $t\in[0,\tau]$ and treatment $a\in\{0,1\}$, if Conditions D.1--D.3 in Appendix \ref{subsubapp:cond-local} hold with the nuisance
estimators converging to their true counterparts, and if the oracle
selection conditions for the federated weights $\bd\eta_{t,a}$ hold
(see Appendix \ref{subapp:theory-fed}), the estimator $\fedEst$ has the asymptotically normal distribution 
\begin{align*}
    \sqrt{n/\widehat{\mc V}_{t,a}^{\text{fed}}}\left\{\fedEst-\estimand\right\}\to_d\mc N(0,1),
\end{align*}
where $\widehat{\mc V}_{t,a}^{\text{fed}}$ is an influence-function-based consistent estimator of the asymptotic variance of $\fedEst$, ${\mc V}_{t,a}^{\text{fed}}$. Under the oracle selection conditions, the federated estimator has asymptotic variance no greater than that of the target-only estimator $\tgtEst$. Moreover, if at least one source site provides a consistent local estimator, ${\mc V}_{t,a}^{\text{fed}}$ is strictly smaller than that of $\tgtEst$ under the stated conditions. The exact form of $\widehat{\mc V}_{t,a}^{\text{fed}}$ is presented in Appendix \ref{subapp:theory-fed}. 
\end{theorem}

The proof of Theorem \ref{thm:RAL-fed} is provided in Appendix \ref{subapp:theory-fed}. Here, we briefly summarize the main ideas and implications. Following \citet{han2025federated, han2023multiply}, we quantify the efficiency gain of $\fedEst$ relative to the target-only estimator. Letting $\mc S=\{1,\dots,K-1\}$ denote the set of source sites, we define the oracle selection space for $\bd\eta_{t,a}$ as $\mc S^*_{t,a} = \{k\in\mc S:\theta^k(t,a)=\estimand\}$, where $\theta^k(t,a)$ denotes the probability limit of $\widehat\theta_n^{k,0}(t,a)$, and the corresponding weight space as $\mathbb R^{S^*_{t,a}} = \{\bd\eta_{t,a}\in\mathbb R^{K-1}:\eta^j_{t,a}=0,\forall j\not\in\mc S^*_{t,a}\}$. Appendix \ref{subapp:theory-fed} gives useful lemmata and conditions under which solving \eqref{eq:fed-obj} recovers the following optimal weights: 
$$ 
\bar{\bd\eta}_{t,a} = \underset{\eta^k_{t,a}=0,\forall k\not\in\mc S^*_{t,a}}{\text{arg min}} \mc V_{t,a}^{\text{fed}}(\bd\eta_{t,a}),
$$
where $\mc V_{t,a}^{\text{fed}}(\bd\eta_{t,a})$ is the asymptotic variance of the federated estimator under a fixed weight vector $\bd\eta_{t,a}$. Since the target-only estimator corresponds to a special case $\bd\eta_{t,a}=(1,0,\dots,0)$, its asymptotic variance cannot be exceeded by that of any federated estimator. Moreover, if the bias term $\widehat\chi_{n,t,a}^{k,0}$ is asymptotically non-zero, the corresponding weight $\eta^k_{t,a}$ converges to zero (dropping site $k$). Note that the asymptotic behavior of $\fedEst$ relies on selection consistency of the comparable source sites under the oracle criteria above, ensuring that post-selection inference with the variance estimator $\widehat{\mc V}_{t,a}^{\text{fed}}$ is valid. This oracle property is pointwise in the underlying data-generating distribution. Under local alternatives where source-target discrepancies are of order $n^{-1/2}$, post-selection uncertainty may not be asymptotically negligible, particularly when the number of observed events is small. 

\begin{remark}[Double robustness of the local estimator]\label{rmk:DR-fed}
    For consistency of the local estimator $\skEst$ under site-$k$ CCOD, we show that it enjoys a double robustness property: at any fixed time point $t$, $\skEst$ is consistent if either (i) the target-site conditional survival model $S^0$ is correctly specified, or (ii) the other nuisance functions $G^k$, $\pi^k$, and $\omega^{k,0}$ are correctly specified. {Here, double robustness concerns the consistency of the local estimator, whereas the efficiency gain arises from the federated aggregation in Theorem \ref{thm:RAL-fed}. Thus, when either condition above holds, site $k$ can provide a consistent source estimator that may be validly borrowed by FED.} A technical version and proof of this remark can be found in Remark D.6 of Appendix \ref{subapp:source-est}.
\end{remark}

\section{Extension to Other Causal Contrasts}\label{sec:estimands}

While our primary focus has been on estimating the treatment-specific survival functions of a target region in the AMP trials, the proposed framework readily extends to other causal estimands of clinical relevance. In the context of the AMP trials, these include several widely used and interpretable contrasts: the survival difference (SD), the survival ratio (SR) at a fixed follow-up time, and the restricted mean survival time (RMST) difference. These functionals provide complementary perspectives on treatment efficacy and enhance the clinical interpretability of our results.

At a fixed $t \in [0,\tau]$, define the SD of the target site at time $t$ as $\delta^0(t) = \theta^0(t,1) - \theta^0(t,0).$ This estimand represents the difference in survival probabilities at time $t$ between the treatment and control groups in the target site. Similarly, define the SR of the target site at time $t$ as $\rho^0(t) = {\theta^0(t,1)}/{\theta^0(t,0)},$ which captures the treatment effect on the survival functions on a ratio scale. Prevention efficacy, defined as $PE(t)=1-F_1(t)/F_0(t)$ with $F_a(t)=1-\theta^0(t,a)$, is another efficacy measure and can be accommodated analogously.

The target-site RMST under treatment $a$ with a pre-specified horizon $\tau<\infty$ is 
$\Theta^0(\tau,a) = \displaystyle\int_0^\tau \theta^0(t,a)\, dt,$
which represents the expected event-free time up to $\tau$ in the target site under treatment $a$. The causal effect of treatment can be summarized by the RMST difference: $\Delta^0(\tau) = \Theta^0(\tau,1) - \Theta^0(\tau,0).$

Both SD and RMST are linear functionals of the treatment-specific survival curves $(\theta^0(t,0), \theta^0(t,1))$. Hence, their EIFs are linear functionals of the corresponding treatment-specific EIFs for the survival curves. Specifically, let $\varphi^{*\theta}_{t,a}(\mc O;\Px)$ denote the EIF of $\theta^0(t,a)$ under a given setting (e.g., the target-only estimator). The EIFs of SD and RMST are given by, respectively,
\begin{align}\label{eq:eif-SD} 
\varphi^{*\delta}_t(\mc O;\Px) = \varphi^{*\theta}_{t,1}(\mc O;\Px)-\varphi^{*\theta}_{t,0}(\mc O;\Px),\text{ and } \varphi^{*\Theta}_{\tau,a}(\mc O;\Px) = \int_0^\tau \varphi^{*\theta}_{t,a}(\mc O;\Px) dt. 
\end{align} 
The EIF of the RMST difference is then given by
\begin{align}\label{eq:eif-RMST-diff} 
\varphi^{*\Delta}_\tau(\mc O;\Px) & = \int_0^\tau \{\varphi^{*\theta}_{t,1}(\mc O;\Px)-\varphi^{*\theta}_{t,0}(\mc O;\Px)\} dt = \int_0^\tau \varphi^{*\delta}_t(\mc O;\Px) dt. 
\end{align}
These representations extend SD and RMST estimation to the proposed CCOD and federated methods. Although SR is nonlinear, the functional delta method gives its EIF as
\begin{align}\label{eq:eif-SR}
    \varphi^{*\rho}_t(\mc O;\Px) & = \varphi^{*\theta}_{t,1}(\mc O;\Px)\frac{1}{\theta^0(t,0)}-\varphi^{*\theta}_{t,0}(\mc O;\Px)\frac{\theta^0(t,1)}{\theta^0(t,0)^2}. 
\end{align}
Therefore, for the CCOD estimators of the SD, RMST, RMST difference and SR, the corresponding EIFs can be obtained by substituting the treatment-specific EIFs from Theorem~\ref{thm:CCOD-EIF} into \eqref{eq:eif-SD}--\eqref{eq:eif-SR}. Solving the estimating equations defined by these EIFs yields the CCOD estimators for these estimands. Variance estimates are then obtained as the empirical mean squares of the corresponding estimated EIFs. 

For the federated estimators of these estimands, one can first compute the density-ratio-adjusted estimates at each local site by substituting the site-specific influence functions from Theorem~\ref{thm:site-k-EIF} into \eqref{eq:eif-SD}--\eqref{eq:eif-SR} to obtain the corresponding local influence functions. These local influence functions are then incorporated into the federated weighting scheme \eqref{eq:fed-obj} to derive the site-specific weights. Finally, the federated estimators for SD, SR, RMST, and RMST difference are obtained as the weighted averages of their corresponding local estimates. The theoretical results in Section \ref{subsec:FED-theory} are pointwise in $(t,a)$; a full uniform-in-time theory for integrated functionals such as RMST is left as future work. 

\begin{remark}
    For each causal estimand, we recompute its federated weights from \eqref{eq:fed-obj}. Although reusing the influence functions of the federated treatment-specific survival curves would avoid recalculating weights, this distinction is important: as shown in Section~\ref{subsec:FED-theory}, the federated weights derived from \eqref{eq:fed-obj} are specifically optimized to minimize the asymptotic variance of a given estimand, ensuring efficiency within the class of weighted local estimators. Consequently, the optimal weights for one estimand (e.g., the survival curve) may not be optimal for others (e.g., RMST, SD, or SR). Moreover, the survival-curve weights are time-varying, making them difficult to directly adapt for alternative estimands. 
\end{remark}

\section{Application to the AMP Trials}\label{sec:data}

We revisit the two coordinated AMP trials, HVTN 704/HPTN 085 and HVTN 703/HPTN 081 \citep{corey2021two}. The primary analyses, conducted separately for each trial and in a pooled population, showed that VRC01 (the bnAb treatment combining the low- and high-dose groups) did not significantly reduce overall HIV-1 acquisition. However, the pooled results provided compelling proof-of-concept, demonstrating approximately 75\% {(95\% CI: 45.5\% to 88.9\%)} protection against HIV-1 isolates highly sensitive to the antibody (VRC01 80\% inhibitory concentration $[\text{IC}_{80}]<1\,\mu$g/mL). 

What remains missing from the primary analysis is precise, region-specific causal survival inference that can (i) estimate treatment effects across geographic regions with heterogeneous circulating strains, while remaining applicable to settings in which individual-level data sharing is restricted, and (ii) provide causal contrasts such as the SR and RMST difference. 


\subsection{Pre-analysis diagnostics and model checking}\label{subsec:data-diag}

We conduct diagnostics for covariate balance and positivity of the estimated nuisance functions. Conditional survival and censoring were estimated using an ensemble of Kaplan--Meier, Cox regression, and survival random forest models implemented via the \texttt{survSuperLearner} package \citep{westling2024inference}. {We used the ensemble rather than selecting a single learner to reduce reliance on any one working model; this choice is not required by the proposed framework.} The propensity score and density ratio models for the federated method were fitted using logistic regression (GLM). Because individual-level AMP data were jointly available in this analysis, the density ratio was estimated directly using logistic regression with pooled target- and source-region covariates. Under data-sharing constraints, the exponential-tilt implementation in Remark \ref{rmk:dens-ratio} can instead be used with shared target-site summary statistics. All nuisance models included age (years), weight (kg), BMI, and the standardized ML risk score as predictors, and were fitted using 5-fold cross-fitting.

Figure \ref{fig:ps-love}(A) shows that, after weighting by the estimated site-specific treatment propensity scores, all candidate propensity score models achieve excellent covariate balance, with absolute standardized mean differences well below the conventional threshold of 0.1 \citep{austin2015moving} and, in most cases, close to 0. Among them, GLM performs best. Figure \ref{fig:ps-love}(B) shows that, within each region, the estimated treatment propensity score distributions for the two treatment groups exhibit good overlap across all fitted models, with no strong evidence of propensity scores near 0 or 1. These suggest that treatment positivity is not a major concern in this randomized trial setting. Based on the covariate balance results and the stability of the estimated propensity score distributions, we use GLM as the treatment propensity score model in the main analysis. 

\begin{figure}[ht]
    \centering
    \includegraphics[width=\linewidth]{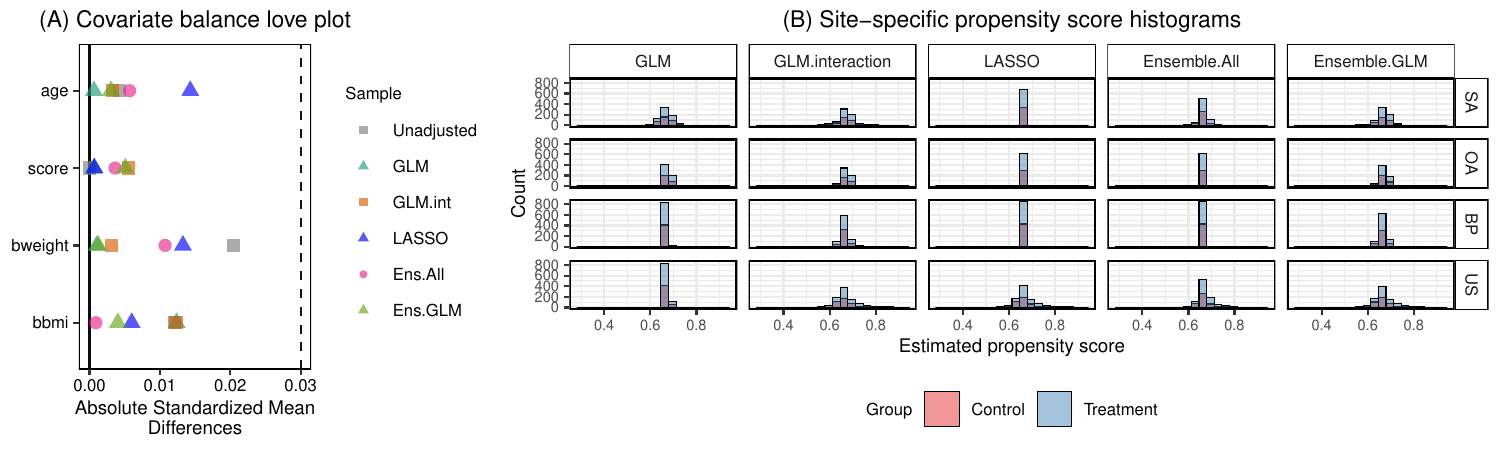}
    \caption{Covariate balance and estimated site-specific treatment propensity score distributions by different models.}
    \label{fig:ps-love}
\end{figure}

Because treatment was randomized, unmeasured confounding of treatment assignment is likewise not a concern. We conducted additional diagnostics to assess the sensitivity to potential violations of the non-informative censoring assumption (Assumption \ref{asp:inde-cen}). These analyses target the pooled AMP treatment comparison adjusted for measured baseline covariates and study region, rather than the federated estimator itself. Using the selection-bias framework of \citet{smith2019bounding}, this pooled adjusted treatment association corresponded to a selection-bias E-value of 1.50. Under the equal-strength calibration of \citet{smith2019bounding}, which assigns a common magnitude to the associations of an unmeasured factor with HIV-1 acquisition and remaining uncensored across treatment groups, an unmeasured factor would need to be associated with both processes by approximately 1.50 within each treatment arm to move the observed point estimate to the null. For comparison, the hazard ratios relating the measured baseline covariates to censoring after adjustment for study region ranged from 0.957 to 1.055, substantially closer to the null value of 1. Therefore, the E-value of 1.50 represents a considerably stronger association than those observed for the measured covariates. We additionally conducted a $\delta$-adjusted multiple-imputation sensitivity analysis for non-informative censoring for the same pooled, covariate- and region-adjusted treatment comparison; full results from both analyses are provided in Appendix \ref{subapp:supp-sens-censoring}. 

\begin{figure}[ht]
    \centering
    \includegraphics[width=\textwidth]{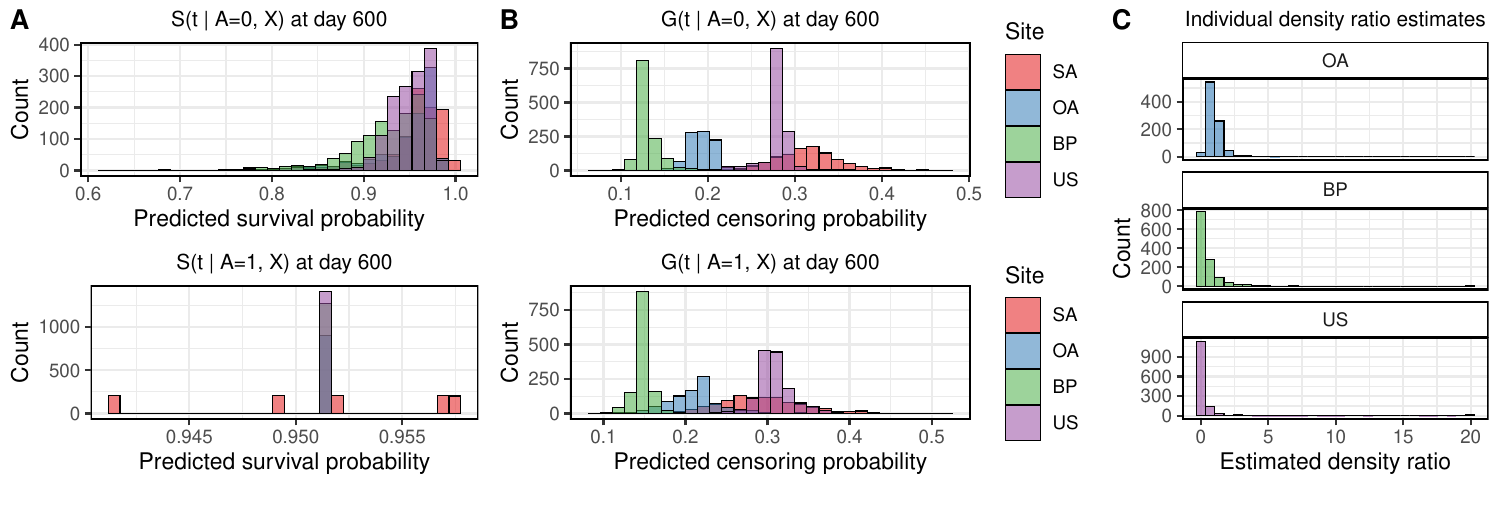}
    \caption{Histograms of individual predicted survival, censoring, and density-ratio nuisance functions for positivity assessment in the AMP analysis. }
    \label{fig:nuis-posit}
\end{figure}


We next examine the distributions of the estimated conditional survival and censoring at day 600, together with the estimated density ratios (Figure \ref{fig:nuis-posit}). We evaluate day 600 rather than day 601 because day 601 is the administrative end of follow-up. The relevant positivity condition at the administrative endpoint concerns $G^k(601-\mid a,\mb X)$, i.e., the censoring survival immediately prior to day 601. On the daily time grid used in our analysis, the day-600 diagnostic provides an empirical assessment of this pre-boundary behavior. 

Figure \ref{fig:nuis-posit}(A) shows that the individual predicted survival probabilities at day 600 remain high and are clearly bounded away from 0 across regions, which is consistent with the relatively low HIV incidence in the AMP trials. Figure \ref{fig:nuis-posit}(B) shows that the predicted censoring survival probabilities at day 600 are also bounded away from 0, suggesting that instability due to near-zero censoring probabilities is unlikely to be severe in our data. Finally, Figure \ref{fig:nuis-posit}(C) shows that the estimated source-to-target density ratios are right-skewed, especially for BP and US, indicating that some source individuals from these regions are less comparable to the target region. However, the estimated density ratios remain within a moderate range (all $\leq 20$). Overall, these diagnostics support the plausibility of the positivity assumption. 

\subsection{Main analysis}\label{subsec:data-main}

We compare the main federated estimator (FED) with several alternatives: the target-only estimator (TGT), the CCOD estimator (Section \ref{subsec:CCOD}), pooling (POOL), inverse-variance weighting (IVW), and clustered Cox PH model (CLCOX). The TGT estimator uses only the target-site data ($R=0$; see also Figure \ref{fig:AMP_siteSurvs}) and is defined in Section \ref{subsec:tgtonlyest}. The CCOD estimator serves as an efficiency benchmark, representing the maximal gain attainable under ideal conditions where CCOD holds and full data pooling is allowed. POOL na\"ively aggregates all site data and applies the TGT estimator to the combined dataset. IVW aggregates site-specific estimators using weights proportional to the inverse of their estimated variances. Finally, CLCOX fits a pooled Cox PH model stratified by region and uses the resulting model-based treatment-specific survival estimates as a simple benchmark.

\begin{figure}[ht]
    \centering
    \includegraphics[width=0.88\textwidth]{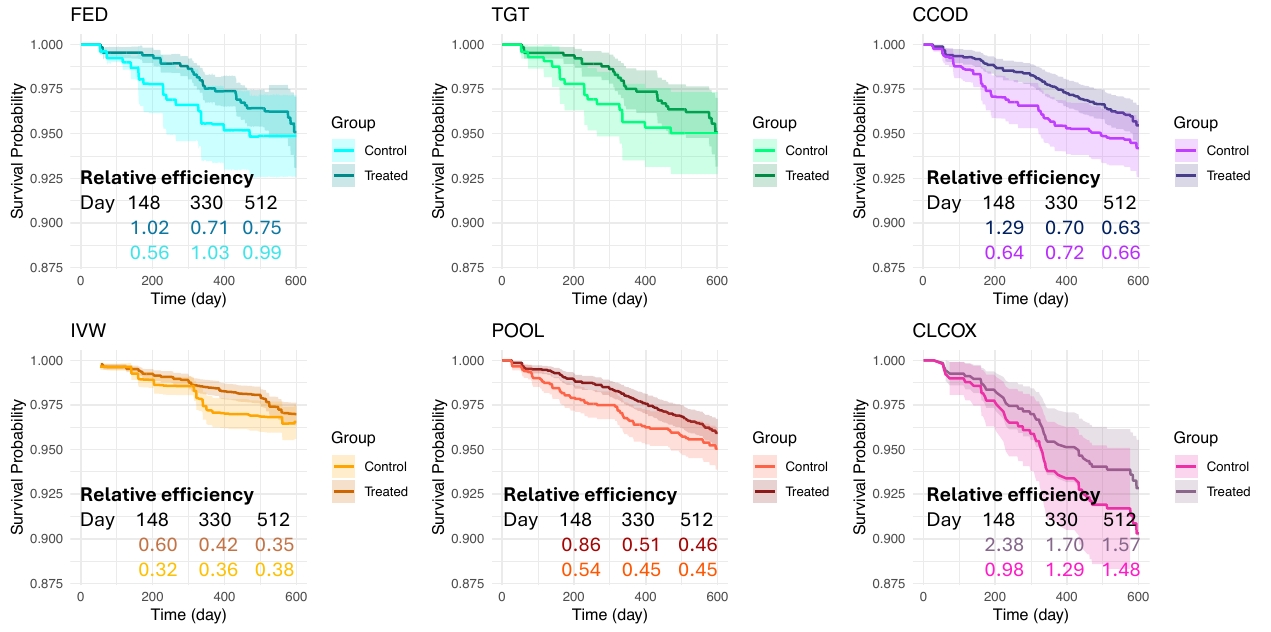}\\
    \textbf{\small (A) Estimated treatment-specific survival curves}
    \includegraphics[width=0.72\textwidth]{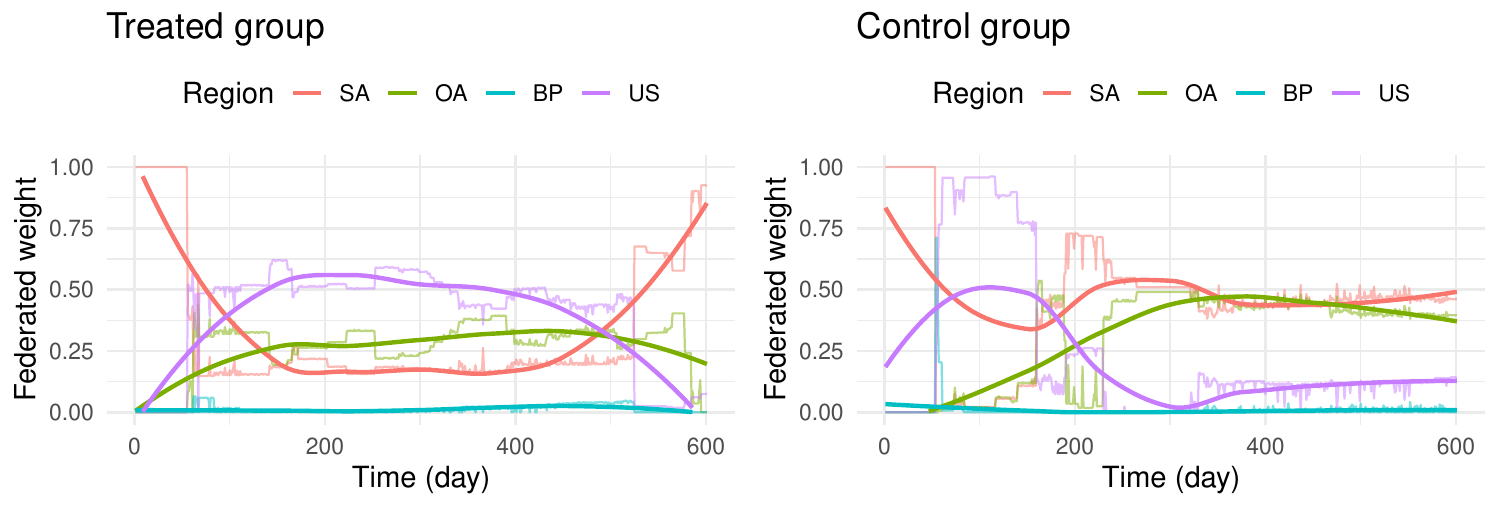}\\
    \textbf{\small (B) Treatment- and time-specific federated weights} 
    \caption{Data analysis results with women in South Africa (SA) treated as the target region. The smoothed weight curves in Panel (B) are obtained using locally weighted regression for visualization only and are not part of the proposed methodology.}
    \label{fig:AMP-SA}
\end{figure}

In this section, we focus on the SA (women in South Africa) region as our target site (Figure \ref{fig:AMP-SA}). Results where other regions are taken as the target are provided in Appendix \ref{subapp:other-targets}. Panel (A) shows that the TGT and FED methods yield nearly identical survival curves. However, TGT can have wider CIs, reflecting lower estimation efficiency. It also fails to yield CIs at certain early time points because of unstable or unavailable variance estimates due to the limited number of observed events at those times. In contrast, FED recovers interval estimates at many of these time points by borrowing external information, with CIs that are generally narrower than those from TGT. CCOD also generates comparable curves, with precision gains that vary across treatment arms and time points.

To assess the efficiency gains of FED and CCOD, we computed a ``relative efficiency'' metric at selected time points, defined as the ratio of the estimated standard error of each method to that of TGT. We evaluated the relative efficiency at days 148, 330, and 512, as shown in Panel (A). The results are broadly consistent with the theoretical considerations in Sections \ref{subsec:CCOD} and \ref{subsec:FED}, with FED showing overall moderate efficiency gains and CCOD showing larger gains at most time points. By comparison, the IVW and POOL methods, while exhibiting higher nominal efficiency (lower relative efficiencies), exhibit a clear discrepancy from TGT, reflecting bias that can arise due to distribution shifts when targeting the SA population. 

We note that POOL and IVW are included as commonly used reference approaches, although they target the overall population rather than the target-site estimand. Thus, differences in their results are expected when regional heterogeneity is present. Together with TGT, they help illustrate the trade-off between efficiency and validity for target-site inference: TGT is consistent for the target-site estimand but may be less efficient, whereas POOL and IVW may be more efficient but are not valid for the target-site estimand. This trade-off also motivates the proposed FED approach, which is designed to preserve consistency for the target-site estimand while achieving efficiency gains via adaptive borrowing.

\begin{table}[ht]
\centering
\singlespacing
\footnotesize
\caption{Estimated SD and SR at days 148, 330, and 512.}
\label{tab:SD-AMP}
\begin{threeparttable}
\begin{tabular}{rrcccccc}
\toprule
Day & Method 
& SD Est. (95\% CI) & SE(SD) & p-value
& SR Est. (95\% CI) & SE(SR) & p-value \\
\midrule
\multirow{3}{*}{148}
& TGT  
& 0.008 (-0.006, 0.021) & 0.007 & 0.257
& 1.008 (0.994, 1.022) & 0.007 & 0.260 \\
& FED  
& 0.009 (0.001, 0.016) & 0.004 & 0.028
& 1.009 (1.001, 1.016) & 0.004 & 0.027 \\
& CCOD 
& 0.008 (-0.001, 0.018) & 0.005 & 0.094
& 1.008 (0.998, 1.018) & 0.005 & 0.097 \\
\midrule
\multirow{3}{*}{330}
& TGT  
& 0.018 (-0.005, 0.041) & 0.012 & 0.125
& 1.019 (0.994, 1.043) & 0.012 & 0.131 \\
& FED  
& 0.019 (-0.002, 0.040) & 0.011 & 0.080
& 1.020 (0.998, 1.042) & 0.011 & 0.079 \\
& CCOD 
& 0.019 (0.004, 0.035) & 0.008 & 0.014
& 1.020 (1.004, 1.036) & 0.008 & 0.015 \\
\midrule
\multirow{3}{*}{512}
& TGT  
& 0.013 (-0.014, 0.041) & 0.014 & 0.339
& 1.014 (0.985, 1.043) & 0.015 & 0.344 \\
& FED  
& 0.015 (-0.011, 0.041) & 0.013 & 0.247
& 1.016 (0.989, 1.043) & 0.014 & 0.234 \\
& CCOD 
& 0.016 (-0.001, 0.034) & 0.009 & 0.070
& 1.017 (0.998, 1.036) & 0.010 & 0.074 \\
\bottomrule
\end{tabular}
\begin{tablenotes}\scriptsize
\item Est.: Estimate; SE: Standard error; CI: confidence interval. The p-value is for testing the null hypothesis of no treatment effect in the target region (SD = 0 or SR = 1).
\end{tablenotes}
\end{threeparttable}
\end{table}

\begin{table}[ht]
\centering
\singlespacing
\footnotesize
\caption{Estimated RMST by treatment group and RMST difference up to day 601.}
\label{tab:RMST-AMP}
\begin{tabular}{rrcccc}
\toprule
& Method & RMST Est. & SE & 95\% CI & p-value \\
\midrule
& TGT  & 582.66 & 4.70 & (573.45, 591.86) & -- \\
Control group
& FED  & 582.51 & 4.73 & (573.24, 591.78) & -- \\
& CCOD & 581.10 & 4.22 & (572.83, 589.36) & -- \\
\midrule
& TGT  & 589.88 & 2.38 & (585.21, 594.55) & -- \\
Treated group
& FED  & 589.98 & 2.10 & (585.87, 594.10) & -- \\
& CCOD & 589.12 & 3.25 & (582.75, 595.48) & -- \\
\midrule
& TGT  & 7.22 & 5.26 & (-3.09, 17.53) & 0.170 \\
RMST difference
& FED  & 7.80 & 4.98 & (-1.96, 17.55) & 0.117 \\
& CCOD & 8.02 & 3.63 & (0.90, 15.14) & 0.027 \\
\bottomrule
\end{tabular}
\begin{tablenotes}\scriptsize
    \item RMST: restricted mean survival time; Est.: Estimate; SE: Standard error; CI: confidence interval. The p-value is for testing the null hypothesis of no treatment effect in the target region (RMST difference = 0). For FED, the RMST difference is estimated directly using estimand-specific federated weights and therefore need not equal the difference between the separately estimated treated- and control-arm RMSTs.
\end{tablenotes}
\end{table}

Finally, we computed the SD and SR at days 148, 330, and 512, as well as the RMST and RMST difference up to day $\tau=601$, the study time horizon, as shown in Tables \ref{tab:SD-AMP} and \ref{tab:RMST-AMP}, using the transformation techniques described in Section \ref{sec:estimands}. For both SD and SR, FED and CCOD yield narrower or comparable 95\% CIs than TGT at each assessed time point. Except for FED at day 148 and CCOD at day 330, all 95\% CIs for SD include 0 and those for SR include 1, indicating no statistically significant treatment effect at the 0.05 level. For the RMST difference, CCOD yields the smallest SEs and p-values, followed by FED and then TGT, a pattern broadly consistent with our theoretical and simulation results. Because multiple time points and causal contrasts were examined, no multiplicity adjustment was applied, and individual nominal p-values should be interpreted with appropriate caution.

Consistent with the conclusion of \cite{corey2021two} that VRC01 provides limited overall protection, our estimated survival curves, SDs, SRs, and RMST differences likewise provide no strong evidence of benefit in this population. However, FED improves precision over the target-only analysis for several treatment contrasts and helps illustrate where na\"ive pooling or IVW may be affected by regional heterogeneity.

\begin{remark}\label{rmk:MLscore}
The ML risk score used in the main analysis is an internally learned prognostic score trained on the AMP control arm rather than a directly observed baseline covariate. Although included to align with the primary analysis \citep{corey2021two}, such internally trained scores require additional care \citep{hojbjerre2026within}. We therefore repeated the analysis excluding the ML risk score and using only age, weight, and BMI. The resulting estimates changed only modestly, and the substantive conclusions remained unchanged; full results are reported in Appendix \ref{subapp:supp-sensitivity}. 
\end{remark}

\section{Simulation Studies}\label{sec:simu}

We conducted simulations comparing FED with TGT, CCOD, POOL, IVW, and CLCOX. Event times were generated on a 365-day scale, with administrative censoring at $\tau=200$ and additional covariate-dependent right censoring. Full details of the data generations, performance metrics, and results are provided in Appendix \ref{app:experiments}. {Our simulations focus on assumptions most directly related to FED, including cross-site distribution shifts and limited treatment overlap.}

Across the simulation settings, FED generally provides a favorable bias--variance trade-off. TGT remains valid but can be less precise because it uses only target-site data, whereas POOL and IVW can exhibit substantial bias under distributional heterogeneity. By adaptively borrowing information from source sites, FED reduces RMSE relative to TGT in many settings, with reductions of up to approximately 23\% in the main simulations. CCOD can provide additional precision when CCOD assumption holds, but may be biased when it is violated. CLCOX is generally more sensitive to distributional heterogeneity. Under censoring shift alone, most methods remain relatively stable because the conditional outcome distribution is unchanged.

Some finite-sample departures from the asymptotic behavior are observed, particularly when event information is limited. This can lead to occasional under- or over-coverage, with over-coverage occurring mainly at early time points and under-coverage in a small number of scenarios. In the poorer-overlap setting, FED reduces RMSE by up to about 67\% relative to TGT, although coverage can fall below the nominal level in some cases. Results for other causal contrasts show generally reasonable finite-sample performance, with small to moderate bias in some scenarios. Overall, these findings reflect a bias--variance trade-off rather than uniform dominance across all metrics. Increasing source-site sample sizes yields only modest additional gains for FED because the estimator remains target-anchored and nuisance functions are estimated locally; thus, larger source samples improve source-specific precision but need not translate into proportional gains in the final federated estimator.

\section{Concluding Remarks}\label{sec:conclude}

Inspired by the AMP trials, we have developed data fusion approaches for causal survival analysis. Our methods aim to improve estimation efficiency for survival curves and their contrast functionals in the target region while accommodating cross-region heterogeneity and respecting privacy constraints. 

To evaluate the maximal efficiency attainable in this setting, we first consider a natural extension of the doubly robust single-site estimator \citep{westling2024inference}, namely, the CCOD estimator. The CCOD estimator is consistent and semiparametrically efficient when the CCOD assumption holds, data can be pooled, and nuisance functions converge to their true limits. Because these conditions are often restrictive, we develop a federated (FED) approach that relaxes them while maintaining valid inference across heterogeneous regions. FED accommodates shifts in covariate, outcome, and censoring distributions, requires only summary-level exchange, and achieves strict efficiency gains under oracle site selection and regularity conditions. Importantly, borrowing is not fixed a priori but learned adaptively from the data: FED anchors on the target-site estimator and assigns source-site weights according to their empirical alignment with the target, downweighting sources whose local estimates exhibit greater discrepancy from the target-site estimate while allowing greater contribution from sources that appear more compatible and can improve precision. 

As a practical guideline, our method is best suited for multi-site or multi-region time-to-event studies in which inference is desired for a specific target population, the target site has adequate sample size, and the number of clinically meaningful regions is fixed and moderate. This is important because the federated estimator is anchored on the target-site analysis, and external sites are used primarily to improve precision and stability rather than to replace the target-site signal. Potential applications include multi-regional phase 3 clinical trials, pragmatic trials across health systems, and observational studies with survival outcomes subject to right-censoring. When meaningful regional heterogeneity is anticipated, future studies will benefit from collecting baseline covariates and the study outcome consistently across regions, so that investigators can better assess cross-region comparability and determine whether adaptive borrowing is appropriate. More broadly, when pooled analyses may obscure important heterogeneity, it is beneficial to prespecify region-specific estimands in statistical analysis plans. Practical implementation also involves a trade-off between robustness and reliability, especially in small regional samples: more flexible learners may reduce bias from model misspecification in larger samples, but they can be unstable when the sample size or number of observed events is limited, and cross-validation risk estimates may be noisy. Extending the method and its guarantees to settings with a substantially larger number of regions is an important direction for future work.

Several limitations highlight directions for future research. First, although we demonstrate clear efficiency gains from FED, additional improvements may be obtained through more efficient covariate-adaptive weighting strategies \citep{li2023efficient}. Second, even when individual-level data sharing is feasible, it remains unclear under CCOD violation (or without certain cross-site homogeneity assumptions) how to use individual-level source-site information to construct an estimator that improves efficiency over both the target-only and FED estimators while preserving estimation consistency for the target-site estimand. Even when de-identified individual-level data are available, sharing them across regions may still raise privacy, regulatory, or data-use concerns. The proposed FED framework requires only summary-level quantities to be exchanged, offering an added benefit over methods that rely on access to individual-level data across sites, while remaining anchored on the target-site estimand. Although the proposed framework avoids transfer of individual-level data, it does not provide a formal privacy guarantee, e.g., in the sense of differential privacy. In particular, summary-level quantities may still carry some disclosure risk. Developing federated survival procedures with formal privacy guarantees is an important direction for future research. Third, the time-specific FED weights, while flexible, may yield non-smooth trajectories and higher computational cost in continuous-time settings; developing smoothing procedures to better capture temporal patterns would enhance stability and scalability. Considering subject-specific weighting for borrowing subjects within each source is another promising direction, as individuals, even within the same source, may exhibit substantial heterogeneity and therefore contribute to the target population in different ways \citep{gao2025improving, zhu2025enhancing}. Fourth, extending the framework to incorporate time-varying covariates (including the censoring process) could further improve efficiency by exploiting dynamic predictive information beyond baseline. Developing such extensions within the present continuous-time federated framework remains an important direction for future work. Unifying these ideas for continuous time remains an open challenge. Fifth, our analysis relies on the standard no-interference assumption, namely that one participant’s treatment assignment does not affect another participant’s outcome. Although this assumption is commonly invoked in causal inference, it may be questionable in infectious disease prevention settings, where spillover or contamination effects could arise through changes in exposure networks or indirect protection. Sixth, because HIV-1 infection was ascertained at scheduled study visits, the underlying infection time is subject to interval censoring; consistent with the primary AMP analyses, we use the recorded diagnosis time as the event time and view extension to interval-censored outcomes as future work. 

Lastly, while our framework naturally extends to several analytically tractable causal contrasts of the treatment-specific survival curves, including the SD, SR, and RMST difference, further work is needed to handle more complex nonlinear functionals. Examples include the causal quantile treatment effects on survival time, which are generally non-smooth or non-collapsible \citep{boughdiri2025unified}. Extending our framework for such nonlinear estimands under the multi-source setting will be an important next step.

\subsection*{Data Availability Statement}

The AMP trials data will be made publicly available at Harvard Dataverse. The R code is available at \url{https://github.com/yiliu1998/ACSsubmissionJASA}. A user-friendly R package is available at \url{https://github.com/yiliu1998/FuseSurv}. 

\subsection*{Disclosure Statement}

The authors declare no conflicts of interest. ChatGPT was used for language polishing and assistance with simulation code parallelization, R plotting, and summarizing results. It was not used to develop the methodology, design the core simulation study, or perform the primary data analyses. The authors take full responsibility for the final product.

\subsection*{Acknowledgments}

This work was supported by the National Heart, Lung, and Blood Institute of the National Institutes of Health (NIH) under Award T32HL079896. The content is solely the responsibility of the authors and does not necessarily represent the official views of NIH.

\bibliography{_refs}

@inproceedings{liuprivacy,
  title={Privacy-Protected Causal Survival Analysis Under Distribution Shift},
  author={Liu, Yi and Levis, Alexander W and Zhu, Ke and Yang, Shu and Gilbert, Peter B and Han, Larry},
  booktitle={The Fourteenth International Conference on Learning Representations},
  year={2026}
}

@article{mkhize2023neutralization,
  title={Neutralization profiles of {HIV-1} viruses from the {VRC01} Antibody Mediated Prevention ({AMP}) trials},
  author={Mkhize, Nonhlanhla N and Yssel, Anna EJ and Kaldine, Haajira and van Dorsten, Rebecca T and Woodward Davis, Amanda S and Beaume, Nicolas and Matten, David and Lambson, Bronwen and Modise, Tandile and Kgagudi, Prudence and others},
  journal={PLoS Pathogens},
  volume={19},
  number={6},
  pages={e1011469},
  year={2023},
  publisher={Public Library of Science San Francisco, CA USA}
}

@article{rubin1980randomization,
  title={Randomization Analysis of Experimental Data: The {Fisher} Randomization Test Comment},
  author={Rubin, Donald B.},
  journal={Journal of the American Statistical Association},
  volume={75},
  number={371},
  pages={591--593},
  year={1980},
  doi={10.2307/2287653}
}

@article{austin2015moving,
  title={Moving towards best practice when using inverse probability of treatment weighting {(IPTW)} using the propensity score to estimate causal treatment effects in observational studies},
  author={Austin, Peter C and Stuart, Elizabeth A},
  journal={Statistics in Medicine},
  volume={34},
  number={28},
  pages={3661--3679},
  year={2015},
  publisher={Wiley Online Library}
}

@article{hojbjerre2026within,
  title={“Within-Trial” Prognostic Score Adjustment Is Targeted Maximum Likelihood Estimation},
  author={H{\o}jbjerre-Frandsen, Emilie and Schuler, Alejandro},
  journal={Pharmaceutical Statistics},
  volume={25},
  number={2},
  pages={e70080},
  year={2026},
  publisher={Wiley Online Library}
}

@article{yang2023multiply,
  title={Multiply robust matching estimators of average and quantile treatment effects},
  author={Yang, Shu and Zhang, Yunshu},
  journal={Scandinavian Journal of Statistics},
  volume={50},
  number={1},
  pages={235--265},
  year={2023},
  publisher={Wiley Online Library}
}

@article{boughdiri2025unified,
  title={A unified framework for the transportability of population-level causal measures},
  author={Boughdiri, Ahmed and Berenfeld, Cl{\'e}ment and Josse, Julie and Scornet, Erwan},
  journal={arXiv preprint arXiv:2505.13104},
  year={2025}
}

@article{liu2024multi,
  title={Multi-source conformal inference under distribution shift},
  author={Liu, Yi and Levis, Alexander W and Normand, Sharon-Lise and Han, Larry},
  journal={Proceedings of Machine Learning Research},
  volume={235},
  pages={31344--31382},
  year={2024}
}

@article{han2025federated,
  title={Federated adaptive causal estimation {(FACE)} of target treatment effects},
  author={Han, Larry and Hou, Jue and Cho, Kelly and Duan, Rui and Cai, Tianxi},
  journal={Journal of the American Statistical Association},
  volume={120},
  pages={1503--1516},
  year={2025},
  publisher={Taylor \& Francis}
}

@article{makhija2024federated,
  title={Federated Learning for Estimating Heterogeneous Treatment Effects},
  author={Makhija, Disha and Ghosh, Joydeep and Kim, Yejin},
  journal={arXiv preprint arXiv:2402.17705},
  year={2024}
}

@article{westling2024inference,
  title={Inference for treatment-specific survival curves using machine learning},
  author={Westling, Ted and Luedtke, Alex and Gilbert, Peter B and Carone, Marco},
  journal={Journal of the American Statistical Association},
  volume={119},
  number={546},
  pages={1541--1553},
  year={2024},
  publisher={Taylor \& Francis}
}

@article{corey2021two,
  title={Two randomized trials of neutralizing antibodies to prevent {HIV-1} acquisition},
  author={Corey, Lawrence and Gilbert, Peter B and Juraska, Michal and Montefiori, David C and Morris, Lynn and Karuna, Shelly T and Edupuganti, Srilatha and Mgodi, Nyaradzo M and Decamp, Allan C and Rudnicki, Erika and others},
  journal={New England Journal of Medicine},
  volume={384},
  number={11},
  pages={1003--1014},
  year={2021},
  publisher={Mass Medical Soc}
}

@article{yang2020combining,
  title={Combining multiple observational data sources to estimate causal effects},
  author={Yang, Shu and Ding, Peng},
  journal={Journal of the American Statistical Association},
  volume={115},
  number={531},
  pages={1540--1554},
  year={2020},
  publisher={Taylor \& Francis}
}

@article{bai2013doubly,
  title={Doubly-robust estimators of treatment-specific survival distributions in observational studies with stratified sampling},
  author={Bai, Xiaofei and Tsiatis, Anastasios A and O'Brien, Sean M},
  journal={Biometrics},
  volume={69},
  number={4},
  pages={830--839},
  year={2013},
  publisher={Wiley Online Library}
}

@article{han2023multiply,
  title={Multiply robust federated estimation of targeted average treatment effects},
  author={Han, Larry and Shen, Zhu and Zubizarreta, Jose},
  journal={Advances in Neural Information Processing Systems},
  volume={36},
  pages={70453--70482},
  year={2023}
}

@article{li2023targeting,
  title={Targeting underrepresented populations in precision medicine: A federated transfer learning approach},
  author={Li, Sai and Cai, Tianxi and Duan, Rui},
  journal={The Annals of Applied Statistics},
  volume={17},
  number={4},
  pages={2970--2992},
  year={2023},
  publisher={Institute of Mathematical Statistics}
}

@article{cleveland1988locally,
  title={Locally weighted regression: an approach to regression analysis by local fitting},
  author={Cleveland, William S and Devlin, Susan J},
  journal={Journal of the American Statistical Association},
  volume={83},
  number={403},
  pages={596--610},
  year={1988},
  publisher={Taylor \& Francis}
}

@article{han2024privacy,
  title={Privacy-preserving, communication-efficient, and target-flexible hospital quality measurement},
  author={Han, Larry and Li, Yige and Niknam, Bijan and Zubizarreta, Jos{\'e} R},
  journal={The Annals of Applied Statistics},
  volume={18},
  number={2},
  pages={1337--1359},
  year={2024},
  publisher={Institute of Mathematical Statistics}
}

@article{gill1990survey,
  title={A survey of product-integration with a view toward application in survival analysis},
  author={Gill, Richard D and Johansen, Soren},
  journal={The Annals of Statistics},
  volume={18},
  number={4},
  pages={1501--1555},
  year={1990},
  publisher={Institute of Mathematical Statistics}
}

@article{fan2024fast,
  title={A fast trans-lasso algorithm with penalized weighted score function},
  author={Fan, Xianqiu and Cheng, Jun and Wang, Hailing and Zhang, Bin and Chen, Zhenzhen},
  journal={Computational Statistics \& Data Analysis},
  volume={192},
  pages={107899},
  year={2024},
  publisher={Elsevier}
}

@article{wolock2024framework,
  title={A framework for leveraging machine learning tools to estimate personalized survival curves},
  author={Wolock, Charles J and Gilbert, Peter B and Simon, Noah and Carone, Marco},
  journal={Journal of Computational and Graphical Statistics},
  volume={33},
  number={3},
  pages={1098--1108},
  year={2024},
  publisher={Taylor \& Francis}
}

@article{smith2019bounding,
  title={Bounding bias due to selection},
  author={Smith, Louisa H and VanderWeele, Tyler J},
  journal={Epidemiology (Cambridge, Mass.)},
  volume={30},
  number={4},
  pages={509},
  year={2019}
}

@article{yang2023smim,
  title={{SMIM}: a unified framework of survival sensitivity analysis using multiple imputation and martingale},
  author={Yang, Shu and Zhang, Yilong and Liu, Guanghan Frank and Guan, Qian},
  journal={Biometrics},
  volume={79},
  number={1},
  pages={230--240},
  year={2023},
  publisher={Wiley Online Library}
}

@book{bickel1993efficient,
  title={Efficient and adaptive estimation for semiparametric models},
  author={Bickel, Peter J and Klaassen, Chris A J and Ritov, Ya’acov and and Wellner, Jon A},
  volume={4},
  year={1993},
  publisher={Johns Hopkins University Press Baltimore}
}

@article{li2023efficient,
  title={Efficient estimation under data fusion},
  author={Li, Sijia and Luedtke, Alex},
  journal={Biometrika},
  volume={110},
  number={4},
  pages={1041--1054},
  year={2023},
  publisher={Oxford University Press}
}

@article{lee2022doubly,
  title={Doubly robust estimators for generalizing treatment effects on survival outcomes from randomized controlled trials to a target population},
  author={Lee, Dasom and Yang, Shu and Wang, Xiaofei},
  journal={Journal of Causal Inference},
  volume={10},
  number={1},
  pages={415--440},
  year={2022},
  publisher={De Gruyter}
}

@article{cao2024transporting,
  title={Transporting randomized trial results to estimate counterfactual survival functions in target populations},
  author={Cao, Zhiqiang and Cho, Youngjoo and Li, Fan},
  journal={Pharmaceutical Statistics},
  volume={23},
  number={4},
  pages={442--465},
  year={2024},
  publisher={Wiley Online Library}
}

@article{cox1972regression,
  title={Regression models and life-tables},
  author={Cox, David R},
  journal={Journal of the Royal Statistical Society: Series B (Methodological)},
  volume={34},
  number={2},
  pages={187--202},
  year={1972},
  publisher={Wiley Online Library}
}

@article{van2006statistical,
  title={Statistical inference for variable importance},
  author={Van der Laan, Mark J},
  journal={The international journal of biostatistics},
  volume={2},
  number={1},
  year={2006},
  publisher={De Gruyter}
}

@article{kaplan1958nonparametric,
	Author = {Kaplan, Edward L and Meier, Paul},
	Journal = {Journal of the American Statistical Association},
	Number = {282},
	Pages = {457--481},
	Publisher = {Taylor \& Francis},
	Title = {Nonparametric estimation from incomplete observations},
	Volume = {53},
	Year = {1958}}

@article{zou2006adaptive,
	Author = {Zou, Hui},
	Journal = {Journal of the American Statistical Association},
	Number = {476},
	Pages = {1418--1429},
	Publisher = {Taylor \& Francis},
	Title = {The adaptive lasso and its oracle properties},
	Volume = {101},
	Year = {2006}}

@Article{Rubin1974,
  author =       {Rubin, DB},
  title =        {Estimating causal effects of treatments in randomized
and nonrandomized studies},
  journal =      {Journal of Educational Psychology},
  year =         {1974},
  volume =       {66},
  pages =        {688-701},
}

@manual{R,
  address = {Vienna, Austria},
  author = {{R Core Team}},
  organization = {R Foundation for Statistical Computing},
  title = {\text{R}: A Language and Environment for Statistical Computing},
  url = {https://www.R-project.org/},
  year = {2020}
}

@article{chernozhukov2018double,
  title={Double/debiased machine learning for treatment and structural parameters},
  author={Chernozhukov, Victor and Chetverikov, Denis and Demirer, Mert and Duflo, Esther and Hansen, Christian and Newey, Whitney and Robins, James},
  journal={The Econometrics Journal},
  year={2018},
  volume={21},
  pages={C1--C68},
  publisher={Oxford University Press Oxford, UK}
}

@book{van1996weak,
  title={Weak Convergence and Empirical Processes.},
  author={Van der Vaart, AW and Wellner, JA},
  year={1996},
  publisher={Springer \& Verlag New York}
}

@article{austin2012generating,
  title={Generating survival times to simulate {Cox} proportional hazards models with time-varying covariates},
  author={Austin, Peter C},
  journal={Statistics in Medicine},
  volume={31},
  number={29},
  pages={3946--3958},
  year={2012},
  publisher={Wiley Online Library}
}

@article{gao2025improving,
  title={Improving randomized controlled trial analysis via data-adaptive borrowing},
  author={Gao, Chenyin and Yang, Shu and Shan, Mingyang and Ye, Wenyu and Lipkovich, Ilya and Faries, Douglas},
  journal={Biometrika},
  volume={112},
  number={2},
  pages={asae069},
  year={2025},
  publisher={Oxford University Press}
}

@inproceedings{zhu2025enhancing,
  title={Enhancing statistical validity and power in hybrid controlled trials: a randomization inference approach with conformal selective borrowing},
  author={Zhu, Ke and Yang, Shu and Wang, Xiaofei},
  booktitle={Proceedings of the 42nd International Conference on Machine Learning (ICML)},
  year={2025},
}

\newpage

\begin{center}
    {\Large Online Supplemental Material \textit{for} \bigskip \\ ``Targeted Data Fusion for Region-Specific Survival Effects in the AMP HIV Prevention Trials'' }
\end{center}

\appendix

\numberwithin{figure}{section}
\numberwithin{table}{section}

The Online Supplemental Material for ``Targeted Data Fusion for Region-Specific Survival Effects in the AMP HIV Prevention Trials'' is organized into five Web Appendices.

\begin{itemize}
    \item Appendix~\ref{app:add-data} provides additional analyses of the AMP trials, including results with OA, BP, and US as target regions; sensitivity analyses assessing the non-informative censoring assumption and the inclusion of the ML risk score; and further diagnostics and sensitivity analyses of the federated estimator with SA as the target region. 
    \item Appendix \ref{subapp:notation} summarizes  related literature and the notation used throughout the main paper and this supplemental material.
    \item Appendix \ref{app:theory-CCOD} provides theoretical and technical details for the CCOD estimator.
    \item Appendix \ref{app:theory-FED} presents the theoretical development and implementation details for the federated estimator.
    \item Appendix \ref{app:experiments} reports full simulation details and results.
\end{itemize}

\newpage

\section{Additional Data Analysis Results}\label{app:add-data}




\subsection{Results when treating the other three regions (OA, BP and US) as target regions}\label{subapp:other-targets}

In Figures \ref{fig:AMP-OA}--\ref{fig:AMP-US}, we present the estimated treatment-specific survival curves and federated weights when each of the OA, BP, and US regions is taken in turn as the target population. As in the SA analysis in the main text, we apply locally weighted regression \citep{cleveland1988locally} only to smooth the displayed federated weights and facilitate visualization of their temporal patterns; this smoothing step is not part of the proposed estimation procedure.

Across the three target regions, the FED survival estimates generally track the corresponding TGT estimates while borrowing information adaptively from the remaining regions. FED also provides more stable interval estimation at some earlier follow-up times. In contrast, POOL and IVW can depart more noticeably from the target-anchored estimators, particularly for the BP and US targets, illustrating the potential consequences of ignoring regional heterogeneity. The CCOD estimator is also broadly aligned with the target-region survival patterns, although its finite-sample precision relative to TGT and FED varies across regions, treatment arms, and follow-up times.

The corresponding region-specific causal contrasts are reported in Tables \ref{tab:OA-SD-SR}--\ref{tab:US-RMST}. Estimated SDs, SRs, and RMST differences are generally positive but vary in magnitude across regions and follow-up times. The strongest early signal is observed for the BP target at day 148, whereas the corresponding estimates for the US target remain closer to the null throughout follow-up. Some FED and CCOD estimates are nominally significant at individual time points, but these findings are not uniformly reproduced across estimators or contrasts and should be interpreted in light of the multiple comparisons performed.

The estimated federated weights further illustrate the adaptive borrowing mechanism. Their magnitude and composition vary across target regions, treatment arms, and follow-up times, reflecting estimated statistical alignment with the target estimand rather than demographic or biological similarity between regions. In all cases, the estimator remains anchored to the target region while adaptively incorporating information from source regions when supported by the data.

\begin{figure}[H]
    \centering
    \includegraphics[width=\textwidth]{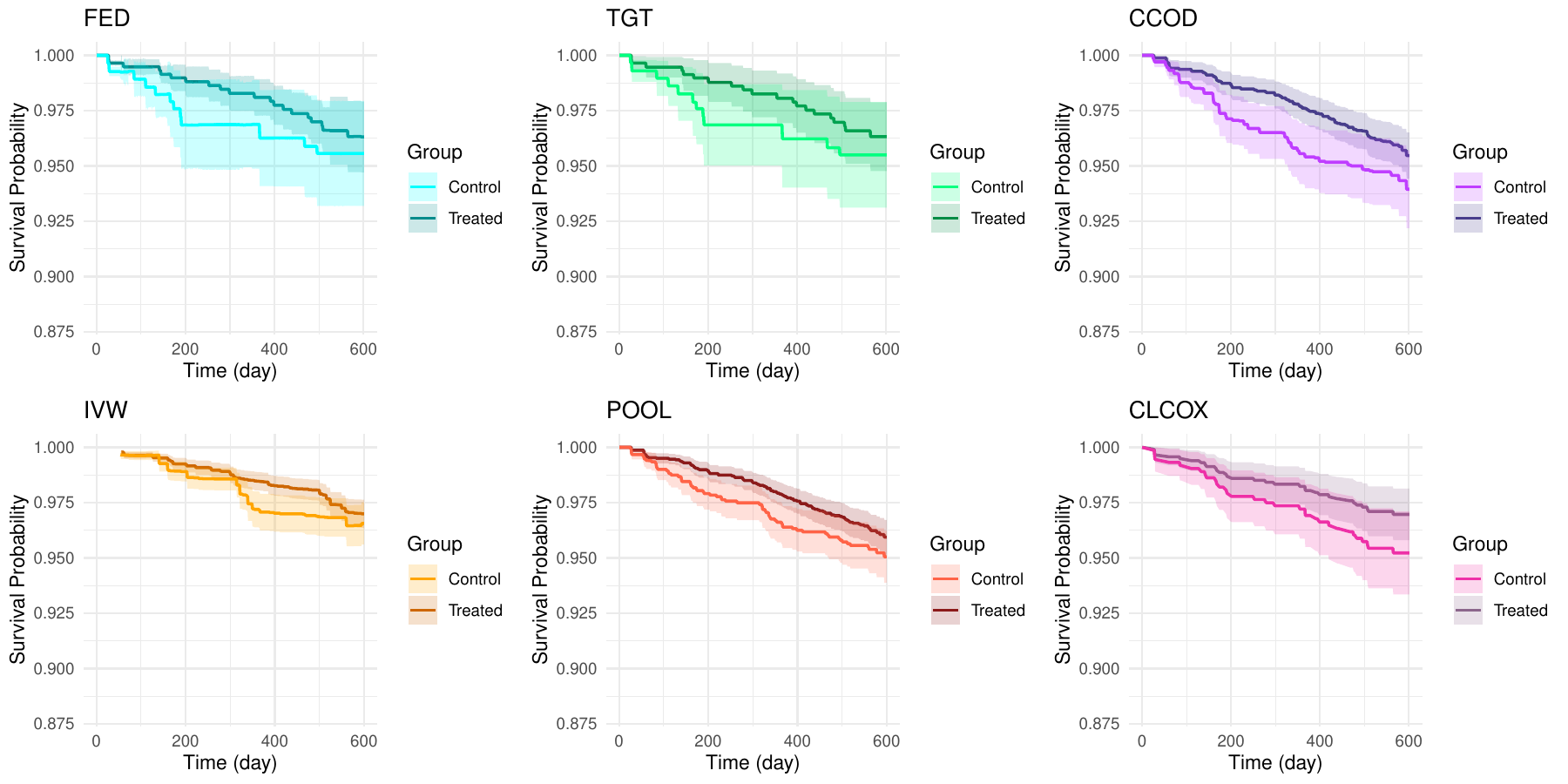}\\
    \textbf{\small (A) Estimated treatment-specific survival curves}\medskip
    
    \includegraphics[width=\textwidth]{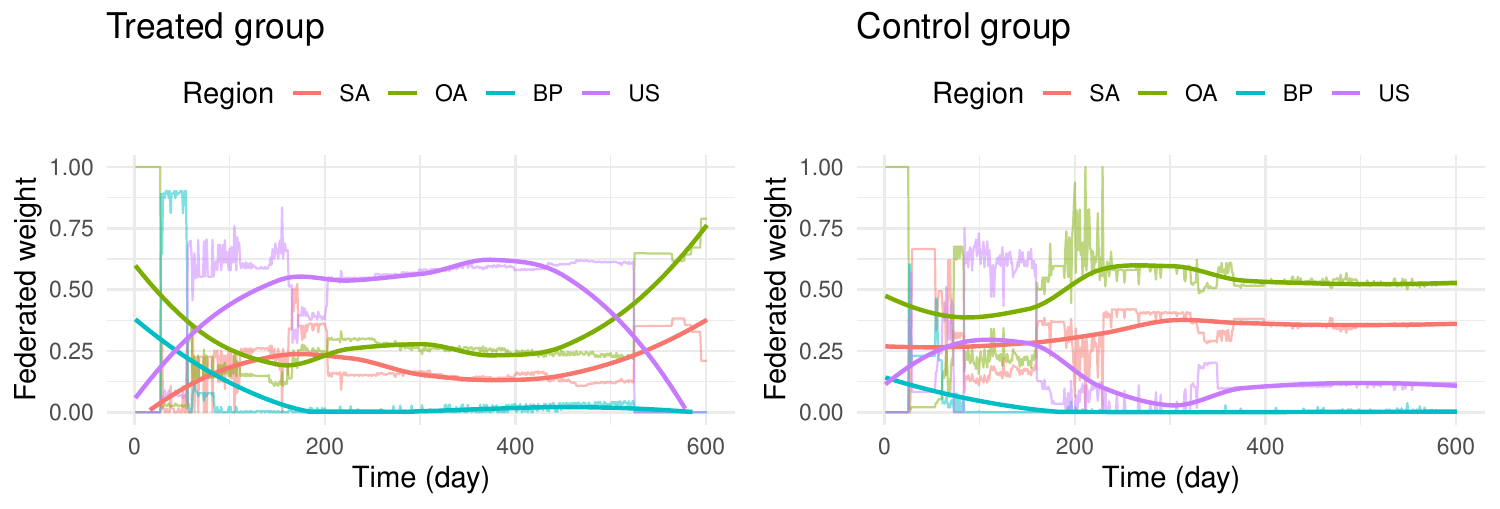}\\
    \textbf{\small (B) Treatment- and time-specific federated weights} 
    \caption{Data analysis results when treating women in sub-Saharan African countries other than South Africa (OA) as the target region. The smoothed weight curves in Panel (B) are obtained by locally weighted regression, which is only a visualization tool (not a part of our methodology).}
    \label{fig:AMP-OA}
\end{figure}

\begin{figure}[H]
    \centering
    \includegraphics[width=\textwidth]{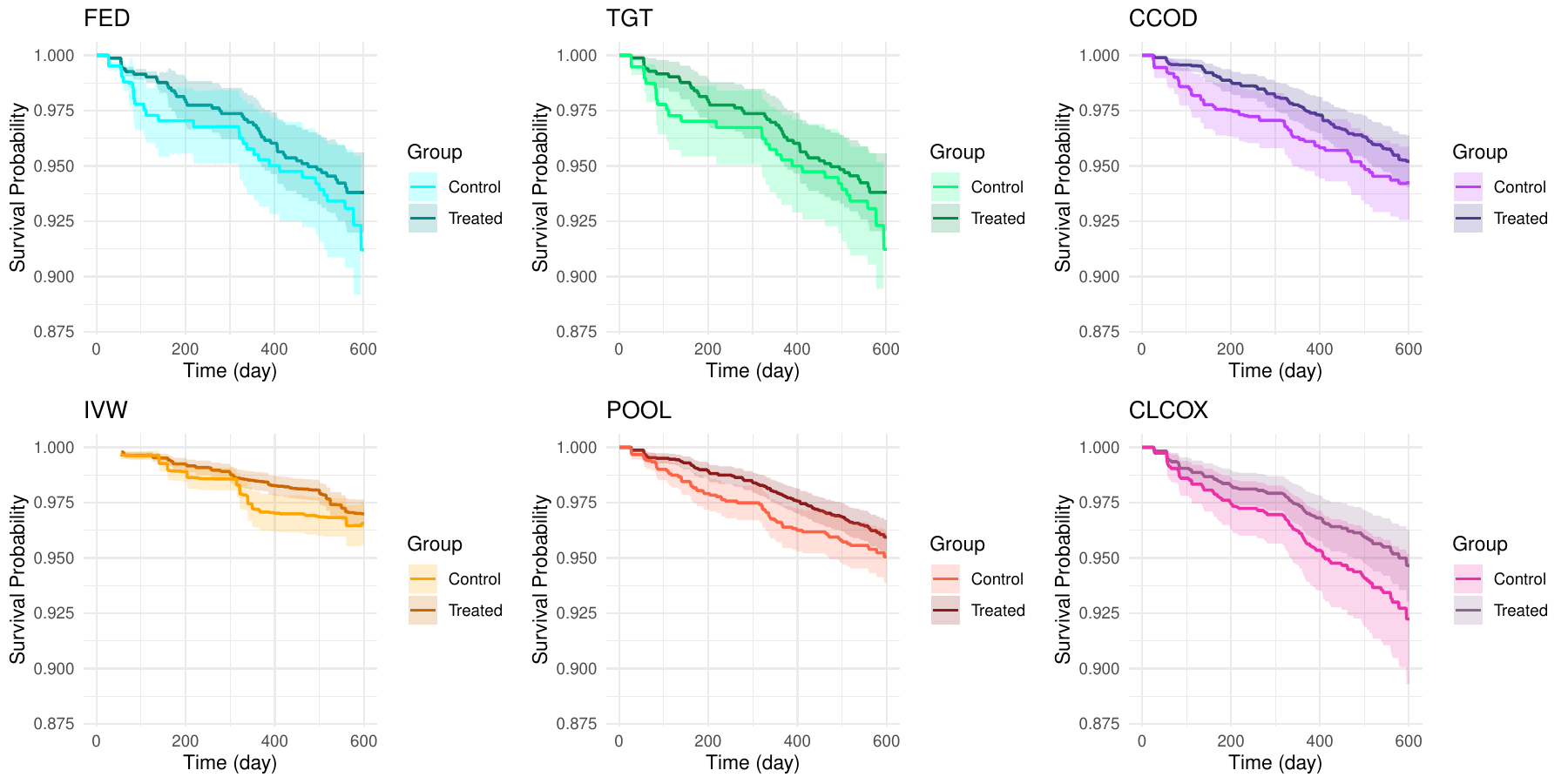}\\
    \textbf{\small (A) Estimated treatment-specific survival curves}\medskip
    
    \includegraphics[width=\textwidth]{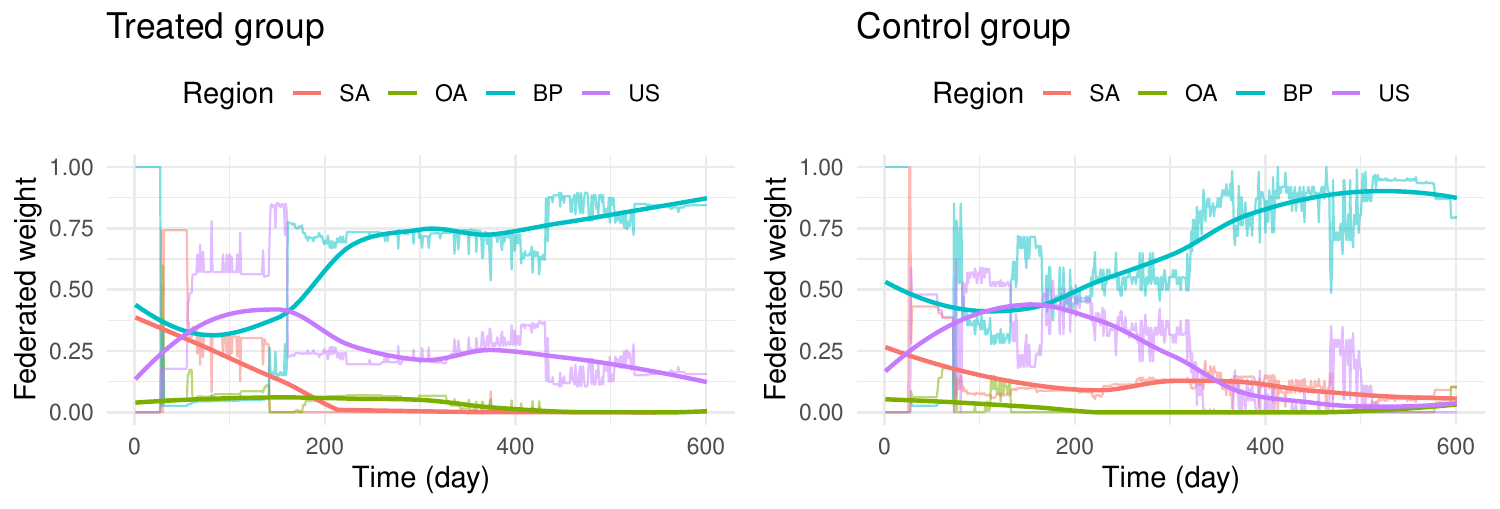}\\
    \textbf{\small (B) Treatment- and time-specific federated weights} 
    \caption{Data analysis results when treating men and transgender persons in Brazil or Peru (BP) / in the United States or Switzerland (US) as the target region. The smoothed weight curves in Panel (B) are obtained by locally weighted regression, which is only a visualization tool (not a part of our methodology).} 
    \label{fig:AMP-BP}
\end{figure}

\begin{figure}[H]
    \centering
    \includegraphics[width=\textwidth]{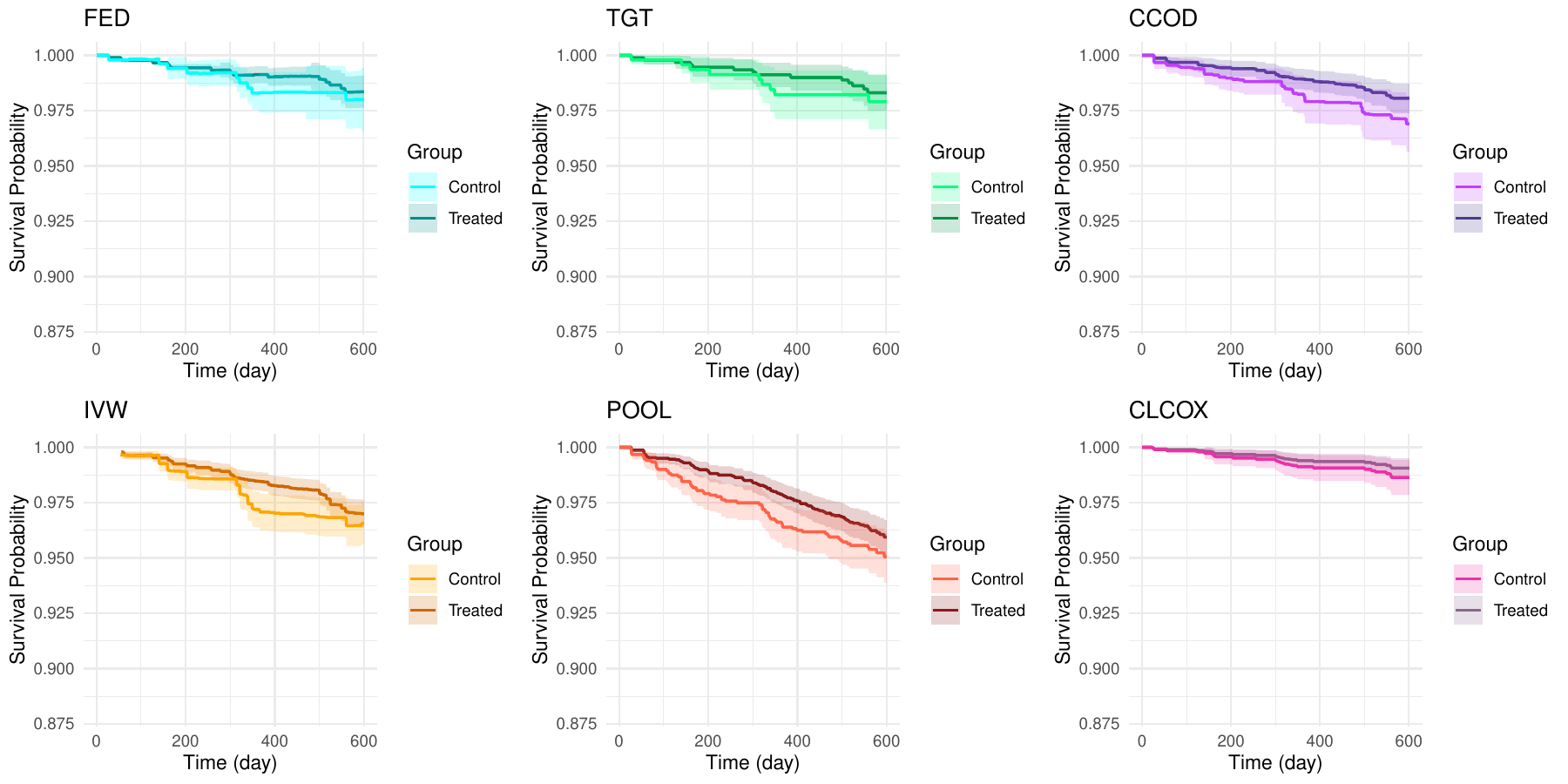}\\
    \textbf{\small (A) Estimated treatment-specific survival curves}\medskip
    
    \includegraphics[width=\textwidth]{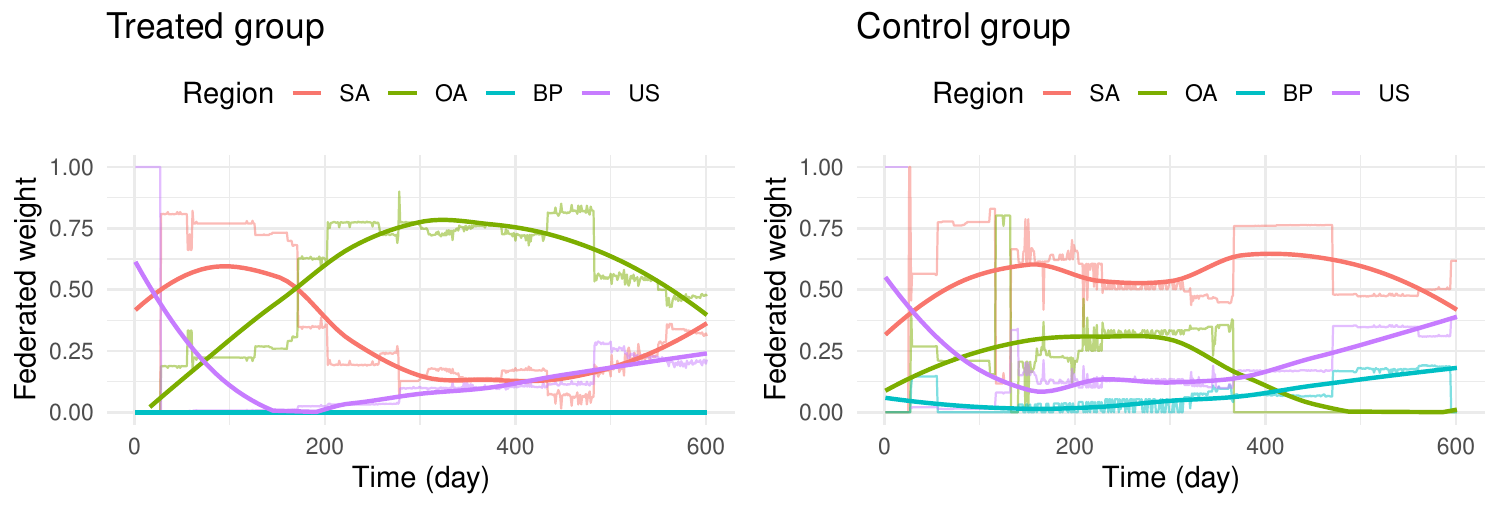}\\
    \textbf{\small (B) Treatment- and time-specific federated weights} 
    \caption{Data analysis results when treating men and transgender persons in Brazil or Peru (BP) / in the United States or Switzerland (US) as the target region. The smoothed weight curves in Panel (B) are obtained by locally weighted regression, which is only a visualization tool (not a part of our methodology). }
    \label{fig:AMP-US}
\end{figure}

\begin{table}[H]
\centering
\singlespacing
\small
\caption{Estimated survival difference (SD) and survival ratio (SR) for the OA target region at days 148, 330, and 512.}
\label{tab:OA-SD-SR}
\resizebox{\textwidth}{!}{
\begin{tabular}{c l c c c c c c}
\toprule
Day & Method & SD Est. (95\% CI) & SE(SD) & p-value
& SR Est. (95\% CI) & SE(SR) & p-value \\
\midrule
148 & TGT  & 0.008 ($-0.009$, 0.025) & 0.009 & 0.337
    & 1.008 (0.991, 1.026) & 0.009 & 0.341 \\
    & FED  & 0.009 ($-0.002$, 0.019) & 0.005 & 0.108
    & 1.009 (0.996, 1.021) & 0.006 & 0.165 \\
    & CCOD & 0.008 ($-0.001$, 0.018) & 0.005 & 0.094
    & 1.008 (0.999, 1.018) & 0.005 & 0.096 \\
\midrule
330 & TGT  & 0.013 ($-0.010$, 0.036) & 0.012 & 0.262
    & 1.013 (0.990, 1.037) & 0.012 & 0.267 \\
    & FED  & 0.013 ($-0.008$, 0.035) & 0.011 & 0.232
    & 1.014 (0.991, 1.036) & 0.011 & 0.230 \\
    & CCOD & 0.019 (0.004, 0.035) & 0.008 & 0.014
    & 1.020 (1.004, 1.036) & 0.008 & 0.015 \\
\midrule
512 & TGT  & 0.011 ($-0.018$, 0.039) & 0.015 & 0.457
    & 1.011 (0.981, 1.042) & 0.015 & 0.461 \\
    & FED  & 0.011 ($-0.017$, 0.038) & 0.014 & 0.448
    & 1.011 (0.983, 1.039) & 0.014 & 0.445 \\
    & CCOD & 0.016 ($-0.002$, 0.033) & 0.009 & 0.086
    & 1.016 (0.997, 1.035) & 0.010 & 0.090 \\
\bottomrule
\end{tabular}
}
\begin{tablenotes}
\scriptsize
\item TGT: target-only estimator; FED: federated estimator; CCOD: estimator under the common conditional outcome distribution assumption. Reported p-values are nominal and are not adjusted for multiple time points or causal contrasts.
\end{tablenotes}
\end{table}

\begin{table}[H]
\centering
\singlespacing
\small
\caption{Estimated restricted mean survival time (RMST) and RMST difference up to day 601 for the OA target region.}
\label{tab:OA-RMST}
\begin{tabular}{l l c c c c}
\toprule
Group & Method & RMST Est. & SE & 95\% CI & p-value \\
\midrule
Control group
& TGT  & 583.80 & 5.11 & (573.78, 593.82) & -- \\
& FED  & 583.88 & 4.89 & (574.31, 593.46) & -- \\
& CCOD & 580.81 & 4.59 & (571.82, 589.80) & -- \\
\midrule
Treated group
& TGT  & 590.53 & 2.58 & (585.47, 595.59) & -- \\
& FED  & 590.59 & 2.47 & (585.76, 595.42) & -- \\
& CCOD & 588.90 & 3.67 & (581.71, 596.10) & -- \\
\midrule
RMST difference
& TGT  & 6.73 & 5.73 & ($-4.50$, 17.96) & 0.240 \\
& FED  & 6.75 & 5.34 & ($-3.71$, 17.21) & 0.206 \\
& CCOD & 8.10 & 3.59 & (1.05, 15.14) & 0.024 \\
\bottomrule
\end{tabular}
\begin{tablenotes}\scriptsize
\item RMST: restricted mean survival time; Est.: Estimate; SE: Standard error; CI: confidence interval. The p-value is for testing the null hypothesis of no treatment effect in the target region (RMST difference = 0). For FED, the RMST difference is estimated directly using estimand-specific federated weights and therefore need not equal the difference between the separately estimated treated- and control-arm RMSTs.
\end{tablenotes}
\end{table}

\begin{table}[H]
\centering
\singlespacing
\small
\caption{Estimated survival difference (SD) and survival ratio (SR) for the BP target region at days 148, 330, and 512.}
\label{tab:BP-SD-SR}
\resizebox{\textwidth}{!}{
\begin{tabular}{c l c c c c c c}
\toprule
Day & Method & SD Est. (95\% CI) & SE(SD) & p-value
& SR Est. (95\% CI) & SE(SR) & p-value \\
\midrule
148 & TGT  & 0.017 ($-0.001$, 0.036) & 0.009 & 0.062
    & 1.018 (0.999, 1.037) & 0.010 & 0.066 \\
    & FED  & 0.017 (0.005, 0.030) & 0.006 & 0.007
    & 1.018 (1.005, 1.031) & 0.007 & 0.008 \\
    & CCOD & 0.015 (0.002, 0.027) & 0.006 & 0.025
    & 1.015 (1.002, 1.028) & 0.007 & 0.027 \\
\midrule
330 & TGT  & 0.012 ($-0.010$, 0.035) & 0.011 & 0.280
    & 1.013 (0.989, 1.037) & 0.012 & 0.285 \\
    & FED  & 0.012 ($-0.009$, 0.034) & 0.011 & 0.252
    & 1.013 (0.991, 1.035) & 0.011 & 0.251 \\
    & CCOD & 0.013 ($-0.002$, 0.028) & 0.008 & 0.088
    & 1.014 (0.998, 1.029) & 0.008 & 0.091 \\
\midrule
512 & TGT  & 0.010 ($-0.019$, 0.039) & 0.015 & 0.489
    & 1.011 (0.980, 1.042) & 0.016 & 0.492 \\
    & FED  & 0.011 ($-0.017$, 0.038) & 0.014 & 0.457
    & 1.011 (0.982, 1.041) & 0.015 & 0.456 \\
    & CCOD & 0.014 ($-0.005$, 0.033) & 0.010 & 0.147
    & 1.015 (0.995, 1.035) & 0.010 & 0.151 \\
\bottomrule
\end{tabular}
}
\begin{tablenotes}
\scriptsize
\item TGT: target-only estimator; FED: federated estimator; CCOD: estimator under the common conditional outcome distribution assumption. Reported p-values are nominal and are not adjusted for multiple time points or causal contrasts.
\end{tablenotes}
\end{table}

\begin{table}[H]
\centering
\singlespacing
\small
\caption{Estimated restricted mean survival time (RMST) and RMST difference up to day 601 for the BP target region.}
\label{tab:BP-RMST}
\begin{tabular}{l l c c c c}
\toprule
Group & Method & RMST Est. & SE & 95\% CI & p-value \\
\midrule
Control group
& TGT  & 577.50 & 4.91 & (567.87, 587.13) & -- \\
& FED  & 577.44 & 4.40 & (568.81, 586.07) & -- \\
& CCOD & 581.81 & 4.11 & (573.75, 589.87) & -- \\
\midrule
Treated group
& TGT  & 583.19 & 2.89 & (577.53, 588.85) & -- \\
& FED  & 583.17 & 2.94 & (577.41, 588.93) & -- \\
& CCOD & 588.64 & 2.81 & (583.12, 594.15) & -- \\
\midrule
RMST difference
& TGT  & 5.69 & 5.70 & ($-5.47$, 16.86) & 0.318 \\
& FED  & 5.71 & 5.43 & ($-4.94$, 16.36) & 0.293 \\
& CCOD & 6.83 & 3.81 & ($-0.63$, 14.29) & 0.073 \\
\bottomrule
\end{tabular}
\begin{tablenotes}\scriptsize
\item RMST: restricted mean survival time; Est.: Estimate; SE: Standard error; CI: confidence interval. The p-value is for testing the null hypothesis of no treatment effect in the target region (RMST difference = 0). For FED, the RMST difference is estimated directly using estimand-specific federated weights and therefore need not equal the difference between the separately estimated treated- and control-arm RMSTs.
\end{tablenotes}
\end{table}

\begin{table}[H]
\centering
\singlespacing
\small
\caption{Estimated survival difference (SD) and survival ratio (SR) for the US target region at days 148, 330, and 512.}
\label{tab:US-SD-SR}
\resizebox{\textwidth}{!}{
\begin{tabular}{c l c c c c c c}
\toprule
Day & Method & SD Est. (95\% CI) & SE(SD) & p-value
& SR Est. (95\% CI) & SE(SR) & p-value \\
\midrule
148 & TGT  & 0.001 ($-0.006$, 0.008) & 0.004 & 0.763
    & 1.001 (0.994, 1.008) & 0.004 & 0.763 \\
    & FED  & 0.001 ($-0.006$, 0.009) & 0.004 & 0.749
    & 1.001 (0.994, 1.009) & 0.004 & 0.753 \\
    & CCOD & 0.004 ($-0.004$, 0.012) & 0.004 & 0.328
    & 1.004 (0.996, 1.012) & 0.004 & 0.330 \\
\midrule
330 & TGT  & 0.004 ($-0.008$, 0.016) & 0.006 & 0.486
    & 1.004 (0.992, 1.017) & 0.006 & 0.488 \\
    & FED  & 0.004 ($-0.007$, 0.015) & 0.005 & 0.465
    & 1.004 (0.993, 1.015) & 0.005 & 0.464 \\
    & CCOD & 0.005 ($-0.004$, 0.015) & 0.005 & 0.280
    & 1.006 (0.995, 1.016) & 0.005 & 0.282 \\
\midrule
512 & TGT  & 0.007 ($-0.007$, 0.021) & 0.007 & 0.357
    & 1.007 (0.992, 1.021) & 0.007 & 0.359 \\
    & FED  & 0.007 ($-0.008$, 0.021) & 0.007 & 0.370
    & 1.007 (0.992, 1.022) & 0.008 & 0.379 \\
    & CCOD & 0.011 ($-0.002$, 0.025) & 0.007 & 0.108
    & 1.011 (0.997, 1.025) & 0.007 & 0.112 \\
\bottomrule
\end{tabular}
}
\begin{tablenotes}
\scriptsize
\item TGT: target-only estimator; FED: federated estimator; CCOD: estimator under the common conditional outcome distribution assumption. Reported p-values are nominal and are not adjusted for multiple time points or causal contrasts.
\end{tablenotes}
\end{table}

\begin{table}[H]
\centering
\singlespacing
\small
\caption{Estimated restricted mean survival time (RMST) and RMST difference up to day 601 for the US target region.}
\label{tab:US-RMST}
\begin{tabular}{l l c c c c}
\toprule
Group & Method & RMST Est. & SE & 95\% CI & p-value \\
\midrule
Control group
& TGT  & 594.34 & 2.38 & (589.68, 599.00) & -- \\
& FED  & 594.35 & 2.58 & (589.29, 599.41) & -- \\
& CCOD & 592.06 & 3.00 & (586.19, 597.93) & -- \\
\midrule
Treated group
& TGT  & 596.51 & 1.40 & (593.77, 599.25) & -- \\
& FED  & 596.49 & 1.25 & (594.04, 598.94) & -- \\
& CCOD & 595.54 & 2.40 & (590.83, 600.25) & -- \\
\midrule
RMST difference
& TGT  & 2.18 & 2.76 & ($-3.23$, 7.58) & 0.430 \\
& FED  & 2.18 & 2.85 & ($-3.41$, 7.76) & 0.445 \\
& CCOD & 3.48 & 2.48 & ($-1.38$, 8.34) & 0.161 \\
\bottomrule
\end{tabular}
\begin{tablenotes}\scriptsize
\item RMST: restricted mean survival time; Est.: Estimate; SE: Standard error; CI: confidence interval. The p-value is for testing the null hypothesis of no treatment effect in the target region (RMST difference = 0). For FED, the RMST difference is estimated directly using estimand-specific federated weights and therefore need not equal the difference between the separately estimated treated- and control-arm RMSTs.
\end{tablenotes}
\end{table}

\subsection{Sensitivity analysis to the non-informative censoring assumption}
\label{subapp:supp-sens-censoring}

The proposed analyses rely on a non-informative censoring assumption conditional on treatment assignment, measured baseline covariates, and study region (Assumption \ref{asp:inde-cen} in Section \ref{sec:setup}). To assess the sensitivity of the empirical findings to potential violations of this assumption, we conducted two complementary analyses. First, we used the selection-bias sensitivity framework of \citet{smith2019bounding} to quantify the magnitude of selection mechanisms that could explain away the observed treatment association. Second, we conducted a direct sensitivity analysis for potentially informative censoring using the $\delta$-adjusted multiple-imputation procedure of \citet{yang2023smim}. The first analysis provides a general calibration of sensitivity to selection bias, whereas the second directly perturbs the post-censoring HIV-1 event hazard and evaluates the resulting change in the RMST treatment contrast.

\subsubsection{Selection-bias E-value}

We first considered the selection-bias bounds of \citet{smith2019bounding}. We fit a Cox proportional hazards model for HIV-1 acquisition adjusting for treatment assignment, baseline covariates $\mb X$, and study region $R$,
\[
\lambda(t\mid A,\mb X,R)
=
\lambda_0(t)
\exp\left\{
\beta_A A+\beta_{\mb X}^{\top}\mb X+\beta_R^{\top}R
\right\},
\]
where $\mb X$ includes baseline age, body weight, BMI, and the ML risk score, and $\lambda_0(t)$ denotes the baseline hazard. The estimated treatment hazard ratio (HR) was
\[
\widehat{\mathrm{HR}}=0.792,
\qquad
95\%~\mathrm{CI}=(0.583,1.075).
\]
The estimated cumulative incidence of HIV-1 acquisition over follow-up was 0.043. Because the outcome was rare, we used the rare-outcome approximation and treated the hazard ratio as an approximation to the corresponding risk ratio for purposes of the selection-bias bound.

For inference in the presence of potentially differential selection, the general bound of \citet{smith2019bounding} depends on four sensitivity parameters describing associations between an unmeasured factor $U$, the outcome $Y$, and selection $S^\ast$, separately by treatment group:
\[
RR_{UY\mid A=0},\qquad
RR_{UY\mid A=1},\qquad
RR_{S^\ast U\mid A=0},\qquad
RR_{S^\ast U\mid A=1}.
\]
These sensitivity parameters characterize the strength of the association between the unmeasured factor and the outcome, and between the unmeasured factor and selection, within each treatment arm. The measured baseline covariates $\mb X$ and study region $R$ are included in the observed-data model, whereas $U$ represents additional unmeasured factors that could induce residual dependence between censoring and the event time. Table~\ref{tab:selection-notation} summarizes the correspondence between the notation in the selection-bias framework of \citet{smith2019bounding} and the quantities in our AMP censoring setting.

\begin{table}[H]
\centering
\singlespacing
\footnotesize
\caption{Correspondence between the selection-bias framework of \citet{smith2019bounding} and the AMP censoring setting.}
\label{tab:selection-notation}
\begin{tabular}{p{0.42\textwidth} p{0.50\textwidth}}
\toprule
\citet{smith2019bounding} framework & AMP censoring setting \\
\midrule
$A$: exposure or treatment
& $A$: randomized treatment assignment \\

$Y$: outcome
& HIV-1 acquisition by the follow-up time, determined by the event time $T$ \\

$S$: selection into the observed sample
& Remaining uncensored through the follow-up time, determined by the censoring time $C$ \\

$U$: unmeasured factor associated with the outcome and selection
& An unmeasured factor that may be associated with both HIV-1 acquisition and censoring \\

$RR_{UY\mid A=a}$: association between $U$ and the outcome within treatment group $a$
& Association between $U$ and HIV-1 acquisition within treatment arm $a$ \\

$RR_{SU\mid A=a}$: association between $U$ and selection within treatment group $a$
& Association between $U$ and remaining uncensored within treatment arm $a$ \\
\bottomrule
\end{tabular}
\begin{tablenotes}\scriptsize
\item In the AMP analysis, the selection-bias framework is applied to censoring by viewing continued observation as the selection mechanism. The analysis conditions on the measured baseline covariates $\mb X$ and study region $R$.
\end{tablenotes}
\end{table}

Following Result 1B of \citet{smith2019bounding}, we consider an equal-strength calibration in which the four sensitivity parameters are assigned a common value,
\[
RR_{UY\mid A=0}
=
RR_{UY\mid A=1}
=
RR_{S^\ast U\mid A=0}
=
RR_{S^\ast U\mid A=1}
=
r.
\]
The selection-bias E-value is the minimum value of $r$ sufficient, under this calibration, to move the observed treatment association to the null. Let
\[
\widetilde{RR}_{\mathrm{obs}}
=
\max\left(RR_{\mathrm{obs}},RR_{\mathrm{obs}}^{-1}\right),
\]
where $RR_{\mathrm{obs}}$ denotes the observed treatment--outcome association on the risk-ratio scale. Result 1B of \citet{smith2019bounding} gives
\[
E_{\mathrm{sel}}
=
\sqrt{\widetilde{RR}_{\mathrm{obs}}}
+
\sqrt{
\sqrt{\widetilde{RR}_{\mathrm{obs}}}
\left(
\sqrt{\widetilde{RR}_{\mathrm{obs}}}-1
\right)
}.
\]

For the observed treatment association, the selection-bias E-value was 1.50. Under the equal-strength calibration defined above, this means that a hypothetical unmeasured factor $U$ would need to be associated, on the risk-ratio scale, with both HIV-1 acquisition and remaining uncensored by approximately 1.50 within each treatment arm for selection bias of this magnitude to be sufficient to move the observed point estimate to the null. Thus, an unmeasured factor generating associations of this joint magnitude with both the event process and the censoring process could, in principle, explain away the observed point estimate. 

To provide empirical context, we additionally examined the associations of the measured baseline covariates $\mb X$ with HIV-1 acquisition and censoring, conditional on study region $R$. Each continuous covariate was standardized to have mean zero and standard deviation one. We fit separate multivariable Cox models for HIV-1 acquisition and censoring, both adjusting for study region $R$. The resulting hazard ratios are reported in Table~\ref{tab:censoring-calibration}.

\begin{table}[H]
\centering
\singlespacing
\caption{Associations of measured baseline covariates with HIV-1 acquisition and censoring after adjustment for study region.}
\label{tab:censoring-calibration}
\begin{tabular}{lccc|ccc}
\toprule
& \multicolumn{3}{c}{HIV-1 acquisition}
& \multicolumn{3}{c}{Censoring} \\
\cmidrule(lr){2-4}\cmidrule(lr){5-7}
Covariate & HR & Lower & Upper & HR & Lower & Upper \\
\midrule
Age           & 0.673 & 0.544 & 0.832 & 0.978 & 0.943 & 1.013 \\
ML risk score & 1.018 & 0.835 & 1.242 & 0.994 & 0.955 & 1.034 \\
Body weight   & 0.666 & 0.439 & 1.011 & 0.957 & 0.890 & 1.029 \\
Baseline BMI  & 1.083 & 0.755 & 1.555 & 1.055 & 0.987 & 1.129 \\
\bottomrule
\end{tabular}
\begin{tablenotes}\scriptsize
\item HR: hazard ratio. Continuous covariates were standardized before fitting the Cox models, so each reported HR corresponds to a one-standard-deviation increase in the corresponding covariate. Both models adjust for study region $R$. These associations are provided only as empirical context and are not direct estimates of the sensitivity parameters in the selection-bias bounds of \citet{smith2019bounding}.
\end{tablenotes}
\end{table}

Among the measured baseline covariates, age showed the clearest association with HIV-1 acquisition, whereas the associations for the ML risk score, body weight, and BMI were weaker as their 95\% CIs cover 1. In contrast, the associations of all four measured covariates with censoring were close to the null. These observed associations provide useful empirical context for the E-value of 1.50, while recognizing that they are not direct estimates of the unmeasured selection mechanisms characterized by the sensitivity parameters. 

\subsubsection{Delta-adjusted sensitivity analysis for informative censoring}

We next conducted a direct sensitivity analysis using the $\delta$-adjusted multiple-imputation procedure of \citet{yang2023smim}. This analysis allows the HIV-1 event hazard after premature censoring to differ systematically from that implied under the censoring-at-random benchmark and evaluates the resulting change in the RMST difference contrast.

We considered the RMST difference through $\tau=601$ days, consistent with the primary AMP analysis. Participants were classified into three observation patterns:
\[
\begin{split}
\text{Pattern 1: } & \text{an observed HIV-1 event by day 601 } (n=174),\\
\text{Pattern 2: } & \text{observed event-free through day 601 } (n=1{,}052),\\
\text{Pattern 3: } & \text{censored before day 601 } (n=3{,}385).
\end{split}
\]
All observations censored before day 601 were conservatively allowed to be potentially informatively censored. Among the 3,385 participants censored before day 601, censoring was concentrated near the end of scheduled follow-up. The median censoring time was 561 days, and 2,527 participants (74.7\%) were censored between days 540 and 599. This pattern suggests that a substantial proportion of censoring before day 601 was administrative rather than premature dropout. Accordingly, departures from the non-informative censoring assumption are most directly relevant to participants who discontinued follow-up prematurely. 

The imputation model included baseline age, body weight, BMI, the ML risk score, and study-region indicators. Let $\delta_a$ denote the sensitivity parameter for participants censored before $\tau$ in treatment arm $a$. The value $\delta_a=1$ corresponds to the censoring-at-random benchmark implied by the imputation model. Values $\delta_a>1$ assume that participants censored before $\tau$ subsequently experience a higher HIV-1 event hazard than predicted under the benchmark model, whereas $\delta_a<1$ assume a lower post-censoring event hazard.

We considered three classes of departures from this benchmark. First, we imposed common departures in the two treatment groups,
\[
(\delta_0,\delta_1)
\in
\{(0.5,0.5),(0.75,0.75),(1.25,1.25),(1.5,1.5),(2,2)\}.
\]
Second, to examine departures adverse to the estimated treatment benefit, we retained the benchmark assumption in the control group and progressively increased the post-censoring hazard in the treated group,
\[
(\delta_0,\delta_1)
\in
\{(1,1.25),(1,1.5),(1,2),(1,3)\}.
\]
Third, we considered differential adverse scenarios in which the post-censoring hazard was decreased in the control group and simultaneously increased in the treated group,
\[
(\delta_0,\delta_1)
\in
\{(0.9,1.25),(0.75,1.5),(0.5,2)\}.
\]
The latter two sets of scenarios deliberately perturb the censoring mechanism in a direction expected to attenuate the estimated treatment benefit.

The $\delta$-adjusted analysis was applied to the pooled AMP survival data, conditional on the measured baseline covariates and study region. It therefore provides a direct assessment of the sensitivity of the underlying AMP treatment comparison to potentially informative censoring.

For each scenario, we estimated the control- and treated-group RMSTs and their difference. Multiple-imputation standard errors were calculated using Rubin's combining rule, and wild-bootstrap standard errors were additionally obtained for inference. We used 50 multiple imputations and 500 wild-bootstrap replicates, and implemented the analysis using the R package \texttt{smim} developed by \citet{yang2023multiply}. Table \ref{tab:censoring-smim} reports the results. 

\begin{table}[H]
\centering
\singlespacing
\small
\caption{Delta-adjusted sensitivity analysis of RMST through day 601 under potentially informative censoring.}
\label{tab:censoring-smim}
\resizebox{\textwidth}{!}{
\begin{tabular}{lccccccccc}
\toprule
Scenario &
$\delta_0$ &
$\delta_1$ &
RMST$_0$ &
RMST$_1$ &
Difference &
SE$_{\mathrm{Rubin}}$ &
SE$_{\mathrm{WB}}$ &
95\% CI &
p-value \\
\midrule
MAR benchmark
& 1.00 & 1.00 & 584.92 & 590.22 & 5.31 & 2.39 & 2.27 & (0.85, 9.76) & 0.019 \\
\midrule
Common $\delta$
& 0.50 & 0.50 & 585.44 & 590.59 & 5.15 & 2.34 & 2.27 & (0.70, 9.59) & 0.023 \\
Common $\delta$
& 0.75 & 0.75 & 585.19 & 590.41 & 5.22 & 2.37 & 2.27 & (0.77, 9.67) & 0.022 \\
Common $\delta$
& 1.25 & 1.25 & 584.64 & 590.05 & 5.41 & 2.42 & 2.28 & (0.95, 9.87) & 0.017 \\
Common $\delta$
& 1.50 & 1.50 & 584.40 & 589.87 & 5.47 & 2.44 & 2.28 & (1.01, 9.94) & 0.016 \\
Common $\delta$
& 2.00 & 2.00 & 583.92 & 589.51 & 5.59 & 2.48 & 2.29 & (1.11, 10.07) & 0.014 \\
\midrule
Treatment arm adverse
& 1.00 & 1.25 & 584.92 & 590.05 & 5.13 & 2.40 & 2.27 & (0.68, 9.59) & 0.024 \\
Treatment arm adverse
& 1.00 & 1.50 & 584.92 & 589.87 & 4.96 & 2.40 & 2.27 & (0.50, 9.41) & 0.029 \\
Treatment arm adverse
& 1.00 & 2.00 & 584.92 & 589.51 & 4.59 & 2.41 & 2.27 & (0.14, 9.05) & 0.043 \\
Treatment arm adverse
& 1.00 & 3.00 & 584.92 & 588.84 & 3.93 & 2.44 & 2.28 & ($-0.53$, 8.38) & 0.084 \\
\midrule
Differential adverse
& 0.90 & 1.25 & 585.03 & 590.05 & 5.01 & 2.39 & 2.27 & (0.56, 9.47) & 0.027 \\
Differential adverse
& 0.75 & 1.50 & 585.19 & 589.87 & 4.68 & 2.39 & 2.27 & (0.23, 9.13) & 0.039 \\
Differential adverse
& 0.50 & 2.00 & 585.44 & 589.51 & 4.07 & 2.38 & 2.27 & ($-0.38$, 8.52) & 0.073 \\
\bottomrule
\end{tabular}
}
\begin{tablenotes}\scriptsize
\item RMST$_0$ and RMST$_1$ denote the estimated RMSTs in the control and treated groups, respectively. Difference is RMST$_1-$RMST$_0$. SE$_{\mathrm{Rubin}}$ denotes the standard error obtained using Rubin's combining rule, and SE$_{\mathrm{WB}}$ denotes the wild-bootstrap standard error. The reported 95\% CIs and p-values are Wald-type quantities based on SE$_{\mathrm{WB}}$. The benchmark $(\delta_0,\delta_1)=(1,1)$ corresponds to the assumption referred to as ``missing at random'' (MAR) in the sensitivity-analysis framework of \citet{yang2023multiply}, under which the post-censoring event hazard follows that implied by the fitted imputation model.
\end{tablenotes}
\end{table}

Under the benchmark assumption referred to as ``missing at random'' (MAR) in the sensitivity-analysis framework of \citet{yang2023multiply}, corresponding to $(\delta_0,\delta_1)=(1,1)$, the estimated RMST difference was 5.31 days, with a wild-bootstrap standard error of 2.27 and a 95\% CI of $(0.85, 9.76)$. The Rubin and wild-bootstrap standard errors were similar in magnitude across the considered scenarios.

When the same departure from the benchmark was imposed in both treatment groups, the estimated treatment contrast was highly stable. Across common values of $\delta$ ranging from 0.5 to 2, the estimated RMST difference ranged from 5.15 to 5.59 days, and all corresponding 95\% CIs excluded zero. Thus, common changes in the assumed post-censoring HIV-1 hazard affected the two treatment groups in broadly parallel directions and had relatively little impact on their contrast.

Departures specifically adverse to the treated group progressively attenuated the estimated RMST benefit. Relative to 5.31 days under the benchmark, the estimated difference decreased to 5.13 days under $(\delta_0,\delta_1)=(1,1.25)$, 4.96 days under $(1,1.5)$, and 4.59 days under $(1,2)$. Even when the post-censoring HIV-1 hazard in the treated group was assumed to be twice that implied by the benchmark model, the estimated difference remained positive and the 95\% CI remained above zero, $(0.14,9.05)$. Under the more extreme scenario $(1,3)$, the point estimate remained positive at 3.93 days, although the 95\% CI included zero.

The differential adverse scenarios produced a similar pattern. The estimated RMST difference was 5.01 days under $(\delta_0,\delta_1)=(0.9,1.25)$ and 4.68 days under $(0.75,1.5)$, with both 95\% CIs excluding zero. Under the strongest differential scenario considered, $(\delta_0,\delta_1)=(0.5,2)$, in which the assumed post-censoring hazard was reduced by half in the control group and doubled in the treated group relative to the benchmark model, the estimated RMST difference remained positive at 4.07 days, although the 95\% CI included zero, $(-0.38,8.52)$.

Overall, the $\delta$-adjusted analysis indicates that the estimated direction of the treatment effect is stable across the considered departures from the benchmark censoring assumption. Common departures had little influence on the RMST contrast, whereas treatment-specific departures designed to attenuate the estimated benefit progressively reduced the RMST difference. The point estimate remained positive under every scenario considered, while the 95\% CI included zero only under the strongest treatment-adverse scenarios, $(\delta_0,\delta_1)=(1,3)$ and $(0.5,2)$.

\subsection{Sensitivity analysis to the main analysis by excluding the baseline ML risk score from covariates}\label{subapp:supp-sensitivity}

In this section, we conduct an additional sensitivity analysis to assess whether excluding the baseline machine learning (ML) risk score from the main analysis materially changes the results. Figure \ref{fig:AMP-SA-supp} presents the treatment-specific survival curves obtained when only baseline age and weight are included as covariates in all nuisance models. We find that the survival curves from all methods have shapes similar to those in the main analysis. However, including the risk score yields overall narrower CI bands, suggesting that the score is predictive of the outcome and improves estimation efficiency. Tables \ref{tab:SD-AMP-supp} and \ref{tab:RMST-AMP-supp} report the corresponding estimates of SD, SR, RMST, and RMST difference, analogous to those in the main text. FED generally yields smaller or comparable SEs than TGT across the reported estimands.

\begin{figure}[H]
    \centering
    \includegraphics[width=\textwidth]{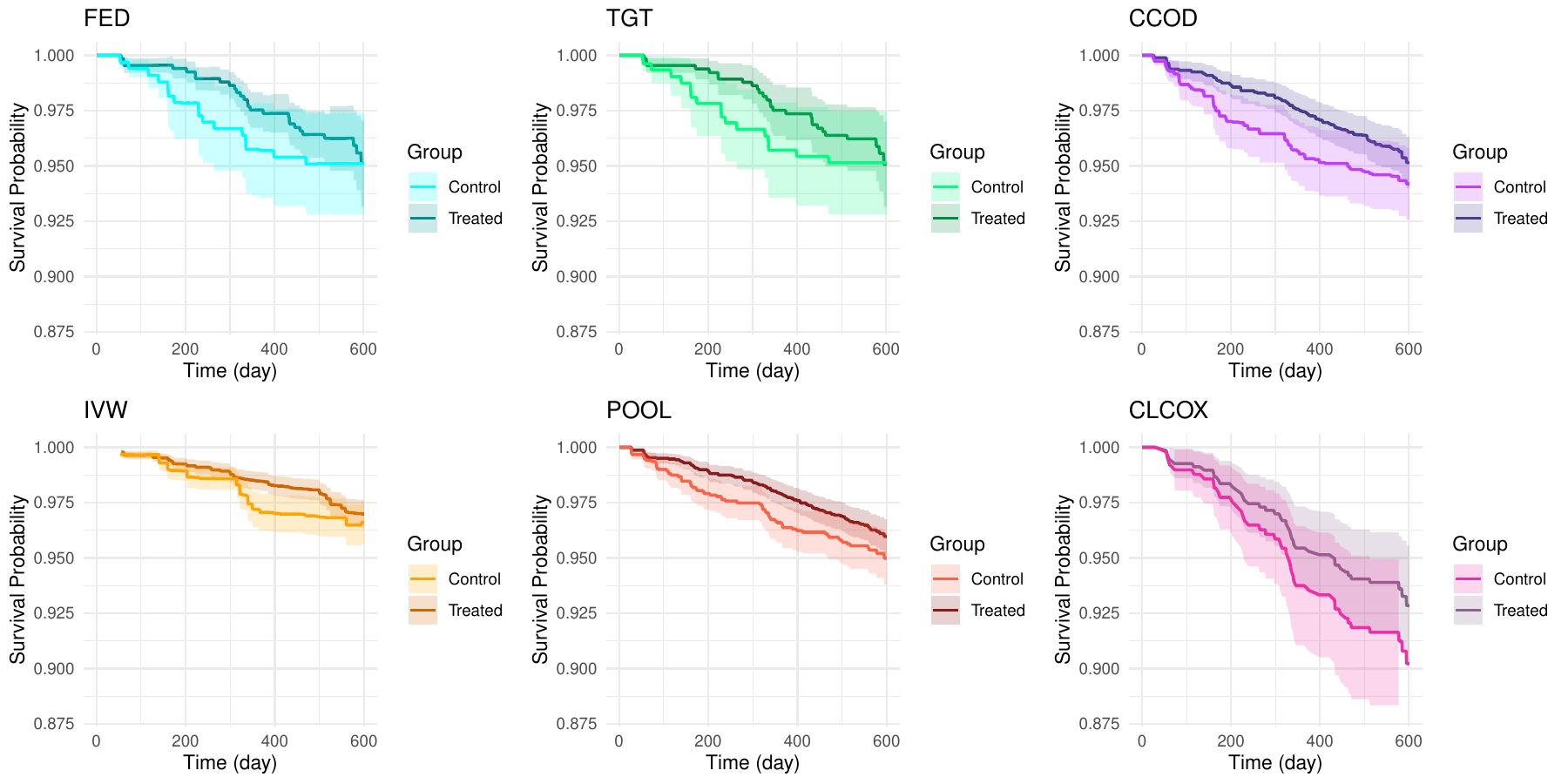}\\
    \textbf{\small (A) Estimated treatment-specific survival curves}\medskip
    
    \includegraphics[width=5in]{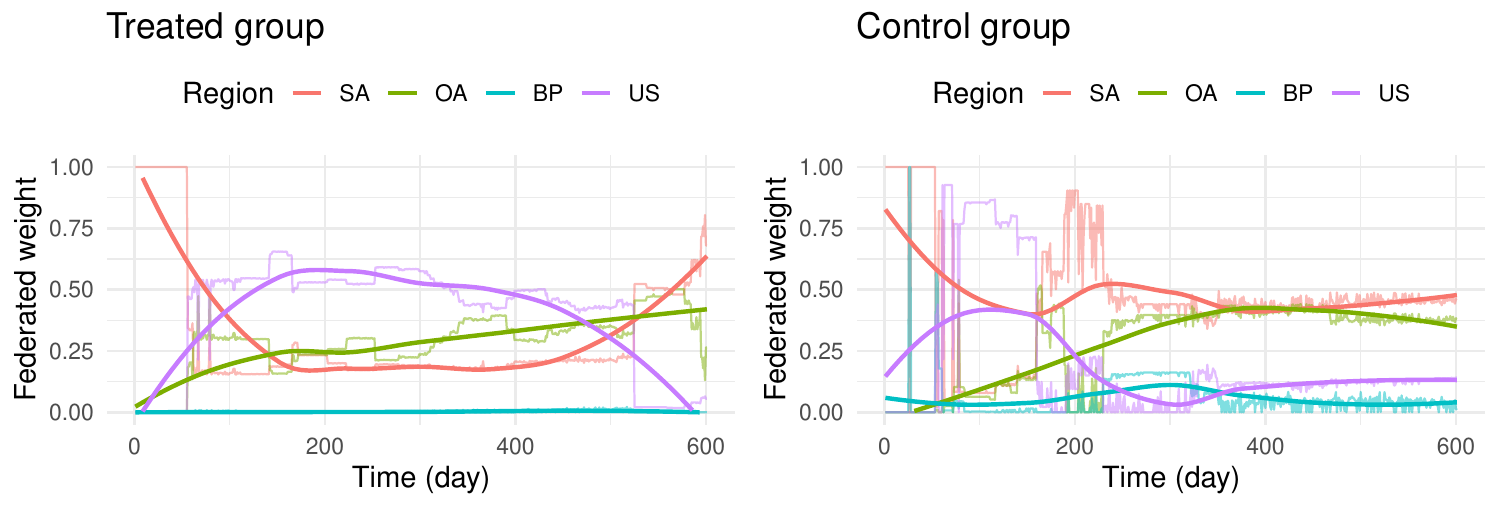}\\
    \textbf{\small (B) Treatment- and time-specific federated weights} 
    \caption{Supplementary results using only age and weight at baseline as covariates in all nuisance models, treating women in South Africa (SA) as the target region. The smoothed weight curves in Panel (B) are obtained by locally weighted regression, which is only a visualization tool (not a part of our methodology). }
    \label{fig:AMP-SA-supp}
\end{figure}

\begin{table}[H]
\centering
\singlespacing
\small
\caption{Estimated SD and SR at days 148, 330, and 512 by excluding ML risk score from covariates.}
\label{tab:SD-AMP-supp}
\begin{threeparttable}
\begin{tabular}{rrcccccc}
\toprule
Day & Method 
& SD Est. (95\% CI) & SE(SD) & p-value
& SR Est. (95\% CI) & SE(SR) & p-value \\
\midrule
\multirow{3}{*}{148}
& TGT  
& 0.008 (-0.005, 0.021) & 0.007 & 0.237
& 1.008 (0.995, 1.022) & 0.007 & 0.241 \\
& FED  
& 0.008 (-0.000, 0.016) & 0.004 & 0.061
& 1.008 (1.000, 1.016) & 0.004 & 0.061 \\
& CCOD 
& 0.009 (-0.001, 0.020) & 0.005 & 0.077
& 1.010 (0.999, 1.020) & 0.005 & 0.079 \\
\midrule
\multirow{3}{*}{330}
& TGT  
& 0.018 (-0.005, 0.041) & 0.012 & 0.127
& 1.019 (0.994, 1.043) & 0.012 & 0.133 \\
& FED  
& 0.018 (-0.003, 0.038) & 0.010 & 0.089
& 1.018 (0.997, 1.040) & 0.011 & 0.092 \\
& CCOD 
& 0.019 (0.003, 0.035) & 0.008 & 0.019
& 1.020 (1.003, 1.036) & 0.009 & 0.021 \\
\midrule
\multirow{3}{*}{512}
& TGT  
& 0.012 (-0.015, 0.040) & 0.014 & 0.380
& 1.013 (0.984, 1.042) & 0.015 & 0.384 \\
& FED  
& 0.013 (-0.013, 0.039) & 0.013 & 0.329
& 1.014 (0.987, 1.041) & 0.014 & 0.325 \\
& CCOD 
& 0.015 (-0.003, 0.033) & 0.009 & 0.103
& 1.016 (0.997, 1.035) & 0.010 & 0.108 \\
\bottomrule
\end{tabular}
\begin{tablenotes}\scriptsize
\item Est.: Estimate; SE: Standard error; CI: confidence interval. The p-value is for testing the null hypothesis of no treatment effect in the target region (SD = 0 or SR = 1).
\end{tablenotes}
\end{threeparttable}
\end{table}

\begin{table}[H]
\centering
\singlespacing
\small
\caption{Estimated RMST by treatment group and RMST difference up to day 601 by excluding ML risk score from covariates.}
\label{tab:RMST-AMP-supp}
\begin{tabular}{rrcccc}
\toprule
& Method & RMST Est. & SE & 95\% CI & p-value \\
\midrule
& TGT  & 582.90 & 4.69 & (573.71, 592.09) & -- \\
Control group
& FED  & 582.95 & 4.44 & (574.25, 591.65) & -- \\
& CCOD & 580.37 & 4.27 & (571.99, 588.75) & -- \\
\midrule
& TGT  & 589.89 & 2.38 & (585.23, 594.56) & -- \\
Treated group
& FED  & 589.96 & 2.22 & (585.60, 594.31) & -- \\
& CCOD & 588.16 & 3.28 & (581.73, 594.60) & -- \\
\midrule
& TGT  & 7.00 & 5.25 & (-3.30, 17.30) & 0.183 \\
RMST difference
& FED  & 7.00 & 4.97 & (-2.73, 16.74) & 0.159 \\
& CCOD & 7.79 & 3.73 & (0.48, 15.10) & 0.037 \\
\bottomrule
\end{tabular}
\begin{tablenotes}\scriptsize
\item RMST: restricted mean survival time; Est.: Estimate; SE: Standard error; CI: confidence interval. The p-value is for testing the null hypothesis of no treatment effect in the target region (RMST difference = 0). For FED, the RMST difference is estimated directly using estimand-specific federated weights and therefore need not equal the difference between the separately estimated treated- and control-arm RMSTs.
\end{tablenotes}
\end{table}

\subsection{Federated weight diagnostics and further sensitivity analyses for the SA target}
\label{subapp:SA-weight-diagnostics}

To further examine the adaptive borrowing behavior of the proposed federated estimator, we report the estimated source discrepancies and selected federated weights for the SA target at days 148, 330, and 512. Table \ref{tab:SA-weight-diagnostics} summarizes the results separately for the control and treated groups.

\begin{table}[H]
\centering
\singlespacing
\caption{Estimated federated weights and source discrepancies for the SA target region at selected follow-up times.}
\label{tab:SA-weight-diagnostics}
\begin{tabular}{clccccccccccc}
\toprule
Day & Group &
$\widehat\eta_{\mathrm{SA}}$ &
$\widehat\eta_{\mathrm{OA}}$ &
$\widehat\eta_{\mathrm{BP}}$ &
$\widehat\eta_{\mathrm{US}}$ &
$\widehat\chi_{\mathrm{OA}}$ &
$\widehat\chi_{\mathrm{BP}}$ &
$\widehat\chi_{\mathrm{US}}$ \\
\midrule
148 & Control
& 0.108 & 0.119 & 0.000 & 0.772
& -0.0026 & -0.0272 & 0.0069 \\

330 & Control
& 0.413 & 0.431 & 0.000 & 0.156
& 0.0027 & -0.0288 & 0.0142 \\

512 & Control
& 0.465 & 0.398 & 0.012 & 0.125
& -0.0027 & -0.0265 & 0.0136 \\
\midrule
148 & Treated
& 0.190 & 0.194 & 0.000 & 0.616
& -0.0015 & -0.0157 & 0.0020 \\

330 & Treated
& 0.169 & 0.310 & 0.010 & 0.511
& 0.0064 & -0.0225 & 0.0097 \\

512 & Treated
& 0.216 & 0.257 & 0.045 & 0.482
& 0.0078 & -0.0233 & 0.0210 \\
\bottomrule
\end{tabular}

\begin{tablenotes}\scriptsize
\item $\widehat\eta_k$ denotes the selected federated weight for region $k$, and
$\widehat\chi_k$ denotes the estimated discrepancy between the corresponding source-region local estimator and the SA target estimator. The weights sum to one within each treatment group and follow-up time.
\end{tablenotes}
\end{table}

The estimated weights show a clear source-selection pattern. The BP region receives zero or only minimal weight across the evaluated settings, consistent with its comparatively large absolute discrepancies from the SA target estimator. In particular, $|\widehat{\chi}_{\mathrm{BP}}(t,a)|$ ranges from approximately 0.016 to 0.029 across the evaluated treatment groups and follow-up times. In contrast, the estimated discrepancies for the OA and US regions are generally smaller, and these regions consequently receive greater federated weights.

The borrowing pattern also changes over follow-up. In the control group, the SA target weight increases from 0.108 at day 148 to 0.465 at day 512, while the OA weight increases from 0.119 to 0.398 and the US weight decreases from 0.772 to 0.125. In the treated group, the US region remains an important source throughout follow-up, with weights of 0.616, 0.511, and 0.482 at days 148, 330, and 512, respectively. The OA contribution also remains non-negligible, whereas the BP contribution remains small. These results illustrate that the federated weights are treatment- and time-specific and are driven by estimated statistical alignment with the SA target estimand rather than by demographic or geographic similarity between regions.

\begin{table}[H]
\centering
\singlespacing
\caption{Sensitivity of the SA-target FED survival estimates to the tuning parameter.}
\label{tab:SA-lambda-sensitivity}
\begin{tabular}{c l c c c c}
\toprule
Day & Group & $\lambda/\lambda_{\rm CV}$ & Estimate & SE & 95\% CI \\
\midrule
148 & Control & 0.5 & 0.9869 & 0.0033 & (0.9805, 0.9934) \\
    &         & 1.0 & 0.9869 & 0.0035 & (0.9801, 0.9938) \\
    &         & 2.0 & 0.9870 & 0.0038 & (0.9796, 0.9944) \\
148 & Treated & 0.5 & 0.9955 & 0.0017 & (0.9921, 0.9989) \\
    &         & 1.0 & 0.9955 & 0.0018 & (0.9920, 0.9990) \\
    &         & 2.0 & 0.9955 & 0.0018 & (0.9919, 0.9990) \\
\midrule
330 & Control & 0.5 & 0.9626 & 0.0101 & (0.9428, 0.9824) \\
    &         & 1.0 & 0.9627 & 0.0104 & (0.9423, 0.9831) \\
    &         & 2.0 & 0.9628 & 0.0107 & (0.9417, 0.9838) \\
330 & Treated & 0.5 & 0.9815 & 0.0036 & (0.9745, 0.9886) \\
    &         & 1.0 & 0.9815 & 0.0038 & (0.9742, 0.9889) \\
    &         & 2.0 & 0.9815 & 0.0040 & (0.9736, 0.9894) \\
\midrule
512 & Control & 0.5 & 0.9488 & 0.0118 & (0.9257, 0.9719) \\
    &         & 1.0 & 0.9489 & 0.0120 & (0.9254, 0.9723) \\
    &         & 2.0 & 0.9490 & 0.0124 & (0.9248, 0.9732) \\
512 & Treated & 0.5 & 0.9642 & 0.0058 & (0.9528, 0.9756) \\
    &         & 1.0 & 0.9642 & 0.0063 & (0.9519, 0.9765) \\
    &         & 2.0 & 0.9641 & 0.0071 & (0.9502, 0.9780) \\
\bottomrule
\end{tabular}
\end{table}

\begin{table}[H]
\centering
\singlespacing
\caption{Sensitivity of the SA-target FED survival estimates to exclusion of the BP and US source regions.}
\label{tab:SA-source-sensitivity}
\begin{tabular}{c l c c c c}
\toprule
Day & Group & Original FED & SE & FED excluding BP/US & SE \\
\midrule
148 & Control & 0.9869 & 0.0035 & 0.9872 & 0.0065 \\
148 & Treated & 0.9955 & 0.0018 & 0.9955 & 0.0027 \\
\addlinespace
330 & Control & 0.9627 & 0.0104 & 0.9628 & 0.0107 \\
330 & Treated & 0.9815 & 0.0038 & 0.9813 & 0.0050 \\
\addlinespace
512 & Control & 0.9489 & 0.0120 & 0.9490 & 0.0124 \\
512 & Treated & 0.9642 & 0.0063 & 0.9641 & 0.0075 \\
\bottomrule
\end{tabular}
\end{table}

We further examined the sensitivity of the SA-target FED estimator to the tuning parameter and to the set of candidate source regions. First, we varied the cross-validated tuning parameter over $0.5\lambda_{\rm CV}$, $\lambda_{\rm CV}$, and $2\lambda_{\rm CV}$. As shown in Table \ref{tab:SA-lambda-sensitivity}, the resulting survival estimates were highly stable across the three choices, although the selected weights and corresponding standard errors showed somewhat greater variation. For example, for the treated group at day 512, the US weight decreased from 0.407 under $0.5\lambda_{\rm CV}$ to 0.106 under $2\lambda_{\rm CV}$, while the FED survival estimate changed only from 0.9642 to 0.9641.

Second, we repeated the analysis after excluding the BP and US regions from the candidate source set, so that OA was the only available source region in addition to the SA target. The resulting estimates remained very close to those from the original FED analysis at all evaluated time points (Table \ref{tab:SA-source-sensitivity}). The largest absolute difference in the estimated survival probabilities was approximately 0.0003, occurring at day 148 in the control group. As expected, restricting the available source information generally increased the standard errors, particularly at day 148. Overall, these sensitivity analyses show that the SA-target FED survival estimates are stable across the considered tuning-parameter values and remain similar when borrowing from BP and US is disallowed.

\newpage

\section{Summary of Literature and Notation}\label{subapp:notation}

In this section, we provide a concise comparison with some related methodological work and summarize the notation used throughout the paper and Appendix.

\subsection{Comparison to related work}\label{subapp:related-work}

Table \ref{tab:related-work} summarizes the main distinctions between the proposed framework and three  related lines of work. \cite{westling2024inference} provide semiparametric and doubly robust methods for causal survival estimation within a single data source and serve as an important building block for our local survival estimators. \cite{han2025federated} develop adaptive data-fusion methodology for borrowing information across heterogeneous data sources, which motivates our federated weighting strategy. \cite{liuprivacy} consider federated causal survival analysis under distribution shift without sharing individual-level data. Building on these developments, the present work focuses on region-specific causal survival inference with right censoring, introduces a CCOD-based efficiency benchmark, and extends the federated framework to clinically interpretable survival contrasts. 

\begin{table}[H]
\centering
\footnotesize
\caption{Comparison with related methodological work.}
\label{tab:related-work}
\begin{tabular}{p{0.17\linewidth}
                p{0.17\linewidth}
                p{0.1\linewidth}
                p{0.11\linewidth}
                p{0.11\linewidth}
                p{0.17\linewidth}}
\toprule
\makecell{\bf Literature}
&
\makecell[l]{\bf Target estimand}
&
\makecell[l]{\bf Addresses \\ \bf censoring}
&
\makecell[l]{\bf Adaptive \\ \bf data fusion}
&
\makecell[l]{\bf CCOD \\ \bf benchmark}
&
\makecell[l]{\bf Causal contrast}
\\
\midrule

\cite{westling2024inference}
&
Treatment-specific survival in a single population
&
$\checkmark$
&
$\times$
&
$\times$
&
$\checkmark$: SD
\\
\addlinespace

\cite{han2025federated}
&
ATE under continuous outcomes in the target site
&
N/A
&
$\checkmark$
&
$\times$
&
$\checkmark$: ATE
\\
\addlinespace

\cite{liuprivacy}
&
Target-site treatment-specific survival
&
$\checkmark$
&
$\checkmark$
&
$\times$
&
$\times$
\\
\addlinespace

\textbf{Proposed framework}
&
Target-site treatment-specific survival and related causal contrasts
&
$\checkmark$
&
$\checkmark$
&
$\checkmark$
&
$\checkmark$: SD, SR, RMST, and RMST difference
\\
\bottomrule
\end{tabular}\scriptsize
\begin{tablenotes}
    \item N/A: not applicable; ATE: average treatment effect; SD: survival difference; SR: survival ratio; RMST: restricted mean survival time. 
\end{tablenotes}
\end{table}

\subsection{Summary of notation}\label{subapp:notation-summary}

We summarize all notation used in the main text (Table \ref{tab:notation-main}) and this supplemental material (Table \ref{tab:notation-supp}) for readers' convenience and to maintain consistency across the theoretical developments.

\begin{center}
\small
\setlength{\LTleft}{0pt}
\setlength{\LTright}{0pt}
\begin{longtable}{p{0.24\linewidth} p{0.73\linewidth}}
\caption{Summary of notation used in the main text.}
\label{tab:notation-main} \\
\toprule
\textbf{Notation} & \textbf{Meaning} \\
\midrule
\endfirsthead

\caption{(continued).} \\
\toprule
\textbf{Notation} & \textbf{Meaning} \\
\midrule
\endhead

\midrule
\endfoot

\bottomrule
\endlastfoot

\multicolumn{2}{l}{\textbf{Observed data and setup}} \\
$K$ & Total number of data sources/sites/regions. \\
$n$ & Total sample size across all sites. \\
$n_k$ & Sample size from site $k$ ($k=0,1,\dots,K-1$). \\
$R$ & Site indicator; $R=0$ denotes the target site, and $R=1,\dots,K-1$ denote source sites. \\
$\mb X$ & Baseline covariates. \\
$\Ix(\cdot)$ & Indicator function. \\
$A$ & Treatment indicator; $A=1$ for active bnAb treatment and $A=0$ for placebo. \\
$T^{(a)}$ & Potential event time under treatment $a$. \\
$C^{(a)}$ & Potential censoring time under treatment $a$. \\
$T=T^{(A)}$ & Observed event time under the received treatment. \\
$C=C^{(A)}$ & Observed censoring time under the received treatment. \\
$Y=\min(T,C)$ & Observed follow-up time. \\
$\Delta=\Ix(T\le C)$ & Event indicator. \\
$\mc O=(\mb X,A,Y,\Delta,R)$ & Observed data vector for one participant. \\
$\Px$, $\Ex$, $\Vx$ & Probability measure, expectation, and variance under the data-generating distribution. \\
$\widehat\Px$ & Plug-in version of a functional obtained by replacing nuisance functions with their estimators. \\
$\Px_n$ & Full-sample empirical average, $n^{-1}\sum_{i=1}^n (\cdot)$. \\
$\bigCI$ & Statistical independence. \\
$\estimand$ & Target-site treatment-specific survival probability, $\Px(T^{(a)}>t \mid R=0)$. \\

\addlinespace
\multicolumn{2}{l}{\textbf{Site-specific nuisance functions}} \\
$S^k(t\mid a,\mb X)$ & Site-$k$ conditional survival function for the event time, $\Px(T>t\mid A=a,\mb X,R=k)$. \\
$\Lambda^k(t\mid a,\mb X)$ & Site-$k$ conditional cumulative hazard function for the event time. \\
$N^k_{\delta}(t\mid a,\mb X)$ & Site-$k$ conditional cumulative incidence function for event type $\delta$, where $\delta=1$ denotes the event and $\delta=0$ denotes censoring. \\
$D^k(t\mid a,\mb X)$ & Site-$k$ conditional at-risk probability, $\Px(Y\ge t\mid A=a,\mb X,R=k)$. \\
$\pi^k(a\mid\mb X)$ & Site-$k$ propensity score, $\Px(A=a\mid \mb X,R=k)$. \\
$G^k(t\mid a,\mb X)$ & Site-$k$ conditional survival function of the censoring time, $\Px(C>t\mid A=a,\mb X,R=k)$. \\
$\prodi$ & Product-integral operator \citep{gill1990survey}. \\

\addlinespace
\multicolumn{2}{l}{\textbf{Target-only and CCOD estimation}} \\
$\varphi^{*0}_{t,a}(\mc O;\Px)$ & EIF for the target-only estimator of $\theta^0(t,a)$. \\
$\tgtEst$ & Target-only estimator of $\theta^0(t,a)$. \\
$\bar S(t\mid a,\mb X)$ & Global conditional survival function under the CCOD assumption. \\
$\bar\Lambda(t\mid a,\mb X)$ & Global conditional cumulative hazard under the CCOD assumption. \\
$\bar\pi(a\mid\mb X)$ & Global propensity score, $\Px(A=a\mid \mb X)$. \\
$\bar G(t\mid a,\mb X)$ & Global censoring survival function, $\Px(C>t\mid A=a,\mb X)$. \\
$q^0(\mb X)$ & Target-site selection probability, $\Px(R=0\mid \mb X)$. \\
$\varphi^{*\text{CCOD}}_{t,a}(\mc O;\Px)$ & EIF of $\theta^0(t,a)$ under the CCOD model. \\
$\glbEst$ & CCOD estimator of $\theta^0(t,a)$. \\
$M$ & Number of folds used in cross-fitting. \\
$\widehat{\bar\pi}_m$, $\widehat{\bar G}_m$, $\widehat{\bar S}_m$, $\widehat{\bar\Lambda}_m$, $\widehat q^0_m$ & Fold-$m$ estimators of the corresponding global nuisance functions. \\
$\bar\pi_\infty$, $\bar G_\infty$, $\bar S_\infty$, $\bar\Lambda_\infty$, $q^0_\infty$ & Generic probability limits of the estimated global nuisance functions. \\

\addlinespace
\multicolumn{2}{l}{\textbf{Federated estimation}} \\
$\omega^{k,0}(\mb X)$ & Density ratio between the target site and site $k$, $\Px(\mb X\mid R=0)/\Px(\mb X\mid R=k)$. \\
$\varphi^{*k,0}_{t,a}(\mc O;\Px)$ & Site-$k$ influence function for estimating the target-site survival function under a working partial CCOD assumption. \\
$\skEst$ & Site-$k$ covariate density ratio-adjusted local estimator targeting $\theta^0(t,a)$. \\
$\widehat\chi_{n,t,a}^{k,0}$ & Site-specific discrepancy measure, $\widehat\theta_n^{k,0}(t,a)-\widehat\theta_n^0(t,a)$. \\
$\bd\eta_{t,a}=(\eta^0_{t,a},\dots,\eta^{K-1}_{t,a})$ & Vector of treatment- and time-specific federated weights. \\
$Q(\bd\eta_{t,a})$ & Penalized objective function used to estimate federated weights. \\
$\lambda$ & Tuning parameter in the $\ell_1$-penalized federated objective. \\
$\fedEst$ & Federated estimator of $\theta^0(t,a)$. \\
$\widehat{\mc V}_{t,a}^{\text{fed}}$ & Estimated asymptotic variance of the federated estimator. \\
$\mc S=\{1,\dots,K-1\}$ & Set of source sites. \\
$\mc S^*_{t,a}$ & Oracle set of informative source sites at $(t,a)$, i.e., sites satisfying $\theta^k(t,a)=\theta^0(t,a)$. \\
$\bar{\bd\eta}_{t,a}$ & Oracle federated weight vector minimizing asymptotic variance. \\
\addlinespace
\multicolumn{2}{l}{\textbf{Additional causal contrasts}} \\
$\delta^0(t)$ & Target-site survival survival difference, $\theta^0(t,1)-\theta^0(t,0)$. \\
$\rho^0(t)$ & Target-site survival ratio, $\theta^0(t,1)/\theta^0(t,0)$. \\
$\Theta^0(\tau,a)$ & Target-site RMST under treatment $a$, $\int_0^\tau \theta^0(t,a)\,dt$. \\
$\Delta^0(\tau)$ & Target-site RMST difference, $\Theta^0(\tau,1)-\Theta^0(\tau,0)$. \\
$\varphi^{*\theta}_{t,a}(\mc O;\Px)$ & Generic EIF for $\theta^0(t,a)$ under a given estimation framework. \\
$\varphi^{*\delta}_t(\mc O;\Px)$ & EIF for the survival difference $\delta^0(t)$. \\
$\varphi^{*\rho}_t(\mc O;\Px)$ & EIF for the survival ratio $\rho^0(t)$. \\
$\varphi^{*\Theta}_{\tau,a}(\mc O;\Px)$ & EIF for the RMST $\Theta^0(\tau,a)$. \\
$\varphi^{*\Delta}_{\tau}(\mc O;\Px)$ & EIF for the RMST difference $\Delta^0(\tau)$. \\

\end{longtable}
\end{center}



\begin{center}
\small
\setlength{\LTleft}{0pt}
\setlength{\LTright}{0pt}
\begin{longtable}{p{0.17\linewidth} p{0.78\linewidth}}
\caption{Additional notation used in the supplemental material.}
\label{tab:notation-supp} \\
\toprule
\textbf{Notation} & \textbf{Meaning} \\
\midrule
\endfirsthead

\caption{(continued).} \\
\toprule
\textbf{Notation} & \textbf{Meaning} \\
\midrule
\endhead

\midrule
\endfoot

\bottomrule
\endlastfoot
$\Px_\infty$ & Generic probability limit of estimated nuisance functions. \\
$\Px_n^m$ & Empirical average over the $m$th validation fold. \\
$\Gx_n^m$ & Empirical process indexed by the $m$th validation fold. \\
$\mc V_m$ & Validation set in fold $m$ of cross-fitting. \\
$\mc I_m$ & Index sets associated with the $m$th validation fold. \\
$\mc T_m$ & Training sample associated with the $m$th fold. \\
$\Ex_\Px(\cdot)$ & Expectation taken under distribution $\Px$ in the usual probabilistic sense. \\
$\dot\ell(\mc O)$ & Score function of a regular parametric submodel. \\
$\dot\ell_{V\mid W}$ & Conditional score function of $V$ given $W$. \\
$\mu(\cdot)$ & Distribution of $\mb X$ induced by $\Px$. \\
$\mu^*(\cdot)$ & Dominating measure used in Radon-Nikodym arguments. \\
$\lambda^*(\cdot\mid a,\mb x)$ & Radon-Nikodym derivative of a cumulative hazard with respect to $\mu^*$. \\
$\Qx$ & Generic probability measure used in entropy calculations. \\
$N(\varepsilon,\mc F,\|\cdot\|)$ & Covering number of a function class. \\
$N_{[]}(\varepsilon,\mc F,\|\cdot\|)$ & Bracketing number of a function class. \\
$F$ & Envelope function for a function class. \\
$\bar N_\delta(t\mid a,\mb x)$ & Global conditional cumulative incidence function for type $\delta$ in the CCOD proof. \\
$\bar D(t\mid a,\mb x)$ & Global at-risk probability in the CCOD proof. \\
$\bar H(t,a,\mb x)$ & Auxiliary function introduced in the derivation of the CCOD EIF. \\
$\varphi^{\text{CCOD}}_{t,a}(\mc O;\Px)$ & CCOD influence function before centering to the mean. \\
$\varphi^{\text{CCOD}}_{\infty,t,a}(\mc O;\Px_\infty)$ & Limiting CCOD influence function evaluated at nuisance limits. \\
$\widehat\varphi^{\text{CCOD}}_{n,m,t,a}(\mc O;\widehat\Px)$ & Estimated fold-specific CCOD influence function. \\
$\mc V^{\text{CCOD}}$ & Asymptotic variance of the CCOD estimator. \\
$\mc V^{\text{TGT}}$ & Asymptotic variance of the target-only estimator. \\ 
$\varphi^{k,0}_{t,a}(\mc O;\Px)$ & Local site-$k$ influence function before emphasizing efficiency or centering notation. \\
$\varphi^{k,0}_{\infty,t,a}(\mc O;\Px_\infty)$ & Limiting local influence function evaluated at nuisance limits. \\
$\widehat\varphi^{k,0}_{n,m,t,a}(\mc O;\widehat\Px)$ & Estimated fold-specific local influence function.  \\ 
$\widehat{\bd\eta}_{t,a}$ & Estimated federated weight vector. \\
$\mathbb R^{S^*_{t,a}}$ & Weight space supported only on oracle-selected source sites. \\
\end{longtable}
\end{center}

\section{Technical Details of the CCOD estimator}\label{app:theory-CCOD}

\subsection{Proof of Theorem \ref{thm:CCOD-EIF}}\label{subsubapp:proof-ccod-EIF}

Theorem \ref{thm:CCOD-EIF} establishes the EIF for the CCOD estimator. The proof combines semiparametric efficiency theory from \cite{bickel1993efficient} with product-integral arguments for survival functionals from \cite{gill1990survey}. Specifically, we characterize the pathwise derivative of the target estimand under a regular parametric submodel and then exploit the product-integral representation of the survival curve to derive the EIF. 

\begin{proof}
Recall the estimand $\estimand=\Ex\{S^0(t\mid a,\mb X)\mid R=0\}=\Ex\{
\bar S(t\mid a,\mb X)\mid R=0\}$ that follows from Assumption \ref{asp:ccod}. We write $$
\bar\Lambda(t\mid a,\mb x) = \displaystyle\int_0^t\frac{\bar N_1(du\mid a,\mb x)}{\bar D(u\mid a,\mb x)},
$$
where $\bar N_\delta(du\mid a,\mb x) = \Px(Y\leq t,\Delta=\delta\mid A=a, \mb X=\mb x)$ and $\bar D(t\mid a,\mb x)=\Px(Y\geq t\mid A=a,\mb X=\mb x)$.

Recall that a mean zero, finite variance function $\varphi^{*0}_{t,a}(\mc O; \Px)$ is called an \textit{influence function} of the target estimand (a functional) $\theta^0(t,a)=\theta^0(t,a; \Px)$ at $\Px$ if, for any one-dimensional regular parametric submodel $\{\Px_{\epsilon} : \epsilon \in [0,1)\}$ through $\Px \equiv \Px_0$, 
$$
\left.\frac{\partial}{\partial\epsilon} \theta^0(t,a; \Px_{\epsilon})\right|_{\epsilon = 0} = \Px_{\Px}[\varphi^{*0}_{t,a}(\mc O; \Px)\dot{\ell}(\mc O)],
$$
where $\dot{\ell}(\mc O)$ is the score function of the submodel at $\epsilon = 0$ (i.e., typically, $\dot{\ell}(\mc O) = {\partial}\log{\{p_{\epsilon}(\mc O)\}}/{\partial\epsilon}~|_{\epsilon = 0}$), where $p_\epsilon(\cdot)$ denotes the probability density (likelihood) function under submodel $\Px_\epsilon$ \citep{bickel1993efficient}.
Thus, to find the EIF, we begin by writing the equation
\begin{align}\label{eq:barEIF-1}
   0 & = \frac{\partial}{\partial\epsilon}\theta^0(t,a; \Px_{\epsilon})\bigg|_{\epsilon=0} = \frac{\partial}{\partial\epsilon}\Px_{\Px_{\epsilon}}\{\bar S_{\epsilon}(t\mid a,\mb X)\mid R=0\}\bigg|_{\epsilon=0} \nonumber \\ 
   & = \Ex\{[\bar S(t\mid a, \mb X)-\estimand]\dot{\ell}_{\mb X\mid R=0}\mid R=0\} +  \Ex\left\{\int\frac{\partial}{\partial\epsilon}\bar S_{\epsilon}(t\mid a,\mb x)\bigg|_{\epsilon=0}\mu(d\mb x)~\bigg|~ R=0\right\},
\end{align}
where $\mu(\cdot)$ denotes the distribution of $\mb X$ induced by $\Px$ and, for any sets of variables $V$ and $W$, $\dot{\ell}_{V \mid W}$ denotes the conditional score function of $V$ given $W$, i.e., typically $\partial\log{\{p_{\epsilon}(V \mid W)\}}/{\partial\epsilon}~|_{\epsilon = 0}$. Note that such scores always satisfy $\Px_{\Px}(\dot{\ell}_{V \mid W} \mid W) = 0$ \citep{bickel1993efficient}. 

The derivation of Equation \eqref{eq:barEIF-1} follows from the decomposition of the pathwise derivative with respect to $\epsilon$ (chain rule). In particular, $\epsilon$ appears both in the outer expectation $\Ex_{\Px_{\epsilon}}$ and in the conditional survival function $\bar S_{\epsilon}(t \mid a, \mb X)$ inside the expectation. Applying the chain rule, the derivative therefore separates into two parts. 

The first term,
$\Ex\{[\bar S(t\mid a, \mb X)-\estimand]\dot{\ell}_{\mb X\mid R=0}\mid R=0\}$, arises from differentiating the conditional distribution of $\mb X \mid R=0$ while treating $\bar S(t \mid a, \mb X)$ as fixed. It is the standard score representation for the pathwise derivative of a conditional expectation. Particularly, note that 
\[
\Ex_{\Px_{\epsilon}}\{\bar S(t\mid a, \mb X)\mid R=0\}
= \int \bar S(t\mid a, \mb x)\,f_{\epsilon}(\mb x\mid R=0)\,\mu(d\mb x),
\]
we differentiate under the integral sign to obtain
\[
\frac{\partial}{\partial\epsilon}\Ex_{\Px_{\epsilon}}\{\bar S(t\mid a, \mb X)\mid R=0\}\bigg|_{\epsilon=0}
=
\int \bar S(t\mid a, \mb x)\frac{\partial}{\partial\epsilon}f_{\epsilon}(\mb x\mid R=0)\bigg|_{\epsilon=0}\mu(d\mb x).
\]
Using
\[
\frac{\partial}{\partial\epsilon}f_{\epsilon}(\mb x\mid R=0)\bigg|_{\epsilon=0}
=
f(\mb x\mid R=0)\dot{\ell}_{\mb X\mid R=0}(\mb x),
\]
it follows that
\[
\frac{\partial}{\partial\epsilon}\Ex_{\Px_{\epsilon}}\{\bar S(t\mid a, \mb X)\mid R=0\}\bigg|_{\epsilon=0}
=
\Ex\{\bar S(t\mid a, \mb X)\dot{\ell}_{\mb X\mid R=0}\mid R=0\}.
\]
Furthermore, because $\Ex[\estimand \dot{\ell}_{\mb X\mid R=0}] = \estimand\Ex\{\dot{\ell}_{\mb X\mid R=0}\} = 0$ by property of the score function, we have 
$$
\frac{\partial}{\partial\epsilon}\Ex_{\Px_{\epsilon}}\{\bar S(t\mid a, \mb X)\mid R=0\}\bigg|_{\epsilon=0} = \Ex\{[\bar S(t\mid a, \mb X)-\estimand]\dot{\ell}_{\mb X\mid R=0}\mid R=0\}.
$$
The second term arises from differentiating $\bar S{\epsilon}(t \mid a, \mb X)$ with respect to $\epsilon$ while keeping the distribution of $\mb X \mid R=0$ fixed, where we derive more details below. Together, these two components yield Equation \eqref{eq:barEIF-1}. 

For the derivative of $\bar S_\epsilon$ w.r.t. $\epsilon$, again by the chain rule, we decompose it as $(\partial\bar S_\epsilon/\partial\bar\Lambda_\epsilon)\times (\partial\bar \Lambda_\epsilon/\partial\epsilon)$. For the first part $\partial\bar S_\epsilon/\partial\bar\Lambda_\epsilon$, we leverage Theorem 8 in \cite{gill1990survey}. Specifically, the mapping $H\mapsto\bar S(t; H):=\prodi_{(0,t]}\{1+H(du)\}$ is Hadamard differentiable relative to the supremum norm with derivative 
$$
\alpha\mapsto\bar S(t; H)\int_0^t\frac{\bar S(u-; H)}{\bar S(u; H)}\alpha(du)
$$
at $H$. 

Thus, by letting $H(t)=\bar\Lambda_\epsilon(t\mid a,\mb x)$ and the chain rule, the integrand in the second term becomes
\begin{align*}
\frac{\partial}{\partial\epsilon}\Prodi_{(0,t]}\{1-\bar \Lambda_{\epsilon}(du\mid a,\mb x)\}\bigg|_{\epsilon=0} & = -\bar S(t\mid a,\mb x)\int_0^t\frac{\bar S(u-\mid a,\mb x)}{\bar S(u\mid a,\mb x)}\frac{\partial}{\partial\epsilon}\bar\Lambda_\epsilon(du\mid a,\mb x)\bigg|_{\epsilon=0}.
\end{align*}
Furthermore,
\begin{align*}
    \frac{\partial}{\partial\epsilon}\bar \Lambda_\epsilon(du\mid a,\mb x)\bigg|_{\epsilon=0} & = \dfrac{\frac{\partial}{\partial\epsilon}\bar N_{1,\epsilon}(du\mid a,\mb x)\mid_{\epsilon=0}}{\bar D(u\mid a,\mb x)} - \dfrac{\frac{\partial}{\partial\epsilon}\bar D_{\epsilon}(u\mid a,\mb x)\mid_{\epsilon=0}\bar N_{1,\epsilon}(du\mid a,\mb x)}{\bar D(u\mid a,\mb x)^2}.
\end{align*}
In addition,
\begin{align*}
    \frac{\partial}{\partial\epsilon}\bar N_{1,\epsilon}(du\mid a,\mb x)\bigg|_{\epsilon=0} 
    & = \frac{\partial}{\partial\epsilon}\Px_\epsilon(Y\leq u,\Delta=1\mid A=a, \mb X=\mb x)\bigg|_{\epsilon=0} \\ 
    & = \frac{\partial}{\partial\epsilon}\iint\Ix(y\leq u,\delta=1)\Px_\epsilon(dy,d\delta\mid a,\mb x)\bigg|_{\epsilon=0}\\
    & = \iint\Ix(y\leq u,\delta=1)\dot{\ell}(y,\delta\mid a,\mb x)\Px(dy,d\delta\mid a,\mb x) \\
    & = \int_\delta\Ix(\delta=1)\dot{\ell}(u,\delta\mid a,\mb x)\Px(du, d\delta\mid a,\mb x), 
\end{align*}
and 
\begin{align*}
    \frac{\partial}{\partial\epsilon}\bar D_{\epsilon}(u\mid a,\mb x)\bigg|_{\epsilon=0} & = \frac{\partial}{\partial\epsilon}\Px_\epsilon(Y\geq u\mid A=a, \mb X=\mb x)\bigg|_{\epsilon=0} = \frac{\partial}{\partial\epsilon}\iint\Ix(y\geq u)\Px_\epsilon(dy, d\delta\mid a,\mb x)\bigg|_{\epsilon=0}\\
    & = \iint\Ix(y\geq u)\dot{\ell}(y,\delta\mid a,\mb x)\Px(dy,d\delta\mid a,\mb x).
\end{align*}
Therefore, plugging-in the above expressions, 
\begin{align*}
    & \frac{\partial}{\partial\epsilon}\iint \Prodi_{(0,t]}\{1-\bar \Lambda_{\epsilon}(du\mid a,\mb x)\}\mu(d\mb x)\bigg|_{\epsilon=0} \\
    & = \iiint -\Ix(y\leq t, \delta=1)\frac{\bar S(t\mid a, \mb x)\bar S(y-\mid a, \mb x)}{\bar S(y\mid a, \mb x)\bar D(y\mid a,\mb x)}\dot{\ell}(y,\delta\mid a,\mb x)\Px(dy,d\delta\mid a,\mb x)\mu(d\mb x) \\ 
    & \quad + \iiiint\Ix(u\leq t,u\leq y)\frac{\bar S(t\mid a, \mb x)\bar S(u-\mid a, \mb x)}{\bar S(u\mid a, \mb x)\bar D(u\mid a,\mb x)}\\
    & \quad\qquad \times \dot{\ell}(y,\delta\mid a,\mb x)\Px(dy,d\delta\mid a,\mb x)\bar N_1(du\mid a,\mb x)\mu(d\mb x) \\
    & = \iiint -\Ix(y\leq t, \delta=1)\frac{\bar S(t\mid a, \mb x)\bar S(y-\mid a, \mb x)}{\bar S(y\mid a, \mb x)\bar D(y\mid a,\mb x)}\dot{\ell}(y,\delta\mid a,\mb x)\Px(dy,d\delta\mid a,\mb x)\mu(d\mb x) \\ 
    & \quad + \iiint \bar S(t\mid a, \mb x)\int_0^{t\wedge y}\frac{\bar S(u-\mid a, \mb x)}{\bar S(u\mid a, \mb x)\bar D(u\mid a,\mb x)^2}\bar N_1(du\mid a,\mb x)\\
    & \quad\qquad \times \dot{\ell}(y,\delta\mid a,\mb x)\Px(dy,d\delta\mid a,\mb x)\mu(d\mb x) \\
    & = \Ex\left[\bar S(t\mid a,\mb X)\frac{\Ix(A=a)}{\bar\pi(a\mid\mb X)}\left\{\bar H(t\wedge Y, a,\mb X)-\frac{\Ix(Y\leq t,\Delta=1)\bar S(Y-\mid a,\mb X)}{\bar S(Y\mid a,\mb X)\bar D(Y\mid a,\mb X)}\right\}\dot{\ell}(Y,\Delta\mid a,\mb X)\right],
\end{align*}
where
\begin{align*}
    \bar H(t,a,\mb x) = \int_0^t\frac{\bar S(u-\mid a, \mb x)\bar N_1(du\mid a,\mb x)}{\bar S(u\mid a, \mb x)\bar D(u\mid a,\mb x)^2}. 
\end{align*}
Now, we note that 
\begin{align*}
    \Ex\left[\frac{\Ix(Y\leq t,\Delta=1)\bar S(Y-\mid A,\mb X)}{\bar S(Y\mid A,\mb X)\bar D(Y\mid A,\mb X)}~\bigg|~ A=a,\mb X=\mb x\right] = \int_0^t\frac{\bar S(y-\mid a, \mb x)\bar N_1(dy\mid a,\mb x)}{\bar S(y\mid a, \mb x)\bar D(y\mid a,\mb x)},
\end{align*}
and 
\begin{align*}
    & \Ex\{\bar H(t\wedge Y,A,\mb X)\mid A=a,\mb X=\mb x\} \\
    & = \iint^t \Ix(u\leq y)\frac{\bar S(u-\mid a, \mb x)\bar N_1(du\mid a,\mb x)}{\bar S(u\mid a, \mb x)\bar D(u\mid a,\mb x)^2}\Px(dy\mid a,\mb x) \\
    & = \int_0^t\Px(Y\geq u\mid A=a, \mb X=\mb x)\frac{\bar S(u-\mid a,\mb x)\bar N_1(du\mid a,\mb x)}{\bar S(u\mid a, \mb x)\bar D(u\mid a,\mb x)^2}\Px(dy\mid a,\mb x)\\
    & = \int_0^t\frac{\bar S(u-\mid a, \mb x)\bar N_1(du\mid a,\mb x)}{\bar S(u\mid a, \mb x)\bar D(u\mid a,\mb x)}.
\end{align*}
Therefore, 
\begin{align*}
    \Ex\left[\bar H(t\wedge Y,A,\mb X) - \frac{\Ix(Y\leq t,\Delta=1)\bar S(Y-\mid A,\mb X)}{\bar S(Y\mid A,\mb X)\bar D(Y\mid A,\mb X)}~\bigg|~ A,\mb X\right] = 0
\end{align*}
almost surely. By the fact that score functions have (conditional) mean zero, and by the tower property of expectation, the above calculations implies that
\begin{align*}
    & \frac{\partial}{\partial\epsilon}\iint \Prodi_{(0,t]}\{1-\bar \Lambda_{\epsilon}(du\mid a,\mb x)\}\mu(d\mb x)\bigg|_{\epsilon=0} \\
    & =  \Ex\left[\bar S(t\mid a, \mb X)\frac{\Ix(A=a)}{\bar \pi(a\mid\mb X)}\left\{\bar H(t\wedge Y,A,\mb X) - \frac{\Ix(Y\leq t,\Delta=1)\bar S(Y-\mid A,\mb X)}{\bar S(Y\mid A,\mb X)\bar D(Y\mid A,\mb X)}\right\}\dot{\ell}(\mc O)\right].
\end{align*}
Combining these results with the facts that $\bar N_1(du\mid a,\mb x)/\bar D(u\mid a,\mb x) = \bar \Lambda(du\mid a,\mb x)$ and $\bar D(u\mid a,\mb x)=\bar S(u-\mid a,\mb x)\bar G(u-\mid a,\mb x)$, we can rewrite \eqref{eq:barEIF-1} at the beginning as follows:
\begin{align*}
    & \frac{\partial}{\partial\epsilon}\theta^0(t,a; \Px_{\epsilon})\bigg|_{\epsilon=0} \\ 
    & = \Ex\bigg[\frac{\Ix(R=0)}{\Px(R=0)}[\bar S(t\mid a, \mb X)-\estimand]\dot{\ell}(\mc O) - \frac{\Ix(R=0)}{\Px(R=0)}\Ex\bigg\{\bar S(t\mid a, \mb X)\frac{\Ix(A=a)}{\bar \pi(a\mid\mb X)} \\
    & \quad\quad \times \left\{\frac{\Ix(Y\leq t,\Delta=1)}{\bar S(y\mid a,\mb X)\bar G(y-\mid a,\mb X)}-\int_0^{t\wedge y}\frac{\bar \Lambda(du\mid a,\mb X)}{\bar S(u\mid a,\mb X)\bar G(u-\mid a,\mb X)}\right\}\dot{\ell}(\mc O)~\bigg|~\mb X\bigg\}\bigg]\\
    & = \Ex\left[\frac{\Ix(R=0)}{\Px(R=0)}[\bar S(t\mid a, \mb X)-\estimand]\dot{\ell}(\mc O)\right] - \Ex\bigg[\frac{\Px(R=0\mid\mb X)}{\Px(R=0)}\bar S(t\mid a, \mb X)\frac{\Ix(A=a)}{\bar \pi(a\mid\mb X)} \\
    & \quad\quad \times \left\{\frac{\Ix(Y\leq t,\Delta=1)}{\bar S(y\mid\mb X)\bar G(y-\mid a,\mb X)}-\int_0^{t\wedge y}\frac{\bar \Lambda(du\mid a,\mb X)}{\bar S(u\mid a,\mb X)\bar G(u-\mid a,\mb X)}\right\}\dot{\ell}(\mc O)\bigg].
\end{align*}

Therefore, the EIF of the target estimand $\estimand$ at $\Px$ is found as 
\begin{align*}
    \varphi^{\text{CCOD}}_{t,a}(\mc O; \bar S, \bar G, \bar \pi) & = \varphi^{\text{CCOD}}_{t,a}(\mc O;\Px) \\ & = \frac{\Ix(R=0)}{\Px(R=0)}\{\bar S(t\mid a, \mb X)-\estimand\} -\frac{q^0(\mb X)}{\Px(R=0)}\frac{\Ix(A=a)}{\bar\pi(a\mid\mb X)}\bar S(t\mid a, \mb X)
    \\
    & \quad\times
    \left[\frac{\Ix(Y\leq t,\Delta=1)}{\bar S(Y\mid a, \mb X)\bar G(Y-\mid a,\mb X)}-\int_0^{t\wedge Y}\frac{\bar \Lambda(du\mid a,\mb X)}{\bar S(u\mid a, \mb X)\bar G(u-\mid a,\mb X)}\right],
\end{align*}
where $q^0(\mb X) = \Px(R=0\mid\mb X)$ is the target site propensity.
\end{proof}

\subsection{Efficiency comparison with the target-only estimator}\label{subsubsec:effcomp}

The difference between the asymptotic variances by the CCOD and target-only estimators could be characterized by the difference of mean squares of the two EIFs, which we express as follows. Denote $\mc V^{\text{CCOD}}$ and $\mc V^{\text{TGT}}$ as the asymptotic variances of the CCOD and target-only estimators, respectively. 

Note that because the EIF of the target-only estimator is \citep{westling2024inference}
\begin{align*}
    \varphi^{*0}_{t,a}(\mc O;\Px) = \frac{\Ix(R=0)}{\Px(R=0)}\left[\left\{1-\frac{\Ix(A=a)}{\pi^0(a\mid\mb X)}\mc H_{t,a}(\mc O;S^0,G^0)\right\}S^0(t\mid a, \mb X)-\theta^0(t,a)\right], 
\end{align*}
we have 
\begin{align*}
    \mc V^{\text{TGT}} & = \Vx\left\{\frac{\Ix(R=0)}{\Px(R=0)}S^0(t\mid a,\mb X)\right\} + \Vx\left\{\frac{\Ix(R=0)}{\Px(R=0)}S^0(t\mid a,\mb X)\frac{\Ix(A=a)}{\pi^0(a\mid\mb X)}\mc H_{t,a}(\mc O;S^0, G^0)\right\}. 
\end{align*}
It is straightforward to verify that the second term of the right-hand side above equals 
\begin{align*}
    & \frac{1}{\Px(R=0)^2}\Ex\left[S^0(t\mid a,\mb X)^2\frac{q^0(\mb X)}{\pi^0(a\mid\mb X)}\Ex\{\mc H_{t,a}(\mc O;S^0, G^0)^2\mid\mb X\}\right] \\
    & = \frac{1}{\Px(R=0)^2}\Ex\left[S^0(t\mid a,\mb X)^2\frac{q^0(\mb X)}{\pi^0(a\mid\mb X)}\Vx\{\mc H_{t,a}(\mc O;S^0, G^0)\mid\mb X\}\right]. 
\end{align*}
Similarly,  
\begin{align*}
    \mc V^{\text{CCOD}} & = \Vx\left\{\frac{\Ix(R=0)}{\Px(R=0)}S^0(t\mid a,\mb X)\right\} \\
    & \quad + \frac{1}{\Px(R=0)^2}\Ex\left[S^0(t\mid a,\mb X)^2\frac{q^0(\mb X)^2}{\bar\pi(a\mid\mb X)}\Vx\{\mc H_{t,a}(\mc O;\bar S, \bar G)\mid\mb X\}\right],
\end{align*}
since under CCOD assumption, $S^0(t\mid a,\mb X)=\bar S(t\mid a,\mb X)$. Therefore, the difference of the asymptotic variance equals 
\begin{align}
    & \mc V^{\text{TGT}} - \mc V^{\text{CCOD}}\nonumber \\ 
    & = \frac{1}{\Px(R=0)^2}\Ex\bigg[S^0(t\mid a,\mb X)^2q^0(\mb X)\nonumber\\
    & \qquad\times\left\{\frac{1}{\pi^0(a\mid\mb X)}\Vx\{\mc H_{t,a}(\mc O;S^0, G^0)\mid\mb X\} - \frac{q^0(\mb X)}{\bar\pi(a\mid\mb X)}\Vx\{\mc H_{t,a}(\mc O;\bar S, \bar G)\mid\mb X\}\right\}\bigg]\nonumber\\
    & = \frac{1}{\Px(R=0)^2}\Ex\left[S^0(t\mid a,\mb X)^2q^0(\mb X)\{1-q^0(\mb X)\}\frac{\Vx\{\mc H_{t,a}(\mc O;S^0, G^0)\mid\mb X\}}{\pi^0(a\mid\mb X)}\right] \label{eq:eifdiff-I} \\
    & \quad + \frac{1}{\Px(R=0)^2}\Ex\left[S^0(t\mid a,\mb X)^2q^0(\mb X)^2\left\{\frac{\Vx\{\mc H_{t,a}(\mc O;S^0, G^0)\mid\mb X\}}{\pi^0(a\mid\mb X)}-\frac{\Vx\{\mc H_{t,a}(\mc O;\bar S, \bar G)\mid\mb X\}}{\bar\pi(a\mid\mb X)}\right\}\right] \label{eq:eifdiff-II}. 
\end{align}
We can see that the first term \eqref{eq:eifdiff-I} is always non-negative. For given $S^0(t\mid a,\mb X)$ and $\mc H_{t,a}(\mc O;S^0,G^0)$, \eqref{eq:eifdiff-I} is larger when $q^0(\mb X)$ is closer to $0.5$ and when $\pi^0(a\mid\mb X)$ is smaller. 

When $\bar G=G^0$ and $\bar\pi=\pi^0$ with probability 1, \eqref{eq:eifdiff-II} is clearly zero, hence in this special case the efficiency gain is fully characterized by \eqref{eq:eifdiff-I}.

\subsection{RAL and uniform RAL properties of the CCOD estimator}\label{subapp:uniformRAL-CCOD}

In this section, we formally state the RAL and uniform RAL properties of the CCOD estimator. The RAL result has been stated in Theorem \ref{thm:RAL-ccod} in the main text. Below, we state regularity conditions needed for both RAL and uniform RAL results of the CCOD estimator. 

Recall that the subscript $m$ denotes nuisance functions estimated in the $m$-th fold of an $M$-fold cross-fitting procedure, with $m=1,\dots,M$. Denote $(\bar\pi_\infty,\bar G_\infty,\bar\Lambda_\infty,\bar S_\infty,q^0_\infty)$ some general probability limits of the estimated nuisance functions $(\widehat{\bar\pi}_m,\widehat{\bar G}_m,\widehat{\bar S}_m,\widehat{\bar\Lambda}_m,\widehat q^0_m)$. 

\begin{condition}\label{cond:nuisance-ccod}
    There exists $\bar\pi_\infty$, $\bar G_\infty$, $\bar\Lambda_\infty$, $\bar S_\infty$ and $q^0_\infty$ such that for $a\in\{0,1\}$, 
    \begin{align*}
       & \max_m \Px\Bigg[
        \left(\frac{1}{\widehat{\bar\pi}_m(a\mid\mb X)}-\frac{1}{\bar\pi_\infty(a\mid\mb X)}\right)^2
        + \left(\widehat q^0_m(\mb X)-q^0_\infty(\mb X)\right)^2
        + \\
        & \qquad \sup_{u\in[0,t]}\Bigg\{
        \left\vert\frac{1}{\widehat{\bar G}_m(u-\mid a,\mb X)}-\frac{1}{\bar G_\infty(u-\mid a,\mb X)}\right\vert^2
        + \left\vert\frac{\widehat{\bar S}_m(t\mid a,\mb X)}{\widehat{\bar S}_{m}(u\mid a,\mb X)}-\frac{\bar S_\infty(t\mid a,\mb X)}{\bar S_\infty(u\mid a,\mb X)}\right\vert^2
        \Bigg\}
        \Bigg] \to_p 0.
    \end{align*}
\end{condition}

\begin{condition}\label{cond:bound-pi-G-ccod}
    There exists $\eta\in(0,\infty)$ such that for $\Px$-almost all $\mb X$, $\min\{\widehat{\bar\pi}_m(a\mid\mb X), \bar\pi_\infty(a\mid\mb x), \widehat q^0_m(\mb X), q^0_\infty(\mb X), \widehat{\bar G}_m(t-\mid a,\mb X), \bar G_\infty(t-\mid a,\mb X)\}\geq 1/\eta$ 
    with probability tending to 1. 
\end{condition}

\begin{condition}\label{cond:prod-error-ccod}
Define
\begin{align*}
    \bar r_{n,t,a,1} & = \max_m\Px\vert\{\widehat{\bar\pi}_{m}(a\mid\mb X)-\bar\pi_\infty(a\mid\mb X)\}\cdot\{\widehat{\bar S}_{m}(t\mid a,\mb X)-\bar S_\infty(t\mid a,\mb X)\}\vert, \\
    \bar r_{n,t,a,2} & = \max_m\Px\vert\{\widehat q^0_{m}(\mb X)-q^0_\infty(\mb X)\}\cdot\{\widehat{\bar S}_{m}(t\mid a,\mb X)-\bar S_\infty(t\mid a,\mb X)\}\vert, \text{ and}\\
    \bar r_{n,t,a,3} & = \max_m\Px\bigg\vert\widehat{\bar S}_{m}(t\mid a,\mb X)\int_0^t\left\{\frac{\bar G_\infty(u-\mid a,\mb X)}{\widehat {\bar G}_{m}(u-\mid a,\mb X)}-1\right\}\left(\frac{\bar S_\infty}{\widehat{\bar S}_{m}}-1\right)(du\mid a,\mb X)\bigg\vert. 
\end{align*}
Then, it holds that all $\bar r_{n,t,a,1}=o_p(n^{-1/2})$, $\bar r_{n,t,a,2}=o_p(n^{-1/2})$ and $\bar r_{n,t,a,3}=o_p(n^{-1/2})$.
\end{condition}

\begin{condition}\label{cond:uniform-S-ccod}
$$
\max_m\Px\left[\sup_{u\in[0,t]}\sup_{v\in[0,u]}\left\vert\frac{\widehat{\bar S}_{m}(u\mid a,\mb X)}{\widehat{\bar S}_{m}(v\mid a,\mb X)}-\frac{\bar S_\infty(u\mid a,\mb X)}{\bar S_\infty(v\mid a,\mb X)}\right\vert\right]^2\to_p 0. 
$$
\end{condition}

\begin{condition}\label{cond:prod-error-unif-ccod}
     It holds that all $\sup_{u\in[0,t]}\bar r_{n,u,a,1}=o_p(n^{-1/2})$, $\sup_{u\in[0,t]}\bar r_{n,u,a,2}=o_p(n^{-1/2})$, and $\sup_{u\in[0,t]}\bar r_{n,u,a,3}=o_p(n^{-1/2})$.
\end{condition}

\begin{theorem}\label{thm:uniformRAL-ccod}
    If Conditions \ref{cond:nuisance-ccod}--\ref{cond:prod-error-ccod} hold with $(\bar\pi_\infty,\bar G_\infty,\bar\Lambda_\infty,\bar S_\infty,q^0_\infty) = (\bar\pi,\bar G,\bar\Lambda,\bar S, q^0)$,  $\glbEst=\estimand+\Px_n(\varphi^{*\text{CCOD}}_{t,a}) + o_p(n^{-1/2})$. In particular, $n^{1/2}(\glbEst-\estimand)$  converges in distribution to a normal random variable with mean zero and variance $\sigma^2 = \Px[(\varphi^{*\text{CCOD}}_{t,a})^2]$. If Conditions \ref{cond:uniform-S-ccod}--\ref{cond:prod-error-unif-ccod} also hold with $(\bar\pi_\infty,\bar G_\infty,\bar\Lambda_\infty,\bar S_\infty,q^0_\infty) = (\bar\pi,\bar G,\bar\Lambda,\bar S, q^0)$,
    \begin{align*}
        \sup_{u\in[0,t]}\left\vert\widehat\theta^\text{CCOD}_n(u,a)-\theta^0(u,a)-\Px_n(\varphi^{*\text{CCOD}}_{u,a})\right\vert = o_p(n^{-1/2}). 
    \end{align*}
    In particular, $\{n^{1/2}(\widehat\theta^\text{CCOD}_n(u,a)-\theta^0(u,a)):u\in[0,t]\}$ converges weakly as a process in the space $\ell^\infty([0,t])$ of uniformly bounded functions on $[0,t]$ to a tight mean zero Gaussian process with covariance function $(u,v)\mapsto\Px(\varphi^{*\text{CCOD}}_{u,a}\varphi^{*\text{CCOD}}_{v,a})$. 
\end{theorem}

To prove Theorem \ref{thm:uniformRAL-ccod}, we first present some useful lemmata in Section \ref{subsubapp:lem-CCOD}. 

\subsection{Useful Lemmata for Theorem \ref{thm:uniformRAL-ccod}}\label{subsubapp:lem-CCOD}

To establish the RAL related results of the CCOD estimator in Theorem \ref{thm:RAL-ccod}, we first present some useful results and lemmata in this section. We assume that the corresponding nuisance functions converge to some general limits defined by $\Px_\infty$. We start by expressing the difference $\glbEst-\estimand$ as
\begin{align}\label{eq:glb-est-deco}
    & \Px_n[\widehat\varphi^{\text{CCOD}}_{t,a}(\mc O;\widehat\Px)] - \estimand \nonumber \\ 
    & = \frac1n\sum_{m=1}^M\sum_{i\in\mc I_m}\widehat\varphi^{\text{CCOD}}_{t,a}(\mc O_i) - \estimand \nonumber \\
    & = \Px_n[\varphi^{\text{CCOD}}_{\infty,t,a}]-\estimand+\frac1n\sum_{m=1}^M\sum_{i\in\mc I_m}\widehat\varphi^{\text{CCOD}}_{t,a}(\mc O_i) - \Px_n[\varphi^{\text{CCOD}}_{t,a}] \nonumber   \\
    & = \Px_n[\varphi^{*\text{CCOD}}_{\infty,t,a}] + \frac1n\sum_{m=1}^M\sum_{i\in\mc L^m}\left[\widehat\varphi^{\text{CCOD}}_{t,a}(\mc O_i) - \varphi^{\text{CCOD}}_{\infty,t,a}(\mc O_i)\right] \nonumber  \\
    & = \Px_n[\varphi^{*\text{CCOD}}_{\infty,t,a}] + \frac1M\sum_{m=1}^M\frac{Mn_m}{n}\frac1{n_m}\sum_{i\in\mc L^m}\left[\widehat\varphi^{\text{CCOD}}_{t,a}(\mc O_i) - \varphi^{\text{CCOD}}_{\infty,t,a}(\mc O_i)\right] \nonumber \\
    & = \Px_n[\varphi^{*\text{CCOD}}_{\infty,t,a}] + \frac1M\sum_{m=1}^M\frac{Mn_m}{n}\Px_n^m\left[\widehat\varphi^{\text{CCOD}}_{n,m,t,a} - \varphi^{\text{CCOD}}_{\infty,t,a}\right] \nonumber \\
    & = \Px_n[\varphi^{*\text{CCOD}}_{\infty,t,a}] + \frac1M\sum_{m=1}^M\frac{Mn_m}{n}(\Px_n^m-\Px)\left[\widehat\varphi^{\text{CCOD}}_{n,m,t,a} - \varphi^{\text{CCOD}}_{\infty,t,a}\right] \nonumber \\
    & \qquad + \frac1M\sum_{m=1}^M\frac{Mn_m}{n}\Px\left[\widehat\varphi^{\text{CCOD}}_{n,m,t,a} - \varphi^{\text{CCOD}}_{\infty,t,a}\right] \nonumber \\
    & = \Px_n[\varphi^{*\text{CCOD}}_{\infty,t,a}] + \frac1M\sum_{m=1}^M\frac{Mn_m^{1/2}}{n}\Gx_n^m\left[\widehat\varphi^{\text{CCOD}}_{n,m,t,a} - \varphi^{\text{CCOD}}_{\infty,t,a}\right] \nonumber \\
    & \qquad + \frac1M\sum_{m=1}^M\frac{Mn_m}{n}\Px\left[\widehat\varphi^{\text{CCOD}}_{n,m,t,a} - \estimand\right]. 
\end{align} 

Next, in Lemma \ref{lm:L2-bound-ccod-est}, we establish an $\ell_2(\Px)$-norm bound on the distance between the estimated IF and the limiting IF in terms of discrepancies in the nuisance parameters. To this end, we consider the following decompositions.

\begin{lemma}\label{lm:L2-bound-ccod-est}
     Under Condition \ref{cond:bound-pi-G-ccod}, there exists a universal constant $C=C(\eta)$ such that for each $m$, $n$, $t$, and $a$, 
\begin{align*}
\Px[\widehat\varphi^{\text{CCOD}}_{n,m,t,a} - \varphi^{\text{CCOD}}_{\infty,t,a}]^2 \leq C(\eta)\sum_{j=1}^5 \bar A_{j,n,m,t,a},
\end{align*}
where 
\begin{align*}
    \bar A_{1,n,m,t,a} & = \Px\left[\frac{1}{\Px_n^m(R=0)}-\frac{1}{\Px(R=0)}\right]^2,\\
    \bar A_{2,n,m,t,a} & = \Px\left[\widehat q^0_{m}(\mb X)-q^0_\infty(\mb X)\right]^2,\\
    \bar A_{3,n,m,t,a} & = \Px\left[\frac{1}{\widehat{\bar\pi}_{m}(a\mid\mb X)}-\frac{1}{\bar\pi_\infty(a\mid\mb X)}\right]^2,\\
    \bar A_{4,n,m,t,a} & = \Px\left[\sup_{u\in[0,t]}\left\vert\frac{1}{\widehat{\bar G}_{m}(u-\mid a,\mb X)}-\frac{1}{\bar G_\infty(u-\mid a,\mb X)}\right\vert\right]^2, \\
    \bar A_{5,n,m,t,a} & = \Px\left[\sup_{u\in[0,t]}\left\vert\frac{\widehat{\bar S}_{m}(t\mid a,\mb X)}{\widehat{\bar S}_{m}(u\mid a,\mb X)}-\frac{\bar S_\infty(t\mid a,\mb X)}{\bar S_\infty(u\mid a,\mb X)}\right\vert\right]^2. 
\end{align*}
\end{lemma}

\begin{proof}
    We first denote  
\begin{align*} 
   \bar B(\mc V_m) & = \frac{\Ix(A=a)}{\bar\pi(a\mid\mb X)}\bar S(t\mid a, \mb X)\left[\frac{\Ix(Y\leq t,\Delta=1)}{\bar S(Y\mid a, \mb X)\bar G(Y-\mid a,\mb X)}-\int_0^{t\wedge Y}\frac{\bar\Lambda(du\mid a,\mb X)}{\bar S(u\mid a, \mb X)\bar G(u-\mid a,\mb X)}\right]. 
\end{align*}
Then, we have the following decomposition:
\begin{align*}
    \widehat\varphi^{\text{CCOD}}_{t,a} - \varphi^{\text{CCOD}}_{\infty,t,a} = \sum_{j=1}^4 \bar U_{j,n,m,t,a},
\end{align*}
where 
\begin{align*}
    \bar U_{1,n,m,t,a} & = \left[\frac{\Ix(R=0)}{\Px_n^m(R=0)} - \frac{\Ix(R=0)}{\Px(R=0)}\right]\widehat{\bar S}_m(t\mid a,\mb x), \\
    \bar U_{2,n,m,t,a} & = \frac{\Ix(R=0)}{\Px(R=0)}\left[\widehat{\bar S}_m(t\mid a,\mb x) - \bar S_{\infty}(t\mid a,\mb x)\right], \\
    \bar U_{3,n,m,t,a} & = \left[\frac{\widehat q^0_m(\mb X)}{\Px_n^m(R=0)} - \frac{q^0(\mb X)}{\Px(R=0)}\right]\widehat{\bar B}_m(\mc V_m), \\
    \bar U_{4,n,m,t,a} & = \frac{q^0(\mb X)}{\Px(R=0)}\left[\widehat {\bar B}_m(\mc V_m) - \bar B_{\infty}(\mc V_m)\right].
\end{align*}
Note that the expression of $\widehat {\bar B}_m(\mc V_m) - \bar B_{\infty}(\mc V_m)$ can be found in Lemma 3 of \cite{westling2024inference}, while we only need to replace corresponding nuisance functions with the global version here, thus the detail is omitted. Then, by the triangle inequality, we have $\Px\left[\widehat\varphi^{\text{CCOD}}_{t,a} - \varphi^{\text{CCOD}}_{\infty,t,a}\right]^2\leq\left[\sum_{j=1}^4\{\Px\{(\bar U_{j,n,m,t,a})^2\}\}^{1/2}\right]^2$. Therefore, under Assumption \ref{asp:positivity} and Condition \ref{cond:bound-pi-G-ccod}, there exists a universal constant $C=C(\eta)$ such that the result in the statement holds. Thus, the proof is completed. 
\end{proof}

Next, we consider conditions that make the empirical process term $\Gx_n^m\left[\widehat\varphi^{\text{CCOD}}_{t,a} - \varphi^{\text{CCOD}}_{\infty,t,a}\right]$ to be $o_p(n^{-1/2})$. This requires some preliminaries from the empirical process theory \citep{van1996weak}. We first introduce the following notation and Lemma \ref{lm:lm4-west}. Then, we consider results for uniform convergence over $t\in[0,\tau]$. 

Given a class of functions $\mc F$ on a sample space $\mc X$, a norm $\Vert\cdot\Vert$, and $\varepsilon>0$, the covering number, denoted by $N(\varepsilon,\mc F,\Vert\cdot\Vert)$, is the minimal number of $\Vert\cdot\Vert$-balls of radius $\varepsilon$ needed to cover $\mc F$. The centers of these balls need not be in $\mc F$. An $(\varepsilon,\Vert\cdot\Vert)$ bracket is a set of the form $\{f\in\mc F: l(x)\leq f(x)\leq u(x)\text{ for all }x\in\mc X\}$ such that $\Vert u-l\Vert\leq\varepsilon$ and is denoted $[l,u]$. Here $l$ and $u$ need not be elements of $\mc F$. The bracketing number, denoted by $N_{[]}(\varepsilon,\mc F,\Vert\cdot\Vert)$, is the minimal number of $(\varepsilon, \Vert\cdot\Vert)$ brackets needed to cover $\mc F$. Here, the subscript ``$[]$'' in the bracketing number distinguishes it from the covering number.  It is well known that $N(\varepsilon,\mc F,\Vert\cdot\Vert)\leq N_{[]}(2\varepsilon,\mc F,\Vert\cdot\Vert)$. Readers are referred to \cite{van1996weak} for more theory on empirical processes and their applications. 

\begin{lemma}[Lemma 4 in \cite{westling2024inference}]\label{lm:lm4-west}
    Let $\mc F = \{x\mapsto f_t(x):t\in[0,\tau]\}$ be a class of functions on a sample space $\mc X$ such that $f_s(x)\leq f_t(x)$ for all $0\leq s\leq t\leq\tau$ and $x\in\mc X$, and such that an envelope $F$ for $\mc F$ satisfies $\Vert F\Vert_{\Px,2}<\infty$. Then $N_{[]}(\varepsilon\Vert F\Vert_{\Px,2},\mc F, L_2(\Px))\leq\Px(f_\tau-f_0)^2/(\varepsilon^2\Vert F\Vert_{\Px,2}^2)\leq 4/\varepsilon^2$ for all $(0,1]$. If $\mc F$ is uniformly bounded by a constant $C$, then $N(\varepsilon C,\mc F, L_2(\Qx))\leq 1/\varepsilon^2$ for each probability distribution $\mathbb Q$ on $\mc X$. 
\end{lemma}

Then, we introduce the following Lemma \ref{lm:glb-fun-classes}. For simplicity, $p^0$ and $q^0$ correspond to, respectively, $\Px(R=0)$ and $q^0(\mb X)$, and $S,\pi,G$ correspond to, respectively $\bar S$, $\bar\pi$ and $\bar G$ functions in Lemma \ref{lm:glb-fun-classes}. 

\begin{lemma}\label{lm:glb-fun-classes}
    Let $S,\pi, q^0, p^0$, and $G$ be fixed, where $t\mapsto S(t\mid a,\mb x)$ is assumed to be non-increasing for each $(a,\mb x)$, and where $G(t_0-\mid a,\mb x)\geq 1/\eta$, $\pi(a_0\mid\mb x)\geq 1/\eta$, and $p^0\geq 1/\eta$, for some $\eta\in(0,\infty)$. Then the class of influence functions $\mc F_{S,\pi,p^0,q^0,G,t_0,a_0} = \{\varphi_{S,\pi,p^0,q^0,G,t_0,a_0}:t\in[0,t_0]\}$ satisfies
    \begin{align*}
        \sup_{\Qx} N(\varepsilon\Vert F\Vert_{\Qx,2},\mc F_{S,\pi,p^0,q^0,G,t_0,a_0}, L_2(\Qx)) \leq 32/\varepsilon^{10},
    \end{align*}
    for any $\varepsilon\in(0,1]$ where $F=\eta(1+2\eta^2)$ is an envelope of $\mc F_{S,\pi,p^0,q^0,G,t_0,a_0}$, and the supremum is taken over all distributions $\Qx$ on the sample space of the observed data. 
\end{lemma}

\begin{proof}
    The class $\mc F_{S,\pi,p^0,q^0,G,t_0,a_0}$ is uniformly bounded by $\eta(1+2\eta^2)$ because of the assumed upper bounds of $1/p^0, 1/\pi$ and $1/G$. Therefore, the envelop function can be taken as $F=\eta(1+2\eta^2)$. 

    Define the following two functions $f_t$ and $h_t$ pointwise as
    \begin{align*}
        f_t(\mb x,a_0,\delta,y) & = \frac{\Ix(a=a_0,y\leq t,\delta=1)q^0(\mb x)S(t\mid a_0,\mb x)}{p^0\pi(a_0\mid\mb x)S(y\mid a_0,\mb x)G(y-\mid a_0,\mb x)},\\
        h_t(\mb x,a_0,y) & = \int\frac{\Ix(a=a_0,u\leq y,u\leq t)q^0(\mb x)S(t\mid a_0,\mb x)}{p^0\pi(a_0\mid\mb x)S(u\mid a_0,\mb x)G(u-\mid a_0,\mb x)}\Lambda(du\mid a_0,\mb x). 
    \end{align*}
    Consider classes $\mc F^1_{S,p^0,q^0,t_0,a_0} = \{(\mb x,r)\mapsto \Ix(r=0)S(t\mid a_0,\mb x)/p^0:t\in[0,t_0]\}$, $\mc F^2_{S,p^0,q^0, G,t_0,a_0} = \{(\mb x,a_0,\delta,y)\mapsto f_t(\mb x,a_0,\delta,y):t\in[0,t_0]\}$, and $\mc F^3_{S,p^0,q^0, G,t,a_0} = \{(\mb x,a_0,y)\mapsto h_t(\mb x,a,y):t\in[0,t_0]\}$. We can then write 
    \begin{align*}
        \mc F_{S,\pi,p^0,q^0,G,t_0,a_0}\subseteq\{f_1-f_2+f_3:f_1\in\mc F^1_{S,p^0,q^0,t_0,a_0}, f_2\in\mc F^2_{S,p^0,q^0, G,t_0,a_0}, f_3\in\mc F^3_{S,p^0,q^0, G,t_0,a_0}\}.
    \end{align*}
    Since $(\mb x,r)\mapsto \Ix(r=0)S(t\mid a_0,\mb x)/p^0$ is non-increasing over $t\in[0,t_0]$ for all $(\mb x,r)$ and uniformly bounded by $1/\eta$, Lemma \ref{lm:lm4-west} implies that $\sup_{\Qx}(\varepsilon,\mc F^1_{S,p^0,q^0,t_0,a_0},L_2(\Qx))\leq 2\varepsilon^2$. 

    The $\mc F^2_{S,p^0,q^0,G,t_0,a_0}$ is contained in the product of the classes $\{y\mapsto\Ix(y\leq t):t\in[0,t_0]\}$, $\{\mb x\mapsto S(t\mid a_0,\mb x):t\in[0,t_0]\}$, and the singleton class $\{(\mb x,a_0,\delta,y)\mapsto\Ix(a=a_0,\delta=1,y\leq t)q^0(\mb x)/[p^0\pi(a_0\mid\mb x)S(y\mid a_0,\mb x)G(y\mid a_0,\mb x)]\}.$ The first two classes both have covering number $\sup_{\Qx}N(\varepsilon,\cdot,L_2(\Qx))\leq 2\varepsilon^2$ by Lemma \ref{lm:lm4-west}. The third class has uniform covering number 1 for all $\varepsilon$ because it can be covered with a single ball of any positive radius. In addition,  $\mc F^2_{S,p^0,q^0,G,t,a}$ is uniformly bounded by $\eta^3$. Therefore, based on Lemma 5.1 in \cite{van2006statistical},  $\sup_{\Qx}N(\varepsilon\eta^3,\mc F^2_{S,p^0,q^0, G,t,a},L_2(\Qx))\leq 4/\varepsilon^4$.  

    Furthermore, based on Lemma 5.1 in \cite{van2006statistical}  again, we can rewrite $\mc F^3_{S,p^0,q^0, G,t_0,a_0} = \{(\mb x,a_0,y)\mapsto\int m_t(u,\mb x,a_0,y)\mu^*(du):t\in[0,t_0]\}$, where 
    \begin{align*}
        m_t(u,\mb x,a_0,y) = \frac{\Ix(a=a,u\leq y,u\leq t)q^0(\mb x)S(t\mid a_0,\mb x)}{p^0\pi(a_0\mid\mb x)S(u\mid a_0,\mb x)G(u-\mid a_0,\mb x)}\lambda^*(u\mid a_0,\mb x),
    \end{align*}
    with $\mu^*(du)$ is a dominating measure for $\Px$-almost all $\mb x$ and $\lambda^*(du\mid a_0,\mb x)$ the Radon-Nikodym derivative of $\Lambda(\cdot\mid a_0,\mb x)$ with respect to $\mu^*$. Then, the class $\mc M_{t}=\{m_t:t\in[0,t_0]\}$ is contained in the product of the singleton class $\{(\mb x, a_0,y,u)\mapsto\Ix(a=a_0,u\leq y)q^0(\mb x)\lambda^*(u\mid a_0,\mb x)/[p^0\pi(a_0\mid\mb x)S(u\mid a_0,\mb x)G(u-\mid a_0,\mb x)]\}$ and the class $\{u\mapsto\Ix(u\leq t):t\in[0,t_0]\}$ and $\{\mb x\mapsto S(t\mid a_0,\mb x):t\in[0,t]\}$, which as discussed above both have $L_2(\Qx)$ covering number bounded by $2/\varepsilon^2$ for any probability measure $\Qx$. Therefore, $\sup_{\Qx}N(\varepsilon\eta^3,\mc M_{t},L_2(\Qx))\leq 4/\varepsilon^4$. By Jensen's inequality, $\Vert h_t-h_s\Vert_{L_2(\Qx)}\leq\Vert m_t-m_s\Vert_{L_2(\mu^*\times\Qx)}$ for any $\Qx$, which implies that $\sup_{\Qx}N(\varepsilon\eta^3,\mc F^3_{S,p^0,q^0, G,t,a},L_2(\Qx))\leq\sup_{\Qx}N(\varepsilon\eta^3,\mc M_{t},L_2(\mu^*\times\Qx))\leq 4/\varepsilon^4$ for all $\Qx$. 

    Therefore, we have shown that the three classes have covering numbers bounded by $2/\varepsilon^2$, $4/\varepsilon^4$ and $4/\varepsilon^4$, respectively. Therefore, by Lemma 5.1 in 
 \cite{van2006statistical}, 
    \begin{align*}
        \sup_{\Qx} N(\varepsilon\eta(1+2\eta^2),\mc F_{S,\pi,p^0,q^0,G,t_0,a_0}, L_2(\Qx)) \leq 32/\varepsilon^{10}. 
    \end{align*}
\end{proof}

\begin{lemma}\label{lm:RAL-ccod-Gnm}
    If Conditions \ref{cond:nuisance-ccod}--\ref{cond:bound-pi-G-ccod} hold,  $M^{-1}\sum_{m=1}^M n^{-1}Mn_m^{1/2}\Gx_n^m\left[\widehat\varphi^{\text{CCOD}}_{n,m,t,a}-\varphi^{\text{CCOD}}_{\infty,t,a}\right] = o_p(n^{-1/2})$. If Conditions \ref{cond:uniform-S-ccod} holds as well, 
    \begin{align*}
        \frac1M\sum_{m=1}^M\frac{Mn_m^{1/2}}{n}\sup_{u\in[0,t]}\left\vert\Gx_n^m\left[\widehat\varphi^{\text{CCOD}}_{n,m,t,a}-\varphi^{\text{CCOD}}_{\infty,t,a}\right]\right\vert = o_p(n^{-1/2}). 
    \end{align*}
\end{lemma}

\begin{proof}
    We follow notation in Lemma \ref{lm:L2-bound-ccod-est}. First, we note that 
    \begin{align*}
        \frac{Mn_m^{1/2}}{n} & \leq\frac{M(|n_m-n/M|+n/M)^{1/2}}{n}\\
        & \leq\frac{M|n_m-n/M|^{1/2}+M|n/M|^{1/2}}{n}\\
        & \leq\left(\frac{M}{n}\right)^{1/2}+\frac{M}{n},
    \end{align*}
    for all $m$ since $|n_m-n/M|\leq 1$ by assumption on $n_m$. Then, we have that
    \begin{align*}
        & \frac1M\sum_{m=1}^M\frac{Mn_m^{1/2}}{n}\sup_{u\in[0,t]}\left\vert\Gx_n^m\left[\widehat\varphi^{\text{CCOD}}_{n,m,t,a}-\varphi^{\text{CCOD}}_{\infty,t,a}\right]\right\vert \\
        & \leq O(n^{-1/2})\frac1M \sum_{m=1}^M\sup_{u\in[0,t]}\left\vert\Gx_n^m\left[\widehat\varphi^{\text{CCOD}}_{n,m,t,a}-\varphi^{\text{CCOD}}_{\infty,t,a}\right]\right\vert,
    \end{align*}
    since $M=O(1)$. 

    Therefore, for the pointwise claim, we turn to show $M^{-1}\sum_{m=1}^M\left\vert\Gx_n^m\left[\widehat\varphi^{\text{CCOD}}_{n,m,t,a}-\varphi^{\text{CCOD}}_{\infty,t,a}\right]\right\vert = o_p(1)$. Using conditional argument, we write
    \begin{align*}
        \Ex\left\vert\Gx_n^m\left[\widehat\varphi^{\text{CCOD}}_{n,m,t,a}-\varphi^{\text{CCOD}}_{\infty,t,a}\right]\right\vert = \Ex\left[\Ex\left\vert\Gx_n^m\left[\widehat\varphi^{\text{CCOD}}_{n,m,t,a}-\varphi^{\text{CCOD}}_{\infty,t,a}\right]\right\vert\mid\mc T_m\right],
    \end{align*}
    where $\mc T_m = \mc O\backslash\mc V_m$ is the $m$th training set. Note that the randomness in the inner expectation of the right-hand-side above, by conditioning on the training set, is only induced from $\Gx_n^m$ by averaging over the observations on the validation set. Therefore, 
    \begin{align*}
        \Ex\left[\Ex\left\vert\Gx_n^m\left[\widehat\varphi^{\text{CCOD}}_{n,m,t,a}-\varphi^{\text{CCOD}}_{\infty,t,a}\right]\right\vert\mid\mc T_m\right] = \Px\left\vert\Gx_n^m(\widehat\varphi^{\text{CCOD}}_{n,m,t,a}-\varphi^{\text{CCOD}}_{\infty,t,a})\right\vert.
    \end{align*}
    Defining $\mc F^{\text{CCOD}}_{n,m,t,a}$ as the singleton class of functions $\widehat\varphi^{\text{CCOD}}_{n,m,t,a} - \varphi^{k,0}_{\infty,t,a}$, we further have
    \begin{align*}
        \Px\left\vert\Gx_n^m(\widehat\varphi^{\text{CCOD}}_{n,m,t,a}-\varphi^{\text{CCOD}}_{\infty,t,a})\right\vert = \Px\left[\sup_{f\in\mc F^{\text{CCOD}}_{n,m,t,a}}\left\vert\Gx_n^m(f)\right\vert\right].
    \end{align*}
    By Theorem 2.1.14 in \cite{van1996weak}, the covering number of $\mc F^{\text{CCOD}}_{n,m,t,a}$ is 1 for all $\varepsilon$, so the uniform entropy integral $J(1,\mc F^{\text{CCOD}}_{n,m,t,a})$ is 1 relative to the natural envelope $\vert\widehat\varphi^{\text{CCOD}}_{n,m,t,a}-\varphi^{\text{CCOD}}_{\infty,t,a}\vert$. Therefore, there is a universal constant $C'$ such that 
    \begin{align*}
        \Px\left[\sup_{f\in\mc F^{\text{CCOD}}_{n,m,t,a}}\left\vert\Gx_n^m(f)\right\vert\right]\leq C'\left\{\Px(\widehat\varphi^{\text{CCOD}}_{n,m,t,a}-\varphi^{\text{CCOD}}_{\infty,t,a})^2\right\}^{1/2}\leq C''\sum_{j=1}^5 \bar A_{j,n,m,t,a},
    \end{align*}
    following definition of $\bar A_{j,n,m,t,a}$ terms in Lemma \ref{lm:L2-bound-k-est}, so that $M^{-1}\sum_{m=1}^M\Ex\left\vert\Gx_n^m\left[\widehat\varphi^{\text{CCOD}}_{n,m,t,a}-\varphi^{\text{CCOD}}_{\infty,t,a}\right]\right\vert$ is bounded up to $C'''\sum_{j=1}^5 \Ex\{\max_m(\bar A_{j,n,m,t,a})\}$ for some constant $C'''$. It is straightforward that by Conditions \ref{cond:nuisance} and \ref{cond:bound-pi-G}, this upper bound tends to zero, therefore,   $M^{-1}\sum_{m=1}^M\left\vert\Gx_n^m\left[\widehat\varphi^{\text{CCOD}}_{n,m,t,a}-\varphi^{\text{CCOD}}_{\infty,t,a}\right]\right\vert=o_p(1)$. 
    
    Next, we show the uniform statement. The basic argument is the same. We first write 
    \begin{align*}
        \Ex\left[\sup_{u\in[0,t]}\left\vert\Gx_n^m\left\{\widehat\varphi^{\text{CCOD}}_{n,m,u,a}-\varphi^{\text{CCOD}}_{\infty,u,a}\right\}\right\vert~\bigg|~\mc T_m\right] & = \Ex\left[\sup_{u\in[0,t]}\left\vert\Gx_n^m\left[\widehat\varphi^{\text{CCOD}}_{n,m,u,a}-\varphi^{\text{CCOD}}_{\infty,u,a}\right]\right\vert\right] \\
        & = \Ex\left[\sup_{g\in\mc G_{n,m,t,a}}\left\vert\Gx_n^mg\right\vert\right],
    \end{align*}
    where $\mc G_{n,m,t,a} = \{\widehat\varphi^{\text{CCOD}}_{n,m,u,a}-\varphi^{\text{CCOD}}_{\infty,u,a}:u\in[0,t]\}$. When conditioning on $\mc T_m$, the functions $\widehat{\bar S}_m$,  $\widehat{\bar G}_m$, $\widehat{\bar\pi}_m$ and $\widehat q^0_m$ are fixed, so Lemma \ref{lm:glb-fun-classes} implies that 
    \begin{align*}
        \log\sup_{\Qx} N(\varepsilon\Vert\bar G_{n,m,t,a}\Vert_{\Qx,2},\mc G_{n,m,t,a},L_2(\Qx))\leq\widetilde C\log\varepsilon^{-1},
    \end{align*}
    for some constant $\widetilde C$ not depending on $n,m,$ or $\varepsilon$, and where $\bar G_{n,m,t,a}:=\sup_{u\in[0,t]}\vert\widehat\varphi^{\text{CCOD}}_{n,m,u,a}-\varphi^{\text{CCOD}}_{\infty,u,a}\vert$ is the natural envelope function for $\mc G_{n,m,t,a}$. As a result, the uniform entropy integral
    \begin{align*}
        J(1,\mc G_{n,m,t,a},L_2(\Px)) = \sup_{\Qx}\int_0^1[1+\log N(\varepsilon\Vert\bar G_{n,m,t,a}\Vert_{\Qx,2},\mc G_{n,m,t,a},L_2(\Qx))]^{1/2}d\varepsilon
    \end{align*}
    is bounded by a constant not depending on $n$ or $m$. By Theorem 2.1.14 in \cite{van1996weak}, there is therefore a constant $\bar C$ not depending on $n$ or $m$ such that 
    \begin{align*}
        \Ex\left[\sup_{g\in\mc G_{n,m,t,a}}\left\vert\Gx_n^mg\right\vert\right] & \leq\bar C\Px\left[\sup_{u\in[0,t]}\{\varphi^{\text{CCOD}}_{n,m,u,a}(\mc O)-\varphi^{\text{CCOD}}_{\infty,u,a}(\mc O)\}^2\right]^2 \\
        & \leq \bar CC(\eta)\sum_{j=1}^5 \bar A_{j,n,m,t,a},
    \end{align*}
    where the second inequality follows notation and results of Lemma \ref{lm:L2-bound-ccod-est}. We therefore have that 
    \begin{align*}
        \frac1M\sum_{m=1}^M\Ex\left[\sup_{u\in[0,t]}\left\vert\Gx_n^m\left[\widehat\varphi^{\text{CCOD}}_{n,m,u,a}-\varphi^{\text{CCOD}}_{\infty,u,a}\right]\right\vert\right]\leq \bar CC(\eta)\sum_{j=1}^5\max_m\Ex[\bar A_{j,n,m,t,a}]. 
    \end{align*}
    By Conditions \ref{cond:nuisance-ccod}, \ref{cond:bound-pi-G-ccod}, and \ref{cond:prod-error-ccod}, this bound tends to zero.  
\end{proof}

\begin{lemma}\label{lm:bias-ccod-est}
    Consider some general nuisance functions under $\Px_\infty$, denoted by $\bar S_\infty$, $\bar G_\infty$, $\bar\pi_\infty$, and $q^0_\infty$. Then, $\Px[\varphi^{\text{CCOD}}_{t,a}(\mc O;\Px_\infty)] - \estimand$ equals
    \begin{align*}
        & \Ex\bigg[\frac{q^0(\mb X)}{\Px(R=0)}\bar S_\infty(t\mid a, \mb X)\int_0^t\frac{\bar S(y-\mid a,\mb X)}{\bar S_\infty(y\mid a,\mb X)}\\
        & \quad \times\left\{\frac{\bar G(y-\mid a,\mb X)\bar\pi(a\mid\mb X)q^0_\infty(\mb X)}{\bar G_\infty(y-\mid a,\mb X)\bar\pi_\infty(a\mid\mb X)q^0(\mb X)}-1\right\}(\bar\Lambda_\infty-\bar\Lambda)(dy\mid a,\mb X)\bigg].  
    \end{align*} 
\end{lemma}

\begin{proof}
    We first express $\varphi^{\text{CCOD}}_{t,a}(\mc O;\Px_\infty)$ as
\begin{align*}
    \frac{\Ix(R=0)}{\Px(R=0)}\bar S_\infty(t\mid a,\mb X) - \frac{q^0_\infty(\mb X)}{\Px(R=0)}\bar S_\infty(t\mid a, \mb X)\frac{\Ix(A=a)}{\bar\pi_\infty(a\mid\mb X)}H_{\bar S_\infty,\bar G_\infty,t,a}(Y,\Delta,\mb X),
\end{align*}
where
\begin{align*}
    H_{\bar S_\infty,\bar G_\infty,t,a}(Y,\Delta,\mb X) = \frac{\Ix(Y\leq t,\Delta=1)}{\bar S_\infty(Y\mid a, \mb X)\bar G_\infty(Y-\mid a,\mb X)}-\int_0^{t\wedge Y}\frac{\bar\Lambda_\infty(du\mid a,\mb X)}{\bar S_\infty(u\mid a, \mb X)\bar G_\infty(u-\mid a,\mb X)}.
\end{align*}
Following a result in Lemma 1 of \cite{westling2024inference}, $\Ex\{H_{\bar S_\infty,\bar G_\infty,t,a}(Y,\Delta,\mb X)\mid\mb X=\mb x\}$ equals 
\begin{align*}
    -\int_0^t\frac{\bar S(y-\mid a,\mb x)\bar G(y-\mid a,\mb x)}{\bar S_\infty(y\mid a,\mb x)\bar G_\infty(y-\mid a,\mb x)}(\bar\Lambda_\infty-\bar \Lambda)(dy\mid a,\mb x).
\end{align*}
Therefore, it is straightforward that $\Px[\varphi^{\text{CCOD}}_{t,a}(\mc O;\Px_\infty)] - \estimand$ equals
\begin{align*}
& \Ex\left[\frac{q^0(\mb X)}{\Px(R=0)}\{\bar S_\infty(t\mid a,\mb X)-\bar S(t\mid a,\mb X)\}\right] \\
  & + \Ex\bigg[\frac{q^0(\mb X)}{\Px(R=0)}\frac{q^0_\infty(\mb X)}{q^0(\mb X)}\bar S_\infty(t\mid a, \mb X)\frac{\bar\pi(a\mid\mb X)}{\bar\pi_\infty(a\mid\mb X)}\\
  & \quad\times\int_0^t\frac{\bar S(y-\mid a,\mb X)\bar G(y-\mid a,\mb X)}{\bar S_\infty(y\mid a,\mb X)\bar G_\infty(y-\mid a,\mb X)}(\bar\Lambda_\infty-\bar\Lambda)(dy\mid a,\mb X)\bigg].
\end{align*}
Furthermore, by Duhamel equation in \cite{gill1990survey}, we have 
\begin{align*}
    \bar S_\infty(t\mid a,\mb X)-\bar S(t\mid a,\mb X) = -\bar S_\infty(t\mid a,\mb X)\int_0^t\frac{\bar S(y-\mid a,\mb X)}{\bar S_\infty(y-\mid a,\mb X)}(\Lambda^0_\infty-\Lambda^0)(du\mid a,\mb X),
\end{align*}
for each $(t,a,\mb x)$. Therefore, we further have $\Px[\varphi^\text{CCOD}_{t,a}(\mc O;\Px_\infty)] - \estimand$ equals
\begin{align}\label{eq:bias-term-ccod}
    & \Ex\bigg[\frac{q^0(\mb X)}{\Px(R=0)}\bar S_\infty(t\mid a, \mb X)\int_0^t\frac{\bar S(y-\mid a,\mb X)}{\bar S_\infty(y\mid a,\mb X)}\nonumber \\
    & \quad \times\left\{\frac{\bar G(y-\mid a,\mb X)\bar\pi(a\mid\mb X)q^0_\infty(\mb X)}{\bar G_\infty(y-\mid a,\mb X)\bar\pi_\infty(a\mid\mb X)q^0(\mb X)}-1\right\}(\bar\Lambda_\infty-\bar\Lambda)(dy\mid a,\mb X)\bigg]. 
\end{align}
\end{proof}

\subsection{Proof of Theorems \ref{thm:RAL-ccod} and \ref{thm:uniformRAL-ccod}}\label{subsubapp:proof-RAL-ccod}

The proof of Theorem \ref{thm:RAL-ccod} is included in that of Theorem \ref{thm:uniformRAL-ccod}. Hence, we prove Theorem \ref{thm:uniformRAL-ccod} below. 
The proof proceeds by decomposing the estimation error into a leading empirical average of the EIF and higher-order remainder terms, and then showing that these remainder terms are asymptotically negligible under the stated regularity conditions. 

\begin{proof}
    By the result in equation \eqref{eq:glb-est-deco} with $\bar\pi_\infty=\bar\pi$, $\bar G_\infty=\bar G$, and $\bar S_\infty=\bar S$, 
    \begin{align*}
    \glbEst - \estimand & = \Px_n[\varphi^{*\text{CCOD}}_{t,a}] + \frac1M\sum_{m=1}^M\frac{Mn_m^{1/2}}{n}\Gx_n^m\left[\widehat\varphi^\text{CCOD}_{n,m,t,a} - \varphi^\text{CCOD}_{t,a}\right] \\
    & \qquad + \frac1M\sum_{m=1}^M\frac{Mn_m}{n}\Px\left[\widehat\varphi^\text{CCOD}_{t,a}-\estimand\right].
    \end{align*}
    By Conditions \ref{cond:nuisance-ccod} and \ref{cond:bound-pi-G-ccod}, the second summand on the right-hand-side is $o_p(n^{-1/2})$ by Lemma \ref{lm:RAL-ccod-Gnm}. By Lemma \ref{lm:bias-ccod-est}, $\Px[\widehat\varphi^\text{CCOD}_{t,a}]-\estimand$ equals
    \begin{align*}
    & \Ex\bigg[\frac{q^0(\mb X)}{\Px(R=0)}\widehat{\bar S}_m(t\mid a, \mb X)\int_0^t\frac{\bar S(y-\mid a,\mb X)}{\widehat{\bar S}_m(y\mid a,\mb X)} \\
    & \quad \times\left\{\frac{\bar G(y-\mid a,\mb X)\bar\pi(a\mid\mb X)\widehat q^0_m(\mb X)}{\widehat{\bar G}_m(y-\mid a,\mb X)\widehat{\bar\pi}_m(a\mid\mb X)q^0(\mb X)}-1\right\}(\widehat{\bar\Lambda}_m-\bar\Lambda)(dy\mid a,\mb X)\bigg]. 
\end{align*}
By Duhamel equation in \cite{gill1990survey}, we have that 
\begin{align*}
    \frac{\bar S(y-\mid a,\mb X)}{\widehat{\bar S}_m(y\mid a,\mb X)}(\widehat{\bar\Lambda}_m-\bar\Lambda)(dy\mid a,\mb X) = \left(\frac{\bar S}{\widehat{\bar S}_m}-1\right)(du\mid a,\mb X),
\end{align*}
and so the above equals 
\begin{align*}
   & \Ex\left[\frac{q^0(\mb X)}{\Px(R=0)}\widehat{\bar S}_m(t\mid a, \mb X)\int_0^t\left\{\frac{\bar G(y-\mid a,\mb X)\bar\pi(a\mid\mb X)\widehat q^0_m(\mb X)}{\widehat{\bar G}_m(y-\mid a,\mb X)\widehat{\bar\pi}_m(a\mid\mb X)q^0(\mb X)}-1\right\}\left(\frac{\bar S}{\widehat{\bar S}_m}-1\right)(du\mid a,\mb X)\right] \\
    & = \Ex\bigg[\frac{q^0(\mb X)}{\Px(R=0)}\widehat{\bar S}_m(t\mid a, \mb X)\left\{\frac{\bar\pi(a\mid\mb X)}{\widehat{\bar\pi}_m(a\mid\mb X)}-1\right\}\int_0^t\frac{\bar S}{\widehat{\bar S}_m}(du\mid a,\mb X) + \frac{q^0(\mb X)}{\Px(R=0)}\widehat{\bar S}_m(t\mid a, \mb X) \\
    & \qquad\times\int_0^t\frac{\bar\pi(a\mid\mb X)}{\widehat{\bar\pi}_m(a\mid\mb X)}\left\{\frac{\bar G(y-\mid a,\mb X)\widehat q^0_m(\mb X)}{\widehat{\bar G}_m(y-\mid a,\mb X)q^0(\mb X)}-1\right\}\left(\frac{\bar S}{\widehat{\bar S}_m}-1\right)(du\mid a,\mb X)\bigg] \\
    & = \underbrace{\Ex\left[\frac{q^0(\mb X)}{\Px(R=0)}\frac{\{\widehat{\bar\pi}_m(a\mid\mb X)-\bar\pi(a\mid\mb X)\}\{\widehat {\bar S}_m(t\mid a, \mb X)-\bar S(t\mid a,\mb X)\}}{\widehat{\bar\pi}_m(a\mid\mb X)}\right]}_{I_1} \\
    & \quad + \Ex\bigg[\frac{q^0(\mb X)}{\Px(R=0)}\frac{\bar\pi(a\mid\mb X)}{\widehat{\bar\pi}_m(a\mid\mb X)}\widehat{\bar S}_m(t\mid a, \mb X) \\
    & \qquad\underbrace{\quad\times \int_0^t\left\{\frac{\bar G(y-\mid a,\mb X)\widehat q^0_m(\mb X)}{\widehat{\bar G}_m(y-\mid a,\mb X)q^0(\mb X)}-1\right\}\left(\frac{\bar S}{\widehat{\bar S}_m}-1\right)(du\mid a,\mb X)\bigg]}_{I_2}. 
\end{align*}
We note that $I_1\leq\eta^2\Ex\left[\vert\widehat{\bar\pi}_m(a\mid\mb X)-\bar\pi(a\mid\mb X)\vert\cdot\vert\widehat{\bar S}_m(t\mid a, \mb X)-\bar S(t\mid a,\mb X)\vert\right]$. For $I_2$, we further expand it, similarly to the above process, as 
\begin{align*}
    I_2 & = \underbrace{\Ex\left[\frac{q^0(\mb X)}{\Px(R=0)}\frac{\bar\pi(a\mid\mb X)}{\widehat{\bar\pi}_m(a\mid\mb X)}\{\widehat q^0_m(\mb X)-q^0(\mb X)\}\{\widehat{\bar S}_m(t\mid a, \mb X)-\bar S(t\mid a,\mb X)\}\right]}_{I_{21}} \\
    & \quad + \Ex\bigg[\frac{\widehat q^0_m(\mb X)}{\Px(R=0)}\frac{\bar\pi(a\mid\mb X)}{\widehat{\bar\pi}_m(a\mid\mb X)}\widehat{\bar S}_m(t\mid a, \mb X) \\
    & \qquad\underbrace{\quad\times\int_0^t\left\{\frac{\bar G(y-\mid a,\mb X)}{\widehat{\bar G}_m(y-\mid a,\mb X)}-1\right\}\left(\frac{\bar S}{\widehat{\bar S}_m}-1\right)(du\mid a,\mb X)\bigg]}_{I_{22}}.
\end{align*}
Then, we note that $I_{21}\leq\eta^2\Ex\left[|\widehat q^0_m(\mb X)-q^0\mb X)|\cdot|\widehat{\bar S}_m(t\mid a, \mb X)-\bar S(y\mid a,\mb X)|\right]$, and 
\begin{align*}
    I_{22}\leq \eta^2\Ex\left[\widehat{\bar S}_m(t\mid a, \mb X)\int_0^t\left\{\frac{\bar G(y-\mid a,\mb X)}{\widehat{\bar G}_m(y-\mid a,\mb X)}-1\right\}\left(\frac{\bar S}{\widehat{\bar S}_m}-1\right)(du\mid a,\mb X)\right]
\end{align*}

by Condition \ref{cond:bound-pi-G-ccod}. Then, by notation in Condition \ref{cond:prod-error-ccod}, $I_1+I_2\leq\eta^2\{\bar r_{n,t,a,1}+\bar r_{n,t,a,2}+\bar r_{n,t,a,3}\}$. Since $M^{-1}\sum_{m=1}^Mn^{-1}Mn_m\leq 2$, 
\begin{align*}
    \left\vert\frac1M\sum_{m=1}^M\frac{Mn_m}{n}\Px\left[\widehat\varphi^\text{CCOD}_{t,a} - \estimand\right]\right\vert\leq 2\eta^2(1+\eta)\{\bar r_{n,t,a,1}+\bar r_{n,t,a,2}+\bar r_{n,t,a,3}\} = o_p(n^{-1/2}), 
\end{align*}
by Condition \ref{cond:prod-error-ccod}. This established the pointwise RAL property: $\glbEst=\estimand+\Px_n(\varphi^{*\text{CCOD}}_{t,a}) + o_p(n^{-1/2})$. Since $\varphi^{*\text{CCOD}}_{t,a}$ is uniformly bounded, $\Px\{(\varphi^{*\text{CCOD}}_{t,a})^2\}<\infty$ and since $\Px\{\varphi^{*\text{CCOD}}_{t,a}\} = 0$, it follows that 
\begin{align*}
    n^{1/2}\Px_n(\widehat\varphi^{*\text{CCOD}}_{t,a})\to_d\mc N(0,\Px\{(\varphi^{*\text{CCOD}}_{t,a})^2\}). 
\end{align*}

For the uniform RAL, we use the same decomposition above. By conditions \ref{cond:nuisance-ccod}, \ref{cond:bound-pi-G-ccod} and \ref{cond:prod-error-ccod} and Lemma \ref{lm:RAL-ccod-Gnm}, 
\begin{align*}
    \frac1M\sum_{m=1}^M\frac{Mn_m^{1/2}}{n}\sup_{u\in[0,t]}\left\vert\Gx_n^m\left[\widehat\varphi^\text{CCOD}_{t,a}-\varphi^\text{CCOD}_{t,a}\right]\right\vert = o_p(n^{-1/2}).
\end{align*}
Therefore, we have that
\begin{align*}
\sup_{u\in[0,t]}\left\vert\frac1M\sum_{m=1}^M\frac{Mn_m}{n}\Px_n^m(\widehat\varphi^\text{CCOD}_{n,m,t,a}-\estimand)\right\vert \leq\sup_{u\in[0,t]}2\eta^2\{\bar r_{n,u,a,1}+\bar r_{n,u,a,2}+\bar r_{n,u,a,3}\},
\end{align*}
which is $o_p(n^{-1/2})$ by Condition \ref{cond:prod-error-unif-ccod}. Thus, $\sup_{u\in[0,t]}\left\vert\widehat\theta^\text{CCOD}_n(u,a)-\theta^0(u,a)-\Px_n(\varphi^{*\text{CCOD}}_{u,a})\right\vert = o_p(n^{-1/2})$. Since $\{\varphi^{*\text{CCOD}}_{u,a}:u\in[0,t]\}$ is a uniformly bounded $\Px$-Donsker class by Lemma \ref{lm:glb-fun-classes}, $\{n^{1/2}\Px_n(\varphi^{*\text{CCOD}}_{u,a}):u\in[0,t]\}$ converges weakly to a tight mean-zero Gaussian process with covariance $(u,v)\mapsto\Px(\varphi^{*\text{CCOD}}_{u,a}\varphi^{*\text{CCOD}}_{v,a})$. 
\end{proof}

\subsection{Details of Remark \ref{rmk:DR-ccod}}

We now state Remark \ref{rmk:DR-ccod} in more technical details below. 

\begin{remark}[Double robustness of the CCOD estimator]\label{rmk:DR-ccod-tech}
If we are only concerned with consistency of the CCOD estimator, the requirement on $(\bar\pi_\infty,\bar G_\infty,\bar\Lambda_\infty,\bar S_\infty,q^0_\infty) = (\bar\pi,\bar G,\bar\Lambda,\bar S,q^0)$ in Theorem \ref{thm:RAL-ccod} (Theorem \ref{thm:uniformRAL-ccod}) can be weakened as follows. For $\Px$-almost all $\mb X$, we only require that there exist measurable sets $\bar{\mc S}_x, \bar{\mc G}_x \subseteq [0,t]$ such that $\bar{\mc S}_x \cup \bar{\mc G}_x = [0,t]$ and $\bar\Lambda(u\mid a,\mb X) = \bar\Lambda_\infty(u\mid a,\mb X)$ for all $u \in \bar{\mc S}_x$, and $\bar G(u-\mid a,\mb X) = \bar G_\infty(u-\mid a,\mb X)$ for all $u \in \bar{\mc G}_x$.
In addition, if $\bar{\mc S}_x$ is a strict subset of $[0,t]$, then $\bar\pi(a\mid\mb X)$ and $q^0(\mb X)$ should be the probability limits of their estimators. Under these conditions, $\glbEst$ is consistent under Conditions \ref{cond:nuisance-ccod} and \ref{cond:bound-pi-G-ccod}, and it is uniformly consistent if Condition \ref{cond:uniform-S-ccod} additionally holds.
\end{remark}

The proof of Remark \ref{rmk:DR-ccod} (or Remark \ref{rmk:DR-ccod-tech} here) is essentially the same as those in the proof of Theorem 2 in \cite{westling2024inference}, hence the details are omitted. A sketch of the proof is that by decomposing the integral $\int_0^t$ as $\int_{\bar{\mc S}_x}+\int_{\bar{\mc S}^c_x}$, where $\bar{\mc S}^c_x$ is the complement of set $\bar{\mc S}_x$, and $\bar{\mc S}_x^c\subseteq\bar{\mc G}_x$ by definition. Then it is straightforward to verify that when the statement in Remark \ref{rmk:DR-ccod} holds, the following integral (also in the bias term \eqref{eq:bias-term-ccod}) 
\begin{align*}
   & \int_0^t\frac{\bar S(y-\mid a,\mb X)}{\bar S_\infty(y\mid a,\mb X)}\left\{\frac{\bar G(y-\mid a,\mb X)\bar\pi(a\mid\mb X)q^0_\infty(\mb X)}{\bar G_\infty(y-\mid a,\mb X)\bar\pi_\infty(a\mid\mb X)q^0(\mb X)}-1\right\}(\bar\Lambda_\infty-\bar\Lambda)(dy\mid a,\mb X) \\
    & = \left(\int_{\bar{\mc S}_x} + \int_{\bar{\mc G}_x}\right)\frac{\bar S(y-\mid a,\mb X)}{\bar S_\infty(y\mid a,\mb X)}\left\{\frac{\bar G(y-\mid a,\mb X)\bar\pi(a\mid\mb X)q^0_\infty(\mb X)}{\bar G_\infty(y-\mid a,\mb X)\bar\pi_\infty(a\mid\mb X)q^0(\mb X)}-1\right\}(\bar\Lambda_\infty-\bar\Lambda)(dy\mid a,\mb X) = 0,
\end{align*}
which further implies $\Px[\varphi^{\text{CCOD}}_{\infty,t,a}] - \estimand = 0$. To prove the consistency or uniform consistency of the CCOD estimator, one needs results from Lemmata \ref{lm:L2-bound-ccod-est}--\ref{lm:glb-fun-classes}. 

\section{Technical Details of the Federated Estimator}\label{app:theory-FED}

\subsection{Theoretical results of the local estimator}\label{subapp:source-est}

We first prove Theorem \ref{thm:site-k-EIF} about the EIF of the source-site estimator in the following Section \ref{subsubapp:proof-site-k-EIF}. We then provide proofs for the RAL and uniform RAL properties of the source-site estimator. The techniques employed parallel those used for the CCOD estimator, so some proofs are either omitted or provided with less detail.

\subsubsection{Proof of Theorem \ref{thm:site-k-EIF}}\label{subsubapp:proof-site-k-EIF}

\begin{proof}
Under the partial CCOD assumption of site $k$ and the target site $R=0$, $S^0(t\mid a, \mb X)=S^k(t\mid a, \mb X)$ almost surely. Therefore,
\begin{align}\label{eq:sourceEIF-1}
   0 & = \frac{\partial}{\partial\epsilon}\estimand\bigg|_{\epsilon=0}\nonumber  = \frac{\partial}{\partial\epsilon}\Ex\{S^0_{\epsilon}(t\mid a,\mb X)\mid R=0\}\bigg|_{\epsilon=0} \nonumber\\ 
   & = \Ex\{[S^0(t\mid a, \mb X)-\estimand]\dot{\ell}_{\mb X\mid R=0}\mid R=0\} + \Ex\left\{\int\frac{\partial}{\partial\epsilon}S^k_{\epsilon}(t\mid a,\mb x)\bigg|_{\epsilon=0}\mu(d\mb x)~\bigg|~R=0\right\}. 
\end{align}
Following a similar argument in the proof in Section \ref{subsubapp:proof-ccod-EIF} (i.e., the way we derive the derivative for $\bar S_\epsilon$ w.r.t. $\epsilon$), we can express the integrand of the second term above as
\begin{align*}
    & \frac{\partial}{\partial\epsilon}\iint \Prodi_{(0,t]}\{1-\Lambda^k_{\epsilon}(du\mid a,\mb x)\}\mu(d\mb x)\bigg|_{\epsilon=0} \\
    & = \iiint -\Ix(y\leq t, \delta=1)\frac{S^k(t\mid a, \mb x)S^k(y-\mid a, \mb x)}{S^k(y\mid a, \mb x)D^k(y\mid a,\mb x)}\dot{\ell}(y,\delta\mid a,\mb x,k)\Px(dy,d\delta\mid a,\mb x,k)\mu(d\mb x) \\ 
    & \quad + \iiiint\Ix(u\leq t,u\leq y)\frac{S^k(t\mid a, \mb x)S^k(u-\mid a, \mb x)}{S^k(u\mid a, \mb x)D^k(u\mid a,\mb x)} \\
    & \quad\qquad \times \dot{\ell}(y,\delta\mid a,\mb x,k)\Px(dy,d\delta\mid a,\mb x,k)N^k_1(du\mid a,\mb x)\mu(d\mb x) \\
    & = \iiint -\Ix(y\leq t, \delta=1)\frac{S^k(t\mid a, \mb x)S^k(y-\mid a, \mb x)}{S^k(y\mid a, \mb x)D^k(y\mid a,\mb x)}\dot{\ell}(y,\delta\mid a,\mb x,k)\Px(dy,d\delta\mid a,\mb x,k)\mu(d\mb x) \\ 
    & \quad + \iiint S^k(t\mid a, \mb x)\int_0^{t\wedge y}\frac{S^k(u-\mid a, \mb x)}{S^k(u\mid a, \mb x)D^k(u\mid a,\mb x)^2}N^k_1(du\mid a,\mb x)\\
    & \quad\qquad \times \dot{\ell}(y,\delta\mid a,\mb x,k)\Px(dy,d\delta\mid a,\mb x,k)\mu(d\mb x) \\
    & = \Ex\bigg[S^k(t\mid a,\mb X)\frac{\Ix(A=a)}{\pi^k(a\mid\mb X)}\left\{H^k(t\wedge Y,a,\mb X)-\frac{\Ix(Y\leq t,\Delta=1)S^k(Y-\mid a,\mb X)}{S^k(Y\mid a,\mb X)D^k(Y\mid a,\mb X)}\right\} \\
    & \qquad\qquad\times\dot{\ell}(Y,\Delta\mid a,\mb X, R=k)\bigg],
\end{align*}
where
\begin{align*}
    H^k(t,a,\mb x) = \int_0^t\frac{S^k(u-\mid a, \mb x)N_1^k(du\mid a,\mb x)}{S^k(u\mid a, \mb x)D^k(u\mid a,\mb x)^2}. 
\end{align*}
Now, we note that 
\begin{align*}
    \Ex\left[\frac{\Ix(Y\leq t,\Delta=1)S^k(Y-\mid A,\mb X)}{S^k(Y\mid A,\mb X)D^k(Y\mid A,\mb X)}~\bigg|~ A=a,\mb X=\mb x, R=k\right] = \int_0^t\frac{S^k(y-\mid a, \mb x)N_1^k(dy\mid a,\mb x)}{S^k(y\mid a, \mb x)D^k(y\mid a,\mb x)},
\end{align*}
and 
\begin{align*}
    & \Ex\{H^k(t\wedge Y,A,\mb X)\mid A=a,\mb X=\mb x, R=k\}\\
    & = \iint^t \Ix(u\leq y)\frac{S^k(u-\mid a, \mb x)N_1^k(du\mid a,\mb x)}{S^k(u\mid a, \mb x)D^k(u\mid a,\mb x)^2}\Px(dy\mid a,\mb x,k) \\
    & = \int_0^t\Px(Y\geq u\mid A=a, \mb X=\mb x, R=k)\frac{S^k(u-\mid a, \mb x)N_1^k(du\mid a,\mb x)}{S^k(u\mid a, \mb x)D^k(u\mid a,\mb x)^2}\Px(dy\mid a,\mb x,k)\\
    & = \int_0^t\frac{S^k(u-\mid a, \mb x)N_1^k(du\mid a,\mb x)}{S^k(u\mid a, \mb x)D^k(u\mid a,\mb x)}.
\end{align*}
Therefore, 
\begin{align*}
    \Ex\left[H^k(t\wedge Y,A,\mb X) - \frac{\Ix(Y\leq t,\Delta=1)S^k(Y-\mid A,\mb X)}{S^k(Y\mid A,\mb X)D^k(Y\mid A,\mb X)}~\bigg|~ A,\mb X, R=k\right] = 0
\end{align*}
almost surely. By properties of score functions and the tower property, the above implies that
\begin{align*}
    & \frac{\partial}{\partial\epsilon}\iint \Prodi_{(0,t]}\{1-\Lambda^k_{\epsilon}(du\mid a,\mb x)\}\mu(d\mb x)\bigg|_{\epsilon=0} \\
    & =  \Ex\bigg[S^k(t\mid a, \mb X)\frac{\Ix(R=k)}{\Px(R=k\mid\mb X)}\frac{\Ix(A=a)}{\pi^k(a\mid\mb X)} \\ 
    & \qquad\times \left\{H^k(t\wedge Y,A,\mb X) - \frac{\Ix(Y\leq t,\Delta=1)S^k(Y-\mid A,\mb X)}{S^k(Y\mid A,\mb X)D^k(Y\mid A,\mb X)}\right\}\dot{\ell}(\mc O)\bigg].
\end{align*}
Combining these results with the facts that $N_1^k(du\mid a,\mb x)/D^k(u\mid a,\mb x) = \Lambda^k(du\mid a,\mb x)$ and $D^k(u\mid a,\mb x)=S^k(u-\mid a,\mb x)G^k(u-\mid a,\mb x)$, we can rewrite \eqref{eq:sourceEIF-1} at the beginning as follows:
\begin{align*}
    & \frac{\partial}{\partial\epsilon}\estimand\bigg|_{\epsilon=0} \\
    & = \Ex\bigg[\frac{\Ix(R=0)}{\Px(R=0)}[S^k(t\mid a, \mb X)-\estimand]\dot{\ell}(\mc O) - \frac{\Ix(R=0)}{\Px(R=0)}\Ex\bigg\{S^k(t\mid a, \mb X)\frac{\Ix(R=k)}{\Px(R=k\mid\mb X)} \\
    & \quad\quad \times \frac{\Ix(A=a)}{\pi^k(a\mid\mb X)}\left\{\frac{\Ix(Y\leq t,\Delta=1)}{S^k(y\mid a,\mb X)G^k(y-\mid a,\mb X)}-\int_0^{t\wedge y}\frac{\Lambda^k(du\mid a,\mb X)}{S^k(u\mid a,\mb X)G^k(u-\mid a,\mb X)}\right\}\dot{\ell}(\mc O)~\bigg|~\mb X\bigg\}\bigg] \\
    & = \Ex\left[\frac{\Ix(R=0)}{\Px(R=0)}\{S^k(t\mid a, \mb X)-\estimand\}\dot{\ell}(\mc O)\right] - \Ex\bigg[\frac{\Ix(R=k)}{\Px(R=0)}\frac{\Px(R=0\mid\mb X)}{\Px(R=k\mid\mb X)}S^k(t\mid a, \mb X)\\
    & \quad\quad \times \frac{\Ix(A=a)}{\pi^k(a\mid\mb X)} \left\{\frac{\Ix(Y\leq t,\Delta=1)}{S^k(y\mid a,\mb X)G^k(y-\mid a,\mb X)}-\int_0^{t\wedge y}\frac{\Lambda^k(du\mid a,\mb X)}{S^k(u\mid a,\mb X)G^k(u-\mid a,\mb X)}\right\}\dot{\ell}(\mc O)\bigg].
\end{align*}
Therefore, the influence-function representation of $\estimand$ at $\Px$ is found as 
\begin{align*}
   & \varphi^{*k,0}_{t,a}(\mc O;\Px) \\
   & = \frac{\Ix(R=0)}{\Px(R=0)}\{S^0(t\mid a, \mb X)-\estimand\} -\frac{\Ix(R=k)\Px(R=0\mid\mb X)}{\Px(R=0)\Px(R=k\mid\mb X)}S^k(t\mid a, \mb X) \\ 
    & \qquad\times\frac{\Ix(A=a)}{\pi^k(a\mid\mb X)}\left[\frac{\Ix(Y\leq t,\Delta=1)}{S^k(Y\mid a, \mb X)G^k(Y-\mid a,\mb X)}-\int_0^{t\wedge Y}\frac{\Lambda^k(du\mid a,\mb X)}{S^k(u\mid a, \mb X)G^k(u-\mid a,\mb X)}\right].
\end{align*}
Observe that, by Bayes's rule, 
\begin{align*}
    \frac{\Px(R=0\mid\mb X)}{\Px(R=k\mid\mb X)} = \underbrace{\frac{\Px(\mb X\mid R=0)}{\Px(\mb X\mid R=k)}}_{\omega^{k,0}(\mb X)}\cdot\frac{\Px(R=0)}{\Px(R=k)},
\end{align*}
where $\omega^{k,0}(\mb X)$ is a density ratio term, we then have 
\begin{align*}
    \varphi^{*k,0}_{t,a}(\mc O;\Px) & = \frac{\Ix(R=0)}{\Px(R=0)}\{S^0(t\mid a, \mb X)-\estimand\} -\frac{\Ix(R=k)}{\Px(R=k)}\omega^{k,0}(\mb X)S^k(t\mid a, \mb X) \\
    & \quad\quad\times\frac{\Ix(A=a)}{\pi^k(a\mid\mb X)}\left[\frac{\Ix(Y\leq t,\Delta=1)}{S^k(Y\mid a, \mb X)G^k(Y-\mid a,\mb X)}-\int_0^{t\wedge Y}\frac{\Lambda^k(du\mid a,\mb X)}{S^k(u\mid a, \mb X)G^k(u-\mid a,\mb X)}\right].
\end{align*}

\end{proof}

\subsubsection{Regularity conditions for local estimators}\label{subsubapp:cond-local}

\begin{condition}\label{cond:nuisance}
There exist $\pi_\infty^k$, $\omega_\infty^{k,0}$, $G_\infty^k$ and $S_\infty^0$ such that, for all $a\in\{0,1\}$,
\begin{align*}
    & \max_m \Px\Bigg[
    \left(\frac{1}{\widehat\pi_m^k(a\mid\mb X)}-\frac{1}{\pi_\infty^k(a\mid\mb X)}\right)^2
    + \left(\widehat\omega_m^{k,0}(\mb X)-\omega_\infty^{k,0}(\mb X)\right)^2 \\
    & \qquad + \sup_{u\in[0,t]}\Bigg\{
    \left\vert\frac{1}{\widehat G_m^k(u-\mid a,\mb X)}-\frac{1}{G_\infty^k(u-\mid a,\mb X)}\right\vert^2
    + \left\vert
    \frac{\widehat S_m^0(t\mid a,\mb X)}{\widehat S_m^0(u\mid a,\mb X)}
    -
    \frac{S_\infty^0(t\mid a,\mb X)}{S_\infty^0(u\mid a,\mb X)}
    \right\vert^2
    \Bigg\}
    \Bigg] \to_p 0.
\end{align*}
\end{condition}

\begin{condition}\label{cond:bound-pi-G}
There exists $\eta\in(0,\infty)$ such that, for $\Px$-almost all $\mb X$,
$\min\{\widehat\pi_m^k(a\mid\mb X),\, \pi_\infty^k(a\mid\mb X),\, 
\widehat G_m^k(t-\mid a,\mb X),\, G_\infty^k(t-\mid a,\mb X)\}\geq 1/\eta,$ and $\max\{\widehat\omega_m^{k,0}(\mb X),\, \omega_\infty^{k,0}(\mb X)\}\leq\eta$, with probability tending to $1$.
\end{condition}

\begin{condition}\label{cond:prod-error}
Let
\begin{align*}
   \displaystyle r_{n,t,a,1}^k & = \max_m\Px\left\vert\{\widehat\pi^k_{m}(a\mid\mb X)-\pi^k_\infty(a\mid\mb X)\} \{\widehat S^0_{m}(t\mid a,\mb X)-S^0_\infty(t\mid a,\mb X)\} \right\vert, \\
   \displaystyle r_{n,t,a,2}^k & = \max_m\Px\left\vert\{\widehat\omega^{k,0}_{m}(\mb X)-\omega^{k,0}_\infty(\mb X)\} \{\widehat S^0_{m}(t\mid a,\mb X)-S^0_\infty(t\mid a,\mb X)\} \right\vert, \text{ and}\\
   \displaystyle r_{n,t,a,3}^k & = \max_m\Px\left\vert\widehat S^0_{m}(t\mid a,\mb X)\int_0^t\left\{\frac{G^k_\infty(u-\mid a,\mb X)}{\widehat G^k_{m}(u-\mid a,\mb X)}-1\right\}\left(\frac{S^0_\infty}{\widehat S^0_{m}}-1\right)(du\mid a,\mb X)\right\vert. 
\end{align*}
It holds that $r_{n,t,a,1}^k=o_p(n^{-1/2})$, $r_{n,t,a,2}^k=o_p(n^{-1/2})$ and $r_{n,t,a,3}^k=o_p(n^{-1/2})$.
\end{condition}

\subsubsection{Uniform RAL of the local estimator}\label{subsubapp:conds-site-k}

Following notation in Section \ref{subsubapp:cond-local} and assuming site-$k$ CCOD, in addition to Conditions \ref{cond:nuisance}--\ref{cond:prod-error}, we state the following additional conditions for uniform convergence theory for local estimators.  

\begin{condition}\label{cond:uniform-S}
$$
\max_m\Px\left[
\sup_{u\in[0,t]}\sup_{v\in[0,u]}
\left\vert
\frac{\widehat S^0_{m}(u\mid a,\mb X)}
{\widehat S^0_{m}(v\mid a,\mb X)}
-
\frac{S_\infty^0(u\mid a,\mb X)}
{S_\infty^0(v\mid a,\mb X)}
\right\vert
\right]^2
\to_p 0.
$$
\end{condition}

\begin{condition}\label{cond:prod-error-unif}
    It holds that $\sup_{u\in[0,t]}r_{n,u,a,1}^k=o_p(n^{-1/2})$, $\sup_{u\in[0,t]}r_{n,u,a,2}^k=o_p(n^{-1/2})$, and $\sup_{u\in[0,t]}r_{n,u,a,3}^k=o_p(n^{-1/2})$. 
\end{condition}

Then, the following theorem formally states the (uniform) RAL of the source-site estimator. 

\begin{theorem}\label{thm:RAL-site-k}
    If Conditions \ref{cond:nuisance}--\ref{cond:bound-pi-G} hold, with $\pi_\infty^k=\pi^k$, $\omega^{k,0}_\infty=\omega^{k,0}$, $G_\infty^k=G^k$, and $S_\infty^0=S^0$, the site-$k$ CCOD holds, and Condition \ref{cond:prod-error} also holds, then $\skEst=\estimand+\Px_n(\varphi^{*k,0}_{t,a}) + o_p(n^{-1/2})$. In particular, $n^{1/2}(\skEst-\estimand)$ then convergences in distribution to a normal random variable with mean zero and variance $\sigma^2 = \Px[(\varphi^{*k,0}_{t,a})^2]$. If Conditions \ref{cond:uniform-S} and \ref{cond:prod-error-unif} also hold, then 
    \begin{align*}
        \sup_{u\in[0,t]}\left\vert\widehat\theta^{k,0}_n(u,a)-\theta^0(u,a)-\Px_n(\varphi^{*k,0}_{u,a})\right\vert = o_p(n^{-1/2}). 
    \end{align*}
    In particular, $\{n^{1/2}(\widehat\theta^{k,0}_n(u,a)-\theta^0(u,a)):u\in[0,t]\}$ converges weakly as a process in the space $\ell^\infty([0,t])$ of uniformly bounded functions on $[0,t]$ to a tight mean zero Gaussian process with covariance function $(u,v)\mapsto\Px(\varphi^{*k,0}_{u,a}\varphi^{*k,0}_{v,a})$. 
\end{theorem}

To prove Theorem \ref{thm:RAL-site-k}, we adopt similar strategies in Section \ref{app:theory-CCOD}. We start from considering the difference $\widehat\theta_n^{k,0}(t,a) -\estimand$. Recall that $\Px_n^m$ is the empirical distribution corresponding to the $m$th validation set $\mc V_m$ from the entire data $\mc O$, and denote $\Gx_n^m$ the corresponding empirical process. Then, following the similar process used in \eqref{eq:glb-est-deco}, it can be shown that
\begin{align}\label{eq:k-est-deco}
    \widehat\theta_n^{k,0}(t,a) -\estimand & = \Px_n[\varphi^{*k,0}_{\infty,t,a}] + \frac1M\sum_{m=1}^M\frac{Mn_m^{1/2}}{n}\Gx_n^m\left[\widehat\varphi^{k,0}_{n,m,t,a} - \varphi^{k,0}_{\infty,t,a}\right] \nonumber \\
    & \qquad + \frac1M\sum_{m=1}^M\frac{Mn_m}{n}\Px\left[\widehat\varphi^{k,0}_{t,a} - \estimand\right]. 
\end{align}

Next, similar to what we did for the CCOD estimator, in the following Lemma \ref{lm:L2-bound-k-est}, we establish the $\ell_2(\Px)$ norm distance (bound) between the estimated IF and the limiting IF for the source-site estimator. 

\begin{lemma}\label{lm:L2-bound-k-est}
     Under Condition \ref{cond:bound-pi-G}, there exists a universal constant $C=C(\eta)$ such that for each $k$, $m$, $n$, $t$, and $a$, 
\begin{align*}
\Px[\widehat\varphi^{k,0}_{t,a} - \varphi^{k,0}_{\infty,t,a}]^2 \leq C(\eta)\sum_{j=1}^6 A^k_{j,n,m,t,a},
\end{align*}
where 
\begin{align*}
    A^k_{1,n,m,t,a} & = \Px\left[\frac{1}{\Px_n^m(R=0)}-\frac{1}{\Px(R=0)}\right]^2,\\
    A^k_{2,n,m,t,a} & = \Px\left[\frac{1}{\Px_n^m(R=k)}-\frac{1}{\Px(R=k)}\right]^2,\\
    A^k_{3,n,m,t,a} & = \Px\left[\widehat\omega^{k,0}_{m}(\mb X)-\omega^{k,0}_\infty(\mb X)\right]^2,\\
    A^k_{4,n,m,t,a} & = \Px\left[\frac{1}{\widehat\pi^k_{m}(a\mid\mb X)}-\frac{1}{\pi_\infty^k(a\mid\mb X)}\right]^2,\\
    A^k_{5,n,m,t,a} & = \Px\left[\sup_{u\in[0,t]}\left\vert\frac{1}{\widehat G^k_{m}(u-\mid a,\mb X)}-\frac{1}{G_\infty^k(u-\mid a,\mb X)}\right\vert\right]^2, \\
    A^k_{6,n,m,t,a}&=\Px\left[\sup_{u\in[0,t]}\left\vert\frac{\widehat S^0_{m}(t\mid a,\mb X)}{\widehat S^0_{m}(u\mid a,\mb X)}-\frac{S_\infty^0(t\mid a,\mb X)}{S_\infty^0(u\mid a,\mb X)}\right\vert\right]^2. 
\end{align*}
\end{lemma}

\begin{proof}
    Under the site-$k$ CCOD assumption, $S^k=S^0$, and we denote
\begin{align*} 
   B^k(\mc V_m) & = \frac{\Ix(A=a)}{\pi^k(a\mid\mb X)}S^0(t\mid a, \mb X) \\
   & \quad\times\left[\frac{\Ix(Y\leq t,\Delta=1)}{S^0(Y\mid a, \mb X)G^k(Y-\mid a,\mb X)}-\int_0^{t\wedge Y}\frac{\Lambda^k(du\mid a,\mb X)}{S^0(u\mid a, \mb X)G^k(u-\mid a,\mb X)}\right],\\
   C^k(\mc V_m) & = B^k(\mc V_m)\omega^{k,0}(\mb X). 
\end{align*}
Then, we first have the following decomposition:
\begin{align*}
    \widehat\varphi^{k,0}_{t,a} - \varphi^{k,0}_{\infty,t,a} = \sum_{j=1}^4 U^k_{j,n,m,t,a},
\end{align*}
where 
\begin{align*}
    U^k_{1,n,m,t,a} & = \left[\frac{\Ix(R=0)}{\Px_n^m(R=0)} - \frac{\Ix(R=0)}{\Px(R=0)}\right]\widehat S^0_m(t\mid a,\mb X), \\
    U^k_{2,n,m,t,a} & = \frac{\Ix(R=0)}{\Px(R=0)}\left[\widehat S^0_m(t\mid a,\mb X) - S^0_{\infty}(t\mid a,\mb X)\right], \\
    U^k_{3,n,m,t,a} & = \left[\frac{\Ix(R=k)}{\Px_n^m(R=k)} - \frac{\Ix(R=k)}{\Px(R=k)}\right]\widehat C^k_m(\mc V_m), \\
    U^k_{4,n,m,t,a} & = \frac{\Ix(R=k)}{\Px(R=k)}\left[\widehat C^k_m(\mc V_m) - C^k_{\infty}(\mc V_m)\right].
\end{align*}
Now, for $U^k_{4,n,m,t,a}$, we further decompose it as 
\begin{align*}
    U^k_{4,n,m,t,a} & = \frac{\Ix(R=k)}{\Px(R=k)}\sum_{j=1}^2 V^k_{j,n,m,t,a},
\end{align*}
where 
\begin{align*}
    V^k_{1,n,m,t,a} & = B^k_{\infty}(\mc V_m)\left[\widehat\omega^{k,0}_m(\mb X)-\omega^{k,0}_\infty(\mb X)\right], \\
    V^k_{2,n,m,t,a} & = \widehat\omega^{k,0}_m(\mb X)\left[\widehat B^k_m(\mc V_m) - B^k_{\infty}(\mc V_m)\right].
\end{align*}
The expression of
$\widehat B^k_m(\mc V_m)-B^k_{\infty}(\mc V_m)$
follows the same decomposition as Lemma 3 in \cite{westling2024inference}, with the site-$k$ treatment and censoring nuisance functions and the target-site survival nuisance $\widehat S_m^0$ (with probability limit $S_\infty^0$).
Therefore, the details are omitted. By the triangle inequality, we have $\Px[\widehat\varphi^{k,0}_{t,a} - \varphi^{k,0}_{\infty,t,a}]^2\leq\left\{\sum_{j=1}^4\{\Px[(U^{k}_{j,n,m,t,a})^2]\}^{1/2}\right\}^2$. Therefore, under Assumption \ref{asp:positivity} and Condition \ref{cond:bound-pi-G}, there exists a universal constant $C=C(\eta)$ such that the result in the statement holds. Thus, the proof is completed. 
\end{proof}

Furthermore, we investigate how to make the empirical process term $\Gx_n^m\left[\widehat\varphi^{k,0}_{n,m,t,a}-\varphi^{k,0}_{\infty,t,a_0}\right]$ to be $o_p(n^{-1/2})$ to establish the asymptotic normality. We first recall Lemma \ref{lm:lm4-west} in Section \ref{subsubapp:lem-CCOD} and related empirical process notation there. Then, we introduce the following Lemma \ref{lm:site-k-fun-classes}. For simplicity, $\omega$, $p^0$ and $p^k$ correspond to, respectively, $\omega^{k,0}$, $\Px(R=0)$, $\Px(R=k)$, and since $S^0=S^k$ is assumed, we just use $S$ to denote conditional survival of the event time for all $k$, and without loss of generality, $G=G^k$ and $\pi=\pi^k$ are also dependent on $R=k$. 

\begin{lemma}\label{lm:site-k-fun-classes}
    Let $S,\pi,\omega$, $p^0$, $p^k$ and $G$ be fixed, where $t\mapsto S(t\mid a,\mb x)$ is assumed to be non-increasing for each $(a,\mb x)$, and where $G(t_0-\mid a,\mb x)\geq 1/\eta$, $\pi(a_0\mid\mb x)\geq 1/\eta$, $\omega(\mb x)\leq\eta$, $p^0\geq 1/\eta$ and $p^k\geq 1/\eta$, for some $\eta\in(0,\infty)$. Then the class of influence functions $\mc F_{S,\pi,\omega,p^0,p^k,G,t_0,a_0} = \{\varphi_{S,\pi,\omega,p^0,p^k,G,t_0,a_0}:t\in[0,t_0]\}$ satisfies
    \begin{align*}
        \sup_{\Qx} N(\varepsilon\Vert F\Vert_{\Qx,2},\mc F_{S,\pi,\omega,p^0,p^k,G,t_0,a_0}, L_2(\Qx)) \leq 32/\varepsilon^{10},
    \end{align*}
    for any $\varepsilon\in(0,1]$ where $F=\eta(1+2\eta^3)$ is an envelope of $\mc F_{S,\pi,\omega,p^0,p^k,G,t_0,a_0}$, and the supremum is taken over all distributions $\Qx$ on the sample space of the observed data. 
\end{lemma}

\begin{proof}
    The class $\mc F_{S,\pi,\omega,p^0,p^k,G,t_0,a_0}$ is uniformly bounded by $\eta(1+2\eta^3)$ because of the assumed upper bounds of $1/p^0, 1/p^k, \omega, 1/\pi$ and $1/G$. Therefore, the envelope  function can be taken as $F=\eta(1+2\eta^3)$. 

    Define the following two functions $f_t$ and $h_t$ pointwise as
    \begin{align*}
        f_t(\mb x,r,a_0,\delta,y) & = \frac{\Ix(r=k,a=a_0,y\leq t,\delta=1)\omega(\mb x)S(t\mid a_0,\mb x)}{p^k\pi(a_0\mid\mb x)S(y\mid a_0,\mb x)G(y-\mid a_0,\mb x)},\\
        h_t(\mb x,r,a_0,y) & = \int\frac{\Ix(r=k,a=a_0,u\leq y,u\leq t)\omega(\mb x)S(t\mid a_0,\mb x)}{p^k\pi(a_0\mid\mb x)S(u\mid a_0,\mb x)G(u-\mid a_0,\mb x)}\Lambda(du\mid a_0,\mb x). 
    \end{align*}
    Consider classes $\mc F^1_{S,p^0,t_0,a_0} = \{(\mb x,r)\mapsto \Ix(r=0)S(t\mid a_0,\mb x)/p^0:t\in[0,t_0]\}$, $\mc F^2_{S,p^k,\omega, G,t_0,a_0} = \{(\mb x,r,a_0,\delta,y)\mapsto f_t(\mb x,r,a_0,\delta,y):t\in[0,t_0]\}$, and $\mc F^3_{S,p^k,\omega, G,t,a_0} = \{(\mb x,r,a_0,y)\mapsto h_t(\mb x,r,a,y):t\in[0,t_0]\}$. We can then write 
    \begin{align*}
        \mc F_{S,\pi,\omega,p^0,p^k,G,t_0,a_0}\subseteq\{f_1-f_2+f_3:f_1\in\mc F^1_{S,p^0,t_0,a_0}, f_2\in\mc F^2_{S,p^k,\omega, G,t_0,a_0}, f_3\in\mc F^3_{S,p^k,\omega, G,t_0,a_0}\}.
    \end{align*}
    Since $(\mb x,r)\mapsto \Ix(r=0)S(t\mid a_0,\mb x)/p^0$ is non-increasing for all $(\mb x, r)$ and uniformly bounded by $1/\eta$, Lemma \ref{lm:lm4-west} implies that $\sup_{\Qx}(\varepsilon,\mc F^1_{S,p^0,t_0,a_0},L_2(\Qx))\leq 2\varepsilon^2$. 

    The $\mc F^2_{S,p^k,\omega, G,t_0,a_0}$ is contained in the product of the classes $\{y\mapsto\Ix(y\leq t):t\in[0,t_0]\}$, $\{\mb x\mapsto S(t\mid a_0,\mb x):t\in[0,t_0]\}$, and the singleton class $\{(\mb x,r,a_0,\delta,y)\mapsto\Ix(r=k,a=a_0,\delta=1,y\leq t)\omega(\mb x)/[p^k\pi(a_0\mid\mb x)S(y\mid a_0,\mb x)G(y-\mid a_0,\mb x)]\}.$ The first two classes both have covering number $\sup_{\Qx}N(\varepsilon,\cdot,L_2(\Qx))\leq 2\varepsilon^2$ by Lemma \ref{lm:lm4-west}. The third class has uniform covering number 1 for all $\varepsilon$ because it can be covered with a single ball of any positive radius. In addition,  $\mc F^2_{S,p^k,\omega, G,t,a}$ is uniformly bounded by $\eta^4$. Therefore, $\sup_{\Qx}N(\varepsilon\eta^4,\mc F^2_{S,p^k,\omega, G,t,a},L_2(\Qx))\leq 4/\varepsilon^4$.  

    Furthermore, similar to the proof of Lemma \ref{lm:glb-fun-classes}, we can rewrite $\mc F^3_{S,p^k,\omega, G,t_0,a_0} = \{(\mb x,r,a_0,y)\mapsto\int m_t(u,\mb x,r,a_0,y)\mu^*(du):t\in[0,t_0]\}$, where 
    \begin{align*}
        m_t(u,\mb x,r,a_0,y) = \frac{\Ix(r=k,a=a,u\leq y,u\leq t)\omega(\mb x)S(t\mid a_0,\mb x)}{p^k\pi(a_0\mid\mb x)S(u\mid a_0,\mb x)G(u-\mid a_0,\mb x)}\lambda^*(u\mid a_0,\mb x),
    \end{align*}
    with $\mu^*(du)$ is a dominating measure for $\Px$-almost all $\mb x$ and $\lambda^*(du\mid a_0,\mb x)$ the Radon-Nikodym derivative of $\Lambda(\cdot\mid a_0,\mb x)$ with respect to $\mu^*$. Then, the class $\mc M_{t}=\{m_t:t\in[0,t_0]\}$ is contained in the product of the singleton class $\{(\mb x, r,a_0,y,u)\mapsto\Ix(r=k,a=a_0,u\leq y)\omega(\mb x)\lambda^*(u\mid a_0,\mb x)/[p^k\pi(a_0\mid\mb x)S(u\mid a_0,\mb x)G(u-\mid a_0,\mb x)]\}$ and the class $\{u\mapsto\Ix(u\leq t):t\in[0,t_0]\}$ and $\{\mb x\mapsto S(t\mid a_0,\mb x):t\in[0,t]\}$, which as discussed above both have $L_2(\Qx)$ covering number bounded by $2/\varepsilon^2$ for any probability measure $\Qx$. Therefore, $\sup_{\Qx}N(\varepsilon\eta^4,\mc M_{t},L_2(\Qx))\leq 4/\varepsilon^4$. The remainder of the proof is the same as that of Lemma \ref{lm:glb-fun-classes} (except for replacing $\eta^3$ by $\eta^4$ here). 
\end{proof}

\begin{lemma}\label{lm:RAL-k-Gnm}
    If Conditions \ref{cond:nuisance}--\ref{cond:bound-pi-G} hold,  $M^{-1}\sum_{m=1}^M n^{-1}Mn_m^{1/2}\Gx_n^m\left[\widehat\varphi^{k,0}_{n,m,t,a}-\varphi^{k,0}_{\infty,t,a}\right]= o_p(n^{-1/2})$. If Conditions \ref{cond:uniform-S} holds as well, 
    \begin{align*}
        \frac1M\sum_{m=1}^M\frac{Mn_m^{1/2}}{n}\sup_{u\in[0,t]}\left\vert\Gx_n^m\left[\widehat\varphi^{k,0}_{n,m,t,a}-\varphi^{k,0}_{\infty,t,a}\right]\right\vert = o_p(n^{-1/2}). 
    \end{align*}
\end{lemma}

The proof of Lemma \ref{lm:RAL-k-Gnm} is similar to the proof of Lemma \ref{lm:RAL-ccod-Gnm}, thus it is omitted. 

\begin{lemma}\label{lm:bias-sitek-est}
    Consider some general nuisance functions under $\Px_\infty$, denoted by $S^0_\infty$, $G^k_\infty$, $\pi^k_\infty$, $\omega^{k,0}_\infty$, and $\Lambda^k_\infty$ (equals 1 if $k=0$ for $\omega^{k,0}_\infty$). Then, $\Px[\varphi^{k,0}_{t,a}(\mc O;\Px_\infty)] - \estimand$ equals
    \begin{align*}
        & \Ex\bigg[\frac{q^0(\mb X)}{\Px(R=0)}S^0_\infty(t\mid a, \mb X)\int_0^t\frac{S^0(y-\mid a,\mb X)}{S^0_\infty(y\mid a,\mb X)} \\
        & \quad\times\left\{\frac{\omega^{k,0}_\infty(\mb X)G^k(y-\mid a,\mb X)\pi^k(a\mid\mb X)}{\omega^{k,0}(\mb X)G^k_\infty(y-\mid a,\mb X)\pi^k_\infty(a\mid\mb X)}-1\right\}(\Lambda^k_\infty-\Lambda^k)(dy\mid a,\mb X)\bigg]. 
    \end{align*}
\end{lemma}

\begin{proof}
Following the first steps in the proof of Lemma \ref{lm:bias-ccod-est}, we can show $\Px[\varphi^{k,0}_{t,a}(\mc O;\Px_\infty)] - \estimand$ equals
\begin{align*}
    & \Ex\bigg[\frac{\Ix(R=0)}{\Px(R=0)}\{S^0_\infty(t\mid a,\mb X)-S^0(t\mid a,\mb X)\}+\frac{q^k(\mb X)}{\Px(R=k)}\omega^{k,0}_\infty(\mb X)S^0_\infty(t\mid a, \mb X)\frac{\pi^k(a\mid\mb X)}{\pi^k_\infty(a\mid\mb X)} \\
    & \qquad\qquad\times \int_0^t\frac{S^0(y-\mid a,\mb X)G^k(y-\mid a,\mb X)}{S^0_\infty(y\mid a,\mb X)G^k_\infty(y-\mid a,\mb X)}(\Lambda^k_\infty-\Lambda^k)(dy\mid a,\mb X)\bigg]\\
    & = \Ex\bigg[\frac{q^0(\mb X)}{\Px(R=0)}\{S^0_\infty(t\mid a,\mb X)-S^0(t\mid a,\mb X)\} + \frac{q^0(\mb X)}{\Px(R=0)}\frac{\omega^{k,0}_\infty(\mb X)}{\omega^{k,0}(\mb X)}S^0_\infty(t\mid a, \mb X)\frac{\pi^k(a\mid\mb X)}{\pi^k_\infty(a\mid\mb X)} \\
    & \qquad\qquad \times \int_0^t\frac{S^0(y-\mid a,\mb X)G^k(y-\mid a,\mb X)}{S^0_\infty(y\mid a,\mb X)G^k_\infty(y-\mid a,\mb X)}(\Lambda^k_\infty-\Lambda^k)(dy\mid a,\mb X)\bigg].
\end{align*}
In the second ``$\Ex$'' after ``$=$'', we used the following relationship:
\begin{align*}
    \frac{q^0(\mb X)}{q^k(\mb X)} = \omega^{k,0}(\mb X)\frac{\Px(R=0)}{\Px(R=k)}
\end{align*}
by Bayes's rule. Furthermore, by Duhamel equation in \cite{gill1990survey} used in Lemma \ref{lm:bias-ccod-est} again, $\Px[\varphi^{k,0}_{t,a}(\mc O;\Px_\infty)] - \estimand$ equals
\begin{align}\label{eq:bias-term-site-k}
    & \Ex\bigg[\frac{q^0(\mb X)}{\Px(R=0)}S^0_\infty(t\mid a, \mb X)\int_0^t\frac{S^0(y-\mid a,\mb X)}{S^0_\infty(y\mid a,\mb X)} \nonumber \\
    & \qquad\times \left\{\frac{\omega^{k,0}_\infty(\mb X)G^k(y-\mid a,\mb X)\pi^k(a\mid\mb X)}{\omega^{k,0}(\mb X)G^k_\infty(y-\mid a,\mb X)\pi^k_\infty(a\mid\mb X)}-1\right\}(\Lambda^k_\infty-\Lambda^k)(dy\mid a,\mb X)\bigg]. 
\end{align}
\end{proof}

We now prove Theorem \ref{thm:RAL-site-k} below. 

\begin{proof}
    By \eqref{eq:k-est-deco} with $\pi_\infty^k=\pi^k$,
$\omega^{k,0}_\infty=\omega^{k,0}$,
$G_\infty^k=G^k$, and $S_\infty^k=S^0$,
    \begin{align*}
    \skEst - \estimand & = 
        \Px_n[\varphi^{*k,0}_{t,a}] + \frac1M\sum_{m=1}^M\frac{Mn_m^{1/2}}{n}\Gx_n^m\left[\widehat\varphi^{k,0}_{n,m,t,a}-\varphi^{k,0}_{t,a}\right]  \\
        & \qquad + \frac1M\sum_{m=1}^M\frac{Mn_m}{n}\Px\left[\widehat\varphi^{k,0}_{t,a} - \estimand\right].
    \end{align*}
    By Conditions \ref{cond:nuisance} and \ref{cond:bound-pi-G}, the second summand on the right-hand-side is $o_p(n^{-1/2})$ by Lemma \ref{lm:RAL-k-Gnm}. By Lemma \ref{lm:bias-sitek-est}, $\Px[\widehat\varphi^{k,0}_{t,a}] - \estimand$ equals
    \begin{align*}
    & \Ex\bigg[\frac{q^0(\mb X)}{\Px(R=0)}\widehat S^0_m(t\mid a, \mb X)\int_0^t\frac{S^0(y-\mid a,\mb X)}{\widehat S^0_m(y\mid a,\mb X)} \\
    & \qquad\times\left\{\frac{\widehat\omega^{k,0}_m(\mb X)G^k(y-\mid a,\mb X)\pi^k(a\mid\mb X)}{\omega^{k,0}(\mb X)\widehat G^k_m(y-\mid a,\mb X)\widehat\pi^k_m(a\mid\mb X)}-1\right\}(\widehat\Lambda^k_m-\Lambda^k)(dy\mid a,\mb X)\bigg]. 
\end{align*}
 
By a similar decomposition in the proof of Theorem \ref{thm:RAL-ccod} in Section \ref{subsubapp:proof-RAL-ccod} using Duhamel equation in \cite{gill1990survey}, and by notation in Condition \ref{cond:prod-error}, we find that the above bias term can be bounded by $\eta^2\{r_{n,t,a,1}^k+r_{n,t,a,2}^k+r_{n,t,a,3}^k\}$ over $m$. Since $M^{-1}\sum_{m=1}^Mn^{-1}Mn_m\leq 2$, we have 
\begin{align*}
    \left\vert\frac1M\sum_{m=1}^M\frac{Mn_m}{n}\Px\left[\widehat\varphi^{k,0}_{t,a} - \estimand\right]\right\vert\leq 2\eta^2\left\{r_{n,t,a,1}^k+r_{n,t,a,2}^k+r_{n,t,a,3}^k\right\} = o_p(n^{-1/2}), 
\end{align*}
by Condition \ref{cond:prod-error}. This established the pointwise RAL property: $\skEst=\estimand+\Px_n(\varphi^{*k,0}_{t,a}) + o_p(n^{-1/2})$. Since $\varphi^{*k,0}_{t,a}$ is uniformly bounded, $\Px\{(\varphi^{*k,0}_{t,a})^2\}<\infty$ and since $\Px\{\varphi^{*k,0}_{t,a}\} = 0$, it follows that 
\begin{align*}
    n^{1/2}\Px_n(\varphi^{*k,0}_{t,a})\to_d \mc N\left(0,\Px\{(\varphi^{*k,0}_{t,a})^2\}\right).
\end{align*}

For the uniform RAL, by conditions \ref{cond:nuisance}--\ref{cond:prod-error-unif} and Lemma \ref{lm:RAL-k-Gnm}, 
\begin{align*}
    \frac1M\sum_{m=1}^M\frac{Mn_m^{1/2}}{n}\sup_{u\in[0,t]}\left\vert\Gx_n^m(\widehat\varphi^{k,0}_{u,a}-\varphi^{k,0}_{u,a})\right\vert = o_p(n^{-1/2}).
\end{align*}
Therefore, we have that
\begin{align*}
\sup_{u\in[0,t]}\left\vert\frac1M\sum_{m=1}^M\frac{Mn_m}{n}\Px_n^m(\widehat\varphi^{k,0}_{n,m,u,a}-\estimand)\right\vert \leq\sup_{u\in[0,t]}2\eta^2\{r_{n,u,a,1}^k+r_{n,u,a,2}^k+r_{n,u,a,3}^k\},
\end{align*}
which is $o_p(n^{-1/2})$ by Condition \ref{cond:prod-error-unif}. Therefore, $\sup_{u\in[0,t]}\left\vert\widehat\theta^{k,0}_n(u,a)-\theta^0(u,a)-\Px_n(\varphi^{*k,0}_{u,a})\right\vert = o_p(n^{-1/2})$. Since $\{\varphi^{*k,0}_{u,a}:u\in[0,t]\}$ is a uniformly bounded $\Px$-Donsker class by Lemma \ref{lm:site-k-fun-classes}, $\{n^{1/2}\Px_n\{\varphi^{*k,0}_{u,a}\}:u\in[0,t]\}$ converges weakly to a tight mean-zero Gaussian process with covariance $(u,v)\mapsto\Px(\varphi^{*k,0}_{u,a}\varphi^{*k,0}_{v,a})$. 
\end{proof}

\begin{remark}[Double robustness of the local estimator]\label{rmk:DR-site-k-tech}
If we only need the consistency of $\widehat\theta_n^{k,0}(t,a)$ (under site-$k$ CCOD), then condition $\pi_\infty^k=\pi^k$, $\omega^{k,0}_\infty=\omega^{k,0}$, $G_\infty^k=G^k$, and $S_\infty^0=S^0$ can be replaced by the following statement: For $\Px$-almost all $\mb x$, there exist measurable sets $\mc S^0_x, \mc G^k_x\subseteq [0,t]$ such that $\mc S^0_x\cup\mc G^k_x = [0,t]$ and $\Lambda^k(u\mid a,\mb x) = \Lambda_\infty^k(u\mid a,\mb x)$ for all $u\in\mc S^0_x$ and $G^k(u-\mid a,\mb x)=G^k_\infty(u-\mid a,\mb x)$ for all $u\in\mc G^k_x$. In addition, if $\mc S^0_x$ is a strict subset of $[0,t]$, then $\pi^k(a\mid\mb x)=\pi^k_\infty(a\mid\mb x)$ and $\omega^{k,0}(\mb x)=\omega^{k,0}_\infty(\mb x)$ as well. Then,  $\widehat\theta_n^{k,0}(t,a)$ is consistent if Conditions \ref{cond:nuisance} and \ref{cond:bound-pi-G} hold, and it is uniformly consistent if Condition \ref{cond:uniform-S} also holds. This statement could be interpreted as that at a given time $t$, if either (i) the conditional survival model $S^0$; or (ii) all other nuisance functions $G^k$, $\pi^k$ and $\omega^{k,0}$ are correctly specified (with other conditions above),  $\widehat\theta_n^{k,0}(t,a)$ is consistent. 
\end{remark}

To prove Remark \ref{rmk:DR-site-k-tech}, we note that under site-$k$ CCOD $S^0=S^k$, and we decompose the integral $\int_0^t$ as $\int_{\mc S^0_x}+\int_{\mc S^{0,c}_x}$, where $\mc S^{0,c}_x$ is the complement of set $\mc S^0_x$, and $\mc S^{0,c}_x\subseteq\mc G^k_x$ by definition. Then, it is straightforward to verify that when the statement in Remark \ref{rmk:DR-site-k-tech} holds, the following integral
\begin{align*}
   & \int_0^t\frac{S^0(y-\mid a,\mb X)}{S^0_\infty(y\mid a,\mb X)}\left\{\frac{G^k(y-\mid a,\mb X)\pi^k(a\mid\mb X)}{G^k_\infty(y-\mid a,\mb X)\pi^k_\infty(a\mid\mb X)}-1\right\}(\Lambda^k_\infty-\Lambda^k)(dy\mid a,\mb X) \\
    & = \left(\int_{\mc S^0_x} + \int_{\mc G^k_x}\right)\frac{S^0(y-\mid a,\mb X)}{S^0_\infty(y\mid a,\mb X)}\left\{\frac{G^k(y-\mid a,\mb X)\pi^k(a\mid\mb X)}{G^k_\infty(y-\mid a,\mb X)\pi^k_\infty(a\mid\mb X)}-1\right\}(\Lambda^k_\infty-\Lambda^k)(dy\mid a,\mb X) = 0,
\end{align*}
which further implies $\Px[\varphi^{*k,0}_{\infty,t,a}]= 0$. 

\subsection{Oracle selection and efficiency by the federated estimator}\label{subapp:theory-fed}

In this section, we present the theoretical properties of the federated estimator. Given that our proposed weights, $\bd\eta_{t,a}$, are both time- and treatment-specific, we focus on the pointwise convergence properties of the federated estimator.

Let the set of all source site indices be $\mc S = {1, \dots, K-1}$. Let $\theta^k(t,a)$ denote the probability limit of the local estimator $\widehat\theta_n^{k,0}(t,a)$. We then define the oracle selection space for $\bd\eta_{t,a}$, and the corresponding weight space as:
\begin{align*}
    \mc S^*_{t,a} = \{k\in\mc S:\theta^k(t,a)=\estimand\}, ~~\text{and}~~\mathbb R^{S^*_{t,a}} = \{\bd\eta_{t,a}\in\mathbb R^{K-1}:\eta^j_{t,a}=0,\forall j\not\in\mc S^*_{t,a}\},
\end{align*}
respectively. Thus, $S^*_{t,a}$ contains source sites whose local estimators are asymptotically aligned with the target estimand at the given $(t,a)$.

The space $\mc S^*_{t,a}$ is both time- and treatment-varying, indicating that a source site may not consistently be useful or unhelpful across different time points or treatments. However, it offers the advantage of increased flexibility and adaptivity, allowing for more effective borrowing of information at different points along the survival functions. For $k\in S^*_{t,a}$, the probability limit of the site-specific local
estimator coincides with $\theta^0(t,a)$ by definition. 

We begin by assuming fixed $\bd\eta_{t,a}=(\eta_{t,a}^0,\eta_{t,a}^1,\dots,\eta_{t,a}^{K-1})$. We invoke Lemmata 4 and 5 in \cite{han2025federated}, which state that the proposed adaptive estimation for $\eta_{t,a}^k$ as shown in \eqref{eq:fed-obj} allows for (i) the recovery of the optimal $\bar\eta^k_{t,a}$ by the estimator $\widehat\eta^k_{t,a}$, and (ii) the uncertainty induced by $\widehat\eta^k_{t,a}$ is negligible when estimating $\estimand$. We require regularity Conditions \ref{cond:nuisance}, \ref{cond:bound-pi-G} and \ref{cond:prod-error} for the pointwise convergence result in Theorem \ref{thm:RAL-site-k} hold. Let us denote the federated estimator by plugging-in the fixed $\bd\eta_{t,a}$ as 
\begin{align*}
    \fedEstfixed = \left(1-\sum_{k\in\mc S}\eta^k_{t,a}\right)\tgtEst + \sum_{k\in\mc S}\eta^k_{t,a}\skEst.
\end{align*}
Recall that notation $\mc H_{t,a}$ defined in the main text: 
$$
\mc H_{t,a}(\mc O;S,G) = \frac{\Ix(Y\leq t,\delta=1)}{S(Y\mid a,\mb X)G(Y-\mid a, \mb X)} - \int_0^{t\wedge Y}\frac{\Lambda(du\mid a,\mb X)}{S(u\mid a, \mb X)G(u-\mid a,\mb X)}. 
$$
Let us then write 
\begin{align*}
    \xi^{0,(1)}(\mc O) & = S^0(t\mid a, \mb X)\frac{\Ix(A=a)}{\pi^0(a\mid\mb X)}\mc H_{t,a}(\mc O;S^0,\Lambda^0,G^0),\\
    \xi^{k,0,(1)}(\mc O) & = \omega^{k,0}(\mb X)S^k(t\mid a, \mb X)\frac{\Ix(A=a)}{\pi^k(a\mid\mb X)}\mc H_{t,a}(\mc O;S^k,\Lambda^k, G^k), \\
    \xi^{0,(2)}(\mc O) & = S^0(t\mid a, \mb X) - \estimand, 
\end{align*}
and $n_k=\sum_{i=1}^n\Ix(R_i=k)$ for $k=0,1,\dots,K-1$. 

Then, 
\begin{align}\label{eq:decomp-fedfix}
    & \fedEstfixed-\estimand \nonumber \\ 
    & = \left(1-\sum_{k\in\mc S}\eta^k_{t,a}\right)\left\{\tgtEst-\estimand\right\} + \sum_{k\in\mc S}\eta^k_{t,a}\left\{\skEst-\estimand\right\}\nonumber \\
    & = \left(1-\sum_{k\in\mc S}\eta^k_{t,a}\right)\frac{1}{n_0}\sum_{i=1}^n\Ix(R_i=0)\left\{\widehat\xi^{0,(2)}(\mc O_i)-\widehat\xi^{0,(1)}(\mc O_i)\right\} \nonumber\\
    & \quad + \sum_{k\in\mc S}\frac{1}{n_0}\sum_{i=1}^n\Ix(R_i=0)\eta^k_{t,a}\widehat\xi^{0,(2)}(\mc O_i) - \sum_{k\in\mc S}\frac{1}{n_k}\sum_{i=1}^n\Ix(R_i=k)\eta^k_{t,a}\widehat\xi^{k,0,(1)}(\mc O_i) \nonumber\\
    & = \frac{1}{n}\sum_{i=1}^n\left(1-\sum_{k\in\mc S}\eta^k_{t,a}\right)\Ix(R_i=0)\frac{\widehat\xi^{0,(2)}(\mc O_i)-\widehat\xi^{0,(1)}(\mc O_i)}{\widehat\Px(R_i=0)} \nonumber\\
    & \quad + \frac{1}{n}\sum_{i=1}^n\Ix(R_i=0)\left(\sum_{k\in\mc S}\eta^k_{t,a}\right)\frac{\widehat\xi^{0,(2)}(\mc O_i)}{\widehat\Px(R_i=0)} - \frac{1}{n}\sum_{k\in\mc S}\sum_{i=1}^n\Ix(R_i=k)\eta^k_{t,a}\frac{\widehat\xi^{k,0,(1)}(\mc O_i)}{\widehat\Px(R_i=k)}.
\end{align}
The asymptotic variance of $\fedEstfixed$ equals the variance of the influence function of \eqref{eq:decomp-fedfix}. Let us denote it as $\mc V_{t,a}^{\text{fed}} = \mc V_{t,a}^{\text{fed}}(\bd\eta_{t,a})$. We highlight its dependence to the federated weights vector $\bd\eta_{t,a}$ here because in the below \eqref{eq:var-opti}, we consider an optimization program for deriving the weights based on minimizing the (estimated) asymptotic variance. 

Under the assumption of i.i.d. participants within each site, we have
\begin{align}\label{eq:var-fedfix}
    \mc V_{t,a}^{\text{fed}} & = \left(1-\sum_{k\in\mc S}\eta^k_{t,a}\right)^2\frac{\Vx\{\xi^{0,(2)}(\mc O_i)-\xi^{0,(1)}(\mc O_i)\mid R_i=0\}}{\Px(R_i=0)}\nonumber \\
    & \quad + \left(\sum_{k\in\mc S}\eta^k_{t,a}\right)^2\frac{\Vx\{\xi^{0,(2)}(\mc O_i)\mid R_i=0\}}{\Px(R_i=0)}\nonumber \\
    & \quad + 2\left(1-\sum_{k\in\mc S}\eta^k_{t,a}\right)\left(\sum_{k\in\mc S}\eta^k_{t,a}\right)\frac{\text{Cov}\{\xi^{0,(2)}(\mc O_i)-\xi^{0,(1)}(\mc O_i), \xi^{0,(2)}(\mc O_i)\mid R_i=0\}}{\Px(R_i=0)}\nonumber \\
    & \quad + \sum_{k\in\mc S}(\eta^k_{t,a})^2\frac{\Vx\{\xi^{k,0,(1)}(\mc O_i)\mid R_i=k\}}{\Px(R_i=k)}.
\end{align}
With appropriate boundedness conditions on variance and covariance of the influence functions, this variance is finite. Consequently, the asymptotic distribution of $\fedEstfixed$ is expressed as 
\begin{align*}
    \sqrt{n}\left\{\fedEstfixed-\estimand\right\}\to_d\mc N(0,\mc V_{t,a}^{\text{fed}}).
\end{align*}
We further define the optimal adaptive weights $\bar{\bd\eta}_{t,a}$ as follows:
\begin{align}\label{eq:var-opti}
    \bar{\bd\eta}_{t,a} = \underset{\eta^k_{t,a}=0,\forall k\not\in\mc S^*_{t,a}}{\text{arg min}} \mc V_{t,a}^{\text{fed}}(\bd\eta_{t,a}).
\end{align}

We adapt two lemmata below (from \cite{han2025federated}) for recovering the optimal weights $\bar{\bd\eta}_{t,a}$ with negligible uncertainty for estimating $\estimand$ if we estimate ${\bd\eta}_{t,a}$ using \eqref{eq:fed-obj}, akin to adaptive Lasso \citep{zou2006adaptive, fan2024fast}.

\begin{lemma}[adapted from Lemma 4 in \cite{han2025federated}]
\label{lem:han4}
Under Conditions \ref{cond:nuisance} and \ref{cond:prod-error} in the main text,
suppose that the tuning parameter $\lambda=\lambda_n$ satisfies
\[
    \frac{\lambda_n}{n}\to\infty
    \qquad\text{and}\qquad
    \frac{\lambda_n}{n^{3/2}}\to 0.
\]
In addition, assume the following mild conditions on the covariate support and covariance structure: (i) the covariates $\mb X$ lie in a compact set $\mb X\in[-B,B]^p$, and the density ratio satisfies
$\omega^{k,0}(\mb X)\in(0,B]$ for all $k=1,\ldots,K-1$ with probability $1$; and (ii) the variance of $\xi^{k,0,(1)}(\mc O)$ lies in
$[\varepsilon,B]$, and the variance-covariance matrix $\mc V\!\left[(\xi^{0,(1)},\xi^{0,(2)})'\mid R=0\right]$
has eigenvalues in $[\varepsilon,B]$ for some positive constants
$\varepsilon$ and $B$.
Then,
\[
\lim_{n\to\infty}
\mathbb P\!\left(
    \widehat{\bd\eta}_{t,a}
    \in \mathbb R^{S^*_{t,a}}
\right)=1,
\qquad
\left\|
    \widehat{\bd\eta}_{t,a}
    -
    \bar{\bd\eta}_{t,a}
\right\|
=
O_p(n^{-1/2}),
\]
for all $(t,a)\in[0,\tau]\times\{0,1\}$.
\end{lemma}

\begin{lemma}[adapted from Lemma 5 in \cite{han2025federated}]\label{lem:han5}
    Under conditions in Lemma \ref{lem:han4},
$$
\sqrt{n} \left(\widehat{\theta}^{\text{fed}}(t,a; \widehat{\bd\eta}_{t,a}) - \theta^0(t,a)\right) \to_d \mc{N}\big(0, \mc{V}_{t,a}^{\text{fed}}(\bar{\bd\eta}_{t,a})\big),
$$
for all $(t,a) \in [0,\tau]\times\{0,1\}$.
\end{lemma}

We estimate $\mc V_{t,a}^{\mathrm{fed}}$ using the corresponding plug-in
estimator. Let $\widehat{\Px}(R=k)=n_k/n$, and let
$\widehat{\Vx}_k(\cdot)$ and $\widehat{\mathrm{Cov}}_k(\cdot,\cdot)$
denote the empirical variance and covariance among participants from
site $k$. Define
\begin{align}
\widehat{\mc V}_{t,a}^{\mathrm{fed}}
&=
\left(1-\sum_{k\in\mc S}\widehat\eta^k_{t,a}\right)^2
\frac{
\widehat{\Vx}_0
\left\{
\widehat\xi^{0,(2)}-\widehat\xi^{0,(1)}
\right\}
}{
\widehat{\Px}(R=0)
}
\nonumber\\
&\quad+
\left(\sum_{k\in\mc S}\widehat\eta^k_{t,a}\right)^2
\frac{
\widehat{\Vx}_0
\left\{
\widehat\xi^{0,(2)}
\right\}
}{
\widehat{\Px}(R=0)
}
\nonumber\\
&\quad+
2
\left(1-\sum_{k\in\mc S}\widehat\eta^k_{t,a}\right)
\left(\sum_{k\in\mc S}\widehat\eta^k_{t,a}\right)
\frac{
\widehat{\mathrm{Cov}}_0
\left\{
\widehat\xi^{0,(2)}-\widehat\xi^{0,(1)},
\widehat\xi^{0,(2)}
\right\}
}{
\widehat{\Px}(R=0)
}
\nonumber\\
&\quad+
\sum_{k\in\mc S}
(\widehat\eta^k_{t,a})^2
\frac{
\widehat{\Vx}_k
\left\{
\widehat\xi^{k,0,(1)}
\right\}
}{
\widehat{\Px}(R=k)
}.
\label{eq:var-fed-est}
\end{align}

Under the conditions of Lemma~\ref{lem:han5} and consistency of the
empirical variance and covariance components,
\[
\widehat{\mc V}_{t,a}^{\mathrm{fed}}
\overset{p}{\longrightarrow}
\mc V_{t,a}^{\mathrm{fed}}(\bar{\bd\eta}_{t,a}).
\]
Thus,
\begin{align*}
\sqrt{n/\widehat{\mc V}_{t,a}^{\mathrm{fed}}}
\left\{\fedEst-\estimand\right\}
\to_d \mc N(0,1).
\end{align*}
We now proceed to analyze the efficiency gain resulting from the federation process. The estimator relies only on the target data is denoted as $\tgtEst=\widehat\theta^\text{fed}_n(t,a;\bd\eta^0_{t,a})$, where $\bd\eta^0_{t,a}$ assigns all weights to the target and none to the source. In contrast, the estimator that leverages the proposed adaptive ensemble approach is denoted as $\widehat\theta^\text{fed}_n(t,a;\widehat{\bd\eta}_{t,a})$. Here $\widehat{\bd\eta}_{t,a}$ can recover the optimal weights $\bar{\bd\eta}_{t,a}$ that are associated with the minimum asymptotic variance. Consequently, the variance of $\widehat\theta^\text{fed}_n(t,a;\widehat{\bd\eta}_{t,a})$ is no larger than that of the estimator relying solely on the target data since $\bd\eta^0_{t,a}$ is generally not the variance minimizer. 

To establish that the asymptotic variance of $\widehat\theta^\text{fed}_n(t,a;\widehat{\bd\eta}_{t,a})$ is strictly smaller than that of the estimator based solely on the target data $\tgtEst$, we adopt Proposition 1 in \cite{han2025federated} with a modified informative source condition (modified Assumption 3(b) in \cite{han2025federated}). 

Specifically, for each source site $s\in\mc S^*_{t,a}$, we define $\widehat\theta^\text{fed}_n(t,a;\eta^s_{t,a})$ a federated estimator where $\eta^s_{t,a}$ is the optimal ensemble weight of site $s$ if we only consider target site and this source site $s$ for the federation. Then, the modified informative source condition is given as  
$$
\left\vert\text{Cov}\left[\sqrt{n}\tgtEst, \sqrt{n}\left\{\widehat\theta^\text{fed}_n(t,a;\eta^s_{t,a})-\tgtEst\right\}\right]\right\vert\geq\varepsilon,
$$ 
for some $\varepsilon>0$, where $\widehat\theta^\text{fed}_n(t,a;\eta^s_{t,a})-\tgtEst$ can be expressed as 
\begin{align*}
    & \widehat\theta^\text{fed}_n(t,a;\eta^s_{t,a})-\tgtEst \\ 
    & = \left\{\widehat\theta^\text{fed}_n(t,a;\eta^s_{t,a})-\estimand\right\} - \left\{\tgtEst-\estimand\right\} \\
    & = \frac1n\sum_{i=1}^n\Ix(R_i=0)(1-\eta^s_{t,a})\frac{\widehat\xi^{0,(2)}(\mc O_i)-\widehat\xi^{0,(1)}(\mc O_i)}{\widehat\Px(R_i=0)} + \frac1n\sum_{i=1}^n\Ix(R_i=0)\eta^s_{t,a}\frac{\widehat\xi^{0,(2)}(\mc O_i)}{\widehat\Px(R_i=0)} \\
    & \quad - \frac1n\sum_{i=1}^n\Ix(R_i=s)\eta^s_{t,a}\frac{\widehat\xi^{s,0,(1)}(\mc O_i)}{\widehat\Px(R_i=s)} - \frac1n\sum_{i=1}^n\Ix(R_i=0)\frac{\widehat\xi^{0,(2)}(\mc O_i)-\widehat\xi^{0,(1)}(\mc O_i)}{\widehat\Px(R_i=0)} \\
    & = \frac1n\sum_{i=1}^n\Ix(R_i=0)\eta^s_{t,a}\frac{\widehat\xi^{0,(1)}(\mc O_i)}{\widehat\Px(R_i=0)} - \frac1n\sum_{i=1}^n\Ix(R_i=s)\eta^s_{t,a}\frac{\widehat\xi^{s,0,(1)}(\mc O_i)}{\widehat\Px(R_i=s)}.
\end{align*}
To establish strict efficiency improvement, consider an informative source site $s\in\mc S^*_{t,a}$ and the federated estimator based only on the target site and source site $s$. From the asymptotic expansions above,
\begin{align*}
\text{Cov}\left[
\sqrt{n}\tgtEst,
\sqrt{n}\left\{
\widehat\theta_n^{\mathrm{fed}}(t,a;\eta^s_{t,a})
-\tgtEst
\right\}
\right] =
-\eta^s_{t,a}
\frac{
\Vx\{\xi^{0,(1)}(\mc O)\mid R=0\}
}{
\Px(R=0)
},
\end{align*}
where we use
$\text{Cov}\{\xi^{0,(2)}(\mc O),
\xi^{0,(1)}(\mc O)\mid R=0\}=0$.
Thus, whenever
$\Vx\{\xi^{0,(1)}(\mc O)\mid R=0\}>0$
and the optimal source weight satisfies $\eta^s_{t,a}>0$,
incorporating this source yields a strict reduction in asymptotic variance.

\newpage

\section{Full Simulation  Details and Results}\label{app:experiments}

\subsection{Data generation process}\label{subapp:DGP}

With 500 independent synthetic data replications and a total sample size of $n = \sum_{k=0}^{K-1} n_k$ distributed across $K = 5$ sites, each site $k \in \{0,1,\dots,4\}$ contains $n_k$ observations, where site $k=0$ represents the target site. The target-site sample size is fixed at $n_0 = 300$, while source-site sizes vary as $n_k \in \{300, 600, 1000\}$ for $k = 1,\dots,4$, corresponding to limited, moderate, and abundant external data. The choice of $K=5$ is similar to the number of regional sites in the AMP trials. The true survival curves are obtained by averaging over a large simulated population ($n_{\text{super}} = 10^8$) generated from the target-site distribution. 

Three covariates $X_1$, $X_2$, and $X_3$ are sampled as transformations of Beta random variables with site-specific parameters: 
\begin{align*}
    X_1 &\sim 33 \cdot \text{Beta}(1.1 - 0.05\gamma_{\mathrm{sim}}(k), 1.1 + 0.2\gamma_{\mathrm{sim}}(k)) + 9 + 2\gamma_{\mathrm{sim}}(k), \\
    X_2 &\sim 52 \cdot \text{Beta}(1.5 + (X_1 + 0.5\gamma_{\mathrm{sim}}(k))/20, 4 + 2\gamma_{\mathrm{sim}}(k)) + 7 + 2\gamma_{\mathrm{sim}}(k), \\
    X_3 &\sim (4 + 2\gamma_{\mathrm{sim}}(k)) \cdot \text{Beta}(1.5 + |X_1 - 50 + 3\gamma_{\mathrm{sim}}(k)|/20, 3 + 0.1\gamma_{\mathrm{sim}}(k)),
\end{align*}
where $\gamma_{\mathrm{sim}}(k)$ represents some function of site $k$, specified later. We generate treatment assignment probabilities using site-specific logistic models. For the source sites, we define
$$
h^{\text{s}}(\mb X)
=
-1.05+\log\left\{
1.3+\exp(-12+X_1/10)+\exp(-2+X_2/12)+\exp(-2+X_3/3)
\right\}.
$$
For target site, we consider
$h^0(\mb X)=0.7h^{\text{s}}(\mb X)+0.3h^{\text{b}}(\mb X),$
where 
$$
h^{\text{b}}(\mb X)
=
-1.05+\log\left\{
0.3+\exp(-120+X_1)+\exp(-6+X_2/4)+\exp(-6+X_3)
\right\}. 
$$
The treatment assignment probability is therefore
$$
\pi^r(\mb X)
=
\text{expit}\left\{
\Ix(r=0)h^0(\mb X)+\Ix(r\neq 0)h^{\text{s}}(\mb X)
\right\},
$$
where $r=0$ denotes the target site. Treatments are then independently sampled as
$$
A\mid \mb X,R=r\sim\text{Bernoulli}\{\pi^r(\mb X)\}.
$$
Next, we consider the mechanisms of event and censoring times. The hazard rates for event times and censoring times are given by the following $\exp(h_t)$ and $\exp(h_c)$, respectively, where $h_t = -5.02 + 0.1(X_1 - 25) - 0.1(X_2 - 25) + 0.05(X_3 - 2) + D_T(k) \cdot 0.1(X_2 - 25) + A \cdot \delta_T(k)\cdot 0.1(X_1+X_2+X_3-50),$ and $h_c = -4.87 + 0.01(X_1 - 25) - 0.02(X_2 - 25) + 0.01(X_3 - 2) - D_C(k) \cdot 0.1(X_2 - 25) + A \cdot \delta_C(k)\cdot 0.1(X_1+X_2+X_3-50)$. 

Here, $D_T(k)$, $D_C(k)$, $\delta_T(k)$ and $\delta_C(k)$ are some site-specific indicators, specified later, for varying the treatment effects and trends of survival curves for different sites. Then, event times and censoring times are sampled as: 
$$
T = \left(-\frac{\log(U_1)}{\exp(h_t) \cdot \lambda_W}\right)^{1/\rho_W}, \quad
C = \left(-\frac{\log(U_2)}{\exp(h_c) \cdot \lambda_W}\right)^{1/\rho_W},
$$
with $\rho_W = 1.2$, $\lambda_W = 0.6$, and $U_1, U_2 \sim \text{Uniform}(0, 1)$. This technique follows \cite{austin2012generating}. Thus, the observed times and event indicators are $Y = \min(T, C), \Delta = \mathbb{I}(T \leq C)$, respectively. 

Under this DGP, event times are generated on a daily scale over one year (365 days), and we truncate the censoring time at $\tau=200$ days to mimic the time-horizon in survival analysis. Our DGP supports the following scenarios based on site-specific heterogeneity: 
\begin{itemize}
    \item \textbf{Homogeneous}: Homogeneous covariates and hazard rates across sites. We let $\gamma_{\mathrm{sim}}(k)=D_T(k)=D_C(k)=\delta_T(k)=\delta_C(k)=0$ for $k=0,1,\dots,4$. 
    \item \textbf{Covariate Shift}: Covariates $X_1$, $X_2$, and $X_3$ vary across sites. We let $\gamma_{\mathrm{sim}}(k)=k$ and $D_T(k)=D_C(k)=\delta_T(k)=\delta_C(k)=0$, for $k=0,1,\dots,4$. 
    \item \textbf{Outcome Shift}: Conditional outcome distribution varies across sites. We assign $\gamma_{\mathrm{sim}}(k)=0$, $D_T(k)=\delta_T(k)=k$, and $D_C(k)=\delta_C(k)=0$ for $k=0,1,\dots,4$. 
    \item \textbf{Censoring Shift}: Censoring mechanism varies across sites. We let $\gamma_{\mathrm{sim}}(k)=0$, $D_T(k)=\delta_T(k)=0$ and $D_C(k)=\delta_C(k)=k$, for $k=0,1,\dots,4$. 
    \item \textbf{All Shifts}: Covariates and both event and censoring effects vary across sites. We let $\gamma_{\mathrm{sim}}(k)=D_T(k)=D_C(k)=\delta_T(k)=\delta_C(k)=k$, for $k=0,1,\dots,4$. 
\end{itemize}

Figure \ref{fig:curves} plots the true treatment-specific survival curves under the Covariate Shift and Outcome Shift scenarios to illustrate the effect of site differences on survival outcomes. Under Covariate Shift, the curves maintain similar shapes and trends, differing primarily in scale. In contrast, Outcome Shift leads to marked alterations in the shapes of the survival curves.

\begin{figure}[H]
    \centering
    Source sites have covariate shift:
    \medskip
    
    \includegraphics[width=\textwidth]{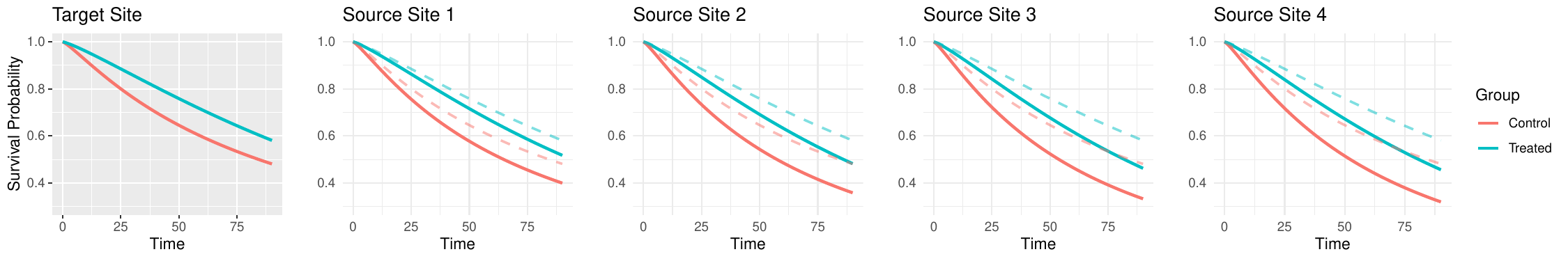}
    
    Source sites have outcome shift:
    \medskip
    
    \includegraphics[width=\textwidth]{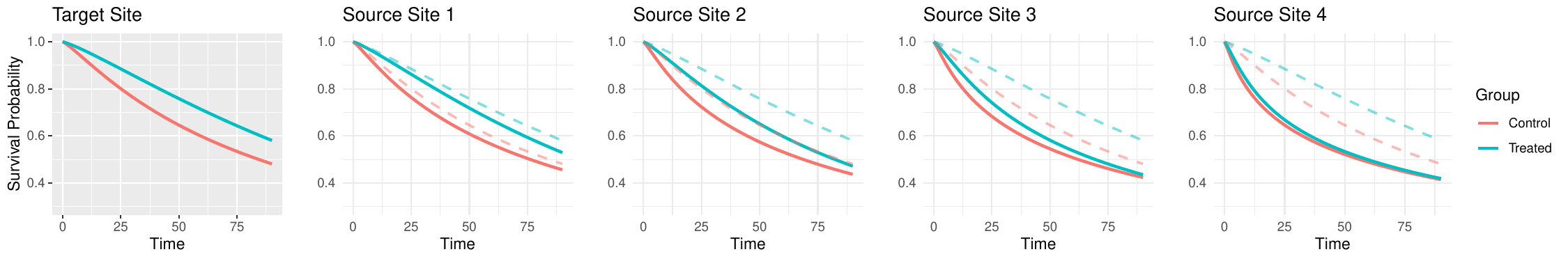}
    \caption{True treatment-specific survival curves across different sites. For illustration, each curve is derived from a random sample of $n=10^5$ generated under the site's own DGP. Dashed lines represent the target site's survival curves for reference.}\label{fig:curves}
\end{figure}

\subsection{Performance metrics}

The simulation results are evaluated using five metrics. Let $\theta$ be the true target parameter, and $\widehat\theta_i$, $\widehat\sigma_i$ denote the point and standard error estimates from the $i$th Monte Carlo replication ($i=1,\dots,500$). 
\begin{itemize}
    \item \textbf{Estimation Bias} (by boxplot): $\widehat\theta_i-\theta$, $i=1,\dots,500$; 
    \item \textbf{ARBias\%:} the absolute percent relative bias, defined by $100\%\cdot\vert500^{-1}\sum_{i=1}^{500}(\widehat\theta_i-\theta)/\theta\vert$; 
    \item \textbf{RRMSE:} the root mean square error (RMSE) of a method relative to that of the TGT estimator, where $\text{RMSE} = \sqrt{500^{-1}\sum_{i=1}^{500}(\widehat\theta_i-\theta)^2}$. By definition, TGT has RRMSE $=1$. Smaller RRMSE values indicate higher efficiency relative to TGT; 
    \item \textbf{CI Width:} the average 95\% CI width across replications, where the 95\% CI from the $i$th replication is $\widehat\theta_i \pm 1.96\widehat\sigma_i$; and 
    \item \textbf{CP\%:} coverage probability in percentage, defined by the percentage of replications in which the CI covers $\theta$. A closer CP\% to 95 indicates more reliable inference. 
\end{itemize}
Note that the CP\% (here denoted by $\widehat p$) is a sample proportion, which follows the below asymptotic distribution:
\begin{align*}
    \widehat p\sim\mc N(p, \sqrt{p(1-p)/M}),
\end{align*}
where $M=500$ is the number of Monte Carlo replicates, and $p=95\%$ is the nominal coverage level. Hence, the Monte Carlo error for $\widehat p$ is $\sqrt{p(1-p)/M} = \sqrt{0.95\cdot 0.05/500}\approx 0.0097 < 0.01$. This means that with 500 replicates, we can trust the first two decimal points (before taking percentage) of the coverage probability. 

\subsection{Simulation results for treatment-specific survival curves under the main setting}

Figures \ref{fig:res-small-bias}--\ref{fig:res-large-cp} present the  simulation results on estimation and inference of treatment-specific survival functions of all simulation scenarios. 

\begin{figure}[H]
    \centering
    \includegraphics[width=\textwidth]{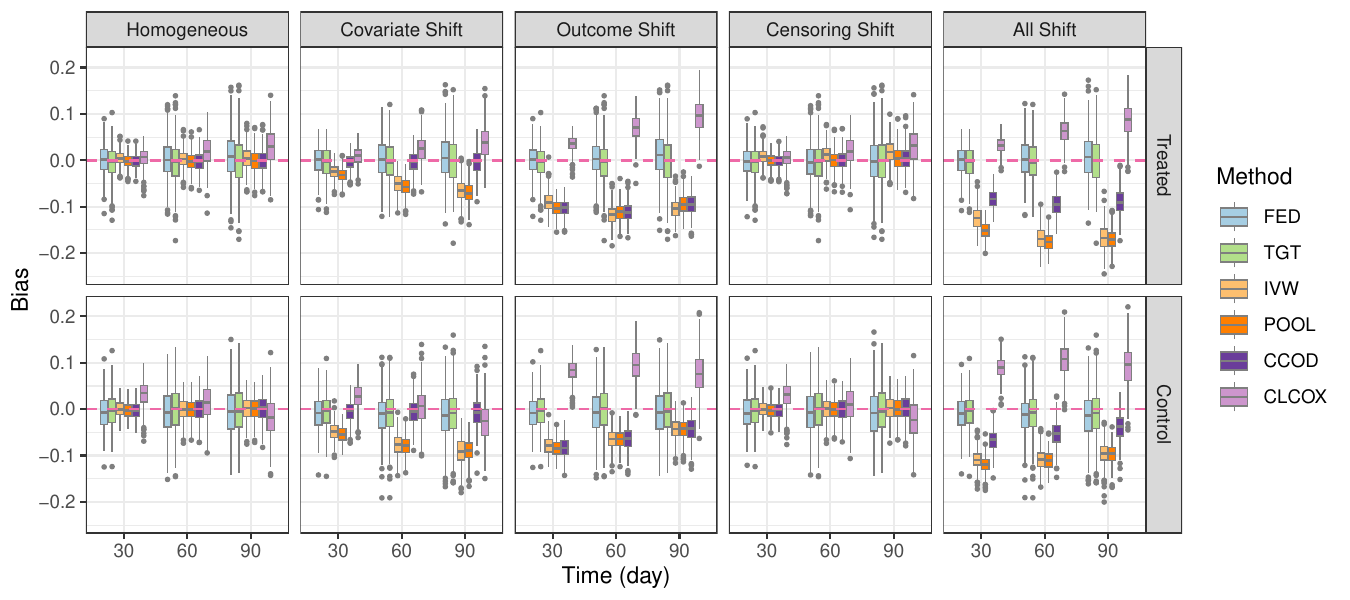}
    \caption{Estimation bias (boxplots) under $n_k=300$ ($k=1,2,3,4$), evaluated at days 30, 60 and 90 in simulation. The true survival probabilities at days 30, 60, and 90 are $(0.86, 0.71, 0.58)$ for the treated group and $(0.76, 0.60, 0.48)$ for the control group. }
    \label{fig:res-small-bias}
\end{figure}

\begin{figure}[H]
    \centering
    \includegraphics[width=\textwidth]{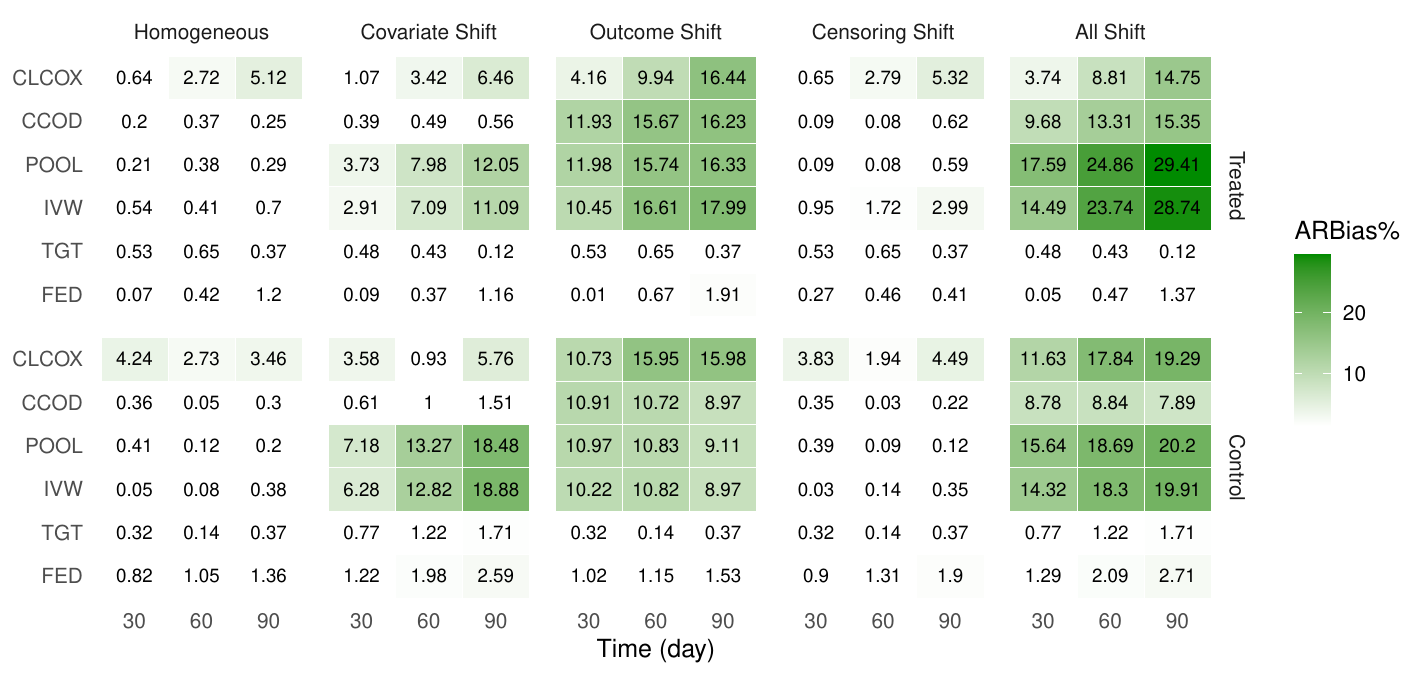}
    \includegraphics[width=\textwidth]{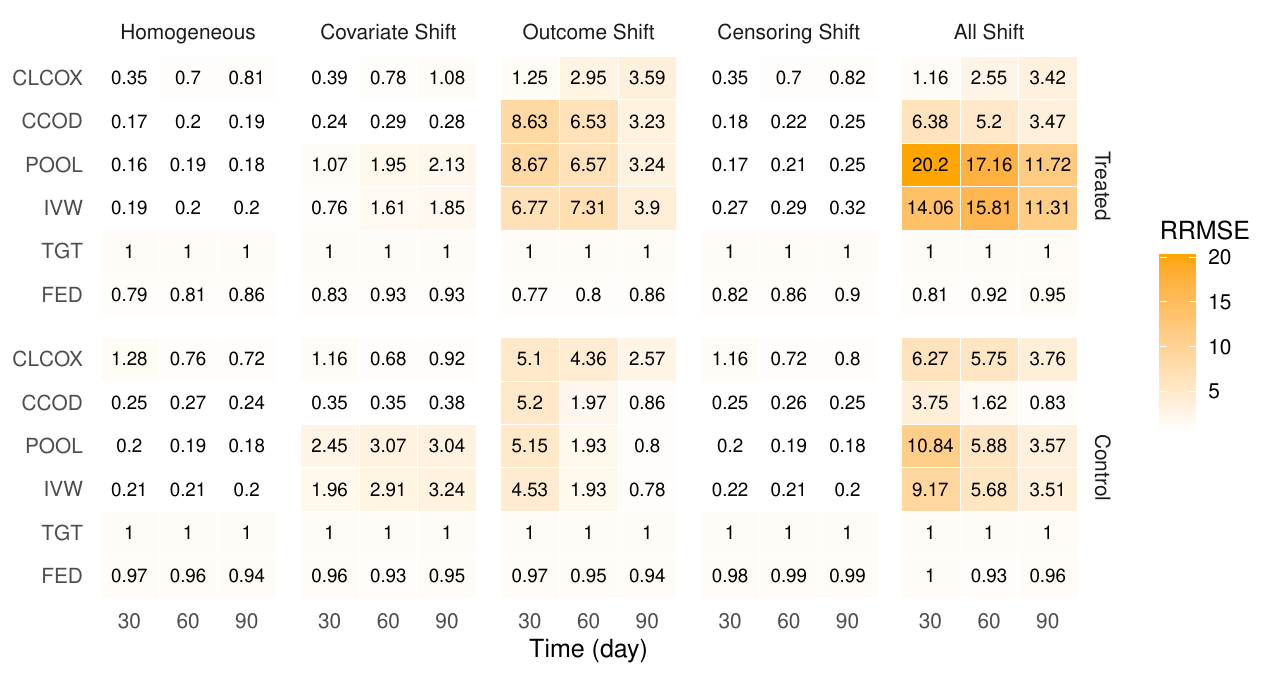}
    \caption{ARBias\% and RRMSE under $n_k=300$ ($k=1,2,3,4$), evaluated at days 30, 60 and 90 in simulation. The true survival probabilities at days 30, 60, and 90 are $(0.86, 0.71, 0.58)$ for the treated group and $(0.76, 0.60, 0.48)$ for the control group. }
    \label{fig:res-small-rrmse}
\end{figure}

\begin{figure}[H]
    \centering
    \includegraphics[width=\textwidth]{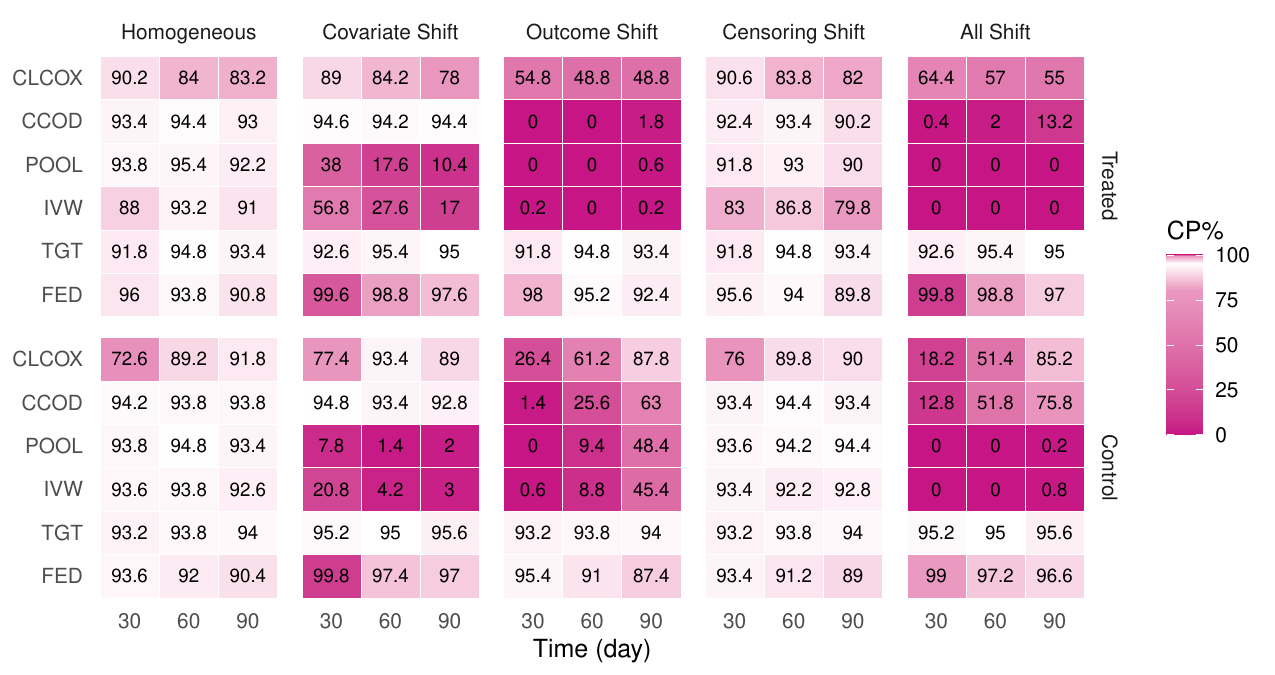}
    \includegraphics[width=\textwidth]{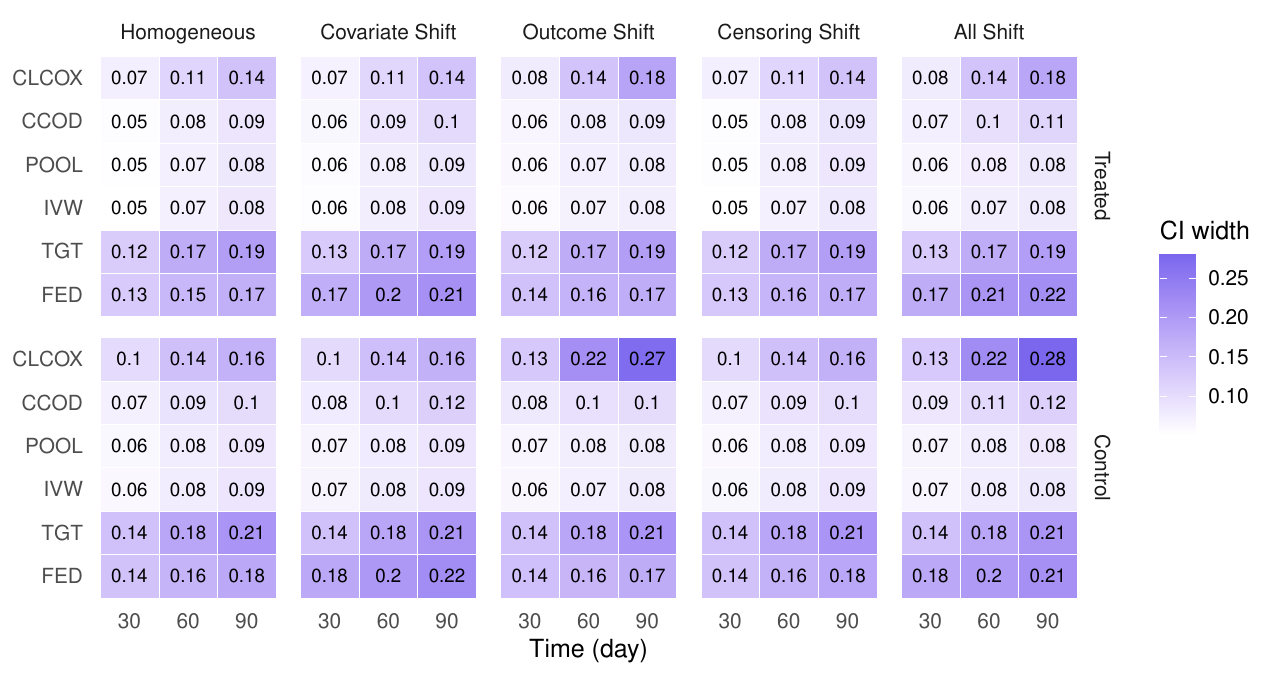}
    \caption{CP\% with 95\% nominal coverage level, and width of 95\% CI under $n_k=300$ ($k=1,2,3,4$), evaluated at days 30, 60 and 90 in simulation. The true survival probabilities at days 30, 60, and 90 are $(0.86, 0.71, 0.58)$ for the treated group and $(0.76, 0.60, 0.48)$ for the control group. }
    \label{fig:res-small-cp}
\end{figure}

\begin{figure}[H]
    \centering
    \includegraphics[width=\textwidth]{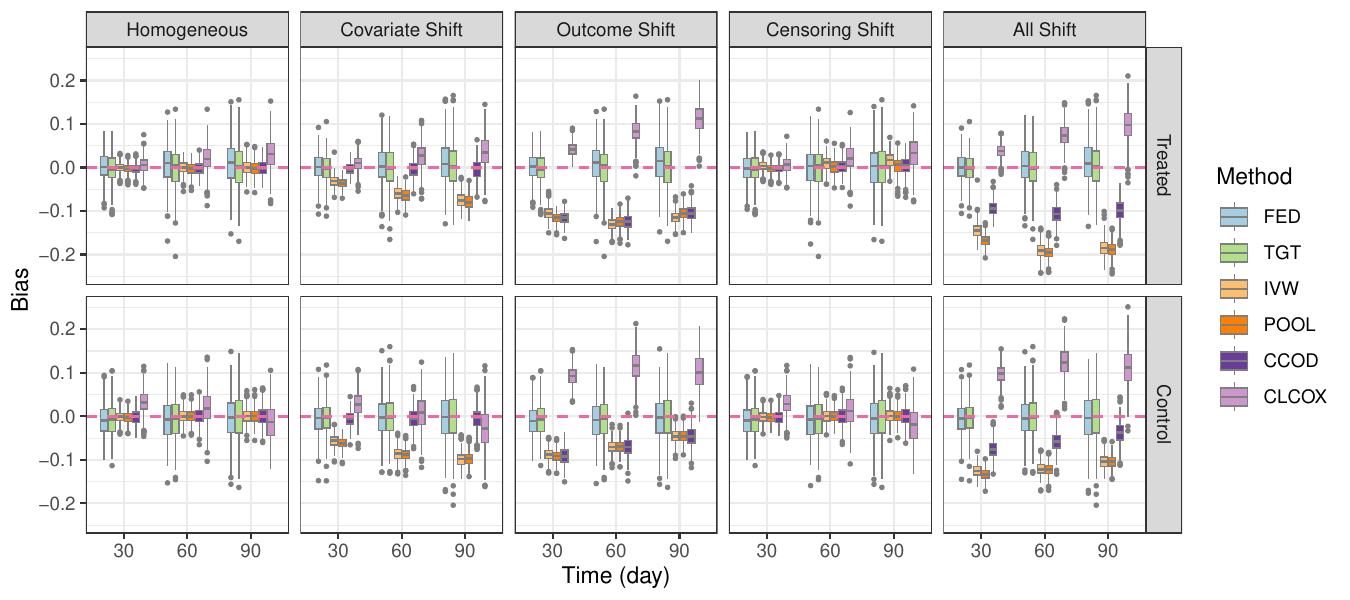}
    \caption{Estimation bias (boxplots) under $n_k=600$ ($k=1,2,3,4$), evaluated at days 30, 60 and 90 in simulation. The true survival probabilities at days 30, 60, and 90 are $(0.86, 0.71, 0.58)$ for the treated group and $(0.76, 0.60, 0.48)$ for the control group. }
    \label{fig:res-med-bias}
\end{figure}

\begin{figure}[H]
    \centering
    \includegraphics[width=\textwidth]{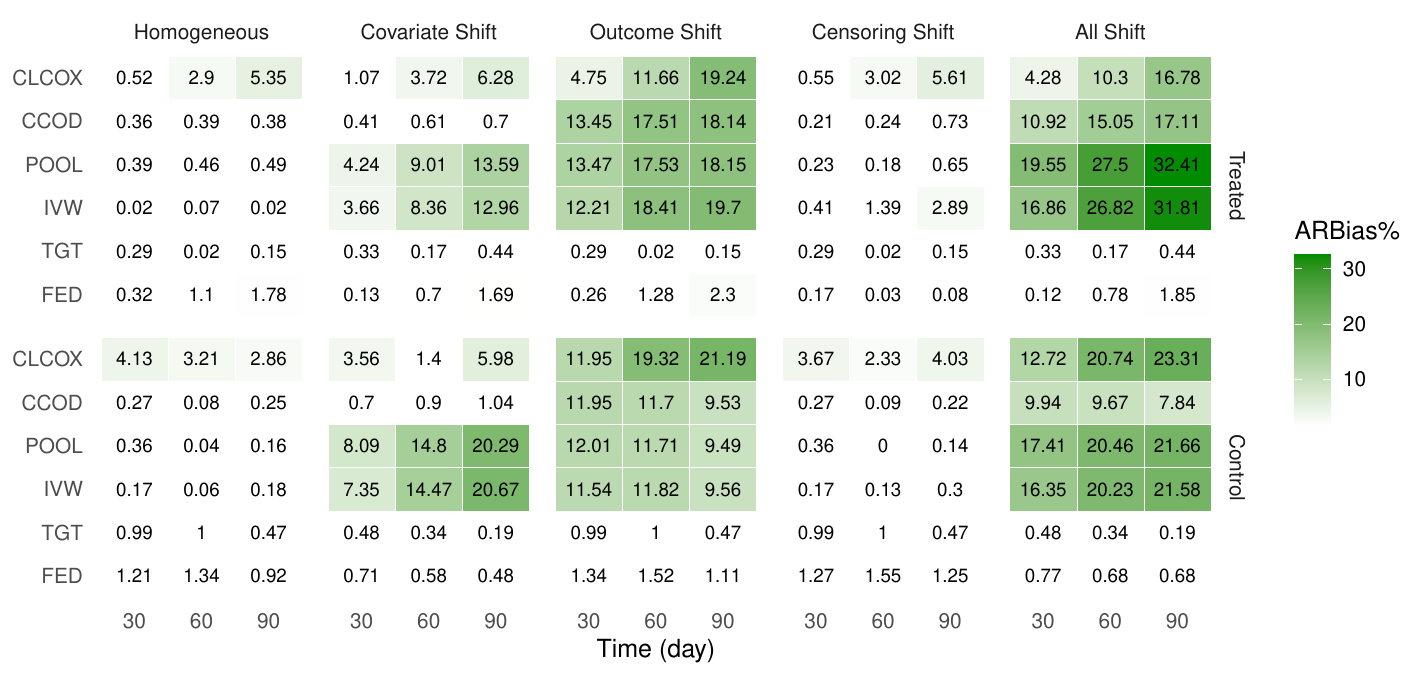}
    \includegraphics[width=\textwidth]{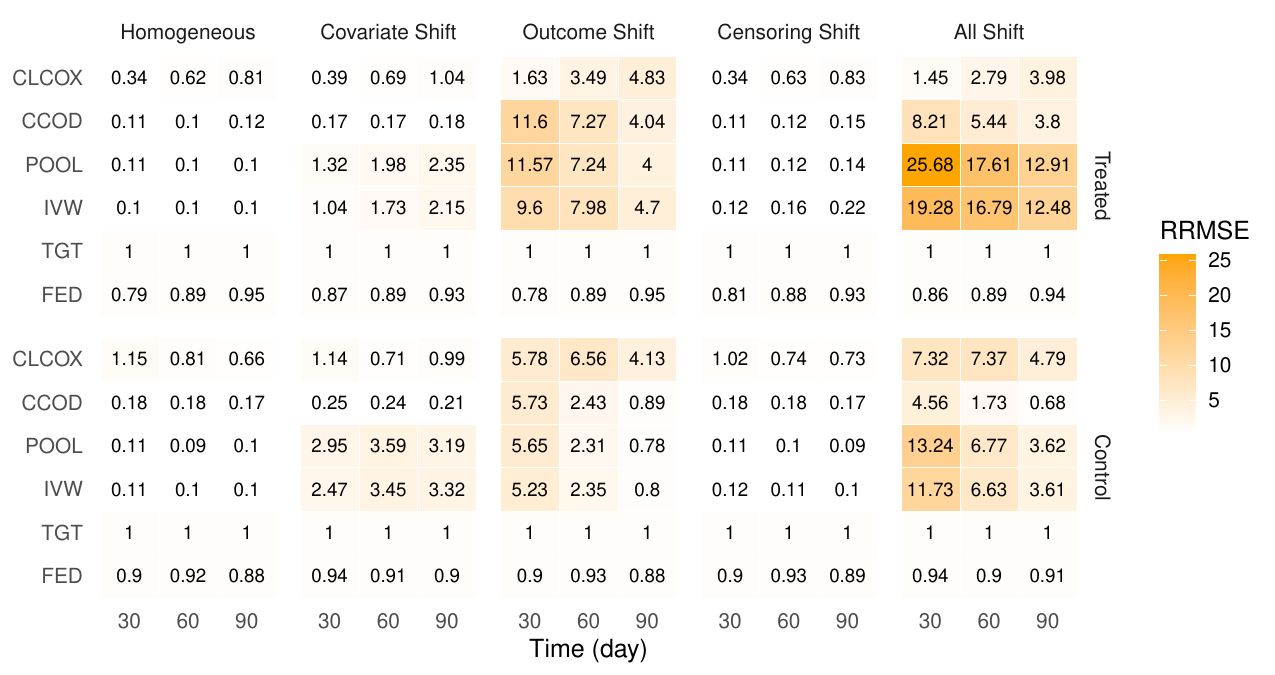}
    \caption{ARBias\% and RRMSE under $n_k=600$ ($k=1,2,3,4$), evaluated at days 30, 60 and 90 in simulation. The true survival probabilities at days 30, 60, and 90 are $(0.86, 0.71, 0.58)$ for the treated group and $(0.76, 0.60, 0.48)$ for the control group. }
    \label{fig:res-med-rrmse}
\end{figure}

\begin{figure}[H]
    \centering
    \includegraphics[width=\textwidth]{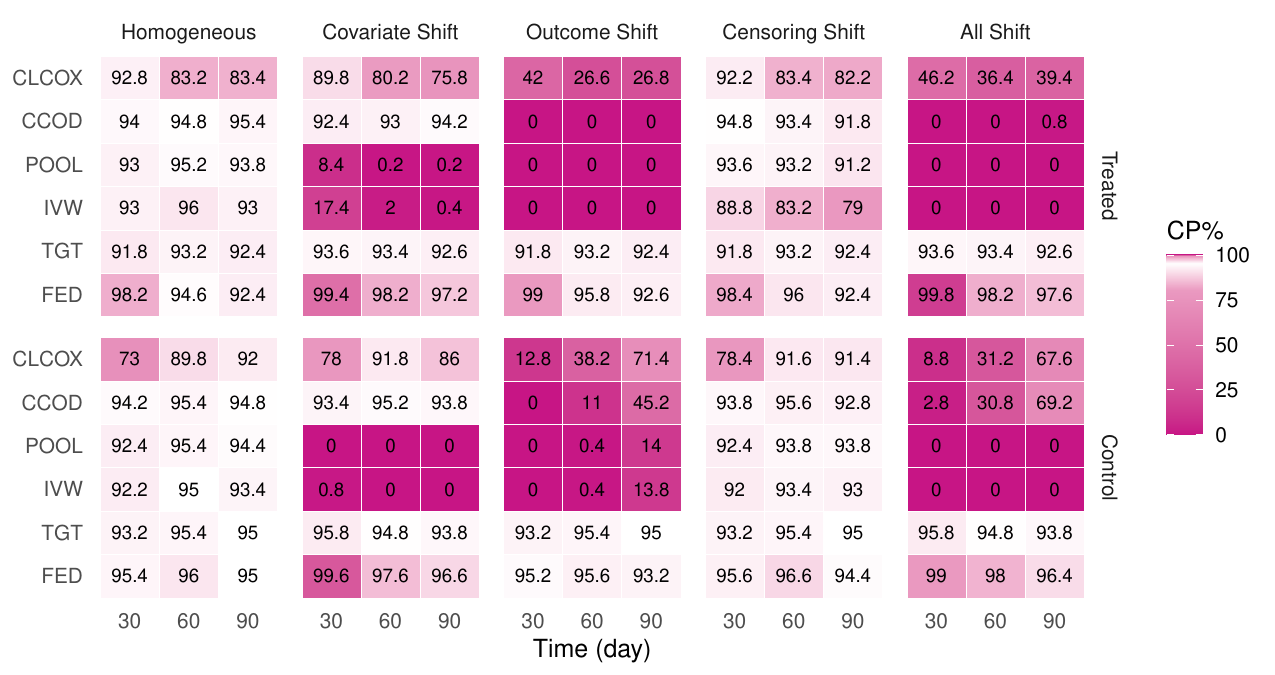}
    \includegraphics[width=\textwidth]{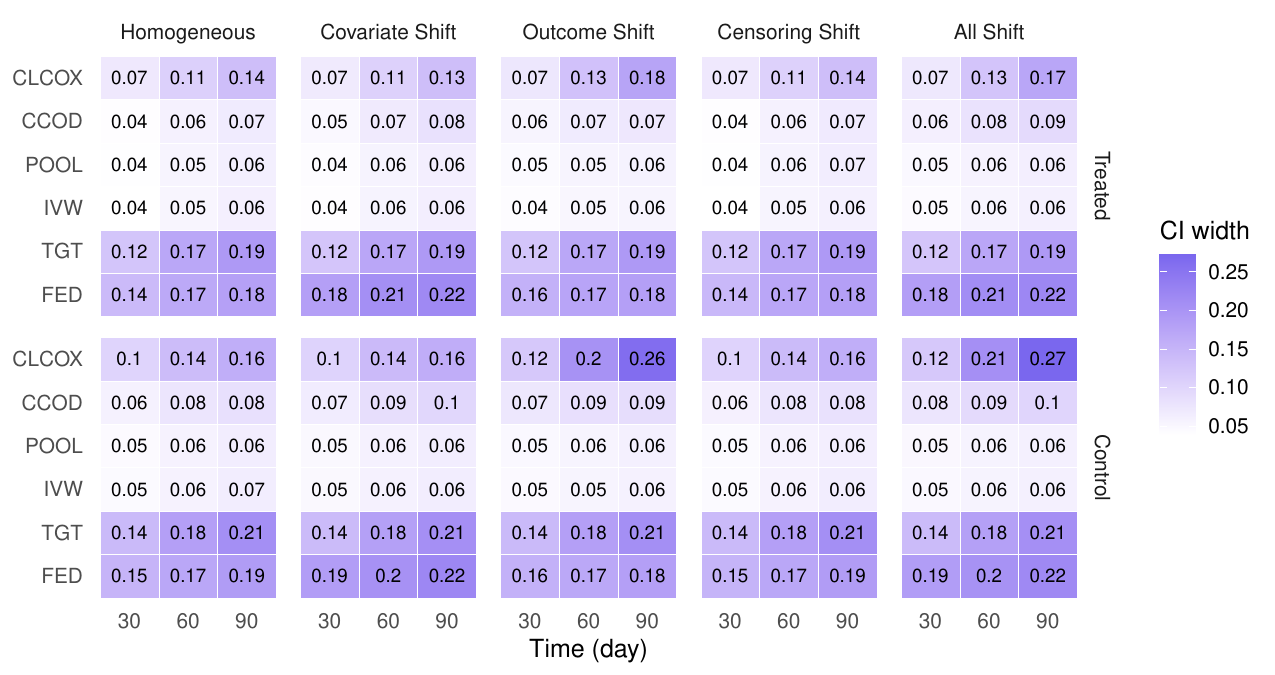}
    \caption{CP\% with 95\% nominal coverage level, and width of 95\% CI under $n_k=600$ ($k=1,2,3,4$), evaluated at days 30, 60 and 90 in simulation. The true survival probabilities at days 30, 60, and 90 are $(0.86, 0.71, 0.58)$ for the treated group and $(0.76, 0.60, 0.48)$ for the control group. }
    \label{fig:res-med-cp}
\end{figure}

\begin{figure}[H]
    \centering
    \includegraphics[width=\textwidth]{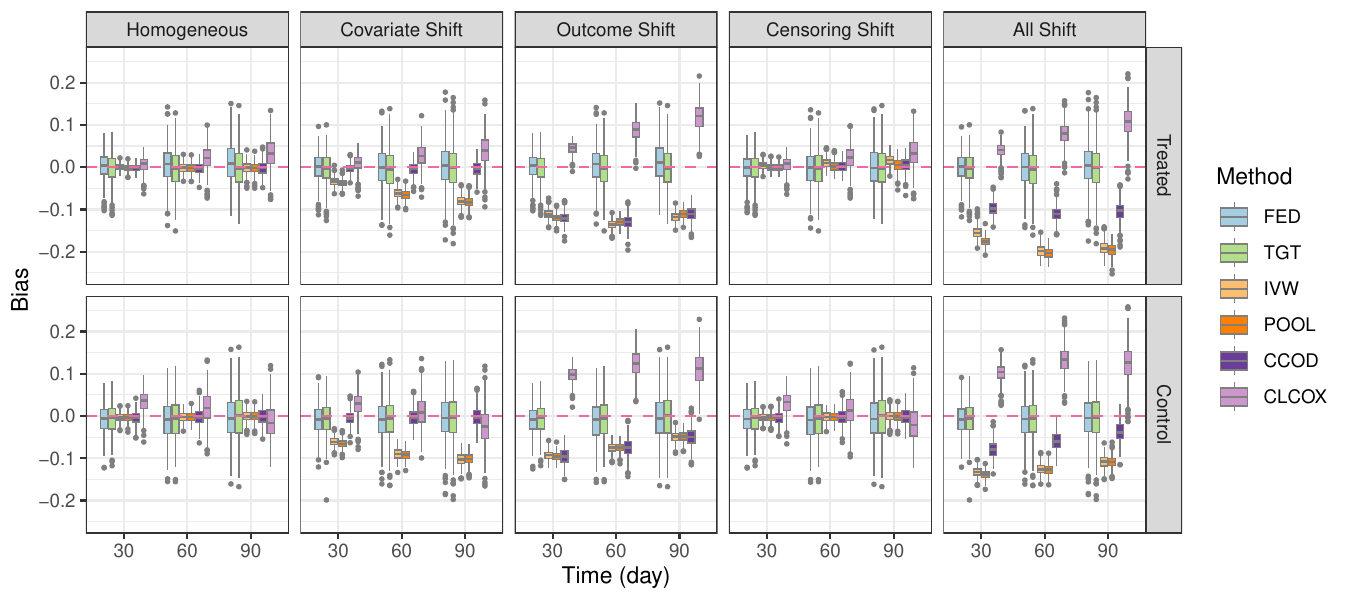}
    \caption{Estimation bias (boxplots) under $n_k=1,000$ ($k=1,2,3,4$), evaluated at days 30, 60 and 90 in simulation. The true survival probabilities at days 30, 60, and 90 are $(0.86, 0.71, 0.58)$ for the treated group and $(0.76, 0.60, 0.48)$ for the control group. }
    \label{fig:res-large-bias}
\end{figure}

\begin{figure}[H]
    \centering
    \includegraphics[width=\textwidth]{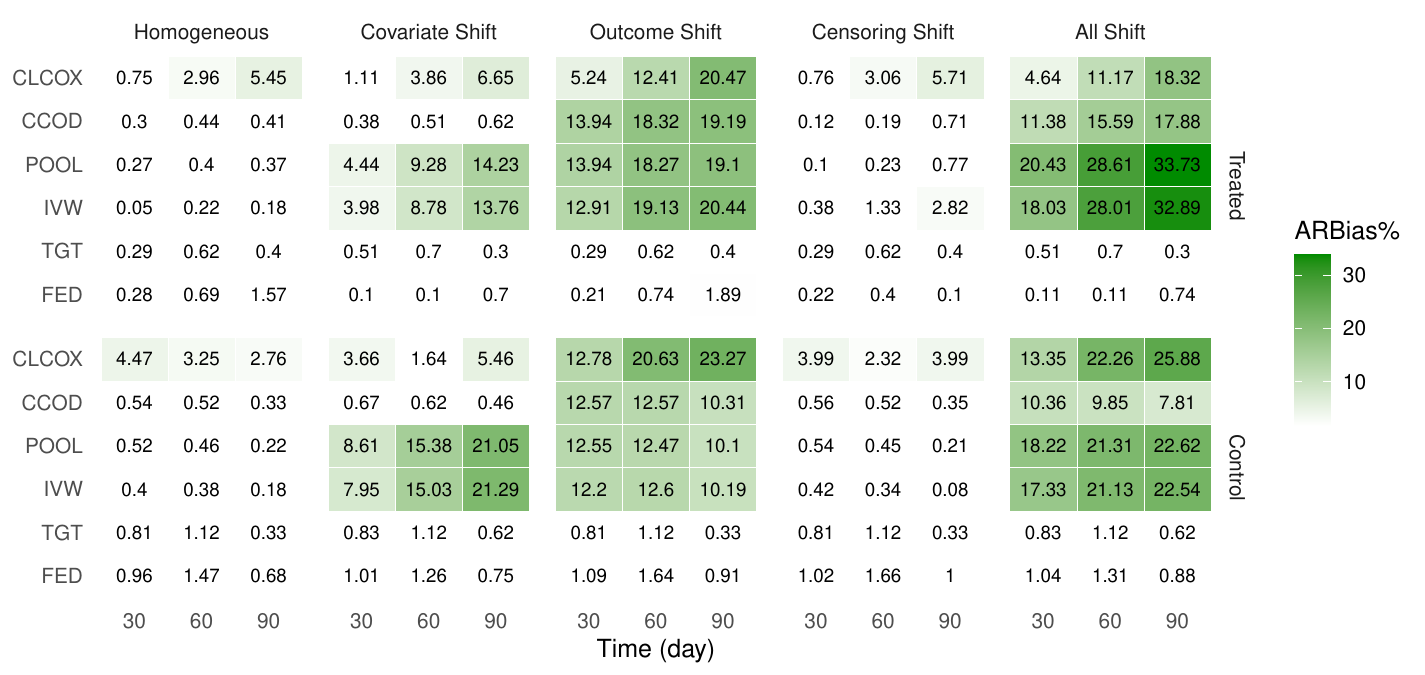}
    \includegraphics[width=\textwidth]{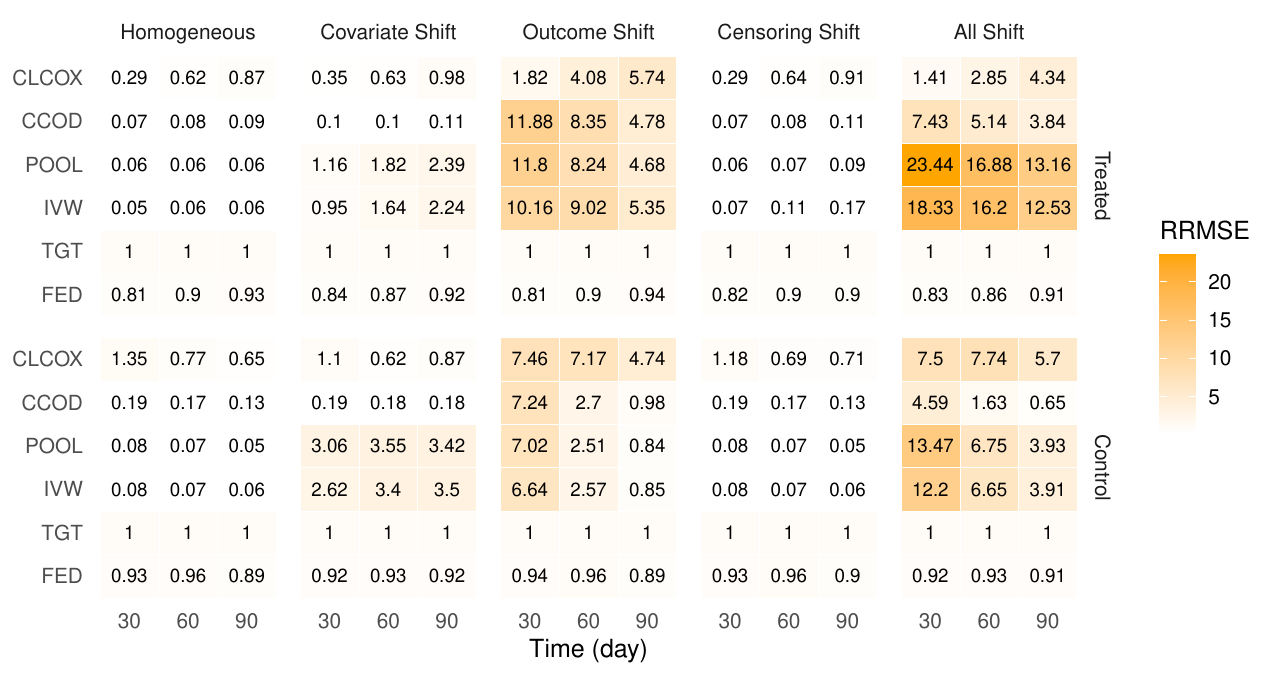}
    \caption{ARBias\% and RRMSE under $n_k=1,000$ ($k=1,2,3,4$), evaluated at days 30, 60 and 90 in simulation. The true survival probabilities at days 30, 60, and 90 are $(0.86, 0.71, 0.58)$ for the treated group and $(0.76, 0.60, 0.48)$ for the control group. }
    \label{fig:res-large-rrmse}
\end{figure}

\begin{figure}[H]
    \centering
    \includegraphics[width=\textwidth]{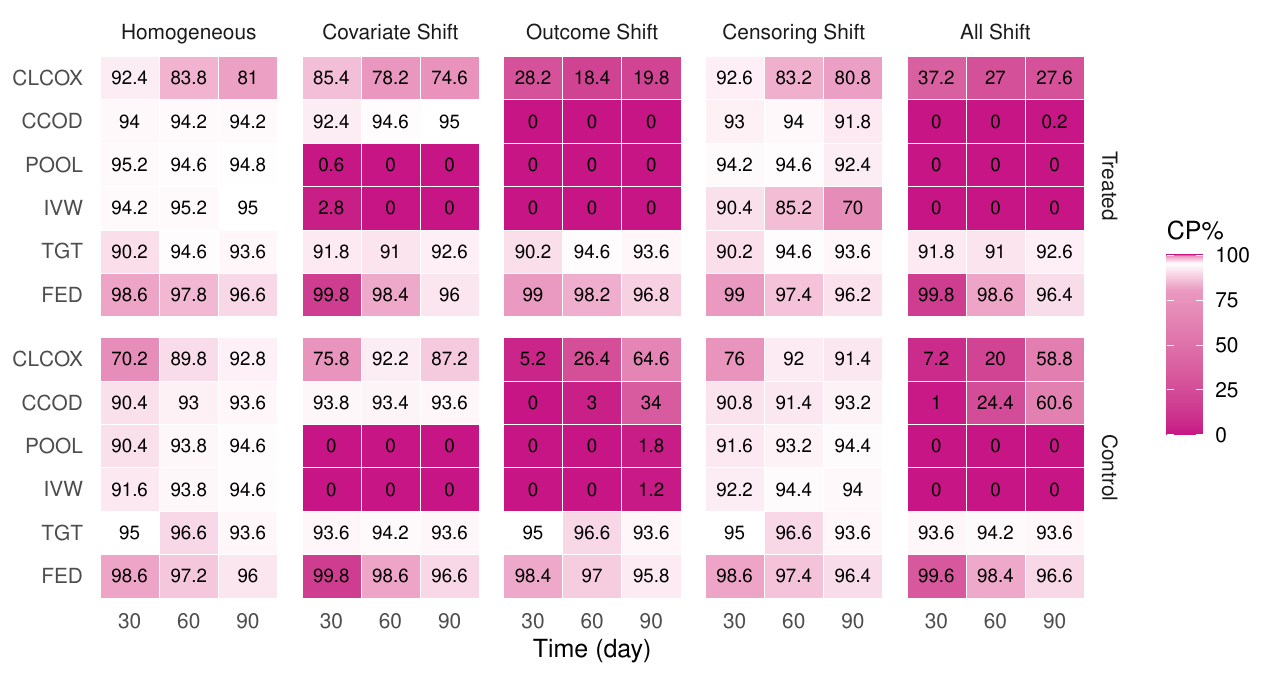}
    \includegraphics[width=\textwidth]{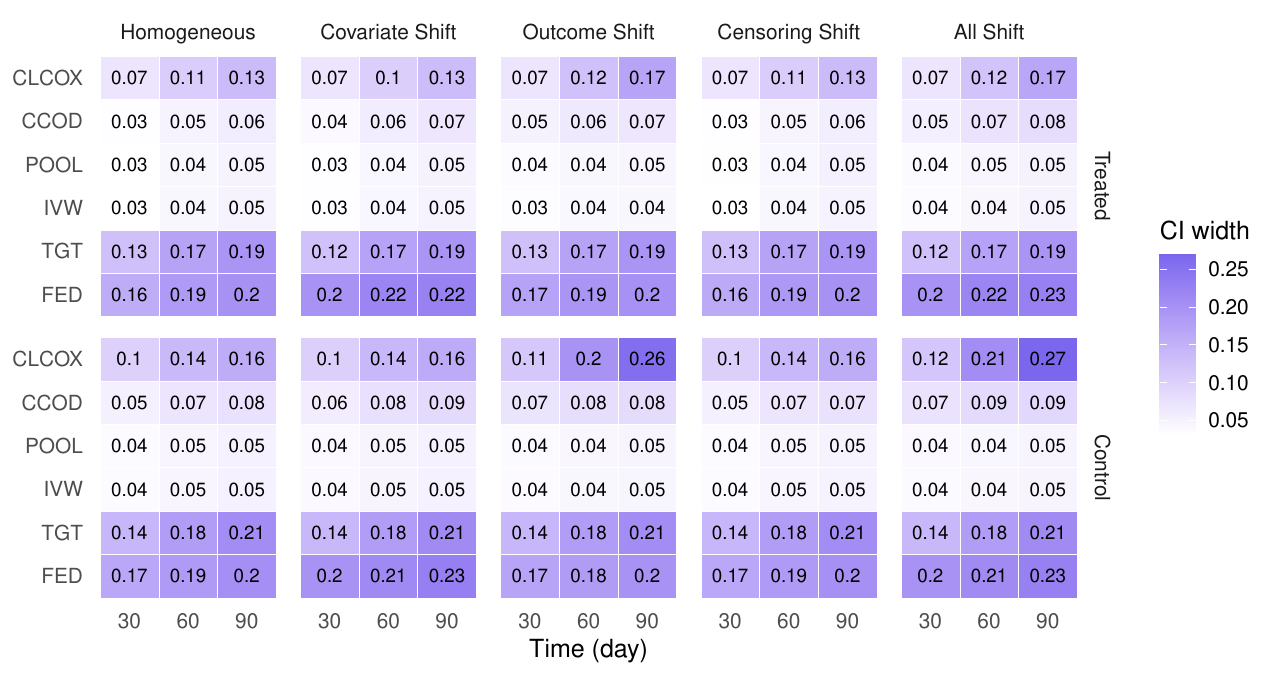}
    \caption{CP\% with 95\% nominal coverage level, and width of 95\% CI under $n_k=1,000$ ($k=1,2,3,4$), evaluated at days 30, 60 and 90 in simulation. The true survival probabilities at days 30, 60, and 90 are $(0.86, 0.71, 0.58)$ for the treated group and $(0.76, 0.60, 0.48)$ for the control group. }
    \label{fig:res-large-cp}
\end{figure}

\subsection{Additional simulations under poor overlap}\label{subapp:limOverlap}

In this section, we consider a setting where the target site has poorer propensity score (treatment) overlap between treatment groups, whereas the source sites have better overlaps. 

To generate data, we modify the target-site propensity score model $\pi^0(\mb X)$ in Section \ref{subapp:DGP} to such that
$$
\text{logit}(\pi^0(\mb X)) = -1.05 + \log\left(0.3 + \exp(-120 + X_1)+\exp(-6+X_2)+\exp(-6 + X_3/4)\right), 
$$
and generate $A \sim \text{Bernoulli}(\pi^0(\mb X))$ accordingly for target-site samples only. This increases the dependence of $A$ on $\mb X$ and induces reduced treatment propensity score overlap. Other sites' DGPs remain unchanged. Full results of this setting are reported below.  

\begin{figure}[H]
    \centering
    \includegraphics[width=\textwidth]{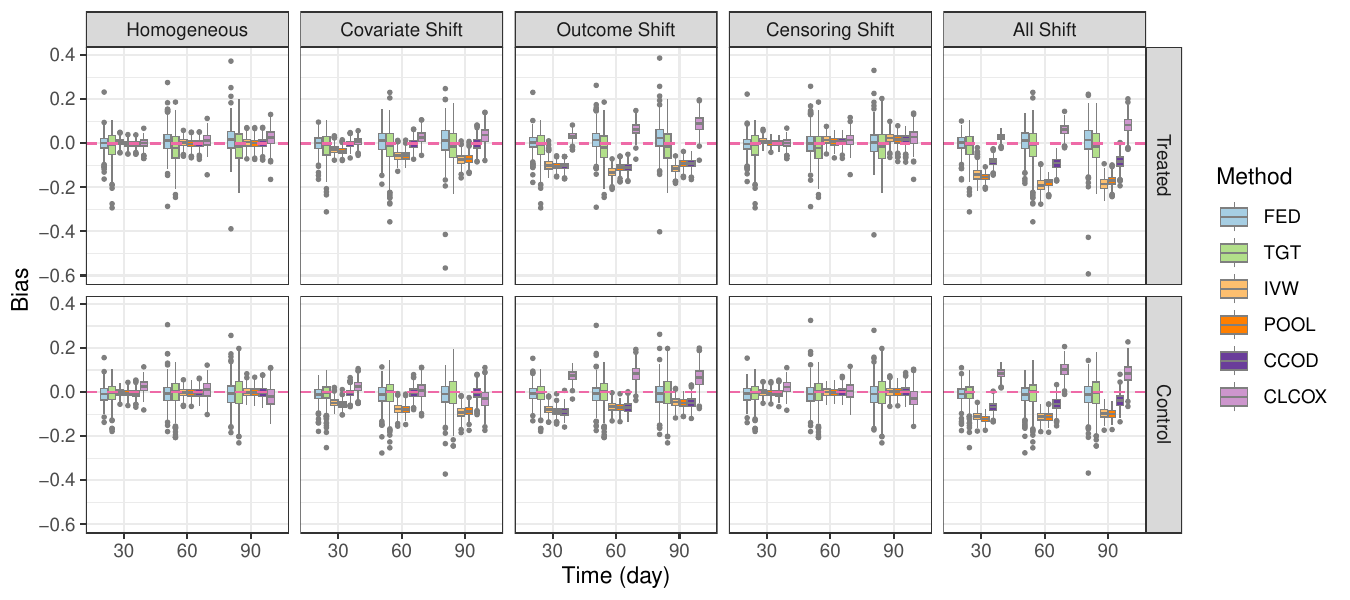}
    \caption{Estimation bias (boxplots) under $n_k=300$ ($k=1,2,3,4$), evaluated at days 30, 60 and 90 in simulation. The true survival probabilities at days 30, 60, and 90 are $(0.86, 0.71, 0.58)$ for the treated group and $(0.76, 0.60, 0.48)$ for the control group. }
    \label{fig:res-limO-bias}
\end{figure}

\begin{figure}[H]
    \centering
    \includegraphics[width=\textwidth]{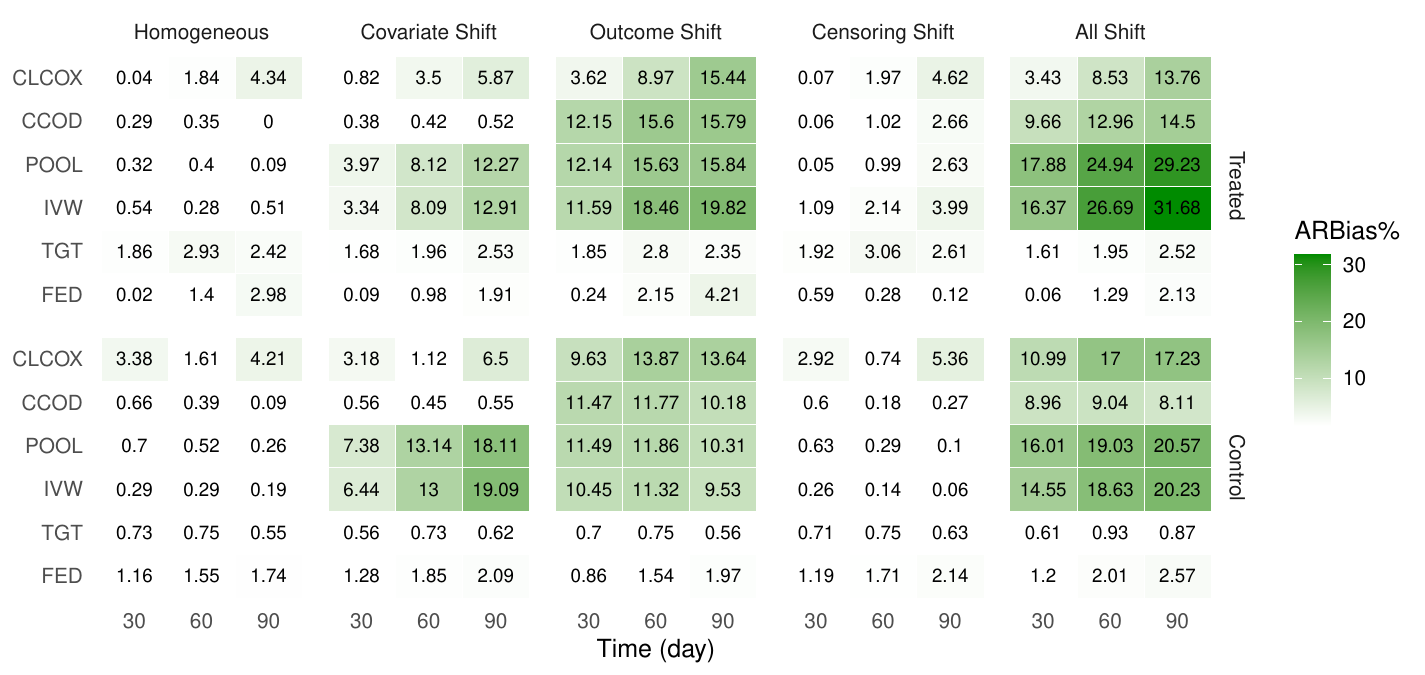}
    \includegraphics[width=\textwidth]{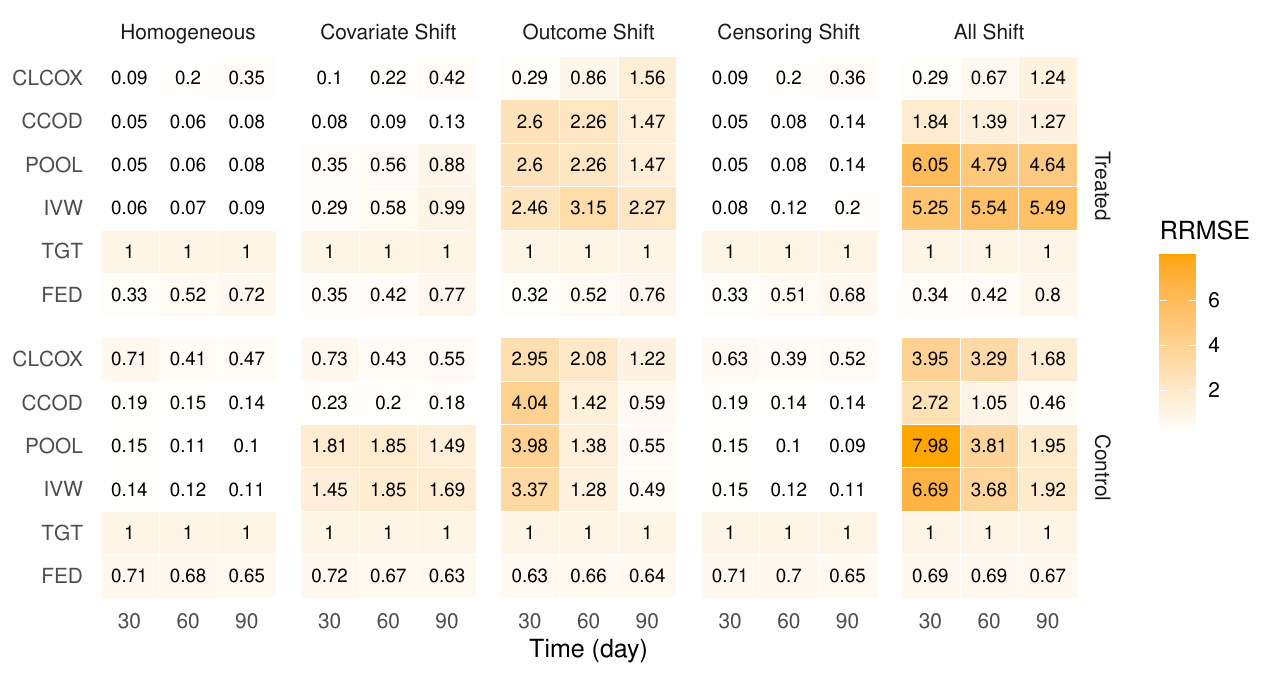}
    \caption{ARBias\% and RRMSE under $n_k=300$ ($k=1,2,3,4$), evaluated at days 30, 60 and 90 in simulation. The true survival probabilities at days 30, 60, and 90 are $(0.86, 0.71, 0.58)$ for the treated group and $(0.76, 0.60, 0.48)$ for the control group. }
    \label{fig:res-limO-rrmse}
\end{figure}

\begin{figure}[H]
    \centering
    \includegraphics[width=\textwidth]{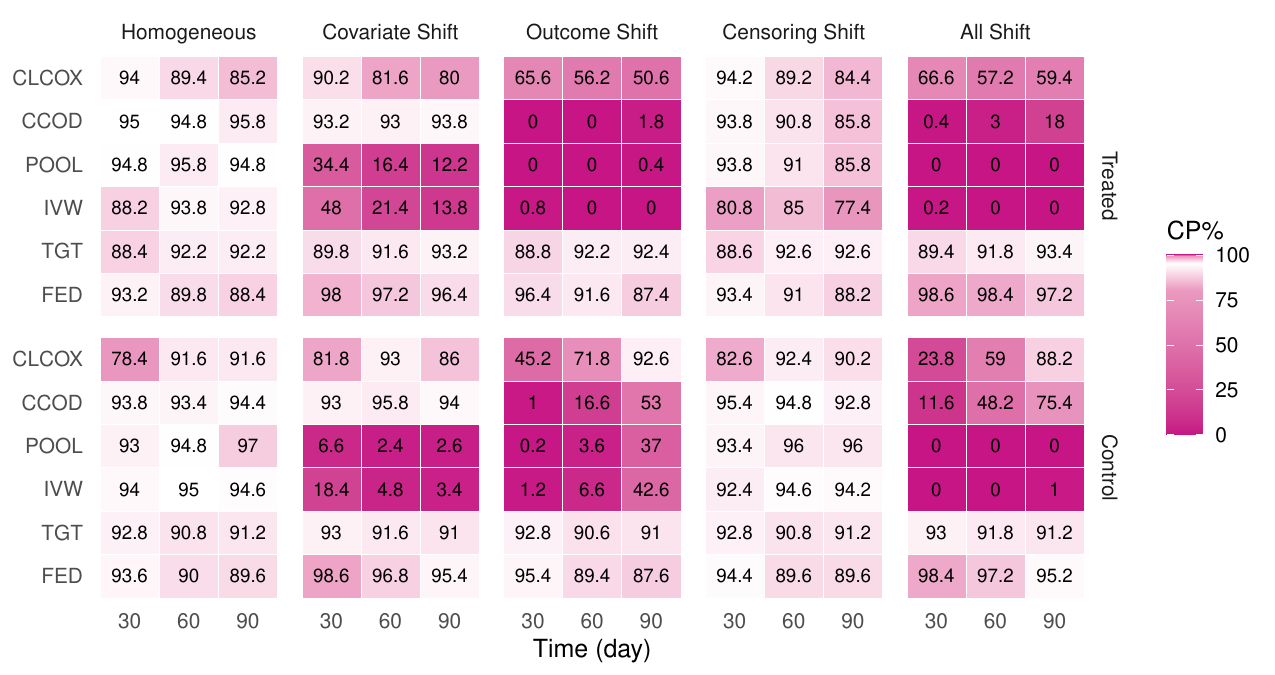}
    \includegraphics[width=\textwidth]{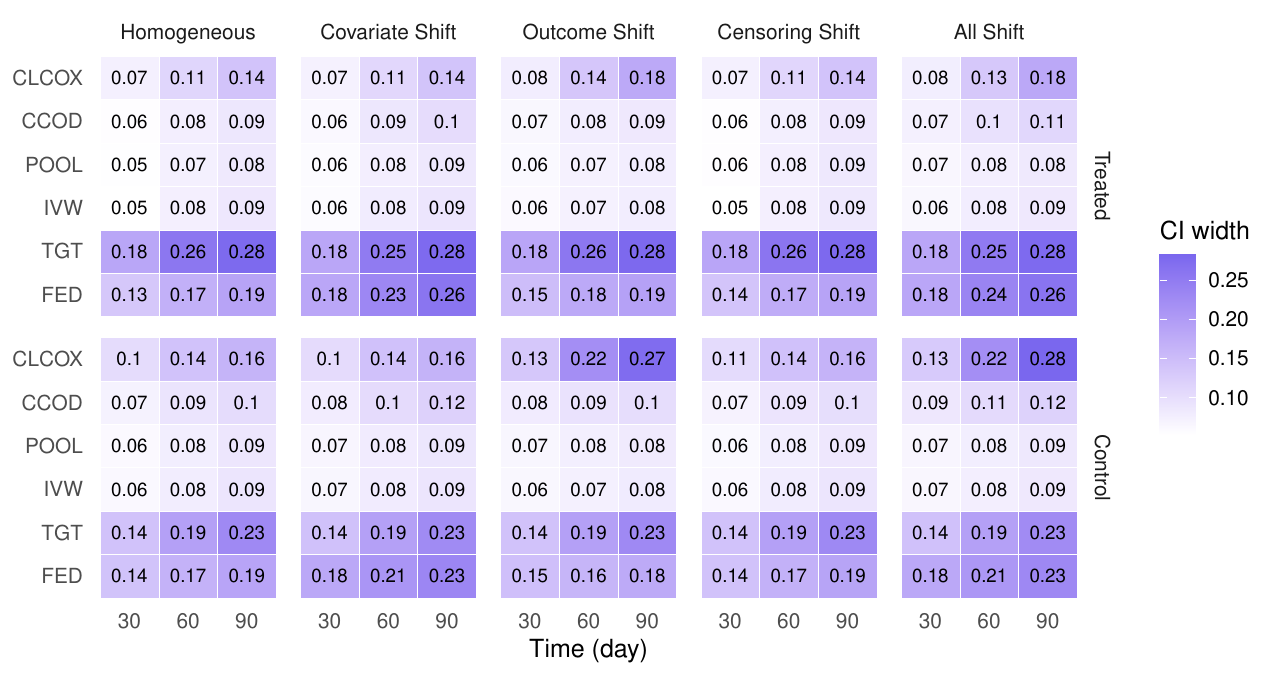}
    \caption{CP\% with 95\% nominal coverage level, and width of 95\% CI under $n_k=300$ ($k=1,2,3,4$), evaluated at days 30, 60 and 90 in simulation. The true survival probabilities at days 30, 60, and 90 are $(0.86, 0.71, 0.58)$ for the treated group and $(0.76, 0.60, 0.48)$ for the control group. }
    \label{fig:res-limO-cp}
\end{figure}

\subsection{Simulation results for other causal contrasts}

To confirm the validity of our methods to other causal contrasts, including SD, SR and RMST, we conduct additional simulations using a few selected settings to check the bias, RRMSE and CP. 

\begin{figure}[H]
    \centering
    \includegraphics[width=\linewidth]{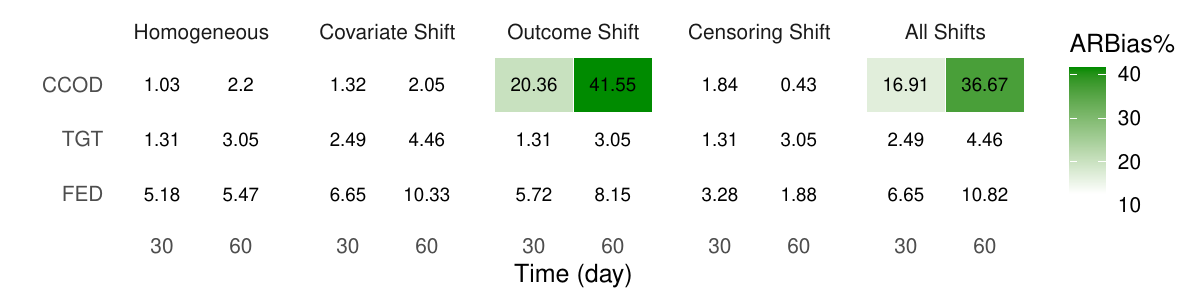}
    \includegraphics[width=\linewidth]{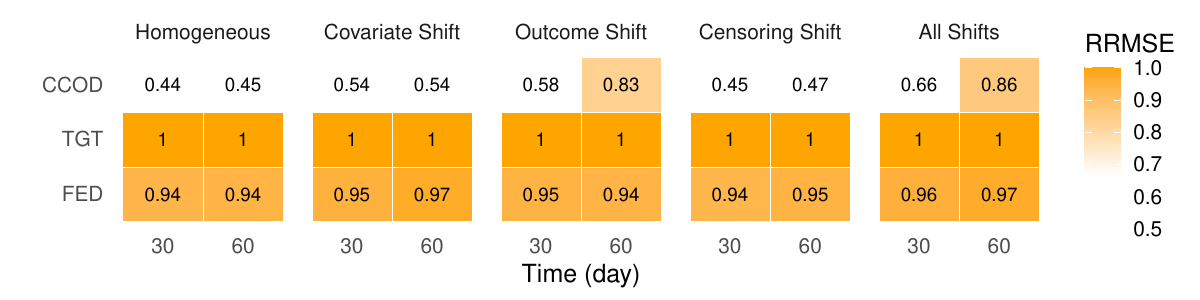}
    \includegraphics[width=\linewidth]{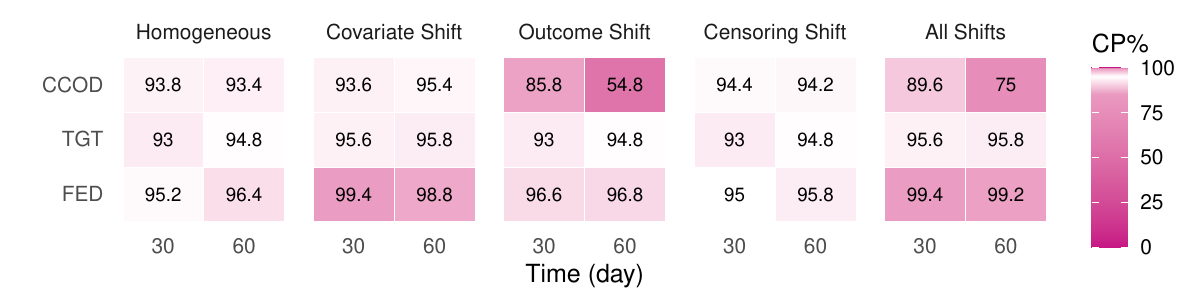}
    \caption{Simulation results for survival difference (SD), with $n_0 = 300$ and $n_k=300$ for $k = 1, 2, 3, 4$. Truths of SD: 0.094 at day 30 and 0.115 at day 60. }
    \label{fig:sim-SD}
\end{figure}

\begin{figure}[H]
    \centering
    \includegraphics[width=\linewidth]{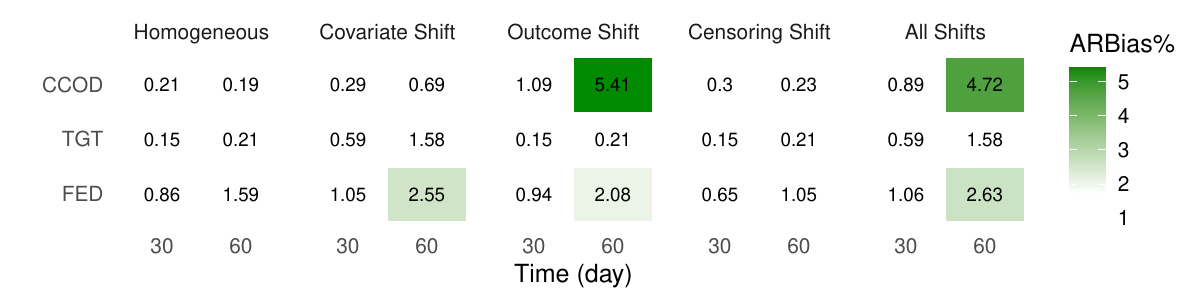}
    \includegraphics[width=\linewidth]{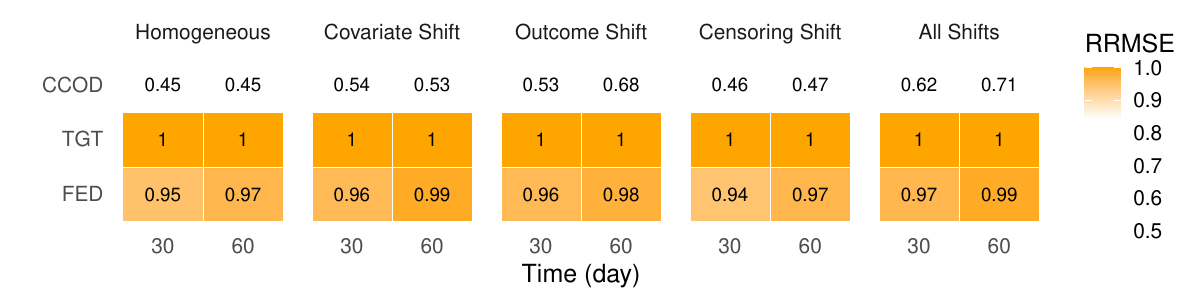}
    \includegraphics[width=\linewidth]{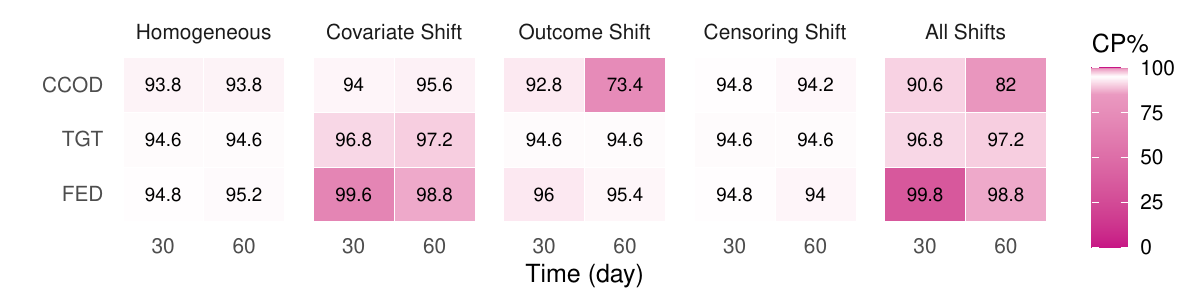}
    \caption{Simulation results for survival ratio (SR), with $n_0 = 300$ and $n_k=300$ for $k = 1, 2, 3, 4$. Truths of SR: 1.123 at day 30 and 1.192 at day 60. }
    \label{fig:sim-SR}
\end{figure}

\begin{figure}[H]
    \centering
    \includegraphics[width=\linewidth]{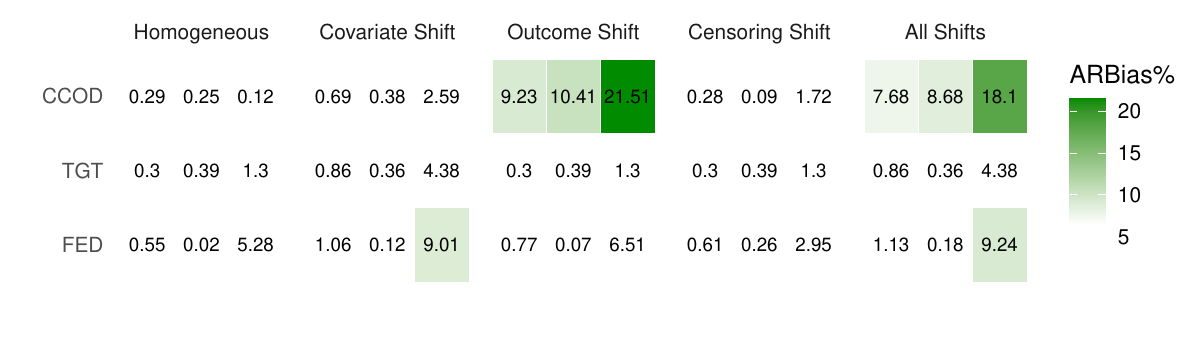}
    \includegraphics[width=\linewidth]{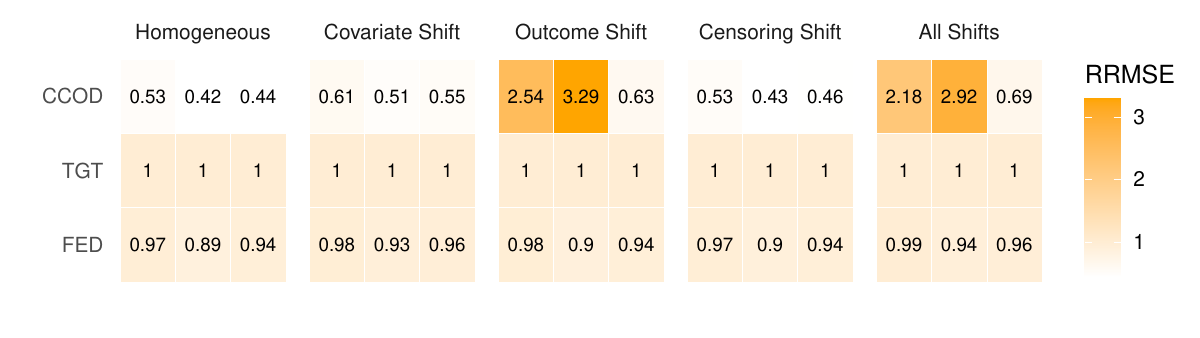}
    \includegraphics[width=\linewidth]{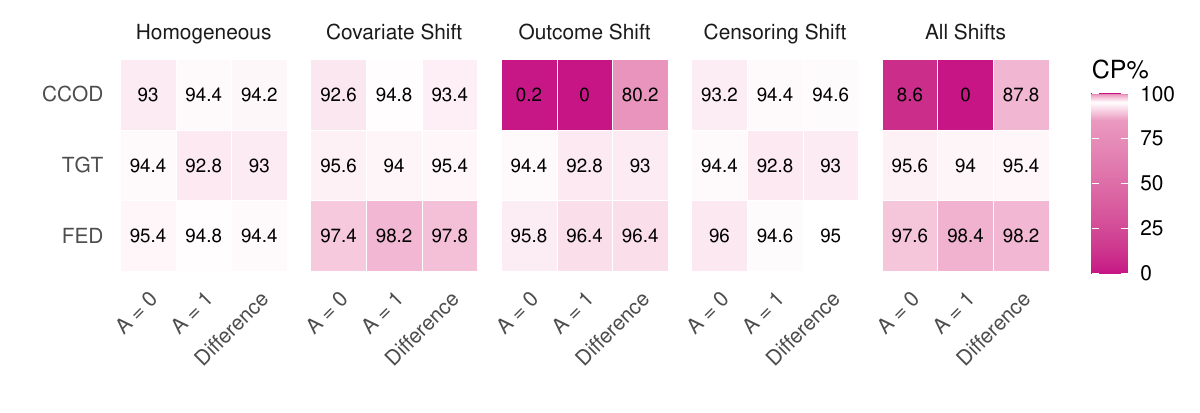}
    \caption{Simulation results for RMST up to day 60 (by treatment group and difference), with $n_0 = 300$ and $n_k=300$ for $k = 1, 2, 3, 4$. Bias is reported as the mean bias across 500 simulation replicates. Truths of RMST: 46.46 for control; 51.38 for treated; 4.92 for difference. }
    \label{fig:sim-RMST}
\end{figure}

\end{document}